\documentclass[twoside]{article}

\usepackage[preprint]{aistats2027}
\usepackage{amsmath,amssymb,amsthm,mathrsfs,dsfont}
\usepackage{graphicx}
\usepackage{xcolor}
\usepackage{booktabs}
\usepackage{array}
\usepackage[authoryear,round]{natbib}
\usepackage{tabularx}
\usepackage{enumitem}
\usepackage{algorithm}
\usepackage{algpseudocode}
\usepackage{placeins}
\usepackage{tikz}
\usepackage{kotex}
\usetikzlibrary{arrows.meta,positioning,fit,calc}

\usepackage[colorlinks=true,linkcolor=blue,citecolor=blue,urlcolor=blue]{hyperref}
\newtheorem{definition}{Definition}[section]
\newtheorem{example}{Example}[section]
\newtheorem{theorem}{Theorem}[section]
\newtheorem{lemma}{Lemma}[section]
\newtheorem{corollary}{Corollary}[section]
\newtheorem{proposition}{Proposition}[section]

\newtheoremstyle{boldremark}
  {3pt}
  {3pt}
  {\normalfont}
  {}
  {\bfseries}
  {.}
  {.5em}
  {}

\theoremstyle{boldremark}
\newtheorem{remark}{Remark}[section]

\newcommand{\cX}{\mathcal X}
\newcommand{\cY}{\mathcal Y}
\newcommand{\cM}{\mathcal M}
\newcommand{\bbE}{\mathbb E}
\newcommand{\bbP}{\mathbb P}
\newcommand{\op}{\mathrm{op}}
\newcommand{\HS}{\mathrm{HS}}
\newcommand{\supp}{\operatorname{supp}}
\newcommand{\ind}{\mathds{1}}
\newcommand{\aref}[2][]{\hyperref[par:#2]{\textnormal{(#2#1)}}}
\newcommand{\astarref}[1]{\hyperref[par:#1star]{\textnormal{(#1$^\star$)}}}
\newcommand{\alphaH}{\alpha_{\mathrm H}}

\begin{document}

%

%

\twocolumn[
\aistatstitle{Semi-Supervised Kernel Ridge Fr\'echet Regression}
\aistatsauthor{Jaehee Seo \And Kwangho Kim \And Jisu Kim{\renewcommand{\thefootnote}{\textcolor{black}{$\star$}}\footnotemark}}
\aistatsaddress{Dept. of Statistics\\ Seoul National University \And Dept. of Statistics\\ Korea University \And Dept. of Statistics\\ Seoul National University}
]

{\renewcommand{\thefootnote}{\textcolor{black}{$\star$}}\footnotetext{Corresponding author; jkim82133@snu.ac.kr}}

\begin{abstract}
Fr\'echet regression provides a framework for predicting random
object-valued responses by minimizing the conditional expected squared
distances.
We propose \emph{semi-supervised kernel ridge Fr\'echet regression}
(SS-KRFR) to exploit unlabeled covariates when labeled samples are scarce.
The method learns kernel spectral features from both labeled and
unlabeled covariates, then utilizes labeled responses to estimate the
conditional Fr\'echet function by kernel ridge regression.
We derive prediction error bounds that distinguish the roles of labeled
and unlabeled samples, and establish polynomial convergence rates under
suitable conditions.
Under additional kernel assumptions, a sharper analysis relaxes the
sample size requirements for learning the features.
We also develop hybrids that combine features, predictions, or both from models fitted with and without unlabeled covariates, establishing uniform stability over mixing weights in Hadamard spaces at common tuning and a separate rate guarantee for prediction mixing.
Simulations and a real-data application show semi-supervised gains and
improvements over competing Fr\'echet regression methods in several settings.
\end{abstract}

\section{Introduction}
\label{sec:introduction}

Statistical regression increasingly involves random object-valued responses,
such as probability distributions, symmetric positive-definite (SPD) matrices,
graphs, and trees. Fr\'echet regression \citep{petersen2019frechet} accounts
for their metric structure by formulating prediction through expected
squared-distance minimization in the response space, following the Fr\'echet mean principle \citep{frechet1948}. Subsequent work has
developed a total variation regularized method \citep{lin2021total}, kernel sliced inverse regression \citep{dong2022frechet}, nonlinear global methods \citep{Bhattacharjee2025}, single index models \citep{bhattacharjee2023single, ghosal2023frechet}, local linear methods for spherical, toroidal, and Riemannian predictors
\citep{im2025spherical,im2026riemannian,im2026toroidal}, local methods using
RKHS response embeddings \citep{iida2025local}, tree-based ensembles \citep{capitaine2024frechet, qiu2024random,zhou2026frechet}, and deep-learning approaches \citep{iao2025deep, kim2025dfnn, cui2026deep}. Methods for inference on Fr\'echet regression also have been developed \citep{xu2025wasserstein, song2026inference}. 

In many applications, labeled pairs are scarce but unlabeled covariates
are abundant. Supervised-only estimators must learn both representations
and regression functions from the small labeled sample, neglecting unlabeled samples.
This motivates semi-supervised learning (SSL) \citep{Zhu2009introduction, van2020survey}, including regression methods \citep{kostopoulos2018semi}. We build on kernel-spectral SSL \citep{Ji2012simple,kim2026semisupervised} by learning basis functions from pooled covariates and fitting the prediction model using labeled responses. Related work studies manifold or cluster structure \citep{belkin2004semi, belkin2008towards, niyogi2013manifold, rigollet2007generalization, singh2008unlabeled, sinha2009semi} and the benefits and limitations of unlabeled data for regression \citep{wasserman2007statistical}.
Graph SS-FR \citep{qiu2024semi} uses pooled covariates to construct graph
distances for local smoothing; we instead learn kernel spectral features
and fit squared-distance responses by ridge regression.

\subsection{Contribution}\label{sec:contributions}

\begin{itemize}[leftmargin=*,itemsep=0.25em]

\item \textbf{SSL for random objects.}
SS-KRFR learns spectral features from pooled covariates and fits squared
distances to labeled responses by ridge regression. Minimizing the
estimated Fr\'echet criterion requires no linear response structure.

\item \textbf{Prediction bounds and the role of unlabeled data.}
Our bounds separate truncation, labeled noise, ridge bias, and feature
estimation, allowing contrasts outside the RKHS. We obtain polynomial
rates; additional spectral assumptions relax pooled-sample requirements.

\item \textbf{Hybrid estimators and stability.}
Representation (RH), prediction (PH), and two-level (TH) hybrids combine
labeled-only and pooled fits, with validation-selected weights in practice.
At common tuning, we prove uniform Hadamard-space stability over mixing
weights for all three, and a polynomial rate for possibly
validation-selected PH from its endpoint bounds. Neither result
guarantees improvement over the better endpoint.

\item \textbf{Experiments.}
We compare methods on SPD, distributional, and Billera--Holmes--Vogtmann
(BHV) tree responses and Healthy Brain Network EEG (HBN-EEG) tasks.

\end{itemize}

\section{Estimation and Theoretical Guarantees}
\label{sec:main-results}

\subsection{Sampling Scheme and Regression Model}
\label{sec:reg-model}

Let $(\cX,d_{\cX})$ and $(\cY,d_{\cY})$ be Polish spaces. Let $(X,Y)\sim\bbP$, with marginal $\bbP_X$ and compact support
$\mathcal S:=\supp(\bbP_X)$. The conditional moment, target,
identification, and curvature conditions below hold on a fixed Borel
set $\cX_0\subseteq\mathcal S$ with $\bbP_X(\cX_0)=1$.
Neither set needs to be known. Draw $X_1,\ldots,X_N\stackrel{\mathrm{i.i.d.}}{\sim}\bbP_X$ and select
$L\subset[N]$, $|L|=n$, uniformly and independently of these covariates.
Conditional on the covariates and $L$, observe independent
$Y_i\sim\bbP_{Y\mid X=X_i}$ for $i\in L$; the remaining $N-n$
observations are unlabeled. Let $X^\circ\sim\bbP_X$ be an independent query.

Using a fixed Borel version of
$\bbP_{Y\mid X=x}$, define the \emph{conditional Fr\'echet function}
\[
    M_\oplus(y,x)
    :=
    \bbE[d_{\cY}^2(y,Y)\mid X=x],
\]
whenever the expectation is finite. The \emph{Fr\'echet regression function} is
\[
    m_\oplus(x)
    :=
    \operatorname*{argmin}_{y\in\cY}
    M_\oplus(y,x),
\]
for $x\in\cX_0$. Condition~\aref{M2} fixes a Borel measurable extension
to $\cX$, with no minimizing requirement outside $\cX_0$.

In a Hilbert space this target on $\cX_0$ is the
conditional mean; in a Hadamard space it is the conditional barycenter.
Both cases have automatic quadratic curvature; see
Proposition~\ref{prop:automatic-quadratic-curvature}.

\subsection{Normalized Contrasts}
\label{subsec:contrasts}

For $z,y,v\in\cY$ with $y\ne z$, define the normalized squared-distance contrast as
\begin{equation}
\label{eq:directional-psi}
\begin{aligned}
    \psi_{z,y}(v)
    &:={}
    \frac{d_{\cY}^2(y,v)-d_{\cY}^2(z,v)}{d_{\cY}(y,z)}
    -d_{\cY}(y,z).
\end{aligned}
\end{equation}
For an observed response $v$, this contrast compares the squared prediction
errors at $y$ and $z$, adjusting for the distance between the two predictions.
In particular, in a Hilbert space, \eqref{eq:directional-psi} becomes
\[
    \psi_{z,y}(v)
    =-2\,\bigg\langle\frac{y-z}{\|y-z\|},v-z\bigg\rangle.
\]
Lemma~\ref{lem:directional-algebra} expresses the normalized population loss
difference as a distance term plus a scalar conditional expectation,
reducing Fr\'echet risk comparisons to scalar regression problems. Thus, we consider squared distances to each $y$ as scalar
responses and minimize the resulting estimated conditional loss over $\cY$.

Define the index set for the normalized contrasts by
\[
    \mathfrak T
    :=\{(z,y):z\in m_\oplus(\cX_0),\ y\in\cY,\ y\ne z\}.
\]
For $t=(z,y)\in\mathfrak T$, write $\psi_t:=\psi_{z,y}$. For
$x\in\cX_0$, define the conditional noise and centered conditional contrast,
respectively, by
\[
\begin{aligned}
    \zeta_t(x,v)
    &:=\psi_{z,y}(v)-\bbE[\psi_{z,y}(Y)\mid X=x],\\
    \dot h_t(x)
    &:=\bbE[\psi_{z,y}(Y)\mid X=x]-\bbE[\psi_{z,y}(Y)].
\end{aligned}
\]
Thus $\zeta_t$ is the conditional fluctuation and $\dot h_t$ the centered
regression signal. Centering separates the intercept from zero-mean
spectral functions (Appendix~\ref{app:RKHS}).

\subsection{Spectral Kernel Structure}\label{subsec:kernel}
Let $\kappa:\mathcal X\times \mathcal X\rightarrow \mathbb R$ be the kernel on $\mathcal X$ satisfying \aref{K1}. Define the centered kernel $\kappa_0:\mathcal X\times\mathcal X\rightarrow \mathbb R$ by
\[
\begin{aligned}
\kappa_0(x,x')
    =&\kappa(x,x')
       -\bbE[\kappa(x,X)]\\[1mm]
       &-\bbE[\kappa(X,x')]
       +\bbE[\kappa(X,X')],
\end{aligned}
\]
where $X'$ is an independent copy of $X$.
Let $\mathcal H_{{\kappa_0}}$ be the
reproducing kernel Hilbert space (RKHS) of
$\kappa_0|_{\mathcal S\times\mathcal S}$, and let
$K_x^0=\kappa_0(x,\cdot)|_{\mathcal S}$ for $x\in\mathcal S$.
For $f\in L^2(\bbP_X)$, define the integral operator by
\[
    (\mathcal T_{\kappa_0}f)(x)
    :=\bbE[\kappa_0(X,x)f(X)],\qquad x\in\mathcal S.
\]
Its RKHS realization is $\bbE[K_X^0\otimes K_X^0]$, where
$(K_x^0\otimes K_x^0)h=h(x)K_x^0$. We use the same symbol for both
realizations, which agree on RKHS inputs. Covariance perturbations and
operator norms act on $\mathcal H_{\kappa_0}$; fractional source powers
act on $L^2(\bbP_X)$.

Let $(\lambda_j,\phi_j)_{j\ge1}$ be the fixed positive eigenpairs of the
$L^2(\bbP_X)$ realization, normalized in that space, with eigenvalues in
nonincreasing order and repeated according to multiplicity; thus
$(\lambda_j,\sqrt{\lambda_j}\phi_j)$ are the corresponding RKHS-normalized
eigenpairs (Appendix~\ref{app:RKHS}).

For the asymptotic results, the kernel and tuning sequences are specified
before observing the data. As $n,N\to\infty$, take $s=s_{n,N}\to\infty$
with $1\le s<N$.
We use $\Pi_{\le s}$ for the orthogonal spectral projector onto the
leading rank-$s$ subspace in either $\mathcal H_{\kappa_0}$ or
$L^2(\bbP_X)$, with
$\Pi_{\le s}f:=\sum_{j\le s}\langle f,\phi_j\rangle_2\phi_j$.
The two realizations agree on RKHS functions; the input determines the
ambient space. Write $\Pi_{>s}:=\mathrm{Id}-\Pi_{\le s}$, where
$\mathrm{Id}$ is the identity on that space.

\subsection{Semi-Supervised Kernel Ridge Fr\'echet Regression}
\label{sec:SS-KRFR}

We adopt a centered version of the kernel spectral feature construction in
\citet{Ji2012simple,kim2026semisupervised}.
The pooled representation is learned from all $N$ covariates,
each receiving equal weight in the empirical eigensystem. Let
\[
\begin{aligned}
    &\mathbf K:=\bigl[\kappa(X_r,X_{r'})\bigr]_{r,r'=1}^N,\\
    &\mathbf H_N:=\mathbf I_N-N^{-1}\mathbf1_N\mathbf1_N^\top,\\
    &\mathbf K_0:=\mathbf H_N\mathbf K \mathbf H_N=\bigl[\widehat\kappa_0(X_r,X_{r'})\bigr]_{r,r'=1}^N,
\end{aligned}
\]
where $\widehat\kappa_0:\mathcal X\times\mathcal X\rightarrow \mathbb R$ is the empirical counterpart of the centered kernel $\kappa_0$, defined by
\[
\begin{aligned}
\widehat\kappa_0(x,x')
=&\kappa(x,x')
-\frac1N\sum_i\kappa(x,X_i)\\
&-\frac1N\sum_i\kappa(X_i,x')
+\frac1{N^2}\sum_{i,j}\kappa(X_i,X_j).
\end{aligned}
\]

We use spectral truncation.
Let $\widehat{\mathbf V}_s=(\widehat{\mathbf v}_1,\ldots,
\widehat{\mathbf v}_s)$ collect the leading $s$ orthonormal eigenvectors of
$N^{-1}\mathbf K_0$, so that
\begin{equation}
\label{eq:full-sample-gram-eigensystem}
\begin{aligned}
    &N^{-1}\mathbf K_0\widehat{\mathbf V}_s
      =\widehat{\mathbf V}_s\widehat{\boldsymbol\Lambda}_s,
      \quad
    &\widehat{\mathbf V}_s^\top\widehat{\mathbf V}_s=\mathbf I_s,\\
    &\widehat{\boldsymbol\Lambda}_s
      :=\operatorname{diag}(\widehat\lambda_1,\ldots,
        \widehat\lambda_s),\quad& \widehat \lambda_1\ge\cdots\ge\widehat\lambda_s.
\end{aligned}
\end{equation}
Let $\widehat\lambda_{s+1}$ denote the $(s+1)$st largest eigenvalue of
$N^{-1}\mathbf K_0$. Define the empirical spectral-gap event as
\[
    \mathcal E_s:=\{\widehat\lambda_s>\widehat\lambda_{s+1}\}.
\]
For a given query point $x\in\cX$, put
\[
\begin{aligned}
    \mathbf k_x
    :=&\,(\kappa(x,X_1),\ldots,\kappa(x,X_N))^\top,\\
    \widehat{\mathbf k}_0(x)
    :=&\,\mathbf H_N\{\mathbf k_x-N^{-1}\mathbf K\mathbf1_N\}\\
    =&\,(\widehat \kappa_0(x,X_1),\ldots,\widehat \kappa_0(x,X_N))^\top.
\end{aligned}
\]
On $\mathcal E_s$, the RKHS-normalized Nystr\"om feature is
\[
    \widehat{\boldsymbol\Phi}_s(x)
    :=\frac1{\sqrt N}\widehat{\boldsymbol\Lambda}_s^{-1/2}
      \widehat{\mathbf V}_s^\top\widehat{\mathbf k}_0(x)
    =\bigl(\widehat\phi_1(x),\ldots,
      \widehat\phi_s(x)\bigr)^\top.
\]
Here $\widehat\phi_j$ is normalized in the empirically centered RKHS,
analogously to $\sqrt{\lambda_j}\phi_j$ in $\mathcal H_{\kappa_0}$.
On $\mathcal E_s$, the full-sample features satisfy
{\small
\begin{equation}
\label{eq:full-sample-feature-matrix}
\begin{aligned}
&\frac1N\sum_{r=1}^N
  \widehat{\boldsymbol\Phi}_s(X_r)\widehat{\boldsymbol\Phi}_s(X_r)^\top
  =\widehat{\boldsymbol\Lambda}_s,\qquad
  \frac1N\sum_{r=1}^N\widehat{\boldsymbol\Phi}_s(X_r)=0.
\end{aligned}
\end{equation}
}
On $\mathcal E_s^c$, Appendix~\ref{app:empirical-spectral-construction}
fixes a measurable fallback only for probability statements.
Proposition~\ref{prop:primitive-spectral-verification} gives
$\bbP(\mathcal E_s^c)\to0$ under Conditions~\aref{K1}--\aref{K3}. In computation,
features are built from the numerical eigensystem, with candidate ranks capped
at the available positive empirical rank and selected by validation.

Write $Z_{y,i}:=d_\cY^2(y,Y_i)$ and $\overline Z_{y,L}:=n^{-1}\sum_{i\in L}Z_{y,i}$. For ridge parameter $\xi>0$, write
\[
\begin{aligned}
 \overline{\boldsymbol\Phi}_L&:=n^{-1}\sum_{i\in L}\widehat{\boldsymbol\Phi}_s(X_i),\\
 \mathbf q_L(x)&:=\widehat{\boldsymbol\Phi}_s(x)-\overline{\boldsymbol\Phi}_L,\\
 \mathbf C_{L,s}&:=n^{-1}\sum_{i\in L}\mathbf q_L(X_i)\mathbf q_L(X_i)^\top,\\
 \widehat{\mathbf b}_{y,L}&:=n^{-1}\sum_{i\in L}
   \mathbf q_L(X_i)(Z_{y,i}-\overline Z_{y,L}),\\
 \widehat{\boldsymbol\theta}_{y,L}&:=
   (\mathbf C_{L,s}+\xi\mathbf I_s)^{-1}\widehat{\mathbf b}_{y,L}.
\end{aligned}
\]
Since pooled centering need not center the labeled subset, we recenter
the learned features over $L$ to fit a separate unpenalized intercept,
unlike the spectral formulations of
\citet{Ji2012simple,kim2026semisupervised}.
This ensures that the Fr\'echet weights sum to one, guaranteeing
constant reproduction and coercivity of the weighted squared-distance
objective function. 

We estimate the Fr\'echet function by
\begin{equation}
\label{eq:ss-krfr-criterion}
 \widehat M_{s,\xi}(y,x)
 :=\overline Z_{y,L}+\mathbf q_L(x)^\top
       \widehat{\boldsymbol\theta}_{y,L}.
\end{equation}
Equivalently, its response-independent weights are
\begin{equation}
\label{eq:calibrated-weights-main}
 a_{s,\xi,i}(x):=\frac1n+
 \frac1n\mathbf q_L(X_i)^\top
 (\mathbf C_{L,s}+\xi\mathbf I_s)^{-1}\mathbf q_L(x).
\end{equation}
These weights sum to one, may be negative, and satisfy
\[
    \widehat M_{s,\xi}(y,x)=\sum_{i\in L}a_{s,\xi,i}(x)Z_{y,i}.
\]
Appendix~\ref{app:full-procedures} gives the coefficient objective and
finite-sample identities. KRFR uses the same fit with the representation
learned only from labeled covariates.

\paragraph{Approximate minimizer selection.}
For a specified tolerance $\varepsilon_{\mathrm{opt}}>0$ tending to zero,
define SS-KRFR as any measurable selector satisfying
\begin{equation}
\label{eq:ss-krfr-estimator}
\begin{split}
    \widehat M_{s,\xi}
      (\widehat m_{s,\xi,\varepsilon_{\mathrm{opt}}}(x),x)
    \le
    \inf_{y\in\cY}\widehat M_{s,\xi}(y,x)
      +\varepsilon_{\mathrm{opt}}.
\end{split}
\end{equation}

One may set $\varepsilon_{\mathrm{opt}}=0$ when an exact measurable
minimizer exists (Corollaries~\ref{cor:proper} and \ref{cor:Hilbert}).
Propositions~\ref{prop:coercivity} and
\ref{prop:measurable-approximate-minimizer} establish coercivity and
measurable selection. The representation is computed once; response
minimization depends on $\cY$ (Algorithm~\ref{alg:ss-krfr-detailed}).
A dense Gram matrix uses $O(N^2)$ memory and kernel evaluations;
$m$ eigensolver matrix--vector products cost $O(N^2m)$, excluding
orthogonalization.

\subsection{Regularity Conditions}
\label{sec:regularity-conditions}

Throughout, all conditional assumptions use the
same full-measure set $\cX_0$ and constants uniform in $x$.
Appendix~\ref{app:notations} specifies conditional norms and measurability
conventions for indexed suprema.

Conditions~\aref{M1}--\aref{M3} concern response moments and the target
Fr\'echet function.

\paragraph{(M1)}\label{par:M1}
\textit{There is $y_*\in\cY$ such that
\[
    \sup_{x\in\cX_0}
    \bbE[d_{\cY}^2(Y,y_*)\mid X=x]<\infty.
\]}

\paragraph{(M2)}\label{par:M2}
\textit{There exists Borel measurable $m_\oplus:\cX\to\cY$ such that,
for every $x\in\cX_0$, the conditional Fr\'echet function
$M_\oplus(\cdot,x)$ has the unique minimizer $m_\oplus(x)$.}

Conditions~\aref{M1}--\aref{M2} ensure uniform square integrability
of the normalized contrasts, so $\zeta_t$ and $\dot h_t$ are
well-defined; see Appendix~\ref{app:normalized-contrast-identities}.

For $x\in\cX_0$, define
\[
    \Delta_\oplus(y,x)
    :=M_\oplus(y,x)-M_\oplus(m_\oplus(x),x).
\]

\paragraph{(M3)}\label{par:M3}
\textit{For every $\eta>0$,
\[
    \inf_{x\in\cX_0}
    \inf_{d_{\cY}(y,m_\oplus(x))\ge\eta}
    \Delta_\oplus(y,x)>0.
\]
Moreover, there are
$c_M,\eta_M>0$ and $\beta_M>1$ such that
\[
    \Delta_\oplus(y,x)
    \ge c_Md_{\cY}(y,m_\oplus(x))^{\beta_M}
\]
whenever $d_{\cY}(y,m_\oplus(x))<\eta_M$, uniformly in $x\in\cX_0$.}

Condition~\aref{M3} requires uniform identification and local curvature.

\paragraph{Hilbert and Hadamard responses.}
The barycenter variance inequality gives
\[
    \Delta_\oplus(y,x)\ge d_\cY^2(y,m_\oplus(x))
\]
in Hadamard spaces, with equality in Hilbert spaces such as Euclidean
spaces, $L^2$, and $H^k=W^{k,2}$. Further Hadamard examples are
log-Euclidean SPD matrices, BHV trees, and univariate distributions on
$[a,b]$ with the $2$-Wasserstein metric
\citep{arsigny2006log,billera2001geometry,villani2021topics}.
Thus a measurable target satisfies \aref{M3} globally with
$\beta_M=2$ and $c_M=1$.

Conditions~\aref{K1}--\aref{K3} provide the kernel-side conditions for basic spectral stability via an ordinary Davis--Kahan argument~\citep{davis1970rotation, yu2015useful}.

\paragraph{(K1)}\label{par:K1}
\textit{The kernel $\kappa$ is measurable, continuous, and positive definite, and
$\sup_{x\in\cX}\kappa(x,x)\le B_\kappa<\infty$.}

Conditions~\aref{K2}--\aref{K3} concern the selected cutoff $s=s_{n,N}$.

\paragraph{(K2)}\label{par:K2}
\textit{$\lambda_s>\lambda_{s+1}>0$.}

Condition~\aref{K2} allows repeated eigenvalues, but the cutoff must not split an eigenspace.

\paragraph{(K3)}\label{par:K3}
\textit{$N^{-1/2}=o(\lambda_s-\lambda_{s+1})$.}

Conditions~\aref{R1}--\aref{R2} control fractional-source regularity and
conditional noise complexity. All indexed processes below, including
approximants and residuals, admit jointly measurable, separable versions.

\paragraph{(R1)}\label{par:R1}
\textit{For some $0<r\le1$ and $B_r>0$, every
$\dot h_t$ belongs to $\overline{\mathcal H_{\kappa_0}}^{\,L^2(\bbP_X)}$, and with
$c_j(t):=\bbE[\dot h_t(X)\phi_j(X)]$,
\[
    \sup_{t\in\mathfrak T}\sum_{j\ge1}
      \frac{c_j(t)^2}{\lambda_j^r}\le B_r^2,
    \qquad \sum_{j\ge1}\lambda_j^r<\infty.
\]}
Note that this uniform fractional source condition is imposed on functions of the
covariate, not on an embedding of the response. 
It gives
\[\dot h_t=\mathcal T_{\kappa_0}^{r/2}u_t,\qquad u_t\in L^2(\bbP_X)\] and
$\sup_t\|u_t\|_2\le B_r$, where the fractional power is taken in
$L^2(\bbP_X)$;
$r=1$ recovers RKHS membership, whereas $0<r<1$ permits targets outside the
RKHS \citep{zhang2023optimality}. The kernel, its population RKHS, $r$, and
$B_r$ remain fixed as $n,N$ grow.

\paragraph{Why the closure?}
Distance-based object models can have nonsmooth conditional Fr\'echet
functions in the covariate. For example, in a BHV tripod, a signal crossing two branches
has squared distance $(a+|\sin x|)^2$ to a point at radius $a$ on the third
branch. A centered contrast retains an absolute-value term,
which need not lie in a smooth covariate RKHS.
Example~\ref{ex:bhv-closure-source} verifies \aref{R1} with $r<1$ for such
a non-RKHS contrast. 

\paragraph{(R2)}\label{par:R2}
\textit{ There is $C_M<\infty$ such that, for every $k\ge1$,
$x_1,\ldots,x_k\in\cX_0$, $\mathbf w=(w_1,\ldots,w_k)\in\mathbb R^k$, and conditionally
independent $Y_i\sim \bbP_{Y\mid X=x_i}$,
{\small
\[
\begin{aligned}
&\left\|\,\sup_{t\in\mathfrak T}
  \left|\sum_{i=1}^kw_i\zeta_t(x_i,Y_i)\right|
  \right\|_{L^2(\,\cdot\mid x_1,\ldots,x_k)}
  \le C_M\|\mathbf w\|_2.
\end{aligned}
\]
}
}
Notice that \citet{im2026localentropy} pointed out that the local entropy conditions used in Fr\'echet regression
\citep{petersen2019frechet,chen2022uniform} can fail in
$\mathcal W^2([a,b])$.
Thus, we alternatively impose \aref{R2}. It holds automatically under \aref{M1} for
Hilbert responses (Proposition~\ref{prop:hilbert-automatic-complexity})
and for univariate $W_2$ responses through the quantile embedding into
$L^2(0,1)$ \citep{bigot2017geodesic}. An entropy-integral bound on the
conditional noise class, with the conditional sub-Gaussian and measurability
conditions of Proposition~\ref{prop:metric-response-entropy}, also suffices.

For the approximants $g_{t,s}:=\Pi_{\le s}\dot h_t$, put
\begin{equation}
\label{eq:R1-regression-truncation}
    b_{r,s}:=B_r\Big(\sum_{j>s}\lambda_j^r\Big)^{1/2},\qquad
    G_{r,s}:=B_r\lambda_s^{(r-1)/2}.
\end{equation}
Lemma~\ref{lem:automatic-spectral-truncation} gives
$\|\sup_t|\dot h_t-g_{t,s}|\|_2\le b_{r,s}$ and
$\sup_t\|g_{t,s}\|_{\mathcal H_{\kappa_0}}\le G_{r,s}$.

For the leading rank-$s$ projector $\widehat\Pi_{\le s}$ learned from $N$
covariates, define the prediction-weighted estimation discrepancy in
$\mathcal H_{\kappa_0}$ by
\[
    E_{N,s}
    :=\|\mathcal T_{\kappa_0}^{1/2}(\widehat\Pi_{\le s}-\Pi_{\le s})\|_{\HS}.
\]
Under Conditions~\aref{K1}--\aref{K3}, the bounded kernel covariance bound
and Davis--Kahan Theorem in
Proposition~\ref{prop:primitive-spectral-verification} give
\[
    E_{N,s}
    =O_p\!\left(
      N^{-1/2}(\lambda_s-\lambda_{s+1})^{-1}
    \right)
    =o_p(1).
\]

\subsection{Prediction Error Bounds}
\label{sec:statistical-rates}

For the basic bound, take $\xi=\xi_{n,N}>0$ satisfying
\begin{equation}
\label{eq:basic-tuning-regime}
    \frac{\xi}{\lambda_s}\to0,\qquad
    \frac{\log(2s)}{n(\lambda_s+\xi)}\to0\quad(N>n).
\end{equation}
The second limit is imposed only when $N>n$ and controls the labeled
covariance inverse (Lemma~\ref{lem:labeled-gram-stability}). Define
{\small\[
\rho_{r,s,n,N,\xi}
    :=b_{r,s}+\sqrt{\frac{s}{n}}
       +G_{r,s}\frac{\xi}{\lambda_s+\xi}+\frac{G_{r,s}}{\sqrt N\,(\lambda_s-\lambda_{s+1})}.
\]}
These four terms represent spectral truncation, labeled response noise,
ridge bias, and representation estimation. The ridge term reflects the
interaction between truncation and regularization \citep{amini2021spectrally}.

A stability argument
controls the entire contrast-residual family before taking its supremum;
it avoids multiplying $b_{r,s}$ by $\sqrt s$ or the labeled-noise term by
$G_{r,s}$ (Lemmas~\ref{lem:hilbert-smoother-stability}--\ref{lem:sharp-fixed-design}).

Conditional on the training-data sigma-field $\mathscr D_{n,N}$, define
\[
    \mathfrak R_{s,\xi,\varepsilon_{\mathrm{opt}}}
    :=\bbE_{X^\circ}\big[
      d_{\cY}^2(\widehat m_{s,\xi,\varepsilon_{\mathrm{opt}}}(X^\circ),m_\oplus(X^\circ))
      \mid\mathscr D_{n,N}\big].
\]
\medskip

\begin{theorem}\label{thm:frechet-error}
Under Conditions~\aref{M1}, \aref{M2}--\aref{M3}, \aref{K1}--\aref{K3},
\aref{R1}--\aref{R2}, and \eqref{eq:basic-tuning-regime}, suppose either
that $\varepsilon_{\mathrm{opt}}>0$ and $\varepsilon_{\mathrm{opt}}\to0$, or
that an exact measurable minimizer is available and
$\varepsilon_{\mathrm{opt}}=0$. Then
\[
\begin{split}
    \mathfrak R_{s,\xi,\varepsilon_{\mathrm{opt}}}
    =O_p\big(\rho_{r,s,n,N,\xi}^{\,2/(\beta_M-1)}
      +\rho_{r,s,n,N,\xi}^{\,2}
      +\varepsilon_{\mathrm{opt}}^{2/\beta_M}\big).
\end{split}
\]
In particular, $\rho_{r,s,n,N,\xi}\to0$ implies consistency.
\end{theorem}
For quadratic curvature, this is
$O_p(\rho_{r,s,n,N,\xi}^2+\varepsilon_{\mathrm{opt}})$.
The target remains the true $m_\oplus$, not a kernel-projected minimizer.
When $r<1$, the rank-$s$ norm bound $G_{r,s}$ can grow; its ridge and
representation costs must still vanish.

\subsection{Spectral Refinement}
Additional kernel conditions~\aref{K4}--\aref{K5} sharpen representation
estimation for the same estimator and weaken sufficient pooled-sample requirements.

\paragraph{(K4)}\label{par:K4}
\textit{For some $C_{X,8}<\infty$,
\[
    \sup_{j\ge1}\bbE\bigl[|\phi_j(X)|^8\bigr]\le C_{X,8}.
\]}

\paragraph{(K5)}\label{par:K5}
\textit{For some $q>1$,
\[
\begin{aligned}
    \lambda_j\asymp j^{-q},&\qquad j\ge1,\\
    \lambda_j-\lambda_k\gtrsim j^{-q}-k^{-q},&
    \qquad 1\le j\le s<k.
\end{aligned}
\]}
The second inequality is an additional cross-boundary separation
assumption, not a consequence of eigenvalue decay alone. It is a
block-level analogue of the relative-gap conditions of
\citet{jirak2020perturbation,jirak2023relative}.
Taking $j=s$, $k=s+1$ shows that \aref{K5} implies \aref{K2}.

For this refinement, take
\begin{equation}
\label{eq:sharp-tuning-regime}
    s\log(es)\sqrt{\frac{\log N}{N}}\to0,\quad
    \frac{s^2}{n}\to0,\quad \frac{\xi}{\lambda_s}\to0.
\end{equation}
The first limit replaces \aref{K3}; $s^2/n\to0$ replaces the labeled
covariance condition in \eqref{eq:basic-tuning-regime}.
Proposition~\ref{prop:sharp-spectral-verification} gives the refined scale
\[
\begin{split}
    \rho_{r,s,n,N,\xi}^{\star}
    :=b_{r,s}&+\sqrt{\frac{s}{n}}
      +G_{r,s}\frac{\xi}{\lambda_s+\xi}\\
      &+G_{r,s}\sqrt{\frac{s^{2-q}\log(es)\log N}{N}}.
\end{split}
\]
\begin{theorem}\label{thm:sharp-frechet-error}
Under Conditions~\aref{M1}, \aref{M2}--\aref{M3}, \aref{K1},
\aref{K4}--\aref{K5}, \aref{R1}--\aref{R2}, and
\eqref{eq:sharp-tuning-regime}, the conclusion of
Theorem~\ref{thm:frechet-error} holds with
$\rho_{r,s,n,N,\xi}^{\star}$ in place of $\rho_{r,s,n,N,\xi}$,
under the same minimizer convention.
\end{theorem}

\paragraph{Polynomial rates.}
Under \aref{K5}, the summability in \aref{R1} requires $qr>1$, and
\[
    b_{r,s}^2\lesssim s^{-(qr-1)},\qquad
    G_{r,s}^2\lesssim s^{q(1-r)}.
\]
Thus the squared refined scale is bounded, up to a constant, by
{\small\[
 s^{-(qr-1)}+\frac{s}{n}
   +s^{q(1-r)}\left(\frac{\xi}{\lambda_s+\xi}\right)^2+\frac{s^{2-qr}\log(es)\log N}{N}.
\]}
\begin{corollary}\label{cor:polynomial-rate}
Assume Conditions~\aref{M1}, \aref{M2}--\aref{M3}, \aref{K1},
\aref{K4}--\aref{K5}, and \aref{R1}--\aref{R2}, and the minimizer convention
of Theorem~\ref{thm:frechet-error}. If $qr>2$, $N\ge n$, and
\[
    s\asymp n^{1/(qr)},\qquad
    0<\xi=O\big(s^{-(3q-1)/2}\big),
\]
then \eqref{eq:sharp-tuning-regime} holds and
\[
\begin{split}
    \mathfrak R_{s,\xi,\varepsilon_{\mathrm{opt}}}
    =O_p\big(n^{-\frac{qr-1}{qr(\beta_M-1)}}
       +n^{-\frac{qr-1}{qr}}+\varepsilon_{\mathrm{opt}}^{2/\beta_M}\big).
\end{split}
\]
For $\beta_M=2$, this becomes
$O_p(n^{-(qr-1)/(qr)}+\varepsilon_{\mathrm{opt}})$.
\end{corollary}
For the RKHS source ($r=1$), the general metric-response rate is recovered. For
$1<qr\le2$, a smaller cutoff still yields consistency
(Remark~\ref{rmk:spectral-decay-regimes}). A vector-valued fractional
source for Hilbert responses improves the squared bias to $s^{-qr}$ and
gives $n^{-qr/(qr+1)}$ under quadratic curvature
(Corollary~\ref{cor:hilbert-fractional-rate}). These are sufficient
rates, not minimax-optimality claims.

\paragraph{Semi-supervised gain.}\label{rmk:ss-gain}
At common $(s,n,\xi)$ and fixed source class, pooling replaces $n$ by $N$
in the basic squared representation term
$G_{r,s}^2/\{N(\lambda_s-\lambda_{s+1})^2\}$ and reduces its refined
counterpart $G_{r,s}^2s^{2-q}\log(es)\log N/N$. In our upper bound, it
does not reduce truncation, response noise, or ridge bias, and does not
remove the refined route's labeled covariance condition $s^2/n\to0$. In particular,
$N\asymp n$ already suffices for Corollary~\ref{cor:polynomial-rate};
more unlabeled data need not improve the leading prediction rate.
When representation error dominates the bound, increasing $N$ can improve
its order until the remaining terms dominate.

\section{Hybrid Estimation}
\label{sec:hybrid-extension}

Since the predictive benefit of pooling varies across tasks, we propose three hybrid methods that adjust the influence of pooling
before ridge fitting through the feature kernel (RH), after fitting
through predictions (PH), or at both stages (TH). They build on model
aggregation and safe SSL
\citep{wolpert1992stacked,breiman1996stacked,li2014towards,li2017learning}.

\paragraph{Representation hybrid (RH).}
For fixed $\alphaH\in[0,1]$, define the concatenated feature
\begin{equation}
\label{eq:hybrid-covariate-law-main}
 \widehat{\boldsymbol\Phi}_{s,\alphaH}(x)
 :=\begin{bmatrix}
 \sqrt{1-\alphaH}\,\widehat{\boldsymbol\Phi}^{\mathrm L}_s(x)\\
 \sqrt{\alphaH}\,\widehat{\boldsymbol\Phi}^{\mathrm N}_s(x)
 \end{bmatrix}.
\end{equation}
The maps are centered and diagonalized separately. Labeled ridge on
this concatenation, including cross-block covariance, defines
$\widehat m^{\alphaH}$. Equivalently, it fits the sum of the two feature
kernels weighted by $1-\alphaH$ and $\alphaH$, in up to $2s$ dimensions.
RH can therefore change model capacity as well as reliance on pooling.
At common tuning, $\alphaH=0,1$ recover the endpoints
(Appendix~\ref{app:hybrid-extension-details}).

\paragraph{Prediction hybrid (PH).}
Denote these labeled-only and pooled endpoint fits, respectively, by
\[
    \widehat m^{\mathrm L}(x)
    :=\widehat m^{0}(x),\qquad
    \widehat m^{\mathrm N}(x)
    :=\widehat m^{1}(x).
\]
For $\eta\in[0,1]$, the prediction-hybrid objective is
\[
    \Gamma_{\eta}^{\mathrm{PH}}(y,x)
    :=(1-\eta)d_\cY^2(y,\widehat m^{\mathrm L}(x))
    +\eta d_\cY^2(y,\widehat m^{\mathrm N}(x)).
\]
Take $\varepsilon_{\mathrm{hyb}}>0$, or zero when an exact measurable
minimizer exists, and define
$\widehat m_\eta^{\mathrm{PH}}$ by
\begin{equation}
\label{eq:prediction-hybrid-main}
    \Gamma_{\eta}^{\mathrm{PH}}
      (\widehat m^{\mathrm{PH}}_{\eta}(x),x)
    \le
    \inf_{y\in\cY}\Gamma_{\eta}^{\mathrm{PH}}(y,x)
    +\varepsilon_{\mathrm{hyb}}.
\end{equation}

\paragraph{Why mix predictions?}
For fitted maps, write
\[\mathfrak D(g,h):=\int d_\cY^2(g(x),h(x))\,d\bbP_X(x),\]
and $
\mathfrak R(g):=\mathfrak D(g,m_\oplus)$. In a Hadamard space, exact
geodesic interpolation $g_\eta$ between $u,v$ satisfies
\[
\begin{split}
 \mathfrak R(g_\eta)\le{}&(1-\eta)\mathfrak R(u)+\eta\mathfrak R(v)
 -\eta(1-\eta)\mathfrak D(u,v).
\end{split}
\]
The same holds for conditional test loss
(Proposition~\ref{prop:hybrid-geodesic-gain}). Disagreement can permit
improvement over both endpoints, but validation need not find such a weight.
For PH, this also bounds risk by the larger endpoint risk.

\begin{corollary}[PH at balanced cutoffs]\label{cor:ph-balanced-rate}
Under Corollary~\ref{cor:polynomial-rate} with Hadamard responses, fit
both endpoints at its fixed common kernel, rank, ridge and tolerance, and use exact
PH interpolation. For any measurable, possibly validation-selected
$\widehat\eta\in[0,1]$,
\[
 \mathfrak R(\widehat m^{\mathrm{PH}}_{\widehat\eta})
 =O_p\big(n^{-(qr-1)/(qr)}+\varepsilon_{\mathrm{opt}}\big).
\]
\end{corollary}

\paragraph{Two-level hybrid (TH).}
TH combines representation and prediction adaptation by replacing the
pooled endpoint in PH with the RH fit:
\[
    \Gamma_{\alphaH,\eta}^{\mathrm{TH}}(y,x)
    :=(1-\eta)d_\cY^2(y,\widehat m^{\mathrm L}(x))
    +\eta d_\cY^2(y,\widehat m^{\alphaH}(x)).
\]
The near-minimizer as for PH defines
$\widehat m_{\alphaH,\eta}^{\mathrm{TH}}$; selection conventions are given
in Appendix~\ref{app:hybrid-extension-details}.

In experiments, PH and TH tune endpoints separately and then select
mixing weights on the same validation set
(Algorithm~\ref{alg:hybrids-detailed}); our rate results do not cover
that full procedure.

\paragraph{Stability at common tuning.}
Fix a kernel $\kappa$, rank $s<n$, ridge $\xi>0$, and common
optimization tolerances. Fits share labeled responses, recentering,
and the unpenalized intercept; Appendix~\ref{app:hybrid-extension-details}
specifies spectral and selection conventions.

For $N\ge n\ge2$, put
\begin{equation}
\label{eq:th-stability-scale}
 \chi_{s,n,N,\xi}
 :=\frac{\xi^{-1}+\xi^{-2}}{\lambda_s-\lambda_{s+1}}
       \sqrt{\frac{N-n}{n(N-1)}}.
\end{equation}

\begin{samepage}
\begin{theorem}[Stability of the two-level hybrid]
\label{thm:th-stability}
Under the common-tuning conventions above, assume the conditions of
Theorem~\ref{thm:frechet-error} or \ref{thm:sharp-frechet-error}, with
$\rho_{\mathrm N}=\rho_{r,s,n,N,\xi}$ or
$\rho_{r,s,n,N,\xi}^{\star}$, respectively.
Suppose $\cY$ is Hadamard. There is one nonnegative
$Z_{n,N}=O_p(1)$, common to all mixing weights, such that
simultaneously for all $(\alphaH,\eta)\in[0,1]^2$,
{\small
\begin{equation}
\label{eq:th-stability-bound}
\begin{aligned}
 &\mathfrak R(\widehat m_{\alphaH,\eta}^{\mathrm{TH}})
   +\mathfrak D(\widehat m_{\alphaH,\eta}^{\mathrm{TH}},\widehat m^{\mathrm N})\\
 &\quad\le Z_{n,N}\left(\rho_{\mathrm N}^2
       +(1-\alphaH\eta)^2\chi_{s,n,N,\xi}^2
       +\varepsilon_{\mathrm{opt}}+\varepsilon_{\mathrm{hyb}}\right).
\end{aligned}
\end{equation}
}
\end{theorem}
\end{samepage}

The common factor allows validation-selected weights. With exact interpolation and shared fits, $\eta=1$
gives RH and $\alphaH=1$ gives PH.

Notice that, along $N>n$, $\chi\to\infty$ at the balanced
rank and ridge choices of Corollaries~\ref{cor:polynomial-rate}
and~\ref{cor:hilbert-fractional-rate}
(Appendix~\ref{app:th-stability-proof}). Thus this bound does not give
their rates for fixed $\alphaH\eta<1$; approaching the pooled endpoint
with $(1-\alphaH\eta)\chi=O_p(\rho_{\mathrm N})$ suffices, up to
optimization error. The PH rate above uses both endpoint bounds instead.
Stability is not a guarantee of improvement over SS-KRFR.

\begin{table*}[!t]
\centering\scriptsize
\setlength{\tabcolsep}{2pt}
\caption{Mean test loss (MCSE) over 200 replications at a common budget: \(U=2000\) for semi-supervised methods and \(U=0\) for supervised methods. S2/S5 units: \(10^{-3}\). Bold marks the lowest mean loss in each setting, not statistical significance.}
\begin{tabular*}{\textwidth}{@{\extracolsep{\fill}}lrrrrrr@{}}
\toprule
Method & S1: Euc.\(\to\)SPD & S2: Euc.\(\to\)Dist. & S3: Euc.\(\to\)BHV & S4: Dist.\(\to\)SPD & S5: Dist.\(\to\)Dist. & S6: Dist.\(\to\)BHV \\
\midrule
KRFR & 2.4098 (0.0142) & 5.4546 (0.0378) & 0.2462 (0.0019) & 0.0340 (0.0006) & 1.1420 (0.0125) & 0.1717 (0.0021) \\
SS-KRFR & 2.4509 (0.0136) & 5.0297 (0.0326) & 0.2306 (0.0019) & 0.0262 (0.0003) & 1.2611 (0.0180) & 0.1567 (0.0024) \\
RH & 1.6079 (0.0107) & 5.1297 (0.0417) & 0.2317 (0.0019) & 0.0256 (0.0003) & 1.1386 (0.0169) & 0.1570 (0.0023) \\
PH & 2.1960 (0.0192) & 4.9917 (0.0307) & \textbf{0.2267 (0.0016)} & 0.0257 (0.0003) & 1.1342 (0.0132) & \textbf{0.1546 (0.0022)} \\
TH & \textbf{1.4981 (0.0110)} & 5.0514 (0.0355) & 0.2271 (0.0017) & \textbf{0.0253 (0.0003)} & \textbf{1.0842 (0.0126)} & 0.1555 (0.0022) \\
\midrule
WCE-NLGFR & 2.2029 (0.0108) & 5.7063 (0.0269) & 0.2746 (0.0018) & 0.0615 (0.0011) & 1.2007 (0.0142) & 0.1767 (0.0024) \\
RF-LC & 3.2938 (0.0219) & \textbf{4.8914 (0.0316)} & 0.2308 (0.0019) & 0.4978 (0.0027) & 1.5596 (0.0151) & 0.3489 (0.0031) \\
RF-LL & 1.9839 (0.0099) & 5.8649 (0.0390) & 0.2615 (0.0018) & -- & -- & -- \\
Graph NW & 6.9672 (0.0375) & 9.1949 (0.0491) & 0.3894 (0.0027) & 0.2487 (0.0027) & 2.2842 (0.0249) & 0.2620 (0.0027) \\
Graph kNN & 7.0637 (0.0371) & 9.2609 (0.0464) & 0.3973 (0.0027) & 0.1963 (0.0017) & 1.7935 (0.0171) & 0.2767 (0.0026) \\
NW-FR & 7.0267 (0.0367) & 9.2243 (0.0497) & 0.3881 (0.0027) & 0.2408 (0.0027) & 2.2614 (0.0243) & 0.2592 (0.0027) \\
kNN-FR & 7.0011 (0.0366) & 9.1966 (0.0441) & 0.3952 (0.0027) & 0.1929 (0.0017) & 1.8189 (0.0171) & 0.2765 (0.0026) \\
LL-FR & 2.4096 (0.0092) & 5.8647 (0.0298) & 0.2712 (0.0016) & -- & -- & -- \\
Global FR & 2.5607 (0.0095) & 6.0875 (0.0310) & 0.2788 (0.0016) & -- & -- & -- \\
\bottomrule
\end{tabular*}
\label{tab:main-object-results}
\end{table*}

\begin{table}[!t]
\centering\scriptsize
\setlength{\tabcolsep}{2pt}
\caption{HBN mean test loss (MCSE) over 200 shared splits: \(U=1314\) for semi-supervised and \(U=0\) for supervised methods. Bold marks the lowest mean loss. MCSE measures fixed-cohort split precision, not population uncertainty.}
\begin{tabular*}{\columnwidth}{@{\extracolsep{\fill}}lrr@{}}
\toprule
Method & HBN-A & HBN-B \\
\midrule
KRFR & 6.6962 (0.1135) & 6.6311 (0.1132) \\
SS-KRFR & 6.4601 (0.1088) & 6.5203 (0.1141) \\
RH & 6.5710 (0.1109) & 6.6451 (0.1139) \\
PH & \textbf{6.3581 (0.0985)} & \textbf{6.3761 (0.1018)} \\
TH & 6.3985 (0.1006) & 6.4245 (0.1083) \\
\midrule
WCE-NLGFR & 6.9164 (0.1272) & 6.8781 (0.1224) \\
RF-LC & 7.1214 (0.0922) & 7.1099 (0.1017) \\
Graph NW & 7.3925 (0.1404) & 7.5625 (0.1436) \\
Graph kNN & 7.8814 (0.1197) & 7.7777 (0.1132) \\
NW-FR & 7.2944 (0.1367) & 7.3169 (0.1361) \\
kNN-FR & 7.6715 (0.1150) & 7.3985 (0.0997) \\
\bottomrule
\end{tabular*}
\label{tab:hbn-main-results}
\end{table}

\section{Experiments}
\label{sec:experiment}

We compare KRFR, SS-KRFR, and hybrids RH, PH, and TH
(Section~\ref{sec:hybrid-extension}) in six simulations and two HBN-EEG tasks.
Tables~\ref{tab:main-object-results}--\ref{tab:hbn-main-results} compare
all methods at a common maximal unlabeled budget, without test-based
selection of methods or \(U\). Appendix~\ref{app:experiment-details}
retains smaller-pool results and descriptive test-best summaries.

Proposed methods and forests use validation loss; local and graph smoothers use
training LOOCV. The primary WCE comparator uses validation, not its
original GCV criterion (HBN-A loss: 6.9164 versus 14.0332;
Appendix~\ref{app:wce-tuning-sensitivity}). PH tunes KRFR and SS-KRFR
separately, then selects their prediction weight on the same validation
set; TH uses KRFR and RH. The fixed-tuning theory does not cover this
full selection procedure.

Test loss averages squared response-space distances to observed responses.
For mean loss \(\overline L_M(U)\) over 200 replications or splits, define
\(\operatorname{SS\!\text{-}gain}_M(U)
=100\{1-\overline L_M(U)/\overline L_M(0)\}\).
Tables report mean loss and Monte Carlo standard error (MCSE).
Appendix~\ref{app:experiment-details} gives paired comparisons, \(R^2\),
interval definitions and numerical conventions.

\subsection{Simulation Study}
\label{subsec:simulation}

The six settings cross Euclidean or distribution-valued covariates with
\(3\times3\) log-Euclidean SPD matrices, distributions on \([0,1]\)
equipped with the 2-Wasserstein metric \(W_2\), and five-taxon
Billera--Holmes--Vogtmann (BHV) tree responses. BHV trees have fixed pendant
lengths and use the intrinsic 15-orthant metric \citep{billera2001geometry}.
S3 crosses multiple topologies; S6 occupies two adjacent orthants.
These are empirical evaluations: the four-taxon construction in
Example~\ref{ex:bhv-closure-source} does not verify \aref{R1}--\aref{R2}
for S3/S6.
Each setting uses 200 independent replications with
\((n_{\rm train},n_{\rm val},n_{\rm test})=(150,50,500)\)
and nested unlabeled sample sizes \(U\in\{0,500,1000,2000\}\).
All methods use the same data within each replication
(Appendix~\ref{app:simulation-design}).

Competitors are WCE-NLGFR \citep{Bhattacharjee2025}, Graph SS-FR
(NW/kNN) \citep{qiu2024semi}, local Fr\'echet regression
(NW-FR, kNN-FR, and LL-FR), global Fr\'echet regression
\citep{petersen2019frechet}, and random-forest weighted local constant
(RF-LC) and local linear (RF-LL) Fr\'echet regression \citep{qiu2024random}.
LL-FR, Global FR and RF-LL are included only for Euclidean covariates.
Forests use 50 trees; tuning grids and implementation conventions are
specified in Appendix~\ref{app:tuning-grid}.

At \(U=2000\), pure SS-KRFR improves on KRFR in S2, S3, S4 and S6,
but its loss increases by 1.70\% in S1 and 10.43\% in S5.
TH instead reduces loss by 37.83\% and 5.06\% in those two settings
(Table~\ref{tab:main-object-results}). PH has lower mean loss than pure
SS-KRFR in all six settings; TH does so in five, with S2 the exception.
These mean comparisons support adapting pooling, not a guarantee of
improvement. Gains need not be monotone in \(U\):
in S3, TH's mean loss is 0.2265 at \(U=1000\) and 0.2271 at
\(U=2000\) (Appendix~\ref{app:simulation-results}).

Although \(\bbP_X\) is shared within S1--S3 and within S4--S6,
gains differ across response mechanisms, so covariate structure alone
does not explain them. Finite-sample feature estimation, the geometry and spectral
alignment of the regression signal, response noise, and regularization
may all matter. Equal weighting of every covariate in the pooled
eigensystem need not produce the best finite-sample predictor.
Hybrids retain a labeled-only component
and adjust reliance on pooled features or predictions; they do not
reweight individual unlabeled observations. RH also uses up to $2s$
feature coordinates. These are plausible
interpretations, not isolated causal effects or guarantees of improvement.

Relative performance is more consistently favorable
for the non-Euclidean, distribution-valued covariates
in S4--S6. TH attains the lowest mean loss in S4 and S5,
and PH in S6, with reductions of 58.87\%, 9.70\%, and 12.53\%
over WCE, the strongest external comparator in these settings. For Euclidean inputs, RF-LC wins S2 (0.00489 versus
PH's 0.00499) and is close to PH in S3. This pattern supports using
kernel representations for object-valued covariates, but does not
establish a geometry-driven advantage: response mechanisms also differ,
the S4/S6 signals are spectrally aligned by construction, and forests
split quantile coordinates whereas our kernels use distances.

\subsection{Real Data Analysis: HBN-EEG}
\label{subsec:hbn-eeg}

Our HBN analysis uses a processed dataset of 2,103 paired eyes-closed and eyes-open observations from HBN-EEG releases R1–R9
\citep{alexander2017hbn,langer2017developing,shirazi2024hbneeg}.\footnote[1]{HBN-EEG
release index: \url{https://neuromechanist.github.io/data/hbn/}.}
For each participant, resting-state eyes-closed (EC) and eyes-open (EO)
power spectra are averaged across 16 channels and normalized on
\([1,45]\) Hz, yielding distributions of power over frequency.

HBN-A predicts the EO spectral distribution \(Y_i=\mu_i^{\rm EO}\)
from the EC distribution \(X_i=\mu_i^{\rm EC}\). HBN-B adds the
regularized \(16\times16\) EC channel correlation matrix \(C_i\), using
\(X_i=(C_i,\mu_i^{\rm EC})\) with log-matrix coordinates and
training-only block scaling.
Both tasks use squared 2-Wasserstein loss, approximated using 80 scaled
midpoint quantiles.
HBN-B selects distance weights from seven candidates by validation;
RF-LC instead selects quantile-only, joint, or SPD-only inputs.

The 789 participants in R1--R4 undergo 200 shared random partitions
without replacement into 284 training, 183 validation, and 322 test
participants. These are fixed-cohort splits, not bootstrap samples or
independent cohorts.
Nested unlabeled pools remain fixed at R5--R6 (\(U=463\)), R5--R7
(\(U=828\)), and R5--R9 (\(U=1314\)), alongside \(U=0\); their
responses are withheld. Pool size and release composition change together,
so this is not a pure sample-size comparison.
Appendix~\ref{app:hbn-details} gives details.

At the common \(U=1314\) budget, pure SS-KRFR reduces mean loss
relative to KRFR by 3.53\% and 1.67\% in HBN-A/B; PH's reductions
are 5.05\% and 3.85\% (Table~\ref{tab:hbn-main-results}). RH slightly
increases HBN-B loss, so hybridization does not ensure improvement.
For both pure SS-KRFR and PH, \(U=463\) gives lower mean losses than
\(U=1314\). PH's corresponding gains over KRFR are 6.44\% and
6.03\% (Appendix~\ref{app:hbn-details}). Thus semi-supervised gains
are observed without monotone improvement as unlabeled data increase;
here pool size is also confounded with release composition.
These are fixed-cohort comparisons. NW-FR uses an empty-neighborhood
fallback for 14.19\% of HBN-A test queries
(Table~\ref{tab:app-exception-coverage}); the KRFR comparison therefore
complements the external-baseline results.

\section{Discussion}\label{sec:conclusion}

Unlabeled covariates support representation learning, but their predictive
value depends on the regression task. Our analysis distinguishes PH's
endpoint-based rate from stability of the jointly refitted RH and TH;
neither explains all empirical hybrid gains. Oracle guarantees for
independently tuned endpoints and validation selection remain open.
HBN evidence is conditional on one cohort: pool size and release composition
vary together, unlike the shared-law theory.






\section*{Acknowledgements}
Jaehee Seo was supported by the Next Generation Scholarship for Basic Studies (Type C) from Seoul National University. Jaehee Seo would like to thank Jeongmin Jeon for helpful comments on assumptions in Fr\'echet regression and Yoonju Jeong for insightful discussions about brain network data.

\bibliographystyle{plainnat}
\bibliography{ref}

\clearpage
\appendix
\thispagestyle{empty}

\onecolumn
\aistatstitle{Supplementary Materials}

\appendix

\section*{Overview}

This supplement provides implementation details, additional experimental
results, and proofs of the theoretical guarantees in the main text.
The algorithms below summarize SS-KRFR and its validation-selected
hybrids, including spectral feature construction, labeled ridge fitting,
and the selection of representation and prediction mixing weights.

Section~\ref{app:experiment-details} describes the experimental designs,
response-space computations, competing methods, tuning and evaluation
protocols, and additional numerical results.
Section~\ref{app:preliminary} introduces the notation and mathematical
preliminaries, establishes the well-posedness of SS-KRFR, and details the
hybrid estimators and normalized-contrast and spectral constructions.
Section~\ref{app:technical-results} develops the population approximation,
empirical spectral, and labeled-design bounds, combines them into a
normalized-contrast oracle inequality, and establishes the escape and
argmin-inversion results used to control prediction error.
Section~\ref{app:main-result-proofs} presents the proofs of the
basic and refined prediction bounds and the polynomial-rate corollary.
Finally, Section~\ref{app:th-stability-proof} proves the TH stability
theorem, deriving the nested-sample spectral discrepancy and propagating
it through the ridge refit and Hadamard interpolation. It also separates
this stability bound from the endpoint-based polynomial rate for PH.

\clearpage

\begin{algorithm*}[t]
\caption[SS-KRFR]{SS-KRFR with pooled spectral features and labeled ridge fitting.}
\label{alg:ss-krfr-detailed}
\small
\begin{algorithmic}[1]
\setlength{\itemsep}{3pt}
\setlength{\lineskip}{2pt}
\Require Pooled training covariates $\{X_i\}_{i=1}^{N}$; labeled indices $L\subseteq[N]$, $|L|=n$;
\Statex \hspace{\algorithmicindent} responses $\{Y_i\}_{i\in L}$ in $(\mathcal Y,d_{\mathcal Y})$; kernel $\kappa$; requested rank $s$;
\Statex \hspace{\algorithmicindent} ridge parameter $\xi>0$; optimization tolerance $\varepsilon_{\mathrm{opt}}\ge0$; Gram eigenvalue threshold $\tau_{\mathrm{eig}}\ge0$.
\Ensure A prediction map $\widehat m^{\mathrm N}:\mathcal X\to\mathcal Y$.
\Statex \textbf{Spectral representation.} The following routine applies to any representation sample $R\subseteq[N]$.
\Function{Spec}{$R;\kappa,s$}
  \State $\displaystyle M=|R|,\qquad K^R=[\kappa(X_r,X_{r'})]_{r,r'\in R},\qquad H_M=I_M-\frac{\mathbf1_M\mathbf1_M^\top}{M},\qquad K_0^R=H_MK^RH_M.$
  \State $\displaystyle M^{-1}K_0^R\widehat v_j^R=\widehat\lambda_j^R\widehat v_j^R,\quad
    \widehat\lambda_1^R\ge\cdots\ge\widehat\lambda_M^R\ge0,\quad
    (\widehat v_j^R)^\top\widehat v_k^R=\ind\{j=k\}.$
  \State $\displaystyle r_R=\min\!\left\{s,\#\{j:M\widehat\lambda_j^R>\tau_{\mathrm{eig}}\}\right\}\ge1,\quad
    \widehat V_R=(\widehat v_1^R,\ldots,\widehat v_{r_R}^R),\quad
    \widehat\Lambda_R=\operatorname{diag}(\widehat\lambda_1^R,\ldots,\widehat\lambda_{r_R}^R).$
  \State For $x\in\mathcal X$, set
  \Statex $\displaystyle\hspace{\algorithmicindent}
    k_R(x)=(\kappa(x,X_r))_{r\in R},\qquad
    \widehat k_{0,R}(x)=H_M\!\left(k_R(x)-\frac{K^R\mathbf1_M}{M}\right),\qquad
    \widehat\Phi^R(x)=\frac{\widehat\Lambda_R^{-1/2}\widehat V_R^\top\widehat k_{0,R}(x)}{\sqrt M}.$
  \State \Return $\widehat\Phi^R$.
\EndFunction
\Statex \textbf{Labeled fitting.} The following routine applies to any feature map $\phi:\mathcal X\to\mathbb R^r$.
\Function{Fit$_L$}{$\phi;\xi,\varepsilon_{\mathrm{opt}}$}
  \State $\displaystyle F=(\phi(X_i)^\top)_{i\in L},\qquad
    \bar\phi_L=\frac1n\sum_{i\in L}\phi(X_i),\qquad
    H_L=I_n-\frac{\mathbf1_n\mathbf1_n^\top}{n},\qquad Q=H_LF.$
  \State $\displaystyle C_L=\frac1n Q^\top Q,\qquad A_L=(C_L+\xi I_r)^{-1},\qquad q_L(x)=\phi(x)-\bar\phi_L.$
  \State For $y\in\mathcal Y$, the scalar ridge fit with an unpenalized intercept is
  \Statex $\displaystyle\hspace{\algorithmicindent}
    Z_y=(d_{\mathcal Y}^2(y,Y_i))_{i\in L},\qquad
    \bar Z_{y,L}=\frac{\mathbf1_n^\top Z_y}{n},\qquad
    \widehat\theta_y=A_L\frac{Q^\top Z_y}{n},\qquad
    \widehat a_y=\bar Z_{y,L}-\bar\phi_L^\top\widehat\theta_y.$
  \Statex $\displaystyle\hspace{\algorithmicindent}
    (\widehat a_y,\widehat\theta_y)
    =\underset{a\in\mathbb R,\,\theta\in\mathbb R^r}{\arg\min}
    \left\{\frac1{2n}\|Z_y-a\mathbf1_n-F\theta\|_2^2+\frac\xi2\|\theta\|_2^2\right\}.$
  \State At each query $x$, compute the weights, which do not depend on $y$:
  \Statex $\displaystyle\hspace{\algorithmicindent}
    w(x)=\frac{\mathbf1_n}{n}+\frac1nQA_Lq_L(x),\qquad
    w_i(x)=\frac1n+\frac1nq_L(X_i)^\top A_Lq_L(x),\qquad
    \sum_{i\in L}w_i(x)=1.$
  \State $\displaystyle \widehat M_{\phi,\xi}(y,x)
    =\sum_{i\in L}w_i(x)d_{\mathcal Y}^2(y,Y_i)
    =\bar Z_{y,L}+q_L(x)^\top\widehat\theta_y.$
  \State Choose a fixed measurable $\widehat m_{\phi,\xi}(x)\in\mathcal Y$ satisfying
  \Statex $\displaystyle\hspace{\algorithmicindent}
    \widehat M_{\phi,\xi}(\widehat m_{\phi,\xi}(x),x)
    \le\inf_{y\in\mathcal Y}\widehat M_{\phi,\xi}(y,x)+\varepsilon_{\mathrm{opt}}.$
  \State \Return $\widehat m_{\phi,\xi}$.
\EndFunction
\State $\displaystyle \widehat\Phi^{\mathrm N}=\mathsf{Spec}([N];\kappa,s),\qquad
  \widehat m^{\mathrm N}=\mathsf{Fit}_L(\widehat\Phi^{\mathrm N};\xi,\varepsilon_{\mathrm{opt}}).$
\State \Return $\widehat m^{\mathrm N}$; using $R=L$ instead gives the labeled-only KRFR fit $\widehat m^{\mathrm L}$.
\end{algorithmic}
\vspace{3pt}
\noindent\footnotesize
The box states the positive retained-rank case; $\tau_{\mathrm{eig}}=0$ gives the exact positive-spectrum version.
Fix eigenvector and minimizer selection rules. Use $\varepsilon_{\mathrm{opt}}=0$ only when an exact measurable minimizer exists.
The weights may be negative, so the final step is a signed weighted Fr\'echet minimization.
\end{algorithm*}
\clearpage

\begin{algorithm*}[t]
\caption[Validation-selected hybrids]{Validation-selected representation (RH),
prediction (PH), and two-level (TH) hybrids.}
\label{alg:hybrids-detailed}
\small
\begin{algorithmic}[1]
\setlength{\itemsep}{3pt}
\setlength{\lineskip}{2pt}
\Require Training data as in Algorithm~\ref{alg:ss-krfr-detailed}; disjoint labeled validation set
\Statex \hspace{\algorithmicindent} $\mathcal D_V=\{(X_j^V,Y_j^V)\}_{j=1}^{n_V}$; finite feature--ridge candidate set $\mathcal C$;
\Statex \hspace{\algorithmicindent} representation-weight grid $\mathcal A_H\subseteq[0,1]$ and prediction-weight grid $\mathcal E\subseteq[0,1]$,
\Statex \hspace{\algorithmicindent} both containing $0,1$; tolerances $\varepsilon_{\mathrm{opt}},\varepsilon_{\mathrm{hyb}}\ge0$; fixed tie-breaking rules.
\Ensure The selected prediction maps $\widehat m^{\mathrm{RH}},\widehat m^{\mathrm{PH}},\widehat m^{\mathrm{TH}}$.
\State $\displaystyle \mathcal V(m)=\frac1{n_V}\sum_{j=1}^{n_V}d_{\mathcal Y}^2(Y_j^V,m(X_j^V)),\qquad R_{\mathrm L}=L,\quad R_{\mathrm N}=[N].$
\Statex \textbf{1. Fit and select the labeled-only and pooled endpoints.}
\For{$c\in\mathcal C$}
  \State For $B\in\{\mathrm L,\mathrm N\}$, use the candidate kernel $\kappa_{B,c}$, requested rank $s_c$, and ridge $\xi_c$:
  \Statex $\displaystyle\hspace{\algorithmicindent}
    \widehat\Phi_c^B=\mathsf{Spec}(R_B;\kappa_{B,c},s_c),\qquad
    \widehat m_c^B=\mathsf{Fit}_L(\widehat\Phi_c^B;\xi_c,\varepsilon_{\mathrm{opt}}).$
\EndFor
\State $\displaystyle \widehat c_B\in\underset{c\in\mathcal C}{\arg\min}\,\mathcal V(\widehat m_c^B),\qquad
  \widehat m^B=\widehat m_{\widehat c_B}^B,\qquad B\in\{\mathrm L,\mathrm N\}.$
\Statex \textbf{2. RH: jointly select the representation weight and feature--ridge candidate.}
\For{$(c,\alpha)\in\mathcal C\times\mathcal A_H$}
  \State $\displaystyle \widehat\Phi_{c,\alpha}^{\mathrm H}(x)=
    \begin{pmatrix}\sqrt{1-\alpha}\,\widehat\Phi_c^{\mathrm L}(x)\\[1pt]
                    \sqrt\alpha\,\widehat\Phi_c^{\mathrm N}(x)\end{pmatrix},\qquad
    q_c^B(x)=\widehat\Phi_c^B(x)-\frac1n\sum_{i\in L}\widehat\Phi_c^B(X_i),\quad B\in\{\mathrm L,\mathrm N\}.$
  \State $\displaystyle q_{c,\alpha}^{\mathrm H}(x)=
    \begin{pmatrix}\sqrt{1-\alpha}\,q_c^{\mathrm L}(x)\\[1pt]\sqrt\alpha\,q_c^{\mathrm N}(x)\end{pmatrix},\qquad
    C_c^{BC}=\frac1n\sum_{i\in L}q_c^B(X_i)q_c^C(X_i)^\top,\quad B,C\in\{\mathrm L,\mathrm N\}.$
  \State $\displaystyle C_{c,\alpha}^{\mathrm H}=
    \begin{pmatrix}
      (1-\alpha)C_c^{\mathrm{LL}} & \sqrt{\alpha(1-\alpha)}\,C_c^{\mathrm{LN}}\\[2pt]
      \sqrt{\alpha(1-\alpha)}\,C_c^{\mathrm{NL}} & \alpha C_c^{\mathrm{NN}}
    \end{pmatrix},\qquad
    \widehat m_{c,\alpha}^{\mathrm{RH}}=\mathsf{Fit}_L(\widehat\Phi_{c,\alpha}^{\mathrm H};\xi_c,\varepsilon_{\mathrm{opt}}).$
\EndFor
\State $\displaystyle (\widehat c_H,\widehat\alpha_H)\in
  \underset{(c,\alpha)\in\mathcal C\times\mathcal A_H}{\arg\min}\,\mathcal V(\widehat m_{c,\alpha}^{\mathrm{RH}}),\qquad
  \widehat m^{\mathrm{RH}}=\widehat m_{\widehat c_H,\widehat\alpha_H}^{\mathrm{RH}}.$
\Statex \textbf{3. PH and TH: select a prediction weight between the already selected endpoints.}
\State For prediction maps $u,v:\mathcal X\to\mathcal Y$, define a measurable two-point interpolation $\mathcal B_\eta[u,v]$ by
\Statex $\displaystyle\hspace{\algorithmicindent}
  \Gamma_\eta(y,x;u,v)=(1-\eta)d_{\mathcal Y}^2(y,u(x))+\eta d_{\mathcal Y}^2(y,v(x)),$
\Statex $\displaystyle\hspace{\algorithmicindent}
  \Gamma_\eta(\mathcal B_\eta[u,v](x),x;u,v)
  \le\inf_{y\in\mathcal Y}\Gamma_\eta(y,x;u,v)+\varepsilon_{\mathrm{hyb}},\qquad
  \mathcal B_0[u,v]=u,\quad\mathcal B_1[u,v]=v.$
\State $\displaystyle \widehat m^{\mathrm{PH}}_\eta=\mathcal B_\eta[\widehat m^{\mathrm L},\widehat m^{\mathrm N}],\qquad
  \widehat\eta_{\mathrm{PH}}\in\underset{\eta\in\mathcal E}{\arg\min}\,\mathcal V(\widehat m^{\mathrm{PH}}_\eta),\qquad
  \widehat m^{\mathrm{PH}}=\widehat m^{\mathrm{PH}}_{\widehat\eta_{\mathrm{PH}}}.$
\State $\displaystyle \widehat m^{\mathrm{TH}}_\eta=\mathcal B_\eta[\widehat m^{\mathrm L},\widehat m^{\mathrm{RH}}],\qquad
  \widehat\eta_{\mathrm{TH}}\in\underset{\eta\in\mathcal E}{\arg\min}\,\mathcal V(\widehat m^{\mathrm{TH}}_\eta),\qquad
  \widehat m^{\mathrm{TH}}=\widehat m^{\mathrm{TH}}_{\widehat\eta_{\mathrm{TH}}}.$
\State \Return $\widehat m^{\mathrm{RH}},\widehat m^{\mathrm{PH}},\widehat m^{\mathrm{TH}}$, retaining their training-sample fits.
\end{algorithmic}
\vspace{3pt}
\noindent\footnotesize
Each $c$ specifies a common kernel family, bandwidth multiplier, requested rank, and ridge value for the two feature blocks;
representation-specific bandwidth scales and retained ranks may differ. Validation covariates do not enter either spectral fit.
In a Hadamard space, $\mathcal B_\eta[u,v](x)=\gamma_{u(x),v(x)}(\eta)$ is the exact geodesic interpolation, with
$\varepsilon_{\mathrm{hyb}}=0$. In general metric spaces, use the displayed near-minimizer definition.
\end{algorithm*}

\clearpage

\section{Experiment Details}
\label{app:experiment-details}

This section specifies the experiments in Section~\ref{sec:experiment}.
Within each simulation covariate type, the generating law is identical
across the three response spaces. Each setting uses 200 generated data
samples. HBN-A and HBN-B use 200 common random partitions without
replacement of the same pooled R1--R4 cohort.
The main tables compare every applicable method at the common maximal
unlabeled budget: \(U=2000\) in simulations and \(U=1314\) in HBN;
supervised methods use \(U=0\). These are common-budget summaries of
the existing runs, not a claim of a preregistered analysis. Descriptive
test-best summaries are retained separately below.

\subsection{Simulation Protocol}
\label{app:simulation-design}

Settings S1--S3 have Euclidean covariates and settings S4--S6 have
distribution-valued covariates. Within each covariate type, the responses
are ordered as \(3\times3\) SPD matrices, univariate distributions,
and five-taxon BHV trees. Each setting uses 200 data draws, indexed 0--199,
from the generating laws below.

For every replication,
\[
 (n_{\rm train},n_{\rm val},n_{\rm test})=(150,50,500),
 \qquad U\in\{0,500,1000,2000\}.
\]
The labeled blocks are independent draws. A separate covariate sample of
size 2,000 supplies nested unlabeled prefixes, independently of the labeled
blocks. All methods share these data within each replication. Euclidean
covariates retain their generated coordinates. For distribution covariates,
each quantile coordinate is centered and divided by its training-sample
standard deviation (divisor \(n_{\rm train}\)); standard deviations below
\(10^{-10}\) are replaced by one. These training transformations are applied
to validation, test, and unlabeled covariates. The latent parameters below
are used only to generate data; fitted methods receive the observed
covariates.

\begin{table}[!htbp]
\centering\scriptsize
\setlength{\tabcolsep}{2pt}
\caption{Simulation settings.}
\begin{tabular*}{\textwidth}{@{\extracolsep{\fill}}lll@{}}
\toprule
Setting & Covariate & Response metric \\
\midrule
S1 & Euclidean, 20 coordinates & log-Euclidean on \(3\times3\) SPD \\
S2 & Euclidean, 20 coordinates & \(W_2\) on distributions \\
S3 & Euclidean, 20 coordinates & intrinsic BHV tree distance \\
S4 & distribution, 100 quantile coordinates & log-Euclidean on \(3\times3\) SPD \\
S5 & distribution, 100 quantile coordinates & \(W_2\) on distributions \\
S6 & distribution, 100 quantile coordinates & intrinsic BHV tree distance \\
\bottomrule
\end{tabular*}
\label{tab:app-sim-settings}
\end{table}

\subsubsection{Covariate Generating Laws}

Write \(\sigma(v)=(1+e^{-v})^{-1}\), \(v_+=\max(v,0)\),
\(s(v)=\log(1+e^v)\), and \(\widetilde E=3\tanh(E/3)\) for
\(E\sim N(0,1)\). Noise arrays have independent entries before any
specified symmetrization. Define
\(\operatorname{wrap}(v)=(v+\pi)\bmod(2\pi)-\pi\).

\paragraph{Euclidean covariates (S1--S3).}
Draw \(C\in\{1,2,3\}\) with probabilities \((0.4,0.4,0.2)\) and centers
\((c_1,c_2,c_3)=(-2.3,-0.3,1.6)\). With independent mixture indicators,
\[
 \theta=\begin{cases}
 \operatorname{wrap}(c_C+0.6\widetilde E_\theta),&\text{probability }0.8,\\
 \operatorname{Unif}[-\pi,\pi],&\text{probability }0.2,
 \end{cases}
 \qquad
 u=\begin{cases}
 0.6\cos\theta+0.4\widetilde E_u,&\text{probability }0.8,\\
 \operatorname{Unif}[-3,3],&\text{probability }0.2.
 \end{cases}
\]
Set \(z_1=(\theta+\pi)/(2\pi)\), \(z_2=\sigma(u)\), and draw
\(z_3,z_4,z_5,W_1,\ldots,W_{15}\) independently from
\(\operatorname{Unif}[0,1]\), independently of \((\theta,u)\). Then
\[
 X=(z_1,z_2,z_3,z_4,z_5,W_1,\ldots,W_{15})\in[0,1]^{20}.
\]
Below, \((a,b,c,d,e)=2(z_1,z_2,z_3,z_4,z_5)-(1,1,1,1,1)\).

\paragraph{Distribution covariates (S4--S6).}
Draw \(z\sim\operatorname{Unif}([0,1]^3)\), put \(u_j=2z_j-1\), and set
\[
 \vartheta_1=2u_1+u_2,\qquad
 \vartheta_2=u_2u_3+u_1,\qquad
 \vartheta_3=(u_1+u_3)_+-u_2.
\]
For
\[
 \phi_1(t)=\sqrt2\cos(2\pi t),\quad
 \phi_2(t)=\sqrt2\sin(2\pi t),\quad
 \phi_3(t)=\sqrt2\cos(4\pi t),
\]
the observed covariate is the distribution with density
\[
 f_X(t\mid z)=
 \frac{\exp(\sum_{k=1}^3\vartheta_k\phi_k(t))}
 {\int_0^1\exp(\sum_{k=1}^3\vartheta_k\phi_k(v))\,dv},
 \qquad 0\le t\le1.
\]
It is represented by
\(Q_X=100^{-1/2}(F_X^{-1}((j-1/2)/100))_{j=1}^{100}\).
The implementation evaluates the density on \(t_k=k/500\),
\(k=0,\ldots,500\), uses trapezoidal integration for normalization and the
CDF, and linearly interpolates the inverse CDF. Thus the representation is a
numerical quantile approximation to the generated distribution.
Observation noise enters through the responses.

\subsubsection{Response Generating Equations}

For S1, define
\[
\begin{gathered}
 h_1=a_+,\quad h_2=b_+,\quad h_3=c_+,\quad
 h_4=(d-e)_+,\quad h_5=(a+b-c)_+,\\
 g_1=2h_1+h_4-h_3,\qquad
 g_2=2h_2+h_5-h_1,\qquad
 g_3=2h_3+h_1h_2+h_4-h_5.
\end{gathered}
\]
For S2, use
\[
 g_1=2ab+cd-e,\qquad g_2=|2ac-bd|+\tfrac12e^2,
 \qquad g_3=ad+be-2c.
\]
For S3, use
\[
 g_1=a_++\tfrac12 d,\qquad
 g_2=b_++\tfrac12 e,\qquad
 g_3=c_+.
\]
For S5, use
\[
 g_1=2\vartheta_1\vartheta_2+\vartheta_3,\qquad
 g_2=\vartheta_1+\vartheta_2-\tfrac12\vartheta_3,\qquad
 g_3=\vartheta_2\vartheta_3+\vartheta_1.
\]

\paragraph{SPD responses (S1, S4).}
For S1, the response is \(Y=\exp(H+\sigma_S\widetilde E_{\rm sym}/\sqrt3)\),
where \(\widetilde E_{\rm sym}=(\widetilde E+\widetilde E^\top)/2\).
For S1,
\[
 H=\begin{pmatrix}
 g_1+g_2&g_1-g_3&g_2+g_3\\
 g_1-g_3&-g_1+g_2&g_1+g_2\\
 g_2+g_3&g_1+g_2&g_3
 \end{pmatrix},\qquad \sigma_S=0.207.
\]
Computations use \(\operatorname{svec}(\log Y)\), with diagonal entries
unchanged and off-diagonal entries multiplied by \(\sqrt2\), preserving
the Frobenius norm.

For S4, an independent input-only reference of size \(M=2048\), generated
with seed 2026092908, defines centered Gaussian Nystr\"om eigenfunctions
at two scales. Reference quantile coordinates are standardized by their own
means and standard deviations (divisor \(M\)); scales at most \(10^{-10}\)
are replaced by one. Bandwidths are \(b\) times the median positive
reference pairwise distance, with \(b\in\{0.5,1.5\}\). For the centered
reference Gram eigenpairs \((\nu_{b,j},v_{b,j})\), ordered by descending
positive eigenvalue, define
\[
 \varphi_{b,j}(X)=\frac{\sqrt M}{\nu_{b,j}}k_{b,0}(X)^\top v_{b,j}.
\]
Here \(k_{b,0}\) uses reference centering, and the largest-magnitude entry
of each eigenvector is made positive. Retain 32 modes per scale. With seed
2026092909, independent Gaussian matrices of sizes \(6\times6\),
\(8\times6\), and \(12\times6\), drawn in that order, yield QR factors
\(Q_6,Q_8,Q_{12}\) with orthonormal columns; put \(C_j=Q_j/\sqrt6\).
Writing eigenfunction vectors as rows, the six-coordinate signal is
\[
\begin{aligned}
 v(X)&=\sqrt{0.3}\,(\varphi_{1.5,1}(X),\ldots,\varphi_{1.5,6}(X))C_6\\
 &\quad+\sqrt{0.7}\,(\varphi_{0.5,9}(X),\ldots,\varphi_{0.5,20}(X))C_{12},\\
 f(X)&=\{v(X)-\bar v_{\rm ref}\}/a_{\rm ref},\qquad
 a_{\rm ref}^2=M^{-1}\sum_{i=1}^M\|v(X_i^{\rm ref})-\bar v_{\rm ref}\|^2,\\
 Y&=\exp\!\left[\operatorname{svec}^{-1}
       \{f(X)+0.12Z/\sqrt6\}\right],\qquad Z\sim N(0,I_6).
\end{aligned}
\]
The noise is independent of \(X\), so
\(\exp\{\operatorname{svec}^{-1}f(X)\}\) is its conditional log-Euclidean
Fr\'echet mean. Unit reference signal energy is an empirical identity,
not population normalization. The reference and its eigenfunctions define
only the generating law and are never learner inputs or extra unlabeled
observations. Learners use training-only standardization. Input/noise streams
use NumPy SeedSequence \([2026092907,r,b,q]\), with \(r=0,\ldots,199\),
block \(b=0,1,2,3\) for training/validation/test/unlabeled and
\(q=0,1\) for input/noise. Unlabeled responses are not used.

\paragraph{Distribution responses (S2, S5).}
Let \(t_j=(j-1/2)/m\) and
\(q_j=m^{-1/2}F_{\operatorname{Beta}(A,B)}^{-1}(t_j)\). The parameters are
\[
\begin{array}{c|c|c|c|c}
 &m&A&B&\sigma_D\\ \hline
 S2&60&1.2+5.2\sigma(g_1+g_2)&1.2+5.0\sigma(-g_2+g_3)&0.07075\\
 S5&100&1+2\sigma(2(g_1+g_2))&1+2\sigma(2(-g_2+g_3))&0.020+0.020\sigma(g_3)
\end{array}
\]
Let \(P_{\rm mono}\) be ordinary equal-weight isotonic regression,
computed by the pool-adjacent-violators algorithm (PAVA), and let
\(\operatorname{clip}_{a,b}\) act coordinatewise. Define the implemented
correction
\[
 \mathcal C_{a,b}(v)
 :=\operatorname{clip}_{a,b}
       \{P_{\rm mono}(\operatorname{clip}_{a,b}(v))\}.
\]
The generated response is
\(\mathcal C_{0,m^{-1/2}}(q+\sigma_D\widetilde E/\sqrt m)\).
It belongs to
\(\mathcal Q_m=\{v:0\le v_1\le\cdots\le v_m\le m^{-1/2}\}\).
This clip-then-PAVA generating operation is not generally the Euclidean
projection of the uncorrected vector
onto \(\mathcal Q_m\).

\paragraph{BHV response space (S3, S6).}
We use unrooted trees on the five labeled taxa \(\{0,1,2,3,4\}\),
with all five pendant edges fixed at length one. Each nontrivial internal
split is identified by its two-element side \(ab\), representing
\(\{a,b\}\mid\{0,1,2,3,4\}\setminus\{a,b\}\).
There are ten such splits. Two are compatible exactly when their
two-element sides are disjoint; each of the 15 compatible pairs defines
an orthant \(\mathbb R_+^2\). Orthants are glued along equal-split boundary
rays, and their common origin is the unresolved star tree. This is the
two-dimensional internal-edge BHV space
\citep{billera2001geometry}. Its CAT(0) geometry gives unique geodesics
\citep{bridson1999metric} and unique probability-weighted Fr\'echet means
\citep{sturm2003probability}; these facts alone do not imply uniqueness
for signed regression weights.

The link of the origin is the Petersen graph: vertices are the ten splits,
and every compatible-pair edge has angular length \(\pi/2\).
For internal-edge radius \(r\) and link position \(q\),
\[
 d_{\rm BHV}^2((r,q),(r',q'))
 =r^2+(r')^2-2rr'\cos\{\min(\pi,d_{\rm link}(q,q'))\}.
\]
The origin has radius zero; link positions there are immaterial.
The ten-coordinate vector of split lengths is only a storage representation,
not a replacement Euclidean metric. Figure~\ref{fig:bhv-geometry} shows the
link and observed tree examples.

\begin{figure}[!htbp]
\centering
\includegraphics[width=\textwidth]{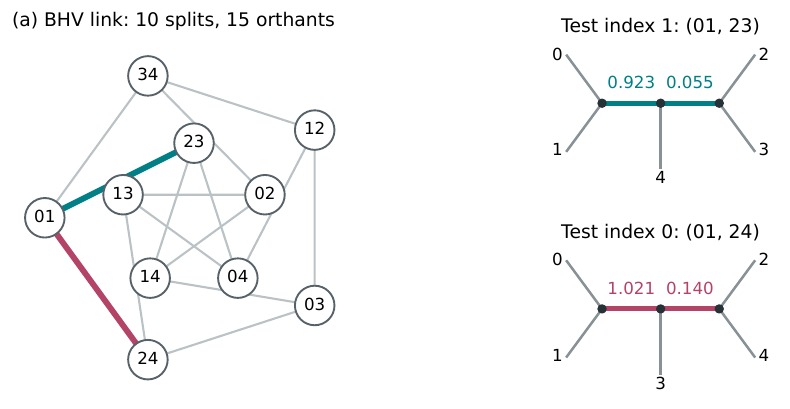}
\caption{Five-taxon BHV geometry and actual S6 tree examples.
Left: the link has ten split vertices and 15 orthant edges; the two S6
orthants are highlighted. Every link edge has intrinsic length \(\pi/2\),
irrespective of this schematic drawing. Right: the first test observation
of each S6 topology in replication 0, with its two internal lengths labeled.
Pendant lengths are one; tree layouts are schematic, not drawn to scale.}
\label{fig:bhv-geometry}
\end{figure}

\paragraph{Euclidean-to-BHV signal (S3).}
Write \(G(T,T',v)\) for the constant-speed intrinsic BHV geodesic at
fraction \(v\in[0,1]\). Let the four anchor trees \(A,B,C,D\) have split
pairs \((01,23),(01,24),(02,34),(04,12)\), respectively, with internal
lengths \((2,1)\) in the listed order. With the S3 scores defined above,
\[
\begin{aligned}
 T_a(X)&=G(A,B,\sigma(2g_1)),\qquad
 T_b(X)=G(C,D,\sigma(2g_2)),\\
 T_0(X)&=\{1+\tfrac12s(g_1+g_2)\}
           G(T_a(X),T_b(X),\sigma(2g_3)).
\end{aligned}
\]
Scalar multiplication rescales internal lengths only. The construction
can cross orthant boundaries, but does not assert positive mass on every
binary topology.

\paragraph{Distribution-to-BHV signal (S6).}
The observed input is the same native log-density distribution as in
S4 and S5, not an encoding of Euclidean covariates.
A fixed independent input-only reference of size \(M=2048\), generated
with seed 2026092241, defines 32 raw centered Nystr\"om eigenfunctions.
The reference quantile coordinates are standardized by their own means
and standard deviations (divisor \(M\)); scales at most \(10^{-10}\)
are replaced by one. The Gaussian bandwidth is 0.75 times the reference
median positive pairwise distance. Let \(K_0\) be its centered Gram matrix,
with descending positive eigenpairs \((\nu_j,v_j)\). The sign of each
\(v_j\) makes its largest-magnitude entry positive. For the correspondingly
reference-centered query kernel vector \(k_0(X)\), set
\[
 \varphi_j(X)=\frac{\sqrt M}{\nu_j}k_0(X)^\top v_j,\qquad
 z(X)=\tfrac12\varphi_1(X)+\sqrt{\frac{0.75}{8}}
                           \sum_{j=21}^{28}\varphi_j(X).
\]
These are raw eigenfunctions, without a \(\tanh\) transformation.
Put \(t(X)=0.6z(X)\) and \(h=0.75\). The signal tree has
\[
 \ell_{01}=h,\qquad \ell_{23}=t(X)_+,\qquad
 \ell_{24}=(-t(X))_+,
\]
with every other internal split length zero. Thus only one of the
incompatible splits 23 and 24 is positive. The signal lies on a
two-topology geodesic, and the observed responses below occupy a
two-orthant subset of BHV. The 75\% high-mode energy allocation is a
method-aligned signal specification, not a general distribution-to-BHV model.
On the fixed reference, \(M^{-1}\sum z^2=1\),
\(M^{-1}\sum t^2=0.36\), and
\(M^{-1}\sum(h^2+t^2)=0.9225\); these are empirical reference identities,
not population normalization or equal noise levels across settings.
Reference observations, eigenfunctions, and latent scores are never
provided to fitted methods as inputs or extra unlabeled observations.
Learners use the training-only standardization specified above.

\paragraph{BHV response noise and reproducibility.}
In both S3 and S6, for the two internal lengths of \(T_0(X)\), generate
\[
 \ell_j(Y)=\ell_j(T_0(X))\{1+0.5\tanh(E_j)\},\qquad
 E_1,E_2\stackrel{\rm iid}{\sim}N(0,1).
\]
The topology is unchanged by this noise. Conditional on \(X\), its
support lies in one convex orthant, and the symmetric multiplier has
mean one; hence \(T_0(X)\) is the conditional intrinsic Fr\'echet mean.
Input and response-noise streams use NumPy SeedSequence
\([2026092105,d,r,b,q]\): \(d=0,1\) for
Euclidean/distribution inputs, \(r=0,\ldots,199\),
\(b=0,1,2,3\) for training/validation/test/unlabeled blocks, and
\(q=0,1\) for input/noise. Unlabeled responses are not used.

\subsubsection{Prediction Rules and Comparators}

All methods use the same response rule for a given response.
For weights \(w_i(x)\) summing to one, log-Euclidean SPD predictions are
\(\exp(\sum_iw_i(x)\log Y_i)\). Distribution predictions are
\[
 \Pi_{[a,b]}^\uparrow\!\left(\sum_iw_i(x)q_i\right)
 =\operatorname{clip}_{[a,b]}
   \left\{\operatorname{PAVA}\!\left(\sum_iw_i(x)q_i\right)\right\}.
\]
Equal-weight PAVA is applied to the raw weighted vector before clipping.
This is the exact Euclidean projection onto the closed convex set of
bounded nondecreasing quantile vectors,
and hence the finite-grid Fr\'echet minimizer for unit-sum signed weights.
Bounds are \([0,m^{-1/2}]\) for S2/S5 and
\([1/\sqrt{80},45/\sqrt{80}]\) for HBN. This prediction rule is distinct
from the clipping/PAVA operation in the data-generating laws.

For BHV responses, the implementation searches all 15 orthants, including
boundary rays and the origin, even when the data occupy fewer orthants.
For signed weights with \(W=\sum_i w_i>0\), the objective in an orthant is
\[
 F(r,\theta)=Wr^2-2rS(\theta)+C,\qquad
 S(\theta)=\sum_i w_i r_i\cos\{\min(\pi,d_{\rm link}(\theta,q_i))\}.
\]
Its radial minimum is at \(r=(S(\theta)/W)_+\).
Between shortest-route changes and \(\pi\)-cutoffs, \(S(\theta)\) is
\(A\cos\theta+B\sin\theta+D\). The solver checks every segment endpoint
and interior maximum, comparing all orthants and the origin. This gives
a global minimum up to floating-point precision, with fixed
orthant/angle traversal resolving ties; it is not Euclidean split-vector
averaging or a positive-weight-only barycenter algorithm.
The specialized BHV5 metric and geodesics were checked against the
Owen--Provan implementation on all 225 orthant pairs (675 random pairs)
\citep{owen2011computing}; signed minimization was checked independently
by numerical search. These checks do not replace a generic numerical
error bound for arbitrary tree spaces.

RH uses concatenated supervised and pooled spectral features weighted by
\(\sqrt{1-\alphaH}\) and \(\sqrt{\alphaH}\), respectively. PH interpolates
the KRFR and SS-KRFR predictions, and TH interpolates KRFR and RH. The
interpolation is along response-space geodesics: log-matrix interpolation,
quantile interpolation, and intrinsic BHV geodesics.
Validation selects \(\alphaH\) and the prediction
mixing weight \(\eta\). Each endpoint is first tuned by
validation; PH and TH then search only their mixing weight between the
selected, already projected endpoint predictions. RH instead searches
its concatenation candidates and \(\alphaH\) together. These sequential
prediction-hybrid searches are not joint searches over every pair of
endpoint hyperparameters. The same validation set selects endpoint
hyperparameters and prediction weights; there is no additional split
between these stages. The common fixed-tuning guarantees of
Theorem~\ref{thm:th-stability} and Corollary~\ref{cor:ph-balanced-rate}
do not cover this full selection procedure.

External comparators are WCE-NLGFR \citep{Bhattacharjee2025},
Graph SS-FR with NW/kNN prediction \citep{qiu2024semi}, NW-FR/kNN-FR,
and random-forest weighted local constant (RF-LC) and local linear
(RF-LL) Fr\'echet regression \citep{qiu2024random}.
LL-FR and Global FR \citep{petersen2019frechet} are included only for
Euclidean covariates (S1--S3), using a common implementation across those
settings. RF-LL is also restricted to S1--S3; RF-LC uses the observed
input coordinates in all eight settings, including quantile coordinates
for distribution inputs. The numerical, graph and forest conventions below are
part of the evaluated procedures, not claims of identity with author software.

\subsubsection{Kernel and Tuning Candidates}
\label{app:tuning-grid}

Proposed-method candidates, forest candidates and WCE regularization are scored by validation
squared response-space loss; local and graph
smoothing parameters use training LOOCV. Global FR has no statistical
tuning parameter. Supervised fits are reused across \(U\), whereas pooled
representations are fitted at each \(U\). In HBN-B, distance-based methods'
input-mixture weights and the forest's input blocks are selected by validation.
No test response enters within-replication parameter selection.

For proposed simulation spectral fits, validation ties within \(10^{-10}\)
prefer smaller requested rank, compact cosine/Gaussian/Laplacian/Mat\'ern
kernel order, bandwidth multiplier closer to one, then larger \(\xi\).
RH first prefers \(\alphaH\) closer to one. The HBN backend uses effective
rank, lexical kernel order, bandwidth multiplier closer to one, then
smaller \(\xi\); its RH key is
\((|\alphaH-1|,s_{\rm eff},\text{kernel},\xi)\).
PH/TH prefer smaller \(\eta\). Remaining ties retain candidate
order. WCE validation and external LOOCV ties retain the first grid candidate.

The proposed spectral kernels are Gaussian, Laplacian, Mat\'ern-\(3/2\),
and compact cosine in simulations; HBN uses the first three.
For distance \(r=\|x-x'\|\) in the specified input representation,
\[
 K_G(r;h)=e^{-r^2/(2h^2)},\quad K_L(r;h)=e^{-r/h},\quad
 K_M(r;h)=(1+\sqrt3r/h)e^{-\sqrt3r/h}.
\]
The bandwidth-free compact-cosine candidate is
\[
 K_C(x,x')=\frac1p\sum_{j=1}^p\sum_{\ell=1}^4
 a_\ell\cos(\pi\ell c(x_j))\cos(\pi\ell c(x'_j)),
 \quad (a_1,a_2,a_3,a_4)=(1,0.5,0.3,0.2),
\]
where \(c(v)=\min(1,\max(0,v))\).
For radial spectral kernels, \(h=c_{\rm bw}\) times the median positive
pairwise distance within the representation sample: training covariates
for KRFR and the applicable training or pooled sample for semi-supervised
representations. WCE has the separate bandwidth rule below.

The proposed-method grids are
\[
\begin{aligned}
 &\mathcal B=\{0.5,0.75,1,1.5,2\},\qquad
 \mathcal S=\{8,16,32,40,48\},\\
 &\mathcal L=\{10^{-4+j/2}:j=0,\ldots,8\}\cup\{0.05,0.5\},\\
 &\mathcal A=\{0,0.1,0.25,0.5,0.75,0.9,1\}.
\end{aligned}
\]
Spectral models use \(c_{\rm bw}\in\mathcal B\), \(s\in\mathcal S\),
and \(\xi>0\) with \(n\xi\in\mathcal L\), where \(n\) is the labeled
fitting size, not the pooled size. Mixing weights
\(\alphaH,\eta\) belong to \(\mathcal A\).
Rank is capped by the available positive empirical eigenvalues.
For \(D=[\mathbf1,F]\), the labeled feature fit solves
\[
 \{D^\top D+n\xi\operatorname{diag}(0,I)\}\widehat\beta=D^\top Z.
\]
The intercept is unpenalized, as in Section~\ref{app:full-procedures}.
WCE's regularization and the numerical stabilization of local/global FR
are distinct from this feature-ridge grid.

\paragraph{Feature construction and numerical conventions.}
For KRFR, the representation sample contains training covariates only;
for pure SS-KRFR it contains training plus the current unlabeled prefix.
Validation and test observations are not representation rows.
Using the centered Gram eigenpairs \((\mu_j,v_j)\), query features are
\(k_0(x)^\top v_j/\sqrt{\mu_j}\).
Since \(\mu_j=N\widehat\lambda_j\) for a pooled sample, this equals
the RKHS normalization in Section~\ref{sec:SS-KRFR}, with no additional
sample-size rescaling. RH uses two separately centered maps with
the same candidate kernel, multiplier and nominal rank, and the centered
labeled covariance includes their cross-block terms.

The stored backends retain Gram eigenvalues greater than \(10^{-10}\)
and cap each requested rank by the retained positive rank.
The HBN backend also divides kernels by their mean diagonal;
the radial candidates here have diagonal one in exact arithmetic.
If no eigenvalue is retained, the numerical backend installs a
one-coordinate fallback rather than applying the asymptotic complete-block
convention. Iterative eigensolvers use tolerance \(10^{-8}\) with a dense
symmetric fallback; solves use the backend's direct/pseudoinverse convention.
These safeguards and validation-selected ranks are computational rules,
not verification of the asymptotic assumptions.

\paragraph{Theory scope.}
The main bounds use a fixed kernel and tuning, shared-law
sampling, and an exact or controlled approximate response minimizer.
Training-standardized coordinates, empirical bandwidths, finite-rank
compact-cosine candidates and sequential validation are not automatically
covered by those bounds. Exact finite-grid distribution projection does
not establish the remaining assumptions. The finite-sample rank grids do
not establish the stated asymptotic tuning regimes. HBN evaluates random partitions
of one fixed cohort, with a separate release-based unlabeled pool;
it is not an iid random-labeling population experiment.

\paragraph{WCE-NLGFR.}
For the Gaussian training Gram matrix \(K\), put
\(Q=I-\mathbf1\mathbf1^\top/n\) and \(G=QKQ\). The RKHS
training-feature centering and two resolvents in
\citet{Bhattacharjee2025}, Sections 4.2--4.3, give
\[
\begin{aligned}
 d_i(x)&=k(X_i,x)-n^{-1}\sum_j k(X_i,X_j),\\
 c(x)&=Q(K+\epsilon I)^{-1}d(x),\\
 w(x)&=n^{-1}\mathbf1+G(G+\epsilon I)^{-1}c(x).
\end{aligned}
\]
The Gaussian squared bandwidth is
\(\binom{n}{2}^{-1}\sum_{i<j}\|X_i-X_j\|^2\), including zero distances.
There is no spectral truncation. The grid is
\(\epsilon\in\{10^{-6},10^{-5},\ldots,10^{-1}\}\) in simulations and
\(\{10^{-5},10^{-4},\ldots,10^{-1}\}\) in HBN. Constant training inputs
use the response mean. In \citet{Bhattacharjee2025}, equation (14) selects
regularization by the training GCV criterion
\[
 \frac{n^{-1}\sum_i d_\cY^2(Y_i,\widehat m_\epsilon(X_i))}
 {[1-\{\operatorname{tr}(QG(G+\epsilon I)^{-1})+1\}/n]^2}.
\]
Our evaluated variant retains the weights above but instead selects
\[
 \widehat\epsilon=\underset{\epsilon\in\mathcal E}{\arg\min}\;
 \frac1{n_{\rm val}}\sum_{i\in I_{\rm val}}
 d_\cY^2\{Y_i,\widehat m_\epsilon(X_i)\},
\]
where each \(\widehat m_\epsilon\) is fitted only on training observations
and uses the specified response minimizer, including bounded isotonic
projection for distributions. Exact ties retain the first grid value;
invalid candidates are not eligible and no valid candidate means failure.
There is no training-plus-validation refit, bandwidth change, spectral
truncation, grid expansion, added null candidate or weight renormalization.
For HBN-B, this selection occurs within each of the seven \(\omega\)
candidates, followed by validation-only selection of \(\omega\).
The variant uses the same held-out prediction criterion as the proposed
methods. GCV instead scores fitted training responses with a trace penalty;
the printed trace is generally not the exact finite-\(\epsilon\) hat trace
of the two-resolvent weights. These distinctions motivate checking tuning
sensitivity, not a claim that the source estimator or its GCV is invalid.
Unless explicitly marked GCV, WCE-NLGFR in all performance tables denotes
this validation-tuned variant, not an exact reproduction of the source
paper's tuning procedure. All other comparator tuning rules remain unchanged.

\paragraph{Local smoothers.}
NW-FR uses the compact Epanechnikov weights
\(w_i(x)\propto(1-\|X_i-x\|^2/h^2)_+\).
Let \(h_0\) be the median distance from each labeled training input to its
nearest other labeled input. NW searches \(h=5^{j/10}h_0\),
\(j=-10,\ldots,20\), giving 31 bandwidths from \(0.2h_0\) to \(25h_0\).
kNN-FR uses uniform neighbor weights with
\[
 \mathcal K_n=\{1,\ldots,10,15,20,30,40,50,75,100,
                 \lfloor n/2\rfloor,n-1\}\cap\{1,\ldots,n-1\}.
\]
After removing duplicates, this gives 18 candidates for \(n=150\)
and 19 for \(n=284\), in increasing order.
Training LOOCV excludes the held-out label; preprocessing remains fixed.
Distance ties use stable training-index order. If a positive finite
bandwidth gives all-zero weights, all eligible labels receive uniform
weight, selecting their Fr\'echet mean. This explicit empty-support
convention is recorded as \texttt{ZERO\_WEIGHT\_TIE}; it does not require
duplicate observations or bootstrap resampling. A zero or nonfinite
bandwidth is instead an invalid candidate.

For S1--S3, Global FR uses
\[
 w_i(x)=n^{-1}\{1+(X_i-\bar X)^\top\widehat\Sigma^{-1}(x-\bar X)\},
 \qquad \widehat\Sigma=n^{-1}\sum_i(X_i-\bar X)(X_i-\bar X)^\top.
\]
LL-FR uses the intercept weights of radial Epanechnikov local weighted
least squares. Its bandwidth is \(h=c\) times the median positive
pairwise training distance, with
\(c\in\{0.25,0.35,0.5,0.75,1,1.5,2,3,4,6,8\}\), selected by training
LOOCV. This multivariate radial application and bandwidth grid are
specified conventions, rather than the scalar local-regression setting
alone in \citet{petersen2019frechet}.

Both covariance solves first attempt unregularized pivoted QR. Full-rank
solutions with relative residual at most \(10^{-8}\) are accepted.
Otherwise a slope-only ridge uses
\[
 \delta=\max\{0,(b-10^{10}a)/(10^{10}-1)\},
\]
rounded upward, where \(a,b\) are extreme covariance eigenvalues and
\(a=0\) under structural rank deficiency. Acceptance requires
\(\delta/b\le10^{-8}\) and stabilized relative residual at most \(10^{-5}\).
The full ridge resolvent retains all null directions; no rank projection
substitutes for it. Structural and numerical cases are recorded as
\texttt{STRUCTURAL\_RIDGE} and \texttt{NUMERICAL\_RIDGE}, not unregularized
fits. Failed checks remain explicit prediction failures.

\paragraph{Graph smoothers and boundary conventions.}
Synthetic Graph SS-FR uses pooled training and current unlabeled inputs
with radius
\[
 \rho=1.2\max_i\min_{j\ne i}\|X_i-X_j\|.
\]
Edges carry input distances. A query attaches individually to vertices
within \(\rho\); shortest paths give distances to labeled vertices.
For HBN, the fitting graph is the smallest connected symmetric union-kNN
graph of degree at least four, chosen from inputs alone. Each query
attaches to its four nearest fitting vertices. This connectivity rule is
an explicit adaptation of the graph construction in \citet{qiu2024semi};
the four-neighbor query rule follows their Remark 2.2.

Graph NW uses the same 31 Epanechnikov bandwidths, with \(h_0\) now the
median nearest-other labeled shortest-path distance. Graph kNN uses
\(k\in\mathcal K_n\). LOOCV retains all fitting covariate vertices and hides
only the held-out response; graph and input scales are not refitted.
Validation and test observations never supply fitting vertices or labels
for other queries. Graph NW uses the same all-zero-weight convention
as NW-FR.

In synthetic Graph kNN, when fewer than the selected \(k\) labels are
reachable, prediction uses the ambient nearest eligible labeled training
response, with stable index ties and self-label exclusion in LOOCV.
The selected \(k\) is retained if the original graph LOOCV has a valid
candidate. Only when all candidates lack valid LOOCV coverage are
they evaluated with this completion rule. This labeled-training 1NN
completion is not a response borrowed from another test observation.
It is recorded as \texttt{AMBIENT\_TRAIN\_1NN\_FALLBACK}.
HBN's connected-graph adaptation is reported separately.

\paragraph{Random-forest weighted Fr\'echet regression.}
RF-LC and RF-LL use the response-aware tree construction in
\citet[Appendix C.1]{qiu2024random}, with 50 trees fixed in every fit.
Each tree draws \(n\) training indices with replacement and uses their
distinct in-bag observations. At each node, a sampled set of coordinates
supplies one scalar two-means split per coordinate; the split minimizing
the sum of child Fr\'echet sums of squares is chosen, provided each child
contains at least five distinct observations. Splitting stops without a
relative improvement above \(10^{-12}\). Our scalar two-means solver finds
the global optimum rather than using randomized R \texttt{kmeans}; this
and the minimum-child constraint are specified implementation conventions.

Validation searches depths \(d\in\{3,5,7\}\) and numbers of candidate
coordinates \(m_{\rm try}\in\{\lceil p/3\rceil,\lceil0.8p\rceil\}\),
where \(p\) counts nonconstant training coordinates. Depth candidates
truncate common depth-seven trees. Exact validation ties prefer smaller
depth, smaller \(m_{\rm try}\), then input order. RF-LC averages uniform
within-leaf weights across trees; RF-LL applies the unregularized local
linear formula to these weights, using relative SVD cutoff
\(\sqrt{\epsilon_{\rm machine}}\). Rank-deficient RF-LL candidates fail
explicitly, without added ridge regularization. Both methods use the
response-space minimizers above, with no training-plus-validation refit.

For HBN-B, RF-LC selects among quantile-only, joint, and SPD-only inputs,
corresponding to \(\omega=0,0.5,1\), in that order. Positive rescaling
of input blocks leaves axis-aligned partitions unchanged in exact
arithmetic, so other interior weights do not define distinct candidates.
NumPy PCG64 streams use SeedSequence
\([2026100301,j,r,t,v]\), where \(j=0,\ldots,7\) indexes
S1--S6, HBN-A and HBN-B, \(r=0,\ldots,199\) is the replication,
\(t=0,\ldots,49\) is the tree, and \(v=0\) generates its bootstrap.
Splitting nodes use binary-heap indices \(v\ge1\). Streams are shared
across depth candidates and HBN-B input choices; the data-generation
and cohort-partition seeds are the same as for all other methods.

All methods use a common finite unit-sum acceptance tolerance
\(|\sum_i w_i-1|\le10^{-6}\). Raw signed weights are not renormalized
to conceal numerical errors. Prediction failures invalidate the
corresponding mean rather than being silently omitted.

\subsubsection{Evaluation and Summary Conventions}

For each fitted model, the test metrics are
\[
 L_r=\frac1{n_{\rm test}}\sum_{i\in I_{{\rm test},r}}
 d_\cY^2(Y_i,\widehat m_r(X_i)),\qquad
 R_r^2=1-\frac{\sum_i d_\cY^2(Y_i,\widehat m_r(X_i))}
 {\sum_i d_\cY^2(Y_i,\widehat m_{{\rm null},r})},
\]
where the sums defining \(R_r^2\) run over the test sample and
\(\widehat m_{{\rm null},r}\) is the training-response Fr\'echet mean.
The SPD loss uses squared Frobenius distance between log matrices,
the distribution loss uses scaled-quantile squared Euclidean distance,
and the BHV loss uses \(d_{\rm BHV}^2\).
BHV null predictors use the same intrinsic solver.
The prediction target is evaluated
against noisy test responses, not the noise-free generating signal.

For simulations and HBN splits, reported values are
\(\overline L=R^{-1}\sum_rL_r\) and \(\overline{R^2}=R^{-1}\sum_rR_r^2\),
with \(R=200\). For a replication-level metric \(Z_r\), the Monte Carlo
standard error is explicitly
\[
 \operatorname{MCSE}(\bar Z)=
 \frac{\operatorname{sd}(Z_1,\ldots,Z_R)}{\sqrt R},
\]
using sample standard deviation (divisor \(R-1\)).
The approximate 95\% Monte Carlo precision interval, denoted MC95, is
\(\bar Z\pm1.96\operatorname{MCSE}(\bar Z)\).
By contrast, E95 denotes the linearly interpolated empirical
2.5th--97.5th percentiles of replication losses. E95 describes variability
across replications or splits, not uncertainty in the mean.

For paired loss vectors \(A_r,B_r\), define the percent gain by the
ratio of means
\[
 \widehat G=100(1-\bar A/\bar B),\qquad
 \psi_r=-\frac{100(A_r-\bar A)}{\bar B}
        +\frac{100\bar A(B_r-\bar B)}{\bar B^2}.
\]
The paired delta-method MCSE is
\(\operatorname{sd}(\psi_1,\ldots,\psi_R)/\sqrt R\), in percentage points;
its MC95 interval is \(\widehat G\pm1.96\operatorname{MCSE}(\widehat G)\).
This is not the mean of per-replication percentage gains.
Paired mean loss reductions use the MCSE of \(B_r-A_r\).
Every comparison uses exactly the same 200 replication IDs.

Within each replication, tuning follows the method-specific training or
validation criterion above. Main Tables~\ref{tab:main-object-results}
and~\ref{tab:hbn-main-results} report all applicable methods at the
common maximal budget; neither methods nor \(U\) are selected by test
loss for those tables. MCSE is computed separately for each method's
200 losses. SS-gain always compares the same method at \(U\) and
\(U=0\), using the ratio of mean losses, not averaged per-run ratios.

The appendix tables explicitly labeled descriptive best instead select
\(U\), and where stated a family representative, by lowest aggregate
test loss, with smaller \(U\) resolving exact ties. Their external gains
compare with the lowest mean loss among allowed external method/\(U\)
pairs. These exploratory minima are not validation-selected procedures;
their MC95 intervals are not adjusted for this selection. A \(U=0\)
tie has gain zero and is not an unlabeled-data benefit.
For HBN, MCSE and MC95 quantify random-split Monte Carlo precision
conditional on the fixed cohort, not population confidence intervals
or uncertainty across independent cohorts. Loss magnitudes remain comparable only
within a setting.

\subsubsection{WCE Tuning Sensitivity}
\label{app:wce-tuning-sensitivity}

Table~\ref{tab:app-wce-selection-sensitivity} compares GCV and validation
selection on identical data and replication IDs, with all non-WCE fits
held fixed. The validation criterion is evaluated after the same response
minimizer; neither test responses nor a training-plus-validation refit
enter the choice. HBN-B also selects its input mixture using validation.

For HBN-A, mean WCE loss decreases from 14.0332 to 6.9164 (50.71\%),
and mean \(R^2\) rises from -0.530 to 0.261. HBN-B improves from 7.1461
to 6.8781 (3.75\%), with mean \(R^2\) rising from 0.233 to 0.265.
Validation also improves S1, but is not universally better: its mean loss
is slightly higher in S2 and S5 and increases by 1.43\% and 3.83\% in
S3 and S6. The S2/S5 tuning-difference MC95 intervals include zero.
All settings use the validation variant in the primary comparison, rather
than choosing between GCV and validation by test performance.

The SPD sensitivity row uses a separate distribution-to-SPD signal,
with the S5 scores \(g_1,g_2,g_3\) and
\[
 H_{\rm sens}=\begin{pmatrix}
 g_1+\tfrac12g_2&\tfrac14(g_1-g_3)&\tfrac14(g_2+g_3)\\
 \tfrac14(g_1-g_3)&-g_1+\tfrac12g_2&\tfrac14(g_1+g_2)\\
 \tfrac14(g_2+g_3)&\tfrac14(g_1+g_2)&\tfrac12g_3
 \end{pmatrix}.
\]
Its response is
\(Y=\exp\{H_{\rm sens}+(0.050+0.050\sigma(g_2))
\widetilde E_{\rm sym}/\sqrt3\}\), with the same distribution input law
and labeled sample sizes. There, validation-WCE loss is 0.8395 versus
TH's 0.8576 at \(U=1000\): TH's external gain is -2.15\%
(MCSE 1.80 percentage points; MC95 [-5.67\%, 1.37\%]).
Its SS-gain is nevertheless 12.45\% (MCSE 1.35 percentage points;
MC95 [9.81\%, 15.10\%]). This signal is distinct from multiscale S4;
no paired GCV-versus-validation effect is claimed for S4.

\begin{table}[!htbp]
\centering\scriptsize
\setlength{\tabcolsep}{2pt}
\caption{WCE tuning sensitivity on 200 matched datasets or partitions. Positive reduction favors validation; negative favors GCV. Reductions are percentages; MCSEs are percentage points. S2/S5 loss units: \(10^{-3}\). SPD sensitivity is separate from multiscale S4; S4 has no GCV-versus-validation comparison.}
\begin{tabular*}{\textwidth}{@{\extracolsep{\fill}}lrrrrr@{}}
\toprule
Setting & GCV loss & Validation loss & \(R^2\): GCV / val. & Loss reduction (MCSE) & MC95 \\
\midrule
S1 & 2.5723 & 2.2029 & 0.781 / 0.812 & 14.36 (0.28) & [13.81, 14.91] \\
S2 & 5.6916 & 5.7063 & 0.587 / 0.586 & -0.26 (0.22) & [-0.68, 0.16] \\
S3 & 0.2707 & 0.2746 & 0.483 / 0.476 & -1.43 (0.29) & [-2.01, -0.86] \\
SPD sensitivity & 1.1378 & 0.8395 & 0.912 / 0.935 & 26.21 (1.68) & [22.93, 29.50] \\
S5 & 1.1960 & 1.2007 & 0.903 / 0.903 & -0.39 (0.83) & [-2.02, 1.24] \\
S6 & 0.1702 & 0.1767 & 0.651 / 0.637 & -3.83 (0.77) & [-5.33, -2.32] \\
HBN-A & 14.0332 & 6.9164 & -0.530 / 0.261 & 50.71 (1.75) & [47.28, 54.15] \\
HBN-B & 7.1461 & 6.8781 & 0.233 / 0.265 & 3.75 (0.72) & [2.35, 5.16] \\
\bottomrule
\end{tabular*}
\label{tab:app-wce-selection-sensitivity}
\end{table}

\subsection{Descriptive Cross-Setting Summaries}
\label{app:descriptive-summaries}

Tables~\ref{tab:app-descriptive-sim-summary} and
\ref{tab:app-descriptive-hbn-summary} retain the exploratory family minima.
They do not evaluate a validation-selected method/\(U\) rule and are
separate from the common-budget comparisons in the main text.
Entries report mean test loss / mean \(R^2\), with \(U\) specified for
semi-supervised families. Advantage is the selected proposed method's
percentage loss reduction over the best allowed external method/\(U\)
pair, computed before rounding.

\begin{table}[!htbp]
\centering\scriptsize
\setlength{\tabcolsep}{2pt}
\caption{Descriptive test-best simulation summary; settings are defined in Table~\ref{tab:app-sim-settings}. Family representatives minimize mean test loss over methods/\(U\), including zero. Bold marks the lowest mean loss per setting. Local FR compares NW-FR, kNN-FR and, for Euclidean inputs, LL-FR; Forest FR compares RF-LC and, for Euclidean inputs, RF-LL.}
\begin{tabular*}{\textwidth}{@{\extracolsep{\fill}}lrrrrrrrr@{}}
\toprule
Setting & \shortstack{Ours\\(best $U$)} & \shortstack{Graph SS-FR\\(best $U$)} & \shortstack{WCE-\\NLGFR} & \shortstack{Forest\\FR} & \shortstack{Local\\FR} & \shortstack{Global\\FR} & SS-gain & Advantage \\
\midrule
S1 & \shortstack{TH, $U=2000$\\\textbf{1.4981 / 0.872}} & \shortstack{$U=2000$\\6.9672 / 0.406} & 2.2029 / 0.812 & \shortstack{RF-LL\\1.9839 / 0.831} & 2.4096 / 0.794 & 2.5607 / 0.781 & 37.83\% & 24.48\% \\
S2 & \shortstack{PH, $U=2000$\\0.00499 / 0.638} & \shortstack{$U=1000$\\0.00919 / 0.334} & 0.00571 / 0.586 & \shortstack{RF-LC\\\textbf{0.00489 / 0.646}} & 0.00586 / 0.575 & 0.00609 / 0.559 & 8.49\% & -2.05\% \\
S3 & \shortstack{TH, $U=1000$\\\textbf{0.22647 / 0.567}} & \shortstack{$U=500$\\0.38804 / 0.261} & 0.27461 / 0.476 & \shortstack{RF-LC\\0.23076 / 0.560} & 0.27117 / 0.482 & 0.27883 / 0.468 & 8.03\% & 1.86\% \\
S4 & \shortstack{TH, $U=2000$\\\textbf{0.0253 / 0.976}} & \shortstack{$U=0$\\0.1951 / 0.815} & 0.0615 / 0.942 & \shortstack{RF-LC\\0.4978 / 0.529} & 0.1929 / 0.817 & -- & 25.54\% & 58.87\% \\
S5 & \shortstack{TH, $U=2000$\\\textbf{0.00108 / 0.912}} & \shortstack{$U=2000$\\0.00179 / 0.854} & 0.00120 / 0.903 & \shortstack{RF-LC\\0.00156 / 0.873} & 0.00182 / 0.852 & -- & 5.06\% & 9.70\% \\
S6 & \shortstack{PH, $U=2000$\\\textbf{0.15457 / 0.683}} & \shortstack{$U=0$\\0.25921 / 0.471} & 0.17670 / 0.637 & \shortstack{RF-LC\\0.34887 / 0.289} & 0.25921 / 0.471 & -- & 9.99\% & 12.53\% \\
\bottomrule
\end{tabular*}
\label{tab:app-descriptive-sim-summary}
\end{table}

\begin{table}[!htbp]
\centering\tiny
\setlength{\tabcolsep}{2pt}
\caption{Descriptive test-best HBN summary. Family representatives minimize mean test loss over methods/\(U\), including zero. Bold marks the lowest mean loss per task. Local FR compares NW-FR and kNN-FR; Forest FR is RF-LC.}
\begin{tabular*}{\textwidth}{@{\extracolsep{\fill}}lrrrrrrr@{}}
\toprule
Setting & \shortstack{Ours\\(best $U$)} & \shortstack{Graph SS-FR\\(best $U$)} & WCE-NLGFR & Forest FR & Local FR & SS-gain & Advantage \\
\midrule
HBN-A & \textbf{PH, 6.2650 / 0.328 / 463} & 7.2566 / 0.222 / 0 & 6.9164 / 0.261 & RF-LC, 7.1214 / 0.238 & 7.2944 / 0.219 & 6.44\% & 9.42\% \\
HBN-B & \textbf{PH, 6.2313 / 0.333 / 463} & 7.2296 / 0.229 / 0 & 6.8781 / 0.265 & RF-LC, 7.1099 / 0.242 & 7.3169 / 0.220 & 6.03\% & 9.40\% \\
\bottomrule
\end{tabular*}
\label{tab:app-descriptive-hbn-summary}
\end{table}

\FloatBarrier

\subsection{Simulation Results}
\label{app:simulation-results}

The common-budget results in Table~\ref{tab:main-object-results} show
pure SS-KRFR gains of -1.70\%, 7.79\%, 6.37\%, 22.86\%, -10.43\%
and 8.72\% in S1--S6, respectively. At the same \(U=2000\), PH's
gains are 8.87\%, 8.49\%, 7.94\%, 24.37\%, 0.68\% and 9.99\%;
TH's are 37.83\%, 7.39\%, 7.76\%, 25.54\%, 5.06\% and 9.47\%.
Positive values indicate lower loss than the same method at \(U=0\).
At this common \(U\), PH has lower mean loss than pure SS-KRFR in all
six settings, and TH in five. The exception is S2, where TH's loss is
\(0.0050514\), versus \(0.0050297\) for pure SS-KRFR. These are
comparisons of means, not claims of statistically significant differences.
Thus the pooled estimator and its hybrids have different empirical
patterns; the fixed-tuning SS-KRFR bounds do not justify the full
validation-selected hybrid procedure.

The remaining tables in this subsection are descriptive best-\(U\)
summaries or explicitly labeled budget sweeps and paired comparisons.
For each best-\(U\) summary, each method is shown at its descriptively best \(U\), using the same
200 generated samples. Parentheses explicitly report MCSE.
Mean precision (MC95) and replication variability (E95) are distinct
columns. Graph NW and Graph kNN abbreviate the two Graph SS-FR variants.

\begin{table}[!htbp]
\centering\scriptsize
\setlength{\tabcolsep}{2pt}
\caption{S1: Euc.\(\to\)SPD. Each method at its descriptive best \(U\), \(R=200\). Parentheses are MCSE; MC95 is the approximate 95\% interval for the mean; E95 is the empirical 2.5--97.5\% loss range.}
\begin{tabular*}{\textwidth}{@{\extracolsep{\fill}}lrrrrrr@{}}
\toprule
Method & \(U\) & Loss (MCSE) & Loss MC95 & Loss E95 & \(R^2\) (MCSE) & \(R^2\) MC95 \\
\midrule
TH & 2000 & \textbf{1.4981 (0.0110)} & [1.4766, 1.5196] & [1.2520, 1.8625] & 0.8721 (0.0009) & [0.8703, 0.8739] \\
PH & 1000 & 2.1635 (0.0203) & [2.1237, 2.2034] & [1.6757, 2.6942] & 0.8154 (0.0017) & [0.8121, 0.8187] \\
RH & 500 & 1.6079 (0.0109) & [1.5864, 1.6293] & [1.3532, 1.9478] & 0.8627 (0.0009) & [0.8610, 0.8645] \\
SS-KRFR & 500 & 2.3920 (0.0165) & [2.3598, 2.4243] & [1.9680, 2.8175] & 0.7959 (0.0013) & [0.7933, 0.7985] \\
KRFR & 0 & 2.4098 (0.0142) & [2.3819, 2.4378] & [1.9895, 2.7503] & 0.7943 (0.0012) & [0.7919, 0.7966] \\
WCE-NLGFR & 0 & 2.2029 (0.0108) & [2.1818, 2.2240] & [1.9256, 2.4834] & 0.8120 (0.0009) & [0.8102, 0.8137] \\
RF-LC & 0 & 3.2938 (0.0219) & [3.2508, 3.3368] & [2.7153, 3.9160] & 0.7191 (0.0017) & [0.7158, 0.7223] \\
RF-LL & 0 & 1.9839 (0.0099) & [1.9644, 2.0034] & [1.7346, 2.2774] & 0.8306 (0.0009) & [0.8287, 0.8324] \\
Graph NW & 2000 & 6.9672 (0.0375) & [6.8938, 7.0407] & [6.0138, 7.9989] & 0.4060 (0.0024) & [0.4013, 0.4106] \\
Graph kNN & 0 & 7.0018 (0.0366) & [6.9299, 7.0736] & [6.0595, 8.0033] & 0.4029 (0.0023) & [0.3984, 0.4075] \\
LL-FR & 0 & 2.4096 (0.0092) & [2.3916, 2.4276] & [2.1711, 2.6494] & 0.7942 (0.0008) & [0.7926, 0.7959] \\
kNN-FR & 0 & 7.0011 (0.0366) & [6.9293, 7.0729] & [6.0595, 8.0033] & 0.4030 (0.0023) & [0.3984, 0.4076] \\
NW-FR & 0 & 7.0267 (0.0367) & [6.9547, 7.0986] & [6.0863, 8.0375] & 0.4009 (0.0023) & [0.3964, 0.4054] \\
Global FR & 0 & 2.5607 (0.0095) & [2.5422, 2.5792] & [2.3157, 2.8113] & 0.7813 (0.0009) & [0.7796, 0.7831] \\
\bottomrule
\end{tabular*}
\label{tab:app-all-methods-s1}
\label{tab:app-all-methods-best-s1-s3}
\end{table}

\begin{table}[!htbp]
\centering\scriptsize
\setlength{\tabcolsep}{2pt}
\caption{S2: Euc.\(\to\)Dist. Each method at its descriptive best \(U\), \(R=200\). Parentheses are MCSE; MC95 is the approximate 95\% interval for the mean; E95 is the empirical 2.5--97.5\% loss range. All loss entries (including MCSE and intervals) are in units of \(10^{-3}\).}
\begin{tabular*}{\textwidth}{@{\extracolsep{\fill}}lrrrrrr@{}}
\toprule
Method & \(U\) & Loss (MCSE) & Loss MC95 & Loss E95 & \(R^2\) (MCSE) & \(R^2\) MC95 \\
\midrule
TH & 2000 & 5.0514 (0.0355) & [4.9818, 5.1210] & [4.2062, 6.1380] & 0.6337 (0.0027) & [0.6283, 0.6390] \\
PH & 2000 & 4.9917 (0.0307) & [4.9316, 5.0517] & [4.2190, 5.9940] & 0.6381 (0.0024) & [0.6334, 0.6427] \\
RH & 2000 & 5.1297 (0.0417) & [5.0480, 5.2113] & [4.2361, 6.4465] & 0.6280 (0.0032) & [0.6217, 0.6342] \\
SS-KRFR & 2000 & 5.0297 (0.0326) & [4.9657, 5.0937] & [4.2328, 6.1654] & 0.6352 (0.0025) & [0.6303, 0.6402] \\
KRFR & 0 & 5.4546 (0.0378) & [5.3805, 5.5287] & [4.6443, 6.5447] & 0.6045 (0.0028) & [0.5990, 0.6101] \\
WCE-NLGFR & 0 & 5.7063 (0.0269) & [5.6535, 5.7591] & [4.9885, 6.4530] & 0.5863 (0.0021) & [0.5822, 0.5904] \\
RF-LC & 0 & \textbf{4.8914 (0.0316)} & [4.8294, 4.9534] & [4.1605, 5.9029] & 0.6455 (0.0023) & [0.6410, 0.6500] \\
RF-LL & 0 & 5.8649 (0.0390) & [5.7883, 5.9414] & [4.8201, 6.9934] & 0.5747 (0.0030) & [0.5689, 0.5806] \\
Graph NW & 1000 & 9.1937 (0.0474) & [9.1008, 9.2865] & [8.0802, 10.4860] & 0.3342 (0.0029) & [0.3285, 0.3399] \\
Graph kNN & 0 & 9.2013 (0.0453) & [9.1126, 9.2900] & [8.0956, 10.6203] & 0.3335 (0.0029) & [0.3278, 0.3391] \\
LL-FR & 0 & 5.8647 (0.0298) & [5.8063, 5.9230] & [5.0343, 6.7450] & 0.5748 (0.0024) & [0.5701, 0.5794] \\
kNN-FR & 0 & 9.1966 (0.0441) & [9.1101, 9.2831] & [8.0970, 10.4402] & 0.3338 (0.0028) & [0.3283, 0.3393] \\
NW-FR & 0 & 9.2243 (0.0497) & [9.1268, 9.3218] & [8.0803, 10.6870] & 0.3320 (0.0031) & [0.3260, 0.3380] \\
Global FR & 0 & 6.0875 (0.0310) & [6.0268, 6.1481] & [5.3107, 6.9504] & 0.5586 (0.0024) & [0.5538, 0.5634] \\
\bottomrule
\end{tabular*}
\label{tab:app-all-methods-s2}
\end{table}

\begin{table}[!htbp]
\centering\scriptsize
\setlength{\tabcolsep}{2pt}
\caption{S3: Euc.\(\to\)BHV. Each method at its descriptive best \(U\), \(R=200\). Parentheses are MCSE; MC95 is the approximate 95\% interval for the mean; E95 is the empirical 2.5--97.5\% loss range.}
\begin{tabular*}{\textwidth}{@{\extracolsep{\fill}}lrrrrrr@{}}
\toprule
Method & \(U\) & Loss (MCSE) & Loss MC95 & Loss E95 & \(R^2\) (MCSE) & \(R^2\) MC95 \\
\midrule
TH & 1000 & \textbf{0.2265 (0.0016)} & [0.2233, 0.2297] & [0.1912, 0.2822] & 0.5672 (0.0032) & [0.5610, 0.5734] \\
PH & 2000 & 0.2267 (0.0016) & [0.2235, 0.2299] & [0.1909, 0.2770] & 0.5670 (0.0031) & [0.5610, 0.5730] \\
RH & 1000 & 0.2303 (0.0018) & [0.2267, 0.2339] & [0.1912, 0.2953] & 0.5599 (0.0036) & [0.5529, 0.5669] \\
SS-KRFR & 2000 & 0.2306 (0.0019) & [0.2269, 0.2342] & [0.1901, 0.2899] & 0.5596 (0.0036) & [0.5526, 0.5666] \\
KRFR & 0 & 0.2462 (0.0019) & [0.2426, 0.2499] & [0.2064, 0.3056] & 0.5299 (0.0033) & [0.5234, 0.5363] \\
WCE-NLGFR & 0 & 0.2746 (0.0018) & [0.2710, 0.2782] & [0.2278, 0.3349] & 0.4756 (0.0033) & [0.4692, 0.4820] \\
RF-LC & 0 & 0.2308 (0.0019) & [0.2270, 0.2345] & [0.1886, 0.2997] & 0.5599 (0.0030) & [0.5540, 0.5658] \\
RF-LL & 0 & 0.2615 (0.0018) & [0.2580, 0.2650] & [0.2247, 0.3119] & 0.4997 (0.0040) & [0.4918, 0.5075] \\
Graph NW & 500 & 0.3880 (0.0026) & [0.3829, 0.3932] & [0.3177, 0.4581] & 0.2606 (0.0031) & [0.2545, 0.2667] \\
Graph kNN & 0 & 0.3951 (0.0027) & [0.3899, 0.4004] & [0.3292, 0.4745] & 0.2468 (0.0035) & [0.2399, 0.2536] \\
LL-FR & 0 & 0.2712 (0.0016) & [0.2681, 0.2743] & [0.2305, 0.3149] & 0.4821 (0.0029) & [0.4764, 0.4877] \\
kNN-FR & 0 & 0.3952 (0.0027) & [0.3899, 0.4004] & [0.3292, 0.4745] & 0.2468 (0.0035) & [0.2399, 0.2536] \\
NW-FR & 0 & 0.3881 (0.0027) & [0.3829, 0.3934] & [0.3177, 0.4580] & 0.2604 (0.0032) & [0.2541, 0.2668] \\
Global FR & 0 & 0.2788 (0.0016) & [0.2756, 0.2820] & [0.2372, 0.3228] & 0.4675 (0.0029) & [0.4619, 0.4731] \\
\bottomrule
\end{tabular*}
\label{tab:app-all-methods-s3}
\end{table}

\FloatBarrier
\begin{table}[!htbp]
\centering\scriptsize
\setlength{\tabcolsep}{2pt}
\caption{S4: Dist.\(\to\)SPD. Each method at its descriptive best \(U\), \(R=200\). Parentheses are MCSE; MC95 is the approximate 95\% interval for the mean; E95 is the empirical 2.5--97.5\% loss range.}
\begin{tabular*}{\textwidth}{@{\extracolsep{\fill}}lrrrrrr@{}}
\toprule
Method & \(U\) & Loss (MCSE) & Loss MC95 & Loss E95 & \(R^2\) (MCSE) & \(R^2\) MC95 \\
\midrule
TH & 2000 & \textbf{0.0253 (0.0003)} & [0.0247, 0.0259] & [0.0217, 0.0335] & 0.9761 (0.0003) & [0.9755, 0.9766] \\
PH & 2000 & 0.0257 (0.0003) & [0.0251, 0.0263] & [0.0218, 0.0311] & 0.9757 (0.0003) & [0.9751, 0.9763] \\
RH & 2000 & 0.0256 (0.0003) & [0.0250, 0.0261] & [0.0219, 0.0338] & 0.9758 (0.0003) & [0.9752, 0.9764] \\
SS-KRFR & 2000 & 0.0262 (0.0003) & [0.0256, 0.0268] & [0.0223, 0.0352] & 0.9752 (0.0003) & [0.9746, 0.9758] \\
KRFR & 0 & 0.0340 (0.0006) & [0.0327, 0.0352] & [0.0242, 0.0518] & 0.9679 (0.0006) & [0.9667, 0.9690] \\
WCE-NLGFR & 0 & 0.0615 (0.0011) & [0.0594, 0.0636] & [0.0456, 0.0872] & 0.9418 (0.0011) & [0.9397, 0.9438] \\
RF-LC & 0 & 0.4978 (0.0027) & [0.4926, 0.5030] & [0.4303, 0.5789] & 0.5292 (0.0021) & [0.5250, 0.5334] \\
Graph NW & 0 & 0.2408 (0.0027) & [0.2354, 0.2461] & [0.1783, 0.3381] & 0.7723 (0.0025) & [0.7675, 0.7771] \\
Graph kNN & 0 & 0.1951 (0.0018) & [0.1917, 0.1985] & [0.1556, 0.2601] & 0.8154 (0.0016) & [0.8122, 0.8186] \\
kNN-FR & 0 & 0.1929 (0.0017) & [0.1895, 0.1962] & [0.1540, 0.2579] & 0.8175 (0.0016) & [0.8144, 0.8206] \\
NW-FR & 0 & 0.2408 (0.0027) & [0.2354, 0.2461] & [0.1783, 0.3381] & 0.7723 (0.0025) & [0.7675, 0.7771] \\
\bottomrule
\end{tabular*}
\label{tab:app-all-methods-s4}
\label{tab:app-all-methods-best-s4-s6}
\end{table}

\begin{table}[!htbp]
\centering\scriptsize
\setlength{\tabcolsep}{2pt}
\caption{S5: Dist.\(\to\)Dist. Each method at its descriptive best \(U\), \(R=200\). Parentheses are MCSE; MC95 is the approximate 95\% interval for the mean; E95 is the empirical 2.5--97.5\% loss range. All loss entries (including MCSE and intervals) are in units of \(10^{-3}\).}
\begin{tabular*}{\textwidth}{@{\extracolsep{\fill}}lrrrrrr@{}}
\toprule
Method & \(U\) & Loss (MCSE) & Loss MC95 & Loss E95 & \(R^2\) (MCSE) & \(R^2\) MC95 \\
\midrule
TH & 2000 & \textbf{1.0842 (0.0126)} & [1.0594, 1.1089] & [0.8094, 1.5007] & 0.9120 (0.0011) & [0.9099, 0.9141] \\
PH & 500 & 1.1333 (0.0132) & [1.1074, 1.1591] & [0.8567, 1.5427] & 0.9080 (0.0011) & [0.9059, 0.9101] \\
RH & 1000 & 1.1386 (0.0176) & [1.1041, 1.1731] & [0.8097, 1.6592] & 0.9075 (0.0015) & [0.9047, 0.9104] \\
SS-KRFR & 0 & 1.1420 (0.0125) & [1.1175, 1.1664] & [0.8848, 1.5621] & 0.9073 (0.0011) & [0.9052, 0.9093] \\
KRFR & 0 & 1.1420 (0.0125) & [1.1175, 1.1664] & [0.8848, 1.5621] & 0.9073 (0.0011) & [0.9052, 0.9093] \\
WCE-NLGFR & 0 & 1.2007 (0.0142) & [1.1728, 1.2285] & [0.9040, 1.6935] & 0.9026 (0.0012) & [0.9002, 0.9049] \\
RF-LC & 0 & 1.5596 (0.0151) & [1.5301, 1.5891] & [1.2209, 2.0124] & 0.8734 (0.0013) & [0.8710, 0.8759] \\
Graph NW & 0 & 2.2614 (0.0243) & [2.2138, 2.3090] & [1.7105, 3.0422] & 0.8166 (0.0019) & [0.8128, 0.8204] \\
Graph kNN & 2000 & 1.7935 (0.0171) & [1.7600, 1.8269] & [1.4179, 2.3271] & 0.8544 (0.0014) & [0.8516, 0.8572] \\
kNN-FR & 0 & 1.8189 (0.0171) & [1.7854, 1.8525] & [1.4275, 2.3207] & 0.8523 (0.0015) & [0.8495, 0.8552] \\
NW-FR & 0 & 2.2614 (0.0243) & [2.2138, 2.3090] & [1.7105, 3.0422] & 0.8166 (0.0019) & [0.8128, 0.8204] \\
\bottomrule
\end{tabular*}
\label{tab:app-all-methods-s5}
\end{table}

\begin{table}[!htbp]
\centering\scriptsize
\setlength{\tabcolsep}{2pt}
\caption{S6: Dist.\(\to\)BHV. Each method at its descriptive best \(U\), \(R=200\). Parentheses are MCSE; MC95 is the approximate 95\% interval for the mean; E95 is the empirical 2.5--97.5\% loss range.}
\begin{tabular*}{\textwidth}{@{\extracolsep{\fill}}lrrrrrr@{}}
\toprule
Method & \(U\) & Loss (MCSE) & Loss MC95 & Loss E95 & \(R^2\) (MCSE) & \(R^2\) MC95 \\
\midrule
TH & 2000 & 0.1555 (0.0022) & [0.1511, 0.1598] & [0.1173, 0.2341] & 0.6814 (0.0047) & [0.6722, 0.6905] \\
PH & 2000 & \textbf{0.1546 (0.0022)} & [0.1503, 0.1589] & [0.1198, 0.2232] & 0.6827 (0.0048) & [0.6732, 0.6922] \\
RH & 2000 & 0.1570 (0.0023) & [0.1524, 0.1615] & [0.1190, 0.2504] & 0.6782 (0.0049) & [0.6686, 0.6878] \\
SS-KRFR & 2000 & 0.1567 (0.0024) & [0.1521, 0.1614] & [0.1200, 0.2533] & 0.6781 (0.0053) & [0.6677, 0.6885] \\
KRFR & 0 & 0.1717 (0.0021) & [0.1676, 0.1758] & [0.1311, 0.2408] & 0.6478 (0.0046) & [0.6389, 0.6567] \\
WCE-NLGFR & 0 & 0.1767 (0.0024) & [0.1721, 0.1813] & [0.1326, 0.2566] & 0.6372 (0.0052) & [0.6270, 0.6475] \\
RF-LC & 0 & 0.3489 (0.0031) & [0.3429, 0.3549] & [0.2744, 0.4445] & 0.2890 (0.0036) & [0.2819, 0.2960] \\
Graph NW & 0 & 0.2592 (0.0026) & [0.2540, 0.2644] & [0.1957, 0.3434] & 0.4708 (0.0044) & [0.4622, 0.4794] \\
Graph kNN & 1000 & 0.2767 (0.0026) & [0.2716, 0.2817] & [0.2142, 0.3518] & 0.4347 (0.0044) & [0.4260, 0.4433] \\
kNN-FR & 0 & 0.2765 (0.0026) & [0.2714, 0.2817] & [0.2131, 0.3593] & 0.4350 (0.0045) & [0.4261, 0.4438] \\
NW-FR & 0 & 0.2592 (0.0027) & [0.2540, 0.2644] & [0.1957, 0.3434] & 0.4708 (0.0044) & [0.4622, 0.4794] \\
\bottomrule
\end{tabular*}
\label{tab:app-all-methods-s6}
\end{table}

\FloatBarrier
\begin{table}[!htbp]
\centering\scriptsize
\setlength{\tabcolsep}{2pt}
\caption{Best proposed method by unlabeled sample size: mean loss / mean \(R^2\).}
\begin{tabular*}{\textwidth}{@{\extracolsep{\fill}}lrrrr@{}}
\toprule
Setting & \(U=0\) & \(U=500\) & \(U=1000\) & \(U=2000\) \\
\midrule
S1 & TH, 2.4098 / 0.794 & TH, 1.5002 / 0.872 & TH, 1.5063 / 0.871 & \textbf{TH, 1.4981 / 0.872} \\
S2 & TH, 0.00545 / 0.605 & PH, 0.00508 / 0.632 & PH, 0.00503 / 0.635 & \textbf{PH, 0.00499 / 0.638} \\
S3 & TH, 0.24624 / 0.530 & TH, 0.22995 / 0.560 & \textbf{TH, 0.22647 / 0.567} & PH, 0.22669 / 0.567 \\
S4 & TH, 0.0340 / 0.968 & TH, 0.0263 / 0.975 & TH, 0.0255 / 0.976 & \textbf{TH, 0.0253 / 0.976} \\
S5 & TH, 0.00114 / 0.907 & TH, 0.00110 / 0.911 & TH, 0.00109 / 0.911 & \textbf{TH, 0.00108 / 0.912} \\
S6 & TH, 0.17172 / 0.648 & PH, 0.15662 / 0.679 & PH, 0.15632 / 0.679 & \textbf{PH, 0.15457 / 0.683} \\
\bottomrule
\end{tabular*}
\label{tab:app-best-ours-sweep}
\label{tab:main-sim-gain}
\end{table}

For the descriptive test-best comparisons, the best proposed methods
are TH in S1, S3, S4 and S5, and PH in S2
and S6. In setting order S1--S6, their SS-gains are 37.83\%, 8.49\%,
8.03\%, 25.54\%, 5.06\% and 9.99\%, respectively. The corresponding
external advantages are 24.48\%, -2.05\%, 1.86\%, 58.87\%, 9.70\%
and 12.53\%. WCE uses validation tuning in all these comparisons;
RF-LL is the best external method in S1, RF-LC in S2/S3, and WCE in
S4--S6. In S2, RF-LC has lower mean loss than every proposed method.
These gains characterize the selected designs, not uniform superiority
over all data-generating mechanisms. Pure SS-KRFR's descriptive best is
\(U=0\) in S5, whereas in S4 its SS-gain is 22.86\%, compared with
TH's 25.54\%. Its S1 gain is only 0.74\% with MC95 crossing zero,
whereas TH's is 37.83\%. Thus semi-supervised benefit is more consistently
expressed by the displayed hybrids than by the pure pooled representation.

The covariate generating law is shared within each input type, yet the
pure-SS and hybrid gains differ across response mechanisms. Covariate
structure alone therefore does not explain the observed gains. The
regression function's shape and spectral alignment with the representation,
response geometry, noise level, finite-sample estimation and regularization
are plausible additional contributors. These experiments do not isolate
their separate causal effects. Hybridization can retain a supervised
component while exploiting pooled structure, but this is an interpretation
of the pattern, not a verified mechanism or a guarantee of improvement.
RH can also change model capacity through its up-to-\(2s\)-coordinate
concatenation; its gains cannot be attributed solely to adjusting pooling.
Hybrid SS-gain MC95 intervals exclude zero in every displayed simulation;
these remain conditional, non-selection-adjusted precision statements.

\begin{table}[!htbp]
\centering\scriptsize
\setlength{\tabcolsep}{2pt}
\caption{Best hybrid paired gains (\%). MCSE is in percentage points. SS-gain uses the same method at \(U=0\); external gain uses the listed best external method. Both use ratios of paired mean losses. MC95 conditions on these descriptive method/\(U\) choices and is not selection-adjusted.}
\begin{tabular*}{\textwidth}{@{\extracolsep{\fill}}llrrlrr@{}}
\toprule
Setting & Method, \(U\) & SS-gain (MCSE) & MC95 & External, \(U\) & Ext. gain (MCSE) & MC95 \\
\midrule
S1 & TH, 2000 & 37.83 (0.53) & [36.80, 38.87] & RF-LL, 0 & 24.48 (0.63) & [23.24, 25.73] \\
S2 & PH, 2000 & 8.49 (0.48) & [7.55, 9.43] & RF-LC, 0 & -2.05 (0.76) & [-3.54, -0.56] \\
S3 & TH, 1000 & 8.03 (0.58) & [6.89, 9.16] & RF-LC, 0 & 1.86 (0.87) & [0.15, 3.56] \\
S4 & TH, 2000 & 25.54 (0.86) & [23.85, 27.23] & WCE-NLGFR, 0 & 58.87 (0.49) & [57.91, 59.82] \\
S5 & TH, 2000 & 5.06 (0.92) & [3.26, 6.87] & WCE-NLGFR, 0 & 9.70 (1.09) & [7.56, 11.85] \\
S6 & PH, 2000 & 9.99 (1.02) & [8.00, 11.98] & WCE-NLGFR, 0 & 12.53 (1.07) & [10.42, 14.63] \\
\bottomrule
\end{tabular*}
\label{tab:app-sim-hybrid-gains}
\end{table}

\begin{table}[!htbp]
\centering\scriptsize
\setlength{\tabcolsep}{2pt}
\caption{Pure SS-KRFR paired gains (\%). MCSE is in percentage points. SS-gain uses the same method at \(U=0\); external gain uses the listed best external method. Both use ratios of paired mean losses. MC95 conditions on these descriptive method/\(U\) choices and is not selection-adjusted.}
\begin{tabular*}{\textwidth}{@{\extracolsep{\fill}}llrrlrr@{}}
\toprule
Setting & Method, \(U\) & SS-gain (MCSE) & MC95 & External, \(U\) & Ext. gain (MCSE) & MC95 \\
\midrule
S1 & SS-KRFR, 500 & 0.74 (0.67) & [-0.58, 2.06] & RF-LL, 0 & -20.57 (0.95) & [-22.44, -18.71] \\
S2 & SS-KRFR, 2000 & 7.79 (0.61) & [6.60, 8.98] & RF-LC, 0 & -2.83 (0.83) & [-4.45, -1.20] \\
S3 & SS-KRFR, 2000 & 6.37 (0.71) & [4.97, 7.77] & RF-LC, 0 & 0.09 (0.93) & [-1.74, 1.91] \\
S4 & SS-KRFR, 2000 & 22.86 (0.93) & [21.04, 24.68] & WCE-NLGFR, 0 & 57.39 (0.49) & [56.42, 58.35] \\
S5 & SS-KRFR, 0 & 0.00 (0.00) & [-0.00, 0.00] & WCE-NLGFR, 0 & 4.89 (1.16) & [2.62, 7.16] \\
S6 & SS-KRFR, 2000 & 8.72 (1.22) & [6.33, 11.11] & WCE-NLGFR, 0 & 11.29 (1.18) & [8.98, 13.60] \\
\bottomrule
\end{tabular*}
\label{tab:app-sim-pure-gains}
\end{table}

For completeness, the following tables hold the unlabeled size fixed
at \(U=2000\) and compare each semi-supervised method with its own
\(U=0\) endpoint. Positive paired mean loss reductions denote improvement.
These fixed-size comparisons need not use each method's descriptive optimum.

\begin{table}[!htbp]
\centering\scriptsize
\setlength{\tabcolsep}{2pt}
\caption{Paired mean loss reductions: \(U=2000\) versus \(U=0\); positive means improvement. MC95 is mean reduction \(\pm1.96\) MCSE. S2/S5 loss reductions, when present, are in units of \(10^{-3}\).}
\begin{tabular*}{\textwidth}{@{\extracolsep{\fill}}lrrrr@{}}
\toprule
Method & S1 reduction (MCSE) & MC95 & S2 reduction (MCSE) & MC95 \\
\midrule
TH & 0.9117 (0.0163) & [0.8798, 0.9436] & 0.4032 (0.0288) & [0.3467, 0.4597] \\
PH & 0.2138 (0.0157) & [0.1831, 0.2445] & 0.4629 (0.0279) & [0.4082, 0.5177] \\
RH & 0.8019 (0.0152) & [0.7722, 0.8317] & 0.3249 (0.0349) & [0.2564, 0.3934] \\
SS-KRFR & -0.0410 (0.0141) & [-0.0686, -0.0134] & 0.4249 (0.0347) & [0.3569, 0.4929] \\
Graph NW & 0.0594 (0.0089) & [0.0420, 0.0769] & 0.0294 (0.0138) & [0.0024, 0.0564] \\
Graph kNN & -0.0619 (0.0082) & [-0.0779, -0.0459] & -0.0597 (0.0168) & [-0.0925, -0.0268] \\
\bottomrule
\end{tabular*}
\label{tab:app-u-gain}
\end{table}

\begin{table}[!htbp]
\centering\scriptsize
\setlength{\tabcolsep}{2pt}
\caption{Paired mean loss reductions: \(U=2000\) versus \(U=0\); positive means improvement. MC95 is mean reduction \(\pm1.96\) MCSE. S2/S5 loss reductions, when present, are in units of \(10^{-3}\).}
\begin{tabular*}{\textwidth}{@{\extracolsep{\fill}}lrrrr@{}}
\toprule
Method & S4 reduction (MCSE) & MC95 & S5 reduction (MCSE) & MC95 \\
\midrule
TH & 0.0087 (0.0004) & [0.0078, 0.0095] & 0.0578 (0.0108) & [0.0367, 0.0789] \\
PH & 0.0083 (0.0004) & [0.0074, 0.0091] & 0.0078 (0.0075) & [-0.0070, 0.0226] \\
RH & 0.0084 (0.0004) & [0.0076, 0.0093] & 0.0033 (0.0149) & [-0.0260, 0.0326] \\
SS-KRFR & 0.0078 (0.0004) & [0.0069, 0.0086] & -0.1191 (0.0151) & [-0.1487, -0.0896] \\
Graph NW & -0.0079 (0.0011) & [-0.0101, -0.0057] & -0.0228 (0.0066) & [-0.0357, -0.0099] \\
Graph kNN & -0.0012 (0.0006) & [-0.0025, 0.0000] & 0.0488 (0.0077) & [0.0337, 0.0639] \\
\bottomrule
\end{tabular*}
\label{tab:app-u-gain-s4-s5}
\end{table}

\begin{table}[!htbp]
\centering\scriptsize
\setlength{\tabcolsep}{2pt}
\caption{Paired mean loss reductions: \(U=2000\) versus \(U=0\); positive means improvement. MC95 is mean reduction \(\pm1.96\) MCSE. S2/S5 loss reductions, when present, are in units of \(10^{-3}\).}
\begin{tabular*}{\textwidth}{@{\extracolsep{\fill}}lrrrr@{}}
\toprule
Method & S3 reduction (MCSE) & MC95 & S6 reduction (MCSE) & MC95 \\
\midrule
TH & 0.0191 (0.0016) & [0.0160, 0.0222] & 0.0163 (0.0016) & [0.0130, 0.0195] \\
PH & 0.0195 (0.0014) & [0.0167, 0.0223] & 0.0172 (0.0018) & [0.0136, 0.0207] \\
RH & 0.0145 (0.0017) & [0.0112, 0.0179] & 0.0148 (0.0018) & [0.0111, 0.0184] \\
SS-KRFR & 0.0157 (0.0018) & [0.0121, 0.0192] & 0.0150 (0.0022) & [0.0108, 0.0192] \\
Graph NW & -0.0011 (0.0005) & [-0.0020, -0.0001] & -0.0027 (0.0007) & [-0.0042, -0.0013] \\
Graph kNN & -0.0022 (0.0005) & [-0.0032, -0.0011] & 0.0047 (0.0015) & [0.0018, 0.0076] \\
\bottomrule
\end{tabular*}
\label{tab:app-bhv-paired}
\end{table}

\FloatBarrier

\subsection{HBN-EEG Design and Results}
\label{app:hbn-details}

\subsubsection{Dataset and Object Construction}

We use resting-state recordings from HBN-EEG releases R1--R9
\citep{alexander2017hbn,langer2017developing,shirazi2024hbneeg}.
The processed cohort has 2,103 distinct participants with usable EC and EO
segments, from 2,118 candidate recordings. Release counts are:

\begin{table}[!htbp]
\centering\scriptsize
\setlength{\tabcolsep}{2pt}
\caption{HBN-EEG analytic cohort.}
\begin{tabular*}{\textwidth}{@{\extracolsep{\fill}}lrrrrrrrrr@{}}
\toprule
Release & R1 & R2 & R3 & R4 & R5 & R6 & R7 & R8 & R9 \\
\midrule
Participants & 136 & 148 & 183 & 322 & 328 & 135 & 365 & 221 & 265 \\
\bottomrule
\end{tabular*}
\label{tab:app-hbn-release-counts}
\end{table}

The first 16 EEG channels are used consistently. EC and EO segments are
identified from event markers, segments shorter than eight seconds are
discarded, and retained segments are concatenated within participant and
condition. Channelwise Welch spectra use four-second windows with 50\%
overlap and are averaged across the 16 channels. Power on 1--45 Hz is
normalized to a probability distribution \(\mu_i^g\),
\(g\in\{\mathrm{EC},\mathrm{EO}\}\), represented by
\[
 Q_i^g=80^{-1/2}\big((F_i^g)^{-1}((j-1/2)/80)\big)_{j=1}^{80}.
\]
For the SPD covariate, each channel is standardized within the participant,
and the empirical EC correlation matrix \(R_i\) is regularized as
\(\widetilde C_i=0.995R_i+0.005I_{16}\). After diagonal normalization,
the SPD object \(C_i\) is represented by
\(S_i=\operatorname{svec}(\log C_i)\).

HBN-A and HBN-B both predict \(Y_i=\mu_i^{\rm EO}\). HBN-A uses \(X_i=\mu_i^{\rm EC}\). HBN-B uses the joint covariate
\(X_i=(C_i,\mu_i^{\rm EC})\). Let \(\delta_S\) and \(\delta_Q\) denote
the median positive pairwise squared distances between the corresponding
training coordinates. Its input representation is
\[
 x_i(\omega)=\left(\sqrt{\omega/\delta_S}\,S_i,
                      \sqrt{(1-\omega)/\delta_Q}\,Q_i^{\rm EC}\right),
 \qquad\omega\in\mathcal A.
\]
Thus the radial kernels operate on the weighted sum of the two scaled
squared distances. The weight acts on the covariate distance before kernel
evaluation. Training data determine both scales.

\subsubsection{Random-Partition Design and Tuning}

The 789 distinct participants in R1--R4 form one labeled cohort.
In each of 200 repetitions, a random permutation without replacement
assigns 284 to training, 183 to validation and 322 to testing.
NumPy SeedSequence \([2026092802,r]\), \(r=0,\ldots,199\), determines
the split. Every method, \(U\), HBN-A/B task and HBN-B mixture candidate
uses the same split in a given repetition. Participants can recur across
repetitions, but never across roles within one split. This is not bootstrap
sampling or a fixed R4 release-held-out evaluation.

The nested unlabeled pools are fixed at R5--R6 (463 participants),
R5--R7 (828), and R5--R9 (1,314), together with \(U=0\).
Unlabeled outcomes are withheld. Validation and test covariates are queries;
only training and current unlabeled inputs enter representations and
graphs. Pool size is confounded with release composition.

KRFR, SS-KRFR, RH, PH, TH, WCE-NLGFR, Graph SS-FR (NW/kNN),
NW-FR, kNN-FR and RF-LC are fitted using the method-specific criteria in
Section~\ref{app:tuning-grid}. Feature fits use \(n=284\) and
\(n\xi\in\mathcal L\). For HBN-B, both input-block scales are recomputed
from the current training split. Each non-forest method/\(U\) is fitted at all seven
\(\omega\in\{0,0.1,0.25,0.5,0.75,0.9,1\}\), and its validation-loss
winner is selected only after all seven candidates are available.
Exact ties retain the first grid value. No valid validation candidate
means selection failure, not selection by test performance. RF-LC instead
uses the three input choices specified in Section~\ref{app:tuning-grid}.

Predictions use \(\Pi_{[1/\sqrt{80},45/\sqrt{80}]}^\uparrow\):
equal-weight PAVA of the raw prediction, followed by clipping.
Loss averages the sum of 80 squared scaled-quantile errors over the
322 test participants. Each split's \(R^2\) denominator uses the test
loss of that split's training-response mean; negative values are retained.
All reported means and intervals concern this fixed cohort and its random
partitions, not independent-cohort or unseen-release generalization.

\subsubsection{Results}

At the common maximal budget \(U=1314\), Table~\ref{tab:hbn-main-results}
reports pure SS-KRFR gains of 3.53\% and 1.67\%, PH gains of 5.05\%
and 3.85\%, and TH gains of 4.45\% and 3.12\% in HBN-A/B.
RH gains 1.87\% in HBN-A but loses 0.21\% in HBN-B. These comparisons
retain all methods at the same available pool size and are distinct
from the descriptive test-best summaries below.

Among the descriptive test-best results, PH at \(U=463\) has the
lowest mean loss in both tasks: 6.2650 in
HBN-A and 6.2313 in HBN-B. Its SS-gains are 6.44\% and 6.03\%,
while external gains are 9.42\% and 9.40\% against validation-tuned
WCE-NLGFR, respectively. Pure SS-KRFR has SS-gains of 6.34\% and
5.40\%; those comparisons are reported separately from PH.
The incremental PH advantage over pure SS is small in HBN-A, so the
stronger hybrid pattern in simulations is not evidence of a significant
hybrid-versus-pure-SS gap in every task. Both hybrid SS-gain intervals
exclude zero, as do their gains over validation-tuned WCE.
The largest unlabeled pool is not the descriptive optimum.
Within-repetition parameters, including HBN-B's \(\omega\), use only
the specified training or validation criteria.

\begin{table}[!htbp]
\centering\footnotesize
\setlength{\tabcolsep}{2pt}
\caption{HBN method comparisons across 200 common random splits: mean loss / mean \(R^2\) / descriptive best \(U\).}
\begin{tabular*}{\textwidth}{@{\extracolsep{\fill}}lrr@{}}
\toprule
Method & HBN-A & HBN-B \\
\midrule
TH & 6.3251 / 0.321 / 463 & 6.2859 / 0.324 / 463 \\
PH & \textbf{6.2650 / 0.328 / 463} & \textbf{6.2313 / 0.333 / 463} \\
RH & 6.5068 / 0.302 / 463 & 6.4568 / 0.306 / 463 \\
SS-KRFR & 6.2718 / 0.327 / 463 & 6.2727 / 0.327 / 463 \\
KRFR & 6.6962 / 0.286 / 0 & 6.6311 / 0.295 / 0 \\
WCE-NLGFR & 6.9164 / 0.261 / 0 & 6.8781 / 0.265 / 0 \\
RF-LC & 7.1214 / 0.238 / 0 & 7.1099 / 0.242 / 0 \\
Graph NW & 7.2566 / 0.222 / 0 & 7.2296 / 0.229 / 0 \\
Graph kNN & 7.7464 / 0.169 / 0 & 7.5131 / 0.196 / 0 \\
kNN-FR & 7.6715 / 0.177 / 0 & 7.3985 / 0.209 / 0 \\
NW-FR & 7.2944 / 0.219 / 0 & 7.3169 / 0.220 / 0 \\
\bottomrule
\end{tabular*}
\label{tab:app-hbn-family-results}
\end{table}

\begin{table}[!htbp]
\centering\scriptsize
\setlength{\tabcolsep}{2pt}
\caption{HBN-A. Each method at its descriptive best \(U\), \(R=200\). Parentheses are MCSE; MC95 is the approximate 95\% interval for the mean; E95 is the empirical 2.5--97.5\% loss range.}
\begin{tabular*}{\textwidth}{@{\extracolsep{\fill}}lrrrrrr@{}}
\toprule
Method & \(U\) & Loss (MCSE) & Loss MC95 & Loss E95 & \(R^2\) (MCSE) & \(R^2\) MC95 \\
\midrule
TH & 463 & 6.3251 (0.0997) & [6.1298, 6.5204] & [4.4707, 10.2572] & 0.3215 (0.0095) & [0.3028, 0.3402] \\
PH & 463 & \textbf{6.2650 (0.0944)} & [6.0800, 6.4500] & [4.3797, 10.1442] & 0.3283 (0.0088) & [0.3112, 0.3455] \\
RH & 463 & 6.5068 (0.1115) & [6.2883, 6.7252] & [4.4917, 10.8168] & 0.3015 (0.0111) & [0.2797, 0.3233] \\
SS-KRFR & 463 & 6.2718 (0.0974) & [6.0809, 6.4627] & [4.2766, 9.8551] & 0.3268 (0.0095) & [0.3082, 0.3454] \\
KRFR & 0 & 6.6962 (0.1135) & [6.4738, 6.9186] & [4.4530, 10.7069] & 0.2861 (0.0094) & [0.2677, 0.3045] \\
WCE-NLGFR & 0 & 6.9164 (0.1272) & [6.6671, 7.1656] & [4.3316, 11.3263] & 0.2606 (0.0119) & [0.2372, 0.2840] \\
RF-LC & 0 & 7.1214 (0.0922) & [6.9408, 7.3020] & [4.7229, 10.0476] & 0.2378 (0.0089) & [0.2203, 0.2554] \\
Graph NW & 0 & 7.2566 (0.1374) & [6.9874, 7.5259] & [4.3354, 11.7694] & 0.2223 (0.0141) & [0.1947, 0.2500] \\
Graph kNN & 0 & 7.7464 (0.1149) & [7.5213, 7.9716] & [5.1102, 11.5278] & 0.1689 (0.0130) & [0.1435, 0.1944] \\
kNN-FR & 0 & 7.6715 (0.1150) & [7.4461, 7.8968] & [4.9599, 11.8496] & 0.1765 (0.0131) & [0.1508, 0.2022] \\
NW-FR & 0 & 7.2944 (0.1367) & [7.0264, 7.5623] & [4.3711, 11.6773] & 0.2187 (0.0138) & [0.1916, 0.2458] \\
\bottomrule
\end{tabular*}
\label{tab:app-hbn-a-precision}
\end{table}

\begin{table}[!htbp]
\centering\scriptsize
\setlength{\tabcolsep}{2pt}
\caption{HBN-B. Each method at its descriptive best \(U\), \(R=200\). Parentheses are MCSE; MC95 is the approximate 95\% interval for the mean; E95 is the empirical 2.5--97.5\% loss range.}
\begin{tabular*}{\textwidth}{@{\extracolsep{\fill}}lrrrrrr@{}}
\toprule
Method & \(U\) & Loss (MCSE) & Loss MC95 & Loss E95 & \(R^2\) (MCSE) & \(R^2\) MC95 \\
\midrule
TH & 463 & 6.2859 (0.0987) & [6.0924, 6.4794] & [4.3799, 10.4892] & 0.3244 (0.0099) & [0.3049, 0.3438] \\
PH & 463 & \textbf{6.2313 (0.0955)} & [6.0442, 6.4184] & [4.4408, 9.6333] & 0.3326 (0.0086) & [0.3159, 0.3494] \\
RH & 463 & 6.4568 (0.1063) & [6.2485, 6.6651] & [4.4481, 10.8168] & 0.3057 (0.0110) & [0.2842, 0.3271] \\
SS-KRFR & 463 & 6.2727 (0.0985) & [6.0797, 6.4657] & [4.4363, 9.5990] & 0.3269 (0.0097) & [0.3080, 0.3459] \\
KRFR & 0 & 6.6311 (0.1132) & [6.4091, 6.8530] & [4.4232, 10.7101] & 0.2947 (0.0086) & [0.2779, 0.3115] \\
WCE-NLGFR & 0 & 6.8781 (0.1224) & [6.6381, 7.1181] & [4.5375, 11.0008] & 0.2652 (0.0113) & [0.2430, 0.2873] \\
RF-LC & 0 & 7.1099 (0.1017) & [6.9106, 7.3091] & [4.6368, 10.4151] & 0.2422 (0.0079) & [0.2268, 0.2577] \\
Graph NW & 0 & 7.2296 (0.1327) & [6.9695, 7.4897] & [4.4050, 11.8129] & 0.2287 (0.0120) & [0.2052, 0.2523] \\
Graph kNN & 0 & 7.5131 (0.1046) & [7.3081, 7.7182] & [5.1150, 11.3943] & 0.1965 (0.0101) & [0.1767, 0.2163] \\
kNN-FR & 0 & 7.3985 (0.0997) & [7.2030, 7.5940] & [5.0370, 11.0573] & 0.2087 (0.0097) & [0.1897, 0.2277] \\
NW-FR & 0 & 7.3169 (0.1361) & [7.0502, 7.5836] & [4.4996, 11.9104] & 0.2199 (0.0124) & [0.1957, 0.2441] \\
\bottomrule
\end{tabular*}
\label{tab:app-hbn-b-precision}
\end{table}

\begin{table}[!htbp]
\centering\scriptsize
\setlength{\tabcolsep}{2pt}
\caption{Best hybrid paired gains (\%). MCSE is in percentage points. SS-gain uses the same method at \(U=0\); external gain uses the listed best external method. Both use ratios of paired mean losses. MC95 conditions on these descriptive method/\(U\) choices and is not selection-adjusted.}
\begin{tabular*}{\textwidth}{@{\extracolsep{\fill}}llrrlrr@{}}
\toprule
Setting & Method, \(U\) & SS-gain (MCSE) & MC95 & External, \(U\) & Ext. gain (MCSE) & MC95 \\
\midrule
HBN-A & PH, 463 & 6.44 (1.01) & [4.46, 8.42] & WCE-NLGFR, 0 & 9.42 (1.39) & [6.70, 12.13] \\
HBN-B & PH, 463 & 6.03 (1.01) & [4.05, 8.01] & WCE-NLGFR, 0 & 9.40 (1.34) & [6.78, 12.03] \\
\bottomrule
\end{tabular*}
\label{tab:app-hbn-hybrid-gains}
\end{table}

\begin{table}[!htbp]
\centering\scriptsize
\setlength{\tabcolsep}{2pt}
\caption{Pure SS-KRFR paired gains (\%). MCSE is in percentage points. SS-gain uses the same method at \(U=0\); external gain uses the listed best external method. Both use ratios of paired mean losses. MC95 conditions on these descriptive method/\(U\) choices and is not selection-adjusted.}
\begin{tabular*}{\textwidth}{@{\extracolsep{\fill}}llrrlrr@{}}
\toprule
Setting & Method, \(U\) & SS-gain (MCSE) & MC95 & External, \(U\) & Ext. gain (MCSE) & MC95 \\
\midrule
HBN-A & SS-KRFR, 463 & 6.34 (1.25) & [3.88, 8.80] & WCE-NLGFR, 0 & 9.32 (1.32) & [6.73, 11.91] \\
HBN-B & SS-KRFR, 463 & 5.40 (1.21) & [3.03, 7.78] & WCE-NLGFR, 0 & 8.80 (1.32) & [6.22, 11.38] \\
\bottomrule
\end{tabular*}
\label{tab:app-hbn-pure-gains}
\end{table}

\FloatBarrier
\begin{table}[!htbp]
\centering\scriptsize
\setlength{\tabcolsep}{2pt}
\caption{Best proposed method by unlabeled release pool: mean loss / mean \(R^2\).}
\begin{tabular*}{\textwidth}{@{\extracolsep{\fill}}lrrrr@{}}
\toprule
Setting & \(U=0\) & \(U=463\) & \(U=828\) & \(U=1314\) \\
\midrule
HBN-A & TH, 6.6962 / 0.286 & \textbf{PH, 6.2650 / 0.328} & PH, 6.3406 / 0.321 & PH, 6.3581 / 0.319 \\
HBN-B & TH, 6.6311 / 0.295 & \textbf{PH, 6.2313 / 0.333} & PH, 6.2884 / 0.326 & PH, 6.3761 / 0.317 \\
\bottomrule
\end{tabular*}
\label{tab:app-hbn-ph-sweep}
\label{tab:hbn-main-unlabeled}
\end{table}

\begin{table}[!htbp]
\centering\footnotesize
\setlength{\tabcolsep}{2pt}
\caption{HBN-A: mean loss / mean \(R^2\) by unlabeled pool. Bold marks the lowest mean loss in each column, including ties.}
\begin{tabular*}{\textwidth}{@{\extracolsep{\fill}}lrrrr@{}}
\toprule
Method & \(U=0\) & \(U=463\) & \(U=828\) & \(U=1314\) \\
\midrule
TH & \textbf{6.6962 / 0.286} & 6.3251 / 0.321 & 6.3934 / 0.314 & 6.3985 / 0.314 \\
PH & \textbf{6.6962 / 0.286} & \textbf{6.2650 / 0.328} & \textbf{6.3406 / 0.321} & \textbf{6.3581 / 0.319} \\
RH & \textbf{6.6962 / 0.286} & 6.5068 / 0.302 & 6.5962 / 0.291 & 6.5710 / 0.294 \\
SS-KRFR & \textbf{6.6962 / 0.286} & 6.2718 / 0.327 & 6.4189 / 0.312 & 6.4601 / 0.307 \\
KRFR & \textbf{6.6962 / 0.286} & 6.6962 / 0.286 & 6.6962 / 0.286 & 6.6962 / 0.286 \\
WCE-NLGFR & 6.9164 / 0.261 & 6.9164 / 0.261 & 6.9164 / 0.261 & 6.9164 / 0.261 \\
RF-LC & 7.1214 / 0.238 & 7.1214 / 0.238 & 7.1214 / 0.238 & 7.1214 / 0.238 \\
Graph NW & 7.2566 / 0.222 & 7.3477 / 0.213 & 7.3509 / 0.213 & 7.3925 / 0.209 \\
Graph kNN & 7.7464 / 0.169 & 7.8109 / 0.162 & 7.7707 / 0.167 & 7.8814 / 0.155 \\
kNN-FR & 7.6715 / 0.177 & 7.6715 / 0.177 & 7.6715 / 0.177 & 7.6715 / 0.177 \\
NW-FR & 7.2944 / 0.219 & 7.2944 / 0.219 & 7.2944 / 0.219 & 7.2944 / 0.219 \\
\bottomrule
\end{tabular*}
\label{tab:app-hbn-a-all-u}
\end{table}

\begin{table}[!htbp]
\centering\footnotesize
\setlength{\tabcolsep}{2pt}
\caption{HBN-B: mean loss / mean \(R^2\) by unlabeled pool. Bold marks the lowest mean loss in each column, including ties.}
\begin{tabular*}{\textwidth}{@{\extracolsep{\fill}}lrrrr@{}}
\toprule
Method & \(U=0\) & \(U=463\) & \(U=828\) & \(U=1314\) \\
\midrule
TH & \textbf{6.6311 / 0.295} & 6.2859 / 0.324 & 6.3529 / 0.317 & 6.4245 / 0.312 \\
PH & \textbf{6.6311 / 0.295} & \textbf{6.2313 / 0.333} & \textbf{6.2884 / 0.326} & \textbf{6.3761 / 0.317} \\
RH & \textbf{6.6311 / 0.295} & 6.4568 / 0.306 & 6.5825 / 0.292 & 6.6451 / 0.287 \\
SS-KRFR & \textbf{6.6311 / 0.295} & 6.2727 / 0.327 & 6.3793 / 0.314 & 6.5203 / 0.299 \\
KRFR & \textbf{6.6311 / 0.295} & 6.6311 / 0.295 & 6.6311 / 0.295 & 6.6311 / 0.295 \\
WCE-NLGFR & 6.8781 / 0.265 & 6.8781 / 0.265 & 6.8781 / 0.265 & 6.8781 / 0.265 \\
RF-LC & 7.1099 / 0.242 & 7.1099 / 0.242 & 7.1099 / 0.242 & 7.1099 / 0.242 \\
Graph NW & 7.2296 / 0.229 & 7.3961 / 0.212 & 7.4852 / 0.202 & 7.5625 / 0.194 \\
Graph kNN & 7.5131 / 0.196 & 7.6990 / 0.175 & 7.6539 / 0.180 & 7.7777 / 0.165 \\
kNN-FR & 7.3985 / 0.209 & 7.3985 / 0.209 & 7.3985 / 0.209 & 7.3985 / 0.209 \\
NW-FR & 7.3169 / 0.220 & 7.3169 / 0.220 & 7.3169 / 0.220 & 7.3169 / 0.220 \\
\bottomrule
\end{tabular*}
\label{tab:app-hbn-b-all-u}
\end{table}

\FloatBarrier

\clearpage

\paragraph{Coverage and exceptional branches.}
Every displayed method/\(U\) has finite test and validation predictions
on all 200 evaluated repetitions; no failed query is removed from a mean.
For HBN-B this statement concerns the validation-selected mixture at
each method/\(U\). The recorded ridge, graph and empty-kernel-support
branches are tabulated separately, so full coverage is not confused
with execution of an unmodified original comparator.
In particular, a \texttt{ZERO\_WEIGHT\_TIE} is an empty compact-kernel
neighborhood even with distinct participants, not evidence of a bootstrap
sample. \texttt{CONNECTED\_GRAPH\_ADAPTATION} counts only queries whose
fitting graph required a kNN construction degree above four and for which
no other exceptional status was recorded.
All methods retain the common \(10^{-6}\) weight-sum tolerance.

\begin{table}[!htbp]
\centering\scriptsize
\setlength{\tabcolsep}{2pt}
\caption{Recorded ridge, empty-support and graph-adaptation branches at each method's displayed best \(U\). These are accepted predictions, not omitted queries. All displayed methods have complete test and validation coverage; HBN-B counts use validation-selected \(\omega\). The unit-sum tolerance is uniformly \(10^{-6}\).}
\begin{tabular*}{\textwidth}{@{\extracolsep{\fill}}lll lrr@{}}
\toprule
Setting & Method & \(U\) & Status & Queries & \% \\
\midrule
S1 & Graph NW & 2000 & \texttt{ZERO\_WEIGHT\_TIE} & 145/100000 & 0.14 \\
S1 & NW-FR & 0 & \texttt{ZERO\_WEIGHT\_TIE} & 100/100000 & 0.10 \\
S2 & Graph NW & 1000 & \texttt{ZERO\_WEIGHT\_TIE} & 176/100000 & 0.18 \\
S2 & NW-FR & 0 & \texttt{ZERO\_WEIGHT\_TIE} & 210/100000 & 0.21 \\
S3 & Graph NW & 500 & \texttt{ZERO\_WEIGHT\_TIE} & 84/100000 & 0.08 \\
S3 & NW-FR & 0 & \texttt{ZERO\_WEIGHT\_TIE} & 93/100000 & 0.09 \\
S4 & Graph NW & 0 & \texttt{ZERO\_WEIGHT\_TIE} & 2633/100000 & 2.63 \\
S4 & Graph kNN & 0 & \texttt{AMBIENT\_TRAIN\_1NN\_FALLBACK} & 241/100000 & 0.24 \\
S4 & NW-FR & 0 & \texttt{ZERO\_WEIGHT\_TIE} & 2633/100000 & 2.63 \\
S5 & Graph NW & 0 & \texttt{ZERO\_WEIGHT\_TIE} & 3979/100000 & 3.98 \\
S5 & Graph kNN & 2000 & \texttt{AMBIENT\_TRAIN\_1NN\_FALLBACK} & 16/100000 & 0.02 \\
S5 & NW-FR & 0 & \texttt{ZERO\_WEIGHT\_TIE} & 3979/100000 & 3.98 \\
S6 & Graph NW & 0 & \texttt{ZERO\_WEIGHT\_TIE} & 3805/100000 & 3.81 \\
S6 & Graph kNN & 1000 & \texttt{AMBIENT\_TRAIN\_1NN\_FALLBACK} & 35/100000 & 0.04 \\
S6 & NW-FR & 0 & \texttt{ZERO\_WEIGHT\_TIE} & 3805/100000 & 3.81 \\
HBN-A & Graph NW & 0 & \texttt{ZERO\_WEIGHT\_TIE} & 8455/64400 & 13.13 \\
HBN-A & Graph NW & 0 & \texttt{CONNECTED\_GRAPH\_ADAPTATION} & 16544/64400 & 25.69 \\
HBN-A & Graph kNN & 0 & \texttt{CONNECTED\_GRAPH\_ADAPTATION} & 19320/64400 & 30.00 \\
HBN-A & NW-FR & 0 & \texttt{ZERO\_WEIGHT\_TIE} & 9141/64400 & 14.19 \\
HBN-B & Graph NW & 0 & \texttt{ZERO\_WEIGHT\_TIE} & 4709/64400 & 7.31 \\
HBN-B & Graph NW & 0 & \texttt{CONNECTED\_GRAPH\_ADAPTATION} & 16925/64400 & 26.28 \\
HBN-B & Graph kNN & 0 & \texttt{CONNECTED\_GRAPH\_ADAPTATION} & 16744/64400 & 26.00 \\
HBN-B & NW-FR & 0 & \texttt{ZERO\_WEIGHT\_TIE} & 4798/64400 & 7.45 \\
\bottomrule
\end{tabular*}
\label{tab:app-exception-coverage}
\end{table}

\FloatBarrier

\paragraph{Computing environment.}
Experiments were conducted on an Apple M2 Pro MacBook Pro
with a 12-core CPU and 16\,GB of unified memory, using
Python 3.13.9 on macOS. All computations were CPU-based,
without GPU acceleration. 
\clearpage

\section{Additional Methodology and Preliminaries}\label{app:preliminary}

Throughout the Appendix, we work under the setting of Section~\ref{sec:reg-model}.
Unless stated otherwise, the basic prediction bound uses
Conditions~\aref{M1}, \aref{M2}--\aref{M3}, \aref{K1}--\aref{K3}, and \aref{R1}--\aref{R2} exactly
as stated in Section~\ref{sec:regularity-conditions}, together with
\eqref{eq:basic-tuning-regime}. The relative-spectral rate result
uses Conditions~\aref{M1}, \aref{M2}--\aref{M3}, \aref{K1},
\aref{K4}--\aref{K5}, and \aref{R1}--\aref{R2}, together with
\eqref{eq:sharp-tuning-regime}; \aref{K5} implies \aref{K2}.

We use the compact support $\mathcal S=\operatorname{supp}(\bbP_X)$ and the
fixed Borel set $\cX_0\subseteq\mathcal S$ with $\bbP_X(\cX_0)=1$ specified
in the main text. Neither set is assumed known. The ambient metric space
$\cX$ need not be compact. All conditional-moment, identification, and
curvature assertions are made on the same set $\cX_0$. The population
RKHSs and their operators are defined using the restrictions of $\kappa$
and $\kappa_0$ to $\mathcal S\times\mathcal S$; their pointwise feature
identities are used for $x\in\mathcal S$. Empirical kernel formulae,
optimization, and measurable selectors remain defined for every ambient
query $x\in\cX$. Since
\[
    \bbP\{X_1,\ldots,X_N,X^\circ\in\cX_0\}=1,
\]
all sample and prediction-risk arguments may be carried out on $\cX_0$.

\subsection{Notation and conventions}
\label{app:notations}

For real \(a,b\), write \(a\land b\) and \(a\lor b\) for their minimum and
maximum, respectively; \(\bigwedge\) and \(\bigvee\) denote indexed extrema.

For \(N\ge1\), write \([N]:=\{1,\ldots,N\}\). For an assertion \(A\), write
\(\ind\{A\}\) for its indicator. For nonnegative numbers
\(a\) and \(b\) depending on \(n,N,s\), write \(a\lesssim b\) when
\(a\le Cb\) for a constant independent of \(n,N,s\). Write
\(a\gtrsim b\) for the reverse inequality and \(a\asymp b\) when both hold.
Constants may depend on fixed quantities in the regularity conditions, but
not on $n,N,s$. The source parameters $r,B_r$ and the population
kernel are fixed; the rank-dependent factors $b_{r,s}$ and $G_{r,s}$
are displayed explicitly. The
sub-Gaussian Orlicz norm of a real random variable \(Z\) under a probability
law \(Q\) is
\[
    \|Z\|_{\psi_2(Q)}
    :=
    \inf\left\{
      c>0:
      \bbE_Q\!\left[\exp\!\left(\frac{Z^2}{c^2}\right)\right]\le2
    \right\},
\]
with \(\inf\varnothing=\infty\). When the underlying law is clear, write
\(\|Z\|_{\psi_2}\). For a bounded operator \(A\), write \(\|A\|_\op\) and
\(\|A\|_\HS\) for its operator and Hilbert--Schmidt norms, and let
\(\mathrm{Id}\) denote the identity on the ambient Hilbert space specified
by context (by default, \(\mathcal H_{\kappa_0}\)).
Unless explicitly stated otherwise, covariance operators, their
perturbations, and the associated operator and Hilbert--Schmidt norms
are interpreted on $\mathcal H_{\kappa_0}$. Fractional source spaces
and their operator powers are defined using the $L^2(\bbP_X)$ realization
of $\mathcal T_{\kappa_0}$.
For a subset \(S\) of a Hilbert space \(\mathcal H\),
\(\overline{\operatorname{span}}(S)\) denotes the closure, in the
\(\mathcal H\)-norm, of the set of all finite linear combinations of
elements of \(S\).

All random elements and maps are measurable with respect to their displayed
Borel and product sigma-fields. Under the measurable-version convention in
Section~\ref{sec:regularity-conditions},
every normalized-contrast process used under an indexed supremum admits a
jointly measurable, separable version, so that the supremum equals
almost surely, and conditionally almost surely when applicable, the supremum
over a fixed countable determining subclass. Consequently, indexed suprema,
conditional expectations, and conditional \(L^2\) norms appearing below are
measurable random variables.

Let
\[
    \mathscr X_{N,L}:=\sigma(X_1,\ldots,X_N,L),
    \qquad \mathscr D_{n,N}
    :=
    \sigma(X_1,\ldots,X_N,L,\{Y_i:i\in L\}).
\]
For a random function \(F(X^\circ)\) and a sigma-field \(\mathscr A\), write
\[
    \|F\|_{2,X\mid\mathscr A}
    :=
    \left\{
        \bbE_{X^\circ}[|F(X^\circ)|^2\mid\mathscr A]
    \right\}^{1/2}.
\]
All expectations involving \(X^\circ\) integrate with respect to the
independent test covariate. When a conditional norm is written as
\(O_p(r_{n,N})\) or \(o_p(r_{n,N})\), the assertion concerns the resulting
random variable under the joint law after the conditioning variables are
restored. Conditional second-moment bounds are converted to this form by
conditional Markov inequalities, and bounds holding pointwise given a
sigma-field are used with respect to the same joint probability.

For normalized contrasts, put $\cM:=m_\oplus(\cX_0)$. The index set introduced
in the main text is then
\[
    \mathfrak T=\{(z,y):z\in\cM,\ y\in\cY,\ y\ne z\}.
\]
For \(t=(z,y)\in\mathfrak T\), set
\[
\begin{alignedat}{2}
    &d_t:=d_{\cY}(y,z),\;&\psi_t&:=\psi_{z,y},\\
    &\overline\psi_t:=\bbE[\psi_t(Y)],\;&H_t(x)&:=\bbE[\psi_t(Y)\mid X=x],\\
    &\zeta_t(x,v):=\psi_t(v)-H_t(x),\;&\dot h_t(x)&:=H_t(x)-\overline\psi_t.
\end{alignedat}
\]
Here $H_t(x)$, $\zeta_t(x,v)$, and $\dot h_t(x)$ are used for
$x\in\cX_0$. The target contrasts belong to the fixed space
$\overline{\mathcal H_{\kappa_0}}^{\,L^2(\bbP_X)}$ under
\aref{R1}, but need not be RKHS functions. For every deterministic cutoff,
\[
 c_j(t)=\bbE[\dot h_t(X)\phi_j(X)],\qquad
 g_{t,s}=\Pi_{\le s}\dot h_t=\sum_{j\le s}c_j(t)\phi_j,\qquad
 r_{t,s}=\dot h_t-g_{t,s}.
\]
The approximants are proof devices and are not computed by the estimator.
Their norm and residual bounds are $G_{r,s}$ and $b_{r,s}$ in
\eqref{eq:R1-regression-truncation}. The source exponent $r$ and radius
$B_r$ do not depend on the sample sizes. In particular, reducing the
residual does not require changing the population RKHS. The common
full-measure agreement of conditional and spectral versions is justified
in Lemma~\ref{lem:population-rkhs-surrogate}. No extension of the
conditional law outside $\cX_0$ is required.

\subsection{RKHS and metric geometry}
\label{app:RKHS}

This subsection records the RKHS facts and operator notation used in the
main text and in the proofs below. We use the same kernel symbols for their
restrictions to $\mathcal S\times\mathcal S$ and write
$K_x=\kappa(x,\cdot)|_{\mathcal S}$ for $x\in\mathcal S$.
Accordingly, $\mathcal H_\kappa$ and $\mathcal H_{\kappa_0}$ below are
spaces of functions on $\mathcal S$.

\begin{definition}[Positive definite kernel and RKHS]
\label{def:app-pd-kernel-rkhs}
A symmetric function
$\kappa:\mathcal S\times\mathcal S\to\mathbb R$ is a \emph{positive definite kernel} if,
for every $m\ge1$, $x_1,\ldots,x_m\in\mathcal S$, and
$c_1,\ldots,c_m\in\mathbb R$,
\[
    \sum_{a=1}^m\sum_{b=1}^m
    c_ac_b\kappa(x_a,x_b)\ge0.
\]
A Hilbert space $\mathcal H_\kappa$ of real-valued functions on $\mathcal S$ is the
\emph{reproducing kernel Hilbert space} associated with $\kappa$ if
$K_x:=\kappa(x,\cdot)\in\mathcal H_\kappa$ for every $x\in\mathcal S$ and
\[
    f(x)
    =\langle f,K_x\rangle_{\mathcal H_\kappa},
    \qquad
    f\in\mathcal H_\kappa.
\]
Every positive definite kernel determines a unique RKHS.
\end{definition}

The reproducing property implies
\[
    |f(x)|
    \le
    \|f\|_{\mathcal H_\kappa}\sqrt{\kappa(x,x)},
    \qquad
    \|K_x-K_{x'}\|_{\mathcal H_\kappa}^2
    =
    \kappa(x,x)+\kappa(x',x')-2\kappa(x,x').
\]
Since $\mathcal S$ is compact and $\kappa$ is continuous, the feature map
$x\mapsto K_x$ is continuous and
$\sup_{x\in\mathcal S}\kappa(x,x)<\infty$.  In particular, every
$f\in\mathcal H_\kappa$ is continuous and bounded on $\mathcal S$.

\paragraph{Population centering.}
Let $X\sim \bbP_X$.  Boundedness of $\kappa$ guarantees that the Bochner mean
embedding
\[
    m_\kappa:=\bbE[K_X]\in\mathcal H_\kappa
\]
is well-defined.  Set
\[
    K_x^0:=K_x-m_\kappa
\]
and define
\[
    \kappa_0(x,x'):=\langle K_x^0,K_{x'}^0\rangle_{\mathcal H_\kappa}
    =\kappa(x,x')-\bbE[\kappa(x,X)]-\bbE[\kappa(X,x')]+\bbE[\kappa(X,X')],
\]
where $X'$ is an independent copy of $X$.

\begin{proposition}[Centered RKHS realization]
\label{prop:app-centered-rkhs-realization}
The function $\kappa_0$ is a continuous positive definite kernel.  Its RKHS
is realized as follows.  Let
\[
    \mathcal H_0
    :=\overline{\operatorname{span}}
      \{K_x-m_\kappa:x\in\mathcal S\}
      \subseteq\mathcal H_\kappa
\]
and define $\mathcal J:\mathcal H_0\to\mathbb R^{\mathcal S}$ by
$(\mathcal J w)(x)=\langle w,K_x^0\rangle_{\mathcal H_\kappa}$.
Then $\mathcal J$ is an isometry of $\mathcal H_0$ onto
$\mathcal H_{\kappa_0}$. Under this identification, the function
$\mathcal J(K_x^0)=\kappa_0(x,\cdot)$ is the representer of evaluation and
\begin{align*}
    f(x)
    &=\langle f,K_x^0\rangle_{\mathcal H_{{\kappa_0}}},\\
    \bbE[K_X^0]
    &=0,
    \qquad
    \bbE[f(X)]=0
    \quad
    \text{for every }f\in\mathcal H_{{\kappa_0}},\\
    \kappa_0(x,x)
    &=\|K_x^0\|_{\mathcal H_{{\kappa_0}}}^2,\\
    \|K_x^0-K_{x'}^0\|_{\mathcal H_{{\kappa_0}}}^2
    &=
    \kappa_0(x,x)+\kappa_0(x',x')-2\kappa_0(x,x').
 \end{align*}
Moreover, with the ambient bound $B_\kappa:=\sup_{x\in\cX}\kappa(x,x)$
from Condition~\nameref{par:K1},
\[
    \sup_{x\in\mathcal S}\kappa_0(x,x)
    \le4B_\kappa.
\]
\end{proposition}

\begin{proof}
The displayed formula for $\kappa_0$ is a feature-map representation, so
$\kappa_0$ is positive definite.  Continuity follows from continuity of
$x\mapsto K_x$.  If $\mathcal J w=0$, then $w$ is orthogonal to every $K_x^0$ and
hence to $\mathcal H_0$, so $w=0$. Thus $\mathcal J$ is injective. Equip
\(\mathcal J(\mathcal H_0)\) with the inner product transported from \(\mathcal H_0\),
so that
\(\langle \mathcal J w,\mathcal J v\rangle_{\mathcal J(\mathcal H_0)}
  :=\langle w,v\rangle_{\mathcal H_\kappa}\).
Then \(\mathcal J\) is an isometry and \(\mathcal J(\mathcal H_0)\) is complete. Moreover,
\[
 \langle \mathcal J(K_x^0),\mathcal J(K_{x'}^0)\rangle_{\mathcal J(\mathcal H_0)}
 =\langle K_x^0,K_{x'}^0\rangle_{\mathcal H_\kappa}
 =\kappa_0(x,x').
\]
To make this identification explicit, recall that the Moore--Aronszajn
construction starts from finite linear combinations of the kernel sections
\(\kappa_0(x,\cdot)\), equips them with the inner product
\[
 \left\langle
   \sum_{i=1}^m c_i\kappa_0(x_i,\cdot),
   \sum_{j=1}^{\ell}d_j\kappa_0(y_j,\cdot)
 \right\rangle
 :=
 \sum_{i=1}^m\sum_{j=1}^{\ell}
 c_id_j\kappa_0(x_i,y_j),
\]
identifies elements whose difference has zero norm, and takes the
Hilbert-space completion.  Here
\(\mathcal J(K_x^0)=\kappa_0(x,\cdot)\), and the definition of \(\mathcal H_0\)
shows that
\[
 \mathcal J(\mathcal H_0)
 =
 \overline{\operatorname{span}}
 \{\kappa_0(x,\cdot):x\in\mathcal S\}
\]
in the transported norm. Thus \(\mathcal J(\mathcal H_0)\) is exactly the complete
space produced by this construction and is canonically isometric to
\(\mathcal H_{\kappa_0}\). Moreover, for \(w\in\mathcal H_0\),
\[
 (\mathcal J w)(x)
 =
 \langle w,K_x^0\rangle_{\mathcal H_\kappa}
 =
 \langle \mathcal J w,\mathcal J(K_x^0)\rangle_{\mathcal J(\mathcal H_0)},
\]
which gives the displayed reproducing identity. Since
$\bbE[K_X^0]=\bbE[K_X]-m_\kappa=0$, taking inner products with
$f\in\mathcal H_{{\kappa_0}}$ gives the displayed mean-zero identities.  The
remaining identities follow from the Hilbert-space norm, and
\[
    \|K_x^0\|_{\mathcal H_\kappa}
    \le
    \|K_x\|_{\mathcal H_\kappa}+\|m_\kappa\|_{\mathcal H_\kappa}
    \le2\sqrt{B_\kappa},
\]
which proves the stated bound.
\end{proof}

Throughout the sequel, write
\[
    C_\kappa
    :=
    \sup_{x\in\mathcal S}\kappa_0(x,x).
\]
The preceding bound gives $C_\kappa\le4B_\kappa<\infty$. Thus uniform
boundedness of the centered
feature map follows from the standing bounded-kernel assumption.

Centering separates the intercept from the spectral component: all functions
in $\mathcal H_{{\kappa_0}}$ have zero $\bbP_X$-mean by the centered RKHS
realization.  This is why the estimator in the main
text estimates the marginal term by an explicit intercept and applies the
kernel expansion to centered contrasts.

\paragraph{Covariance and integral operators.}
For $u,v,h\in\mathcal H_{\kappa_0}$, the rank-one outer-product operator
$u\otimes v:\mathcal H_{\kappa_0}\to\mathcal H_{\kappa_0}$ is
\[
 (u\otimes v)h:=\langle v,h\rangle_{\mathcal H_{\kappa_0}}u.
\]
In particular, for $x\in\mathcal S$, the reproducing property gives
\[
 (K_x^0\otimes K_x^0)h
 =\langle K_x^0,h\rangle_{\mathcal H_{\kappa_0}}K_x^0
 =h(x)K_x^0.
\]
Thus $K_X^0\otimes K_X^0$ is a random operator, not a scalar product;
it maps an input function $h$ to the random RKHS element $h(X)K_X^0$.
For $f\in L^2(\bbP_X)$, define the integral operator
$\mathcal T_{\kappa_0}:L^2(\bbP_X)\to L^2(\bbP_X)$ by the representative
\[
 (\mathcal T_{\kappa_0}f)(x)
 :=\int_{\mathcal S}\kappa_0(x,x')f(x')\,d\bbP_X(x')
 =\bbE[\kappa_0(X,x)f(X)],\qquad x\in\mathcal S.
\]
The input is the function $f$; $x$ is the point at which its output is
evaluated. Changing $f$ on a $\bbP_X$-null set does not change this
integral. Under the bounded-kernel assumption,
$\bbE[f(X)K_X^0]$ is a well-defined RKHS element, and its
$L^2(\bbP_X)$ class is the displayed output.

We use the same symbol $\mathcal T_{\kappa_0}$ for the RKHS realization
$\mathcal H_{\kappa_0}\to\mathcal H_{\kappa_0}$. On an input
$h\in\mathcal H_{\kappa_0}$ it returns the RKHS element
$\bbE[h(X)K_X^0]$, with the same pointwise formula.
In this realization, $\mathcal T_{\kappa_0}=\bbE[K_X^0\otimes K_X^0]$
is the average of the random rank-one operators over $X\sim\bbP_X$.
This Bochner expectation in the trace-class operators is well-defined,
since
\[
 \bbE\|K_X^0\otimes K_X^0\|_{\mathrm{tr}}
 =\bbE\|K_X^0\|_{\mathcal H_{\kappa_0}}^2
 =\bbE\kappa_0(X,X)<\infty,
\]
where $\|\cdot\|_{\mathrm{tr}}$ is the trace norm. Because
$\bbE K_X^0=0$, it is the covariance operator of the centered random
feature $K_X^0$, analogous to $\bbE[ZZ^\top]$ for a centered random vector
$Z$.
The two realizations agree on RKHS inputs, but retain their respective
ambient inner products. In particular, covariance perturbations and
operator and Hilbert--Schmidt norms below use the RKHS realization,
whereas fractional source spaces use the $L^2(\bbP_X)$ realization.
The operator $\mathcal T_{\kappa_0}$ is positive, self-adjoint, and trace class, and for
$f,g\in\mathcal H_{{\kappa_0}}$,
\[
    \langle f,\mathcal T_{{\kappa_0}}g\rangle_{\mathcal H_{{\kappa_0}}}
    =\bbE[f(X)g(X)],
    \qquad
    \operatorname{tr}(\mathcal T_{{\kappa_0}})
    =\bbE[\kappa_0(X,X)]<\infty.
\]
The two realizations have the same nonzero eigenvalues. For a positive
eigenvalue, the corresponding $L^2(\bbP_X)$ eigenfunction has a unique
continuous representative in $\mathcal H_{{\kappa_0}}$; we use the same
symbol for both realizations. We write the positive eigenpairs as
$(\lambda_j,\phi_j)_{j\ge1}$, with positive eigenvalues listed in
nonincreasing order and repeated according to multiplicity:
\[
    \lambda_1\ge\lambda_2\ge\cdots>0,
    \qquad
    \mathcal T_{{\kappa_0}}\phi_j=\lambda_j\phi_j,
    \qquad
    \langle\phi_j,\phi_k\rangle_{L^2(\bbP_X)}
    =\ind\{j=k\}.
\]
Under Condition~\nameref{par:K2}, each truncation level used
below is a spectral-block boundary,
\[
    \lambda_s>\lambda_{s+1}>0.
\]
Thus multiplicities are allowed inside $\mathcal H_s$ and inside its
orthogonal complement, and the cutoff respects repeated-eigenvalue blocks.
With the $L^2(\bbP_X)$ normalization above,
\[
    \left\{
        \sqrt{\lambda_j}\,\phi_j
    \right\}_{j\ge1}
    \quad\text{is orthonormal in }\mathcal H_{{\kappa_0}}.
\]
Thus the population eigenfunctions $\phi_j$ are $L^2(\bbP_X)$-normalized,
whereas their RKHS-normalized versions are
$\sqrt{\lambda_j}\phi_j$.

Under Condition~\nameref{par:K1}, the RKHS realization of
$\mathcal T_{{\kappa_0}}$ has trivial null space. Indeed, for
$f\in\mathcal H_{\kappa_0}$,
$\mathcal T_{{\kappa_0}}f=0$ implies $\bbE[f(X)^2]=0$ by the covariance identity
above; continuity of $f$ and full
support of $\bbP_X$ on $\mathcal S$ then imply $f\equiv0$ on $\mathcal S$.  Hence the positive-eigenvalue
system displayed above is complete in
$\mathcal H_{{\kappa_0}}$ by the spectral theorem.

\paragraph{The $L^2$ closure and fractional source spaces.}
The inclusion $J:\mathcal H_{\kappa_0}\to L^2(\bbP_X)$ is bounded and
injective. The RKHS realization of $\mathcal T_{\kappa_0}$ is $J^*J$,
whereas its $L^2(\bbP_X)$ realization is $JJ^*$; these expressions act on
different Hilbert spaces. Identifying RKHS functions with their
$L^2(\bbP_X)$ classes gives
\[
 \overline{\mathcal H_{\kappa_0}}^{\,L^2(\bbP_X)}
 =\overline{\operatorname{span}\{\phi_j:j\ge1\}}^{L^2(\bbP_X)}
 =\ker(\mathcal T_{\kappa_0})^\perp.
\]
Here the kernel and orthogonal complement are taken in $L^2(\bbP_X)$;
the trivial-null-space statement above concerns only the RKHS realization.
For $0<r\le1$, using fractional powers and ranges in $L^2(\bbP_X)$, define
\[
 [\mathcal H_{\kappa_0}]^r:=\operatorname{Ran}(\mathcal T_{\kappa_0}^{r/2})
 =\left\{\sum_jc_j\phi_j:\sum_jc_j^2/\lambda_j^r<\infty\right\}.
\]
At $r=1$ this is the embedded RKHS; at $r<1$ it can strictly contain it.
These are source/interpolation spaces in the sense used in misspecified
KRR \citep{zhang2023optimality}. In contrast, the closure in the RKHS
norm would not enlarge the RKHS at all. We use the same notation
$\Pi_{\le s}$ for the spectral projector in $L^2(\bbP_X)$ and its RKHS
restriction, defining
\[
 \Pi_{\le s}f:=\sum_{j\le s}\langle f,\phi_j\rangle_2\phi_j,
 \qquad f\in L^2(\bbP_X).
\]
The two realizations agree on RKHS functions: for
$h\in\mathcal H_{\kappa_0}$ and $e_j=\sqrt{\lambda_j}\phi_j$,
\[
 \sum_{j\le s}\langle h,e_j\rangle_{\mathcal H_{\kappa_0}}e_j
 =\sum_{j\le s}\langle h,\phi_j\rangle_2\phi_j
 =\Pi_{\le s}h.
\]
The input determines the ambient space. In particular,
$\Pi_{\le s}\dot h_t$ uses the $L^2$ realization when the target is outside
the RKHS; operator norms and projector products in the RKHS eigenspace
perturbation arguments retain their RKHS interpretation.

If the original kernel is universal on $\mathcal S$, its RKHS is dense
in $C(\mathcal S)$ \citep{steinwart2001influence}. Centering each
approximant shows that $\overline{\mathcal H_{\kappa_0}}^{\,L^2(\bbP_X)}=L_0^2(\bbP_X)$: for
$f\in\mathcal H_\kappa$, $f-\bbE f(X)$ belongs to the centered RKHS
realization, and continuous centered functions are dense in
$L_0^2(\bbP_X)$. Universality alone does not imply the fractional
source bound or the trace summability in \aref{R1}; those are
additional quantitative assumptions.

For a truncation level $s$, define
\[
    \mathcal H_s
    :=\operatorname{span}\{\sqrt{\lambda_1}\phi_1,\ldots,\sqrt{\lambda_s}\phi_s\}
    =\operatorname{span}\{\phi_1,\ldots,\phi_s\},
\]
and use $\Pi_{\le s}$ for the orthogonal projector onto $\mathcal H_s$
in either ambient space, as above. Put
$\Pi_{>s}:=\mathrm{Id}-\Pi_{\le s}$ in the same space. The population
RKHS-normalized feature vector is
\[
    \boldsymbol\Phi_s(x)
    :=
    \bigl(
        \sqrt{\lambda_1}\phi_1(x),\ldots,
        \sqrt{\lambda_s}\phi_s(x)
    \bigr)^\top.
\]
It satisfies
\begin{align*}
    \bbE[\boldsymbol\Phi_s(X)]
    &=0,\\
    \bbE[\boldsymbol\Phi_s(X)\boldsymbol\Phi_s(X)^\top]
    &=\boldsymbol\Lambda_s
    :=\operatorname{diag}(\lambda_1,\ldots,\lambda_s),\\
    \|\boldsymbol\Phi_s(x)\|_2^2
    &=\sum_{j=1}^s\lambda_j\phi_j(x)^2
    \le\kappa_0(x,x).
\end{align*}
The last inequality is Bessel's inequality applied to the
RKHS-orthonormal system
$(\sqrt{\lambda_j}\phi_j)_{j\ge1}$ and the representer $K_x^0$.

\paragraph{Normalization conventions.}
The notation used throughout the paper is summarized as follows:
\[
\begin{array}{c|cc}
& L^2\text{-normalized}
& \text{RKHS-normalized}\\
\hline
\text{population}
& \phi_j
& \sqrt{\lambda_j}\phi_j\\[1mm]
\text{empirical}
& \widehat\varphi_j
& \widehat\phi_j.
\end{array}
\]
More precisely, the empirical eigensystem is constructed in
Section~\ref{app:empirical-spectral-construction}. On \(\mathcal E_s\),
\[
    \widehat\varphi_j
    :=\widehat\lambda_j^{-1/2}\widehat\phi_j,
    \qquad
    \widehat{\boldsymbol\Phi}_s(x)
    :=
    (\widehat\phi_1(x),\ldots,\widehat\phi_s(x))^\top.
\]
Thus $\sqrt{\widehat\lambda_j}$ is already absorbed into the empirical
RKHS-normalized eigenfunction $\widehat\phi_j$.

\paragraph{Kernel-side notation.}
The basic Conditions~\nameref{par:K1}--\nameref{par:K3} and the additional
Conditions~\nameref{par:K4}--\nameref{par:K5} are stated in the main
text. Under Condition~\nameref{par:K4}, fix a population
eigenbasis satisfying its eighth-moment bound. The complete-block projectors
used below are invariant under rotations within repeated-eigenvalue blocks;
no moment bound is imposed on every possible rotated basis. This subsection
records the operator notation used below.
Let
\[
    e_j:=\sqrt{\lambda_j}\phi_j,
    \qquad
    \Pi_j:=e_j\otimes e_j.
\]
Thus $(e_j)_{j\ge1}$ is an orthonormal basis of
$\mathcal H_{{\kappa_0}}$ and
\[
    \Pi_{\le s}=\sum_{j\le s}\Pi_j,
    \qquad
    \Pi_{>s}=\sum_{j>s}\Pi_j.
\]
Let
\[
    \widetilde{\mathcal T}_N
    :=
    \frac1N\sum_{i=1}^N K_{X_i}^0\otimes K_{X_i}^0.
\]
The empirically centered covariance operator satisfies
\[
    \widehat{\mathcal T}_{{\kappa_0},N}
    =
    \widetilde{\mathcal T}_N
    -\overline K_N^0\otimes\overline K_N^0,
    \qquad
    \overline K_N^0:=\frac1N\sum_{i=1}^N K_{X_i}^0.
\]
Define the effective dimension
\[
    \mathcal N(\mu)
    :=
    \operatorname{tr}\{\mathcal T_{{\kappa_0}}
    (\mathcal T_{{\kappa_0}}+\mu \mathrm{Id})^{-1}\}.
\]
For the principal-subspace analysis, define the relative rank across the
selected spectral cutoff by
\[
    \mathfrak r_s^{\partial}
    :=
    \sum_{j\le s}
    \frac{\lambda_j}{\lambda_j-\lambda_{s+1}}
    +
    \sum_{k>s}
    \frac{\lambda_k}{\lambda_s-\lambda_k}.
\]
This quantity is built from interactions crossing the cutoff and therefore
permits arbitrary multiplicities within equal-eigenvalue groups on either side
of the cutoff.
For $I=\{1,\ldots,s\}$, it is exactly the relative rank
$\mathbf r_I(\mathcal T_{\kappa_0})$ of \citet[Definition~1]{jirak2020perturbation}.
The block perturbation theorem of that paper assumes the contraction
condition $x\mathfrak r_s^{\partial}\le 1/8$, where $x$ is the scale of the
relative perturbation; it does not itself assume polynomial eigenvalue decay.
Condition~\nameref{par:K5} is a stronger primitive sufficient condition used
here to obtain explicit rates.  In particular,
Lemma~\ref{lem:cross-boundary-spectrum-consequences} shows that it implies
$\mathfrak r_s^{\partial}\lesssim s\log(es)$, so the relative perturbation
scale $x_N\asymp\sqrt{\log N/N}$ satisfies the required contraction under
\eqref{eq:sharp-tuning-regime}.  The strict boundary
$\lambda_s>\lambda_{s+1}$ only prevents a population eigenspace from being
split at the cutoff; eigenvalue multiplicities within complete blocks on
either side remain allowed.  Thus the block argument below applies
\citet[Theorem~2]{jirak2020perturbation}; the related individual-eigenvalue
relative-rank theory is developed in \citet{jirak2023relative}.

Two population spectral sums will be used below:
\begin{align*}
    \mathcal U_s
    &:={}
    \sum_{j\le s}\sum_{k>s}
    \frac{\lambda_j\lambda_k}
         {(\lambda_j-\lambda_k)^2},\\
    \mathcal W_s^{+}
    &:={}
    \left(\sum_{k>s}\lambda_k\right)
    \left(\sum_{j\le s}
      \frac{\lambda_j}{\lambda_j-\lambda_{s+1}}\right).
\end{align*}
The first is the variance scale for the ordinary Hilbert--Schmidt
principal-subspace error. The second is the non-asymptotic upper-bound scale
for the leading-side PCA excess risk in \citet[
Proposition~2.10]{reiss2016pca}. Their different spectral weights are what produce the
relative-spectral prediction-weighted representation bound.

\begin{example}[A Fourier kernel satisfying the spectral conditions]
\label{ex:fourier-kernel-moment-spectrum}
Let $\cX=\mathbb T:=\mathbb R/(2\pi\mathbb Z)$ with its usual circular
metric, and let $\bbP_X$ be the uniform probability measure. For $m\ge1$, put
\[
    \phi_{2m-1}(\theta)=\sqrt2\cos(m\theta),\qquad
    \phi_{2m}(\theta)=\sqrt2\sin(m\theta).
\]
These functions form an orthonormal basis of $L^2_0(\bbP_X)$ and have zero
mean. For any $q>1$, define
\[
    \kappa(\theta,\theta')
    :=1+\sum_{j\ge1}j^{-q}\phi_j(\theta)\phi_j(\theta').
\]
Since $|\phi_j(\theta)\phi_j(\theta')|\le2$ and
$\sum_jj^{-q}<\infty$, the series converges absolutely and uniformly.
Each summand is a positive semidefinite kernel, so $\kappa$ is continuous,
positive definite, and bounded by
$\sup_\theta\kappa(\theta,\theta)\le1+2\sum_jj^{-q}<\infty$.
Thus Condition~\aref{K1} holds. Population centering removes precisely the
constant term, and orthonormality gives the positive eigenpairs
$(\lambda_j,\phi_j)$ with $\lambda_j=j^{-q}$.
Consequently,
\[
    \sup_{j\ge1}\bbE|\phi_j(X)|^8\le16,
    \qquad
    \lambda_j-\lambda_k=j^{-q}-k^{-q}\quad(j<k).
\]
Condition~\aref{K4} holds, Condition~\aref{K5} holds with equality in its
cross-boundary comparison, and every positive integer cutoff satisfies
Condition~\aref{K2}. This example verifies only the kernel-side conditions;
the response conditions and tuning restrictions remain separate.
\end{example}

\begin{example}[A BHV contrast in the closure but outside the RKHS]
\label{ex:bhv-closure-source}
Consider unrooted four-taxon BHV trees with fixed pendant lengths. Their
internal-edge space is a tripod: three rays $R_1,R_2,R_3$ meeting at an
unresolved star tree $o$ \citep{billera2001geometry}. Denote radius $a\ge0$
on ray $k$ by $aR_k$. The intrinsic distance is $|a-b|$ on the same ray
and $a+b$ on distinct rays.

Take the circular covariate and Fourier kernel in
Example~\ref{ex:fourier-kernel-moment-spectrum} with $q=4$. Thus
$X\sim\operatorname{Unif}(\mathbb T)$, $\lambda_j=j^{-4}$, and
$\overline{\mathcal H_{\kappa_0}}^{\,L^2(\bbP_X)}=L_0^2(\bbP_X)$. Put $u=\sin x$ and let
\[
 T(x)=\begin{cases}
       uR_1,&u\ge0,\\
       (-u)R_2,&u<0.
       \end{cases}
 \qquad Y=W T(X),\qquad
 W\sim\operatorname{Unif}[1/2,3/2],\quad W\perp X.
\]
Scaling changes only the internal-edge length. Conditional on $X=x$,
the response lies on a single ray, has mean radius $|u|$, and has
Fr\'echet mean $m_\oplus(x)=T(x)$. In particular $o$ belongs to the target
image. For the candidate $y=aR_3$, $a>0$,
\[
 M_\oplus(y,x)=a^2+\bbE[W^2]\sin^2x+2a|\sin x|,\qquad
 M_\oplus(o,x)=\bbE[W^2]\sin^2x.
\]
Consequently the normalized centered contrast is
\[
 \dot h_{(o,aR_3)}(x)=2\big(|\sin x|-2/\pi\big).
\]
This is a nonsmooth function of the \emph{covariate}; merely observing
nonsmoothness in the response candidate would not establish
nonmembership in a covariate RKHS.

The Fourier identity
\[
 |\sin x|=\frac2\pi-\frac4\pi
     \sum_{m\ge1}\frac{\cos(2mx)}{4m^2-1}
\]
shows that this centered contrast has coefficients of order $m^{-2}$
along $\phi_{4m-1}=\sqrt2\cos(2mx)$. Thus its RKHS squared norm would
contain a divergent series of order
$\sum_m m^{-4}(4m-1)^4$. It does not belong to
$\mathcal H_{\kappa_0}$, even as an $L^2$ equivalence class with an RKHS
representative, but it does belong to $\overline{\mathcal H_{\kappa_0}}^{\,L^2(\bbP_X)}$.

This example also verifies the \emph{uniform} fractional source rather
than only one contrast. For a point of radius $a$ on ray $k$, write
\[
 F_1(u)=u,\qquad F_2(u)=-u,\qquad F_3(u)=-|u|.
\]
Its squared distance to $WT(x)$ is
$a^2+W^2u^2-2aW F_k(u)$. In the contrast between $y=aR_k$ and
$z=bR_\ell$, the $W^2u^2$ term cancels. If $k=\ell$, division by
$|a-b|$ gives a coefficient of magnitude two; if $k\ne\ell$, division
by $a+b$ gives two coefficients with total absolute value two.
After marginal centering, every contrast therefore has the form
\[
 \dot h_t(x)=\alpha_t\sin x+\beta_t\big(|\sin x|-2/\pi\big),\qquad
 |\alpha_t|+|\beta_t|\le2.
\]
The constant $\beta_t$ here is a coefficient, not the curvature exponent
$\beta_M$. The weighted coefficient sums are uniformly finite whenever
$4r<3$, while $\sum_j\lambda_j^r<\infty$ whenever $4r>1$.
Thus \aref{R1} holds for every $1/4<r<3/4$, although its exact-RKHS
case fails.

Conditions \aref{M1} and \aref{M2} hold on the whole circle. The tripod
is Hadamard, so \aref{M3} holds with $\beta_M=2$. Moreover
\[
 \zeta_t(x,WT(x))=(W-1)\{\alpha_t\sin x+\beta_t|\sin x|\}.
\]
For conditionally independent radii, Cauchy--Schwarz in $\mathbb R^2$
and conditional centering give \aref{R2} with
$C_M^2\le8\operatorname{Var}(W)$. The kernel conditions
\aref{K1}, \aref{K2}, \aref{K4}, and \aref{K5} hold by
Example~\ref{ex:fourier-kernel-moment-spectrum}.
For instance, $r=5/8$ gives $qr=5/2>2$ and the refined source rate
$O_p(n^{-3/5}+\varepsilon_{\mathrm{opt}})$ with the tuning in
Corollary~\ref{cor:polynomial-rate}.

The absolute value arises from crossing branches and comparing with a
third branch; two-ray distances alone can cancel it. Smooth mixture
probabilities or other conditional laws can also smooth a Fr\'echet
function. Hence BHV geometry alone does not imply failure of RKHS
membership. This constructed example verifies the new assumptions; it
is not an assumption verification for every experimental BHV mechanism.
\end{example}

\subsubsection{Metric geometry}\label{app:geom}

\begin{definition}[Hadamard space]
A geodesic metric space $(\cY,d_{\cY})$ is a \emph{CAT(0) space} if, for
all $y_0,y_1,v\in\cY$, every constant-speed geodesic
$\gamma:[0,1]\to\cY$ from $y_0$ to $y_1$ satisfies
\[
    d_{\cY}(v,\gamma(t))^2
    \le
    (1-t)d_{\cY}(v,y_0)^2
    +t d_{\cY}(v,y_1)^2
    -t(1-t)d_{\cY}(y_0,y_1)^2,
    \qquad t\in[0,1].
\]
A \emph{Hadamard space} is a complete CAT(0) space.
\end{definition}

\begin{definition}[Finite second moment]
A Borel probability measure $P$ on $\cY$ has a \emph{finite second
moment} if
\[
    \int_{\cY}d_{\cY}(y,y_0)^2\,dP(y)<\infty
\]
for some, and hence every, $y_0\in\cY$.
\end{definition}

In a Hadamard space, every Borel probability measure with a finite second
moment has a unique Fr\'echet mean, also called its barycenter
\citep{sturm2003probability}.

\begin{proposition}[Automatic quadratic curvature]
\label{prop:automatic-quadratic-curvature}
Under Condition~\nameref{par:M1}, the conditional law
$\bbP_{Y\mid X=x}$ has a finite second moment for every $x\in\cX_0$. The following
assertions hold.
\begin{enumerate}[label=(\roman*),leftmargin=*]
\item If $\cY$ is a Hilbert space, then for every $x\in\cX_0$ and $y\in\cY$,
\[
    \Delta_\oplus(y,x)
    =\|y-m_\oplus(x)\|^2.
\]
\item If $\cY$ is a Hadamard space, then for every $x\in\cX_0$ and $y\in\cY$,
\[
    \Delta_\oplus(y,x)
    \ge d_\cY(y,m_\oplus(x))^2.
\]
\end{enumerate}
Consequently, in either case the local curvature part of
Condition~\nameref{par:M3} holds globally in $y$, uniformly for
$x\in\cX_0$, with \(\beta_M=2\) and \(c_M=1\).
The same inequality also implies uniform identification, because for every
$\eta>0$,
\[
    \inf_{x\in\cX_0}\inf_{d_\cY(y,m_\oplus(x))\ge\eta}
    \Delta_\oplus(y,x)\ge\eta^2.
\]
Thus Condition~\nameref{par:M1}, together with Hilbert or Hadamard geometry,
supplies the curvature and identification parts of
Condition~\nameref{par:M3}; Condition~\nameref{par:M2} records the
measurability of the chosen conditional-mean or conditional-barycenter
version.
The RKHS approximation and response-noise complexity requirements are handled
separately by Conditions~\nameref{par:R1}--\nameref{par:R2}.
\end{proposition}

\begin{proof}
Fix $x\in\cX_0$. In the Hilbert case, put $m=\bbE[Y\mid X=x]$.  Since
$\bbE[Y-m\mid X=x]=0$, expansion of the square gives
\[
\begin{split}
 M_\oplus(y,x)
 &=\bbE[\|y-m+m-Y\|^2\mid X=x]\\
 &=\|y-m\|^2+\bbE[\|Y-m\|^2\mid X=x].
\end{split}
\]
Taking $y=m$ and subtracting proves part~(i). For the Hadamard case, put
$m=m_\oplus(x)$, and let $\gamma_t$ be the geodesic from $m$ to $y$.
The CAT(0) variance inequality \citep{sturm2003probability} gives
\[
    M_\oplus(\gamma_t,x)
    \le (1-t)M_\oplus(m,x)+tM_\oplus(y,x)
       -t(1-t)d_\cY(m,y)^2.
\]
Since $m$ minimizes the conditional Fr\'echet function,
$M_\oplus(m,x)\le M_\oplus(\gamma_t,x)$.  Dividing by $t>0$ yields
\[
    \Delta_\oplus(y,x)
    \ge (1-t)d_\cY(m,y)^2.
\]
Letting $t\downarrow0$ proves part~(ii).
The Hilbert conditional mean and the Hadamard barycenter exist uniquely
because Condition~\nameref{par:M1} gives finite conditional second moments.
Finally, the displayed global bound
implies the uniform-identification display in the statement.  Measurability
of $x\mapsto m_\oplus(x)$ is retained explicitly in
Condition~\nameref{par:M2}, since it concerns the chosen version of the
regular conditional law.
\end{proof}

\subsection{SS-KRFR computation and well-posedness}
\label{app:full-procedures}
\label{app:ss-krfr-full-procedure}

The empirical eigensystem, its measurable repeated-eigenvalue convention,
and its basis-free realization are given in
Section~\ref{app:empirical-spectral-construction}.  Together with the
weights in \eqref{eq:calibrated-weights-main}, these constructions implement
the estimator in \eqref{eq:ss-krfr-estimator}; here we record the
well-posedness results used by the theory.

\paragraph{Ridge-regularized coefficient.}
Let $\mathbf F=(\widehat{\boldsymbol\Phi}_s(X_i)^\top)_{i\in L}$,
$\mathbf H_L=\mathbf I_n-\mathbf1_n\mathbf1_n^\top/n$, and
$\mathbf Z_y=(Z_{y,i})_{i\in L}$. The fit in Section~\ref{sec:SS-KRFR} is
\[
 (\widehat a_y,\widehat{\boldsymbol\theta}_{y,L})
 =\operatorname*{argmin}_{a,\boldsymbol\theta}
 \left\{\frac1{2n}\|\mathbf Z_y-a\mathbf1_n-\mathbf F\boldsymbol\theta\|^2
       +\frac{\xi}{2}\|\boldsymbol\theta\|^2\right\}.
\]
The intercept normal equation gives
$\widehat a_y=\overline Z_{y,L}
-\overline{\boldsymbol\Phi}_L^\top\widehat{\boldsymbol\theta}_{y,L}$.
Substitution yields
\[
 \mathbf C_{L,s}=\mathbf F^\top\mathbf H_L\mathbf F/n,\qquad
 \widehat{\mathbf b}_{y,L}=\mathbf F^\top\mathbf H_L\mathbf Z_y/n,\qquad
 (\mathbf C_{L,s}+\xi\mathbf I_s)\widehat{\boldsymbol\theta}_{y,L}
 =\widehat{\mathbf b}_{y,L}.
\]
The centered slope Hessian is at least $\xi\mathbf I_s$; the intercept
equation then gives a unique joint fit. With
$\mathbf A_s=(\mathbf C_{L,s}+\xi\mathbf I_s)^{-1}$, its weight vector is
\[
 \mathbf a_{s,\xi}(x)=\mathbf1_n/n+
 \mathbf H_L\mathbf F\mathbf A_s\mathbf q_L(x)/n,\qquad
 \|\mathbf a_{s,\xi}(x)\|^2
 =\frac1n+\frac1n\mathbf q_L(x)^\top\mathbf A_s
             \mathbf C_{L,s}\mathbf A_s\mathbf q_L(x).
\]
The two summands are orthogonal and the weights sum to one.
A common translation of the feature map leaves $\mathbf H_L\mathbf F$
and $\mathbf q_L$ unchanged. The pooled diagonal covariance is used only
for learning and analyzing the representation, not as the regression inverse.
The following well-posedness results apply to these unit-sum weights;
they do not assert that every numerical response correction is an argmin.

\begin{proposition}[Coercivity]
\label{prop:coercivity}
For every $x\in\cX$, the map
\[
    y\longmapsto\widehat M_{s,\xi}(y,x)
    =\sum_{i\in L}a_{s,\xi,i}(x)d_{\cY}^2(y,Y_i)
\]
is coercive: it diverges whenever $d_{\cY}(y_*,y)\to\infty$ for any fixed
$y_*\in\cY$.
\end{proposition}

\begin{proof}
Fix $x\in\cX$ and a reference point $y_0\in\cY$.
The reverse triangle inequality gives
\[
    |d_{\cY}(y,Y_i)-d_{\cY}(y,y_0)|
    \le d_{\cY}(Y_i,y_0).
\]
Consequently,
\[
    |d_{\cY}^2(y,Y_i)-d_{\cY}^2(y,y_0)|
    \le
    d_{\cY}(Y_i,y_0)
    \{2d_{\cY}(y,y_0)+d_{\cY}(Y_i,y_0)\}.
\]
Since $\sum_{i\in L}a_{s,\xi,i}(x)=1$,
\[
\begin{split}
    \widehat M_{s,\xi}(y,x)
    ={}&d_{\cY}^2(y,y_0)
    +\sum_{i\in L}a_{s,\xi,i}(x)
      \{d_{\cY}^2(y,Y_i)-d_{\cY}^2(y,y_0)\}\\
    \ge{}&d_{\cY}^2(y,y_0)
    -2d_{\cY}(y,y_0)
      \sum_{i\in L}|a_{s,\xi,i}(x)|d_{\cY}(Y_i,y_0)
    -\sum_{i\in L}|a_{s,\xi,i}(x)|d_{\cY}^2(Y_i,y_0).
\end{split}
\]
For fixed data and fixed $x$, the two sums on the right-hand side are
finite constants. Hence
\[
    \widehat M_{s,\xi}(y,x)\longrightarrow\infty
    \qquad\text{as }d_{\cY}(y,y_0)\longrightarrow\infty,
\]
which proves coercivity.
\end{proof}

For the measurability statements below, let \(\mathbb D_{n,N}\) be the
space of possible observed datasets
\[
    \mathbf d=(x_1,\ldots,x_N,L,(v_i)_{i\in L}),
    \qquad L\subset[N],\quad |L|=n.
\]
With its natural Borel structure, \(\mathbb D_{n,N}\) is a standard Borel
space: it is a finite disjoint union of finite products of the Polish
spaces \(\cX\) and \(\cY\). For \(\mathbf d\in\mathbb D_{n,N}\), write
\[
    \widehat M_{s,\xi}(y,x;\mathbf d)
    :=\sum_{i\in L}a_{s,\xi,i}(x;\mathbf d)d_{\cY}^2(y,v_i)
\]
for the estimated Fr\'echet function computed from \(\mathbf d\). The data argument is
suppressed outside the measurability arguments.

\begin{proposition}[Measurable approximate minimizer]
\label{prop:measurable-approximate-minimizer}
For the estimated Fr\'echet function in \eqref{eq:ss-krfr-criterion}, for every
$\varepsilon_{\mathrm{opt}}>0$, there exists a selector jointly measurable
in the observed data and $x$ satisfying
\eqref{eq:ss-krfr-estimator}. If an exact measurable argmin selector exists,
one may take $\varepsilon_{\mathrm{opt}}=0$.
\end{proposition}

\begin{proof}
Let \(\cY_0=\{y_1,y_2,\ldots\}\) be a fixed countable dense subset of
\(\cY\), and fix \(y_0\in\cY\). The empirical matrix is a Borel function of
\(\mathbf d\), and the fixed Borel eigensystem rule specified in
Section~\ref{app:empirical-spectral-construction} makes every empirical
feature a jointly Borel function of \((\mathbf d,x)\). Thus, for every fixed
\(j\), the map
\((\mathbf d,x)\mapsto\widehat M_{s,\xi}(y_j,x;\mathbf d)\) is measurable.
The lower bound from the proof of Proposition~\ref{prop:coercivity} gives
\[
\begin{aligned}
    \widehat M_{s,\xi}(y,x;\mathbf d)
    &\ge
    d_{\cY}^2(y,y_0)
    -2d_{\cY}(y,y_0)
      \sum_{i\in L}|a_{s,\xi,i}(x;\mathbf d)|
      d_{\cY}(v_i,y_0)
    -\sum_{i\in L}|a_{s,\xi,i}(x;\mathbf d)|d_{\cY}^2(v_i,y_0)\\
    &\ge
    -\left\{
      \sum_{i\in L}|a_{s,\xi,i}(x;\mathbf d)|
      d_{\cY}(v_i,y_0)
    \right\}^2
    -\sum_{i\in L}|a_{s,\xi,i}(x;\mathbf d)|d_{\cY}^2(v_i,y_0).
\end{aligned}
\]
For fixed \((\mathbf d,x)\), the two sums in the final bound are finite.
Together with finiteness at \(y_1\), this shows that the infimum over \(y\)
is a finite real number. Hence
\[
    q(\mathbf d,x):=\inf_{\ell\ge1}
      \widehat M_{s,\xi}(y_\ell,x;\mathbf d)
\]
is jointly measurable as a countable infimum. By the definition of an infimum,
for every $\varepsilon_{\mathrm{opt}}>0$ the set
\[
    \left\{j:\widehat M_{s,\xi}(y_j,x;\mathbf d)
    \le q(\mathbf d,x)+\varepsilon_{\mathrm{opt}}\right\}
\]
is nonempty. Let \(\widehat j(\mathbf d,x)\) be its smallest element. Moreover,
\[
    \{(\mathbf d,x):\widehat j(\mathbf d,x)=j\}
    =\{(\mathbf d,x):\widehat M_{s,\xi}(y_j,x;\mathbf d)
          \le q(\mathbf d,x)+\varepsilon_{\mathrm{opt}}\}
       \cap\bigcap_{\ell<j}
       \{(\mathbf d,x):\widehat M_{s,\xi}(y_\ell,x;\mathbf d)
          >q(\mathbf d,x)+\varepsilon_{\mathrm{opt}}\},
\]
so \(\widehat j\), and therefore
\(\widehat m_{s,\xi,\varepsilon_{\mathrm{opt}}}(x;\mathbf d)
:=y_{\widehat j(\mathbf d,x)}\), is jointly measurable.

For fixed $(\mathbf d,x)$, the objective function is continuous in $y$. Since
$\cY_0$ is dense,
\[
    \inf_{\ell\ge1}\widehat M_{s,\xi}(y_\ell,x;\mathbf d)
    =\inf_{y\in\cY}\widehat M_{s,\xi}(y,x;\mathbf d).
\]
The asserted approximate optimality follows from the definition of
\(\widehat j(\mathbf d,x)\).
\end{proof}

\begin{corollary}[Proper metric space]
\label{cor:proper}
If $(\cY,d_{\cY})$ is proper, that is, every closed ball is compact, then
$\operatorname*{argmin}_{y\in\cY}\widehat M_{s,\xi}(y,x)$ is nonempty.
Thus, there is an exact selector jointly measurable in the observed data and
$x$, so one may take $\varepsilon_{\mathrm{opt}}=0$.
\end{corollary}

\begin{proof}
Fix $\mathbf d\in\mathbb D_{n,N}$ and $x$. The map
$y\mapsto\widehat M_{s,\xi}(y,x;\mathbf d)$ is continuous. By coercivity,
the sublevel set
\[
    K_{\mathbf d,x}:=\{y:\widehat M_{s,\xi}(y,x;\mathbf d)
       \le \widehat M_{s,\xi}(y_1,x;\mathbf d)+1\}
\]
is bounded. It is closed by continuity and compact since $\cY$ is proper. The global infimum equals the infimum over $K_{\mathbf d,x}$, so
continuity and compactness give a nonempty compact argmin.

For joint measurability, use the countable dense set
$\cY_0=\{y_1,y_2,\ldots\}$ and define
\[
    q(\mathbf d,x):=\inf_{j\ge1}
       \widehat M_{s,\xi}(y_j,x;\mathbf d).
\]
Continuity in $y$ shows that this equals the infimum over all $y\in\cY$, and
the countable representation makes $q$ measurable. Since the weights and response coordinates are
Borel and $y\mapsto d_\cY^2(y,v_i)$ is continuous, the objective function is jointly
Borel in $(\mathbf d,x,y)$. Hence the argmin graph
\[
 \mathcal G
 :=\{(\mathbf d,x,y):
   \widehat M_{s,\xi}(y,x;\mathbf d)=q(\mathbf d,x)\}
\]
is Borel and has nonempty closed sections. The space
$\mathbb D_{n,N}\times\cX$ is standard Borel, $\cY$ is Polish, and the
sections are compact by the first part of the proof. The Arsenin--Kunugui
uniformization theorem \citep{kechris2012classical} therefore yields a
jointly Borel exact selector.
\end{proof}

\begin{corollary}[Hilbert space]
\label{cor:Hilbert}
If $\cY$ is a Hilbert space, then the empirical minimizer is unique and
\[
    \widehat m_{s,\xi,0}(x)
    =\sum_{i\in L}a_{s,\xi,i}(x)Y_i.
\]
Thus $\varepsilon_{\mathrm{opt}}=0$ is available.
\end{corollary}

\begin{proof}
Since $\sum_{i\in L}a_{s,\xi,i}(x)=1$, expansion of the squared Hilbert
norm gives
\[
    \widehat M_{s,\xi}(y,x)
    =\left\|y-\sum_{i\in L}a_{s,\xi,i}(x)Y_i\right\|_{\cY}^2+C(x),
\]
where $C(x)$ is independent of $y$. The first term is strictly convex and
has the unique minimizer displayed in the corollary. The weights and the
responses are measurable, so this exact minimizer is measurable and one may
set $\varepsilon_{\mathrm{opt}}=0$.
\end{proof}

\subsection{Validation-selected hybrid estimation}
\label{app:hybrid-extension-details}

Let $\widehat{\boldsymbol\Phi}^{\mathrm L}_{s_L}$ and
$\widehat{\boldsymbol\Phi}^{\mathrm N}_{s_N}$ be learned separately from
labeled-only and pooled covariates. Each map has its own centering,
positive-rank cap and bandwidth; the same bandwidth multiplier need not
produce the same numerical bandwidth. For $\alphaH\in[0,1]$, define
\[
 \widehat{\boldsymbol\Phi}_{s,\alphaH}(x)=
 \begin{bmatrix}
 \sqrt{1-\alphaH}\,\widehat{\boldsymbol\Phi}^{\mathrm L}_{s_L}(x)\\
 \sqrt{\alphaH}\,\widehat{\boldsymbol\Phi}^{\mathrm N}_{s_N}(x)
 \end{bmatrix}.
\]
The experimental candidates cap each endpoint by the same nominal $s$,
so the concatenation has at most $2s$ coordinates. Fit the labeled
coefficient as in Section~\ref{app:full-procedures} to these
features, with a single unpenalized intercept. In particular, the centered
labeled covariance contains both within-block and cross-block terms.
This is ridge on the sum kernel
\[
 (1-\alphaH)\langle\widehat{\boldsymbol\Phi}^{\mathrm L}_{s_L}(x),
                  \widehat{\boldsymbol\Phi}^{\mathrm L}_{s_L}(x')\rangle
 +\alphaH\langle\widehat{\boldsymbol\Phi}^{\mathrm N}_{s_N}(x),
                  \widehat{\boldsymbol\Phi}^{\mathrm N}_{s_N}(x')\rangle,
\]
not one eigensystem of a mixture of empirical covariate laws and not a
mixture of separately fitted slopes.

At fixed common ridge, features and response rule, $\alphaH=0,1$ give
the respective endpoints. If the two endpoint maps agree up to an
orthogonal transformation, their concatenation has the same feature
kernel for every $\alphaH$ and therefore gives the same prediction.
These identities do not compare independently tuned fits.

RH jointly searches its feature, ridge and representation-weight
candidates by validation. PH first tunes KRFR and pure SS-KRFR, then
selects its prediction mixing weight between those fitted endpoints.
TH similarly combines tuned KRFR and tuned RH. Thus prediction mixing
does not perform an exhaustive joint search over all endpoint
hyperparameter pairs. The implemented geodesics are specified in
Appendix~\ref{app:experiment-details}; in a Hadamard space the exact
two-point barycenter is the point at fraction $\eta$ along the unique
geodesic. For the general metric-space definition, the measurable
approximate-selection argument in
Proposition~\ref{prop:measurable-approximate-minimizer} applies to the two
endpoint predictions.

\paragraph{Joint selection over mixing weights.}
At the common tuning of Theorem~\ref{thm:th-stability}, the concatenated
features and ridge weights are jointly Borel in the data, query and
$\alphaH$. For any common $\varepsilon_{\mathrm{opt}}>0$, the proof of
Proposition~\ref{prop:measurable-approximate-minimizer} applies with
$\alphaH$ added to the parameter space: minimize over a fixed countable
dense subset of $\cY$ and choose its first point within
$\varepsilon_{\mathrm{opt}}$ of the infimum. Unit-sum weights ensure a
finite infimum, and continuity in the response argument identifies it
with the full infimum. This yields a jointly Borel family satisfying
the same tolerance for every $\alphaH$; at $0,1$ reuse the endpoint fits.
For $\varepsilon_{\mathrm{opt}}=0$, a jointly Borel family of exact
selectors is required, not merely separate measurability at each weight.
Corollaries~\ref{cor:proper} and \ref{cor:Hilbert} extend to this enlarged
parameter space; Hadamard geometry alone does not establish existence
of a signed-criterion minimizer. In a Hadamard space, exact two-point
interpolation is jointly continuous in its endpoints and weight.
For $\varepsilon_{\mathrm{hyb}}>0$, the same countable selection argument
with $(\alphaH,\eta)$ yields a jointly Borel family with a uniform
interpolation tolerance.

The original prediction bounds concern KRFR and pure SS-KRFR with a
fixed kernel and tuning under the shared-law sampling model.
Theorem~\ref{thm:th-stability} additionally controls the concatenated
learner and prediction mixing through TH for Hadamard responses,
uniformly over mixing weights at common numerical kernel, rank and ridge,
without an additional labeled eigengap condition. RH and PH are boundary cases. The mixing
weights may be selected measurably using training or validation data;
prediction risk integrates over a fresh query independent of these data.
Corollary~\ref{cor:ph-balanced-rate} separately gives PH a balanced
polynomial rate by bounding both endpoints, without using the TH
discrepancy scale $\chi$.
These results do not automatically cover independently tuned endpoints,
data-dependent kernels, the full validation procedure, or covariate-law
mismatch. For distribution responses,
PAVA followed by clipping gives the exact finite-grid Fr\'echet minimizer
for unit-sum signed weights; this does not establish the remaining assumptions.
A CAT(0) squared-distance convexity inequality
can propagate bounds for both fitted endpoints, but does not guarantee
that validation mixing is no worse than the better endpoint.

\subsection{Normalized contrasts and spectral constructions}
\label{app:normalized-contrast-identities}

We record the squared-distance algebra behind the normalized contrast.
Notation for the coefficient-level spectral construction, the empirical
Nystr\"om extension, and the fixed-design proof is introduced in
Subsections~\ref{app:population-spectral-notation}--\ref{app:empirical-spectral-construction}
below.

For $t=(z,y)\in\mathfrak T$ and $x\in\cX_0$, define the normalized
population contrast by
\[
    \dot C_\oplus(t;x)
    :=\frac{M_\oplus(y,x)-M_\oplus(z,x)}{d_t}.
\]
Its empirical counterpart is
\[
\begin{aligned}
    \widehat C_{s,\xi}(t;x)
    &:=\frac{\widehat M_{s,\xi}(y,x)-\widehat M_{s,\xi}(z,x)}{d_t}\\
    &=d_t+\sum_{i\in L}a_{s,\xi,i}(x)\psi_t(Y_i).
\end{aligned}
\]

\begin{lemma}[Normalized squared-distance algebra]
\label{lem:directional-algebra}
For every $t=(z,y)\in\mathfrak T$ and $v\in\cY$,
\[
    d_{\cY}^2(y,v)-d_{\cY}^2(z,v)
    =d_t^2+d_t\psi_t(v),
\]
and
\[
    |\psi_t(v)|\le2d_{\cY}(z,v).
\]
Consequently, for $x\in\cX_0$,
\[
    \dot C_\oplus(t;x)=d_t+H_t(x)
    =d_t+\overline\psi_t+\dot h_t(x).
\]
For any weights $a_1,\ldots,a_n$ satisfying $\sum_i a_i=1$,
\[
    \frac{
        \sum_i a_i d_{\cY}^2(y,v_i)
        -\sum_i a_i d_{\cY}^2(z,v_i)
    }{d_t}
    =
    d_t+\sum_i a_i\psi_t(v_i).
\]
\end{lemma}

\begin{proof}
The first identity is the definition of $\psi_t$. Let
$r=d_{\cY}(y,z)$ and $a=d_{\cY}(z,v)$. The triangle inequality gives
$|a-r|\le d_{\cY}(y,v)\le a+r$. Squaring these bounds gives
\[
    -2ar
    \le
    d_{\cY}^2(y,v)-a^2-r^2
    \le
    2ar,
\]
which proves the envelope after division by $r$. Taking conditional
expectations proves the population identity. The weighted identity follows
by summing the first display and using $\sum_i a_i=1$.
\end{proof}

\begin{remark}[Uniform square integrability of normalized contrasts]
Under Conditions~\nameref{par:M1} and \nameref{par:M2}, put
\[
    \sigma_*^2
    :=
    \sup_{x\in\cX_0}M_\oplus(y_*,x)<\infty.
\]
For \(t=(z,y)\in\mathfrak T\), the first coordinate satisfies
\(z=m_\oplus(x_z)\) for some \(x_z\in\cX_0\). The minimizing property gives
\[
    \bbE[d_{\cY}^2(z,Y)\mid X=x_z]
    \le
    \bbE[d_{\cY}^2(y_*,Y)\mid X=x_z]
    \le\sigma_*^2.
\]
Minkowski's inequality therefore yields
\[
    d_{\cY}(z,y_*)\le2\sigma_*.
\]
For any \(x\in\cX_0\),
\[
\begin{split}
    \bbE[d_{\cY}^2(z,Y)\mid X=x]
    &\le
    2d_{\cY}^2(z,y_*)
    +2\bbE[d_{\cY}^2(y_*,Y)\mid X=x]\\
    &\le10\sigma_*^2.
\end{split}
\]
Combining this bound with Lemma~\ref{lem:directional-algebra} gives
\[
    \bbE[\psi_t(Y)^2\mid X=x]
    \le4\bbE[d_{\cY}^2(z,Y)\mid X=x]
    \le40\sigma_*^2,
\]
uniformly in \(t\in\mathfrak T\) and \(x\in\cX_0\). Hence the
unconditional means $\overline\psi_t$ are finite, and $H_t$, $\zeta_t$,
and the conditional definition of $\dot h_t$ are well-defined on
$\cX_0$ with the measurable versions fixed in the main text.
Condition~\nameref{par:R1} supplies fractional-source control of the
centered conditional contrasts in the fixed $L^2$ closure; only their
finite spectral projections need be RKHS functions.
\end{remark}

\subsubsection{Population spectral coordinates}
\label{app:population-spectral-notation}

Recall that
\[
    \mathcal T_{{\kappa_0}}\phi_j=\lambda_j\phi_j,
    \qquad
    \langle\phi_j,\phi_k\rangle_{L^2(\bbP_X)}=\ind\{j=k\}.
\]
The identity
$\langle f,\mathcal T_{{\kappa_0}}g\rangle_{\mathcal H_{{\kappa_0}}}
=\langle f,g\rangle_{L^2(\bbP_X)}$ gives
\[
    \left\langle
        \sqrt{\lambda_j}\phi_j,
        \sqrt{\lambda_k}\phi_k
    \right\rangle_{\mathcal H_{{\kappa_0}}}
    =\ind\{j=k\}.
\]
Thus the RKHS-normalized population eigenfunction is
$\sqrt{\lambda_j}\phi_j$. Recall from Section~\ref{app:RKHS} the population
feature vector \(\boldsymbol\Phi_s\), the subspace \(\mathcal H_s\), and its
\(\mathcal H_{\kappa_0}\)-orthogonal projector \(\Pi_{\le s}\).

For $t\in\mathfrak T$, the population $L^2$ coefficients are
\[
 c_j(t):=\bbE[\phi_j(X)\dot h_t(X)]
        =\bbE[\phi_j(X)\psi_t(Y)].
\]
The equality uses conditional expectation and $\bbE\phi_j(X)=0$.
For every fixed $s$, the RKHS coordinates of $g_{t,s}=\Pi_{\le s}\dot h_t$
are $c_j(t)/\sqrt{\lambda_j}$, $j\le s$. The corresponding infinite
sequence need not be square summable. Define the population comparison
criterion
\[
 \dot C_{\kappa,s}(t;x):=d_t+\overline\psi_t+g_{t,s}(x),
 \qquad x\in\mathcal S.
\]
This approximates the true normalized criterion; it does not redefine
the regression target.

\subsubsection{Empirical centered-kernel eigensystem and Nystr\"om extension}
\label{app:empirical-spectral-construction}

The empirical features used by the main estimator are constructed from the
empirical law $\widehat{\bbP}_{X,N}=N^{-1}\sum_{i=1}^N\delta_{X_i}$ based on all
$N$ covariates. Hence, conditionally on the pooled covariate sample, the
eigensystem is fixed with respect to the random labeled subset.
Let
\[
    \mathbf K:=[\kappa(X_r,X_{r'})]_{r,r'=1}^N,
    \qquad
    \mathbf H_N:=\mathbf I_N-\frac1N\mathbf1_N\mathbf1_N^\top,
\]
and set
\[
    \mathbf K_{0}:=\mathbf H_N\mathbf K \mathbf H_N.
\]
Let $(\widehat\lambda_{j},\widehat{\mathbf v}_{j})_{j=1}^N$
satisfy the symmetric full-sample eigenproblem
\[
    \frac1N\mathbf K_{0}
    \widehat{\mathbf v}_{j}
    =
    \widehat\lambda_{j}\widehat{\mathbf v}_{j},
    \qquad
    \widehat\lambda_{1}\ge\cdots\ge
    \widehat\lambda_{N}\ge0,
\]
where
$\widehat{\mathbf v}_{j}
=(\widehat v_{j1},\ldots,\widehat v_{jN})^\top$ is
normalized in $\mathbb R^N$. Put
$\widehat\lambda_{N+1}:=0$ and
\[
    \mathcal E_{s}:=
    \{\widehat\lambda_{s}>\widehat\lambda_{s+1}\}.
\]
For every realized symmetric Gram matrix, choose the ordered eigenbasis by the
following fixed Borel rule: within each repeated-eigenvalue eigenspace,
project the canonical vectors of \(\mathbb R^N\) in lexicographic order and
apply Gram--Schmidt after discarding zero projections; choose the sign so that
the first nonzero coordinate is positive. Thus the rank-\(s\) subspace and all
coordinate formulae below are jointly measurable. The objective function is invariant
under any orthogonal change of basis within a complete repeated-eigenvalue
block, because its ridge multiplier is constant on that block.
For $x\in\cX$ and $r=1,\ldots,N$, define the empirically centered kernel
\[
    \widehat\kappa_{0}(x,X_r)
    :=\kappa(x,X_r)-\frac1N\sum_{t=1}^N\kappa(x,X_t)
      -\frac1N\sum_{t=1}^N\kappa(X_t,X_r)
      +\frac1{N^2}\sum_{t=1}^N\sum_{t'=1}^N
       \kappa(X_t,X_{t'}).
\]
The full-sample empirical centered-kernel integral operator is
\[
    (\widehat{\mathcal L}_{\kappa_0}h)(x)
    :=
    \frac1N\sum_{r=1}^N
    \widehat\kappa_{0}(x,X_r)h(X_r).
\]
All $N$ covariates receive the same representation mass; responses enter the estimated Fr\'echet function.

On $\mathcal E_{s}$, for $1\le j\le s$, define the empirical
$L^2(\widehat{\bbP}_{X,N})$-normalized Nystr\"om extension
\[
    \widehat\varphi_{j}(x)
    :=
    \frac1{\sqrt N\,\widehat\lambda_{j}}
    \sum_{r=1}^N
    \widehat v_{jr}\,
    \widehat\kappa_{0}(x,X_r).
\]
The RKHS-normalized empirical eigenfunction used as a regression feature is
\[
    \widehat\phi_{j}(x)
    :=
    \widehat\lambda_{j}^{1/2}\widehat\varphi_{j}(x)
    =
    \frac1{\sqrt{N\widehat\lambda_{j}}}
    \sum_{r=1}^N
    \widehat v_{jr}\,
    \widehat\kappa_{0}(x,X_r).
\]
At the observed covariate points,
\[
    \widehat\varphi_{j}(X_\ell)
    =
    \sqrt N\,\widehat v_{j\ell},
    \qquad
    \widehat\phi_{j}(X_\ell)
    =
    \sqrt{N\widehat\lambda_{j}}\,
    \widehat v_{j\ell},
    \qquad \ell=1,\ldots,N.
\]
Hence
\[
    \frac1N\sum_{r=1}^N
    \widehat\varphi_{j}(X_r)
    \widehat\varphi_{k}(X_r)
    =
    \ind\{j=k\}.
\]
Both normalizations represent the same full-sample empirical eigendirection:
\[
    \widehat{\mathcal L}_{\kappa_0}\widehat\varphi_{j}
    =\widehat\lambda_{j}\widehat\varphi_{j},
    \qquad
    \widehat{\mathcal L}_{\kappa_0}\widehat\phi_{j}
    =\widehat\lambda_{j}\widehat\phi_{j}.
\]
The first eigenfunction is normalized in $L^2(\widehat{\bbP}_{X,N})$,
whereas the second is normalized in the empirical RKHS.
Finally, set
\[
    \widehat{\boldsymbol\Phi}_{s}(x)
    :=
    (\widehat\phi_{1}(x),\ldots,
     \widehat\phi_{s}(x))^\top,
    \qquad
    \widehat{\boldsymbol\Lambda}_{s}
    :=
    \operatorname{diag}(\widehat\lambda_{1},\ldots,
    \widehat\lambda_{s}).
\]
This is the normalization convention used by \citet{kim2026semisupervised}:
$\widehat\phi_{j}$ itself is RKHS-normalized, so the factor
$\sqrt{\widehat\lambda_{j}}$ is already part of the feature.
For proofs, on $\mathcal E_{s}^c$ use the measurable exceptional-event
convention
\[
    \widehat\varphi_{j}=0,
    \qquad
    \widehat\phi_{j}=0,
    \qquad
    \widehat{\boldsymbol\Phi}_{s}=0,
    \qquad
    \widehat{\boldsymbol\Lambda}_{s}=0.
\]
The scalar quantities $\widehat\lambda_j$ always denote the raw empirical
eigenvalues before this convention is applied; the convention zeroes the
displayed truncated feature objects and matrices. In finite-sample
computation, the estimator uses the numerical eigensystem and the
validation-selected effective rank; this convention is used only in
probability arguments.
This is the centered kernel-PCA/Nystr\"om normalization used throughout the
proofs; see also \citet{zwald2005convergence}. It yields
\[
    \frac1N\sum_{r=1}^N
    \widehat{\boldsymbol\Phi}_{s}(X_r)
    \widehat{\boldsymbol\Phi}_{s}(X_r)^\top
    =\widehat{\boldsymbol\Lambda}_{s},
\]
\[
    \frac1N\sum_{r=1}^N
    \widehat{\boldsymbol\Phi}_{s}(X_r)=0.
\]

Fix an ordering \(L=\{\ell_1,\ldots,\ell_n\}\).
The componentwise weights in the main text are collected as
$\mathbf a_{s,\xi}(x)
=(a_{s,\xi,i}(x))_{i\in L}\in\mathbb R^n$.
This is the usual full-sample kernel-PCA/Nystr\"om construction.

\subsubsection{A common RKHS realization of the centered eigensystems}

Let $\mathcal H_\kappa$ be the RKHS of the original kernel restricted to
$\mathcal S\times\mathcal S$. All identities in this common-space
realization are for $x\in\mathcal S$, and the observed covariates lie in
$\cX_0\subseteq\mathcal S$ almost surely. Write
\[
    K_x:=\kappa(x,\cdot)|_{\mathcal S},
    \qquad
    m_\kappa:=\bbE[K_X],
    \qquad
    K_x^0:=K_x-m_\kappa.
\]
We identify $\mathcal H_{{\kappa_0}}$ isometrically with the closed subspace
$\overline{\operatorname{span}}\{K_x^0:x\in\mathcal S\}$ of
$\mathcal H_\kappa$. Under this identification,
\[
    \mathcal T_{{\kappa_0}}=\bbE[K_X^0\otimes K_X^0].
\]
Recall that
\(\overline K_N^0=N^{-1}\sum_{r=1}^N K_{X_r}^0\), and define
\[
    \widehat K_{x}^0
    :=K_x^0-\overline K_N^0
    =K_x-\frac1N\sum_{r=1}^N K_{X_r}.
\]
Then
\[
    \widehat\kappa_{0}(x,X_r)
    =
    \langle\widehat K_{x}^0,
    \widehat K_{X_r}^0\rangle_{\mathcal H_\kappa}.
\]
Set
\[
    \widehat{\mathcal T}_{{\kappa_0},N}
    :=
    \frac1N\sum_{r=1}^N
    \widehat K_{X_r}^0\otimes\widehat K_{X_r}^0.
\]
On $\mathcal E_s$, for $1\le j\le s$, define the common-space eigenvectors by
\[
    \widehat e_j
    :=\frac1{\sqrt{N\widehat\lambda_j}}
      \sum_{r=1}^N\widehat v_{jr}\widehat K_{X_r}^0.
\]
Then
\[
    \widehat{\mathcal T}_{{\kappa_0},N}\widehat e_j
    =\widehat\lambda_j\widehat e_j,
    \qquad
    \langle\widehat e_j,\widehat e_k\rangle=\ind\{j=k\},
    \qquad
    \widehat\phi_j(x)=\langle\widehat e_j,\widehat K_x^0\rangle.
\]
The vector $\widehat e_j\in\mathcal H_{\kappa_0}$ is distinguished from
its empirically centered feature function $\widehat\phi_j$:
\[
 \langle\widehat e_j,K_x^0\rangle
 =\widehat\phi_j(x)+\langle\widehat e_j,\overline K_N^0\rangle.
\]
Thus the common-space eigenvectors, rather than the empirically centered
functions, enter the synthesis maps below. On $\mathcal E_s^c$, set
$\widehat e_j=0$ for $j\le s$, consistently with the proof convention.

\subsubsection{Synthesis maps and normalized-contrast comparison criteria}

All population RKHS evaluations below are on $\mathcal S$, and observed
covariates belong to $\cX_0$ almost surely. In the common
$\mathcal H_{\kappa_0}$ realization define
\[
 E_s\mathbf a=\sum_{j\le s}a_j\sqrt{\lambda_j}\phi_j,\qquad
 \widehat E_s\mathbf a=\sum_{j\le s}a_j\widehat e_j,\qquad
 \Pi_{\le s}=E_sE_s^*,\quad
 \widehat\Pi_{\le s}=\widehat E_s\widehat E_s^*.
\]
On $\mathcal E_s^c$ set $\widehat E_s,\widehat{\boldsymbol\Phi}_s,
\widehat{\boldsymbol\Lambda}_s,\widehat\Pi_{\le s}$ to zero for the
probability convention. Write
\[
 \mathbf B_s:=\widehat{\boldsymbol\Lambda}_s+\xi\mathbf I_s,\qquad
 \mathbf A_s:=(\mathbf C_{L,s}+\xi\mathbf I_s)^{-1}.
\]
The first matrix depends only on the pooled covariates; the second also
depends on $L$. Retain the query second moment
\begin{equation}
\label{eq:supp-overline-C}
 \overline{\mathbf C}_s:=
 \bbE_{X^\circ}[\widehat{\boldsymbol\Phi}_s(X^\circ)
 \widehat{\boldsymbol\Phi}_s(X^\circ)^\top\mid\mathscr X_{N,L}]
 =\widehat E_s^*(\mathcal T_{\kappa_0}
       +\overline K_N^0\otimes\overline K_N^0)\widehat E_s.
\end{equation}
Distinguish the uncentered labeled second moment from its centered version:
\[
 \mathbf C^u_{L,s}:=P_L\widehat{\boldsymbol\Phi}_s
                         \widehat{\boldsymbol\Phi}_s^\top,\qquad
 \mathbf C_{L,s}=\mathbf C^u_{L,s}
       -\overline{\boldsymbol\Phi}_L\overline{\boldsymbol\Phi}_L^\top,
 \qquad P_Lf:=n^{-1}\sum_{i\in L}f(X_i).
\]
With $\overline K_L^0=P_LK^0$, common centering cancels exactly:
\[
 \mathbf q_L(x)=\widehat E_s^*(K_x^0-\overline K_L^0).
\]
For scalar or Hilbert-valued functions on the observed covariates,
define the linear smoother
\[
 (\widehat S f)(x):=\sum_{i\in L}a_{s,\xi,i}(x)f(X_i).
\]
It reproduces constants. In the common population-centered RKHS put
\[
 \mathbf c_{t,s}:=\widehat E_s^*g_{t,s},\qquad
 \ell_{t,s}:=\widehat\Pi_{\le s}g_{t,s},\qquad
 d_{t,s}:=(I-\widehat\Pi_{\le s})g_{t,s}.
\]
Thus $\dot h_t=r_{t,s}+d_{t,s}+\ell_{t,s}$ on $\cX_0$.
For the conditional-mean oracle define
\[
 \overline C_{s,\xi}^L(t;x):=d_t+\widehat S H_t(x)
       =d_t+\overline\psi_t+\widehat S\dot h_t(x).
\]
This replaces $\psi_t(Y_i)$ by $H_t(X_i)$ in the actual labeled fit.
The learned-subspace comparison uses
$\widehat C_s^H(t;x):=d_t+\overline\psi_t+\ell_{t,s}(x)$ and
population-centered evaluation. It is a proof criterion, not a new
fitting step. In every formula $\mathbf C_{L,s}$ is the centered labeled
covariance, never $\mathbf C^u_{L,s}$. The exact smoother decomposition
is recorded in Lemma~\ref{lem:fixed-design-decomposition}.

\section{Technical Results}
\label{app:technical-results}

In this section, pointwise identities involving population RKHS functions
are understood on $\mathcal S$, and those involving conditional contrasts
or conditional response laws are understood on $\cX_0$. We work on the
probability-one event on which all observed covariates lie in $\cX_0$.

The proof has four stages. First, fractional-source expansions express
all truncation residuals through a common Hilbert-valued tail. Second,
the kernel moment and spectral assumptions control the empirical
representation and labeled covariance. Third, a vector-valued ridge
stability inequality and a pooled empirical covariance bound control
truncation and representation residuals even under data reuse; response
noise and ridge bias complete the four-term contrast bound. Finally,
curvature and uniform escape turn that bound into prediction error for
the true Fr\'echet target.

\subsection{Population and approximation lemmas}
\label{app:population-lemmas}

\begin{lemma}[Spectral expansion in the fixed closure]
\label{lem:population-rkhs-surrogate}
Under \aref{K1} and \aref{R1}, write
$a_j(t)=\lambda_j^{-r/2}c_j(t)$. Then
\[
 \dot h_t=\sum_{j\ge1}\lambda_j^{r/2}a_j(t)\phi_j
       \quad\hbox{in }L^2(\bbP_X),\qquad
 \sup_t\sum_j a_j(t)^2\le B_r^2.
\]
The finite projection $g_{t,s}=\Pi_{\le s}\dot h_t$ belongs to
$\mathcal H_{\kappa_0}$, with
\[
 \sup_t\|g_{t,s}\|_{\mathcal H_{\kappa_0}}^2
 =\sup_t\sum_{j\le s}a_j(t)^2\lambda_j^{r-1}
 \le B_r^2\lambda_s^{r-1}=G_{r,s}^2.
\]
The infinite series need not converge in the RKHS when $r<1$.
Under the additional moment and target conditions \aref{M1} and \aref{M2},
all displayed scalar expansions can be made simultaneous on the common
full-measure set $\cX_0$, after a Borel full-measure restriction if needed.
\end{lemma}

\begin{proof}
The positive eigenfunctions form an orthonormal basis of
$\overline{\mathcal H_{\kappa_0}}^{\,L^2(\bbP_X)}$.
Closure membership and Parseval give the $L^2$ expansion. Condition
\aref{R1} gives the coefficient bound. Since $0<r\le1$ and
$\lambda_j\ge\lambda_s$ for $j\le s$, the finite RKHS norm is bounded
as displayed; no RKHS norm is assigned to the infinite target.

For the version assertion additionally use \aref{M1} and \aref{M2}.
Tonelli and \aref{R1} give a Borel full-measure
set on which $\sum_j\lambda_j^r\phi_j(x)^2<\infty$. On this set,
Cauchy--Schwarz gives, uniformly in $t$,
\[
 \left|\sum_{j>J}\lambda_j^{r/2}a_j(t)\phi_j(x)\right|
 \le B_r\left(\sum_{j>J}\lambda_j^r\phi_j(x)^2\right)^{1/2}\longrightarrow0.
\]
The series therefore converges uniformly in $t$ at each such $x$.
The chosen conditional contrasts are continuous in $t=(z,y)$ on
$\mathfrak T$ with its product metric: the contrast is pointwise continuous
off the diagonal, and Lemma~\ref{lem:directional-algebra} and the uniform
bound $d_{\cY}(z,y_*)\le2\sigma_*$ give an integrable dominating function,
both conditionally at any $x\in\cX_0$ and unconditionally. The same bound
makes $c_j(t)$ continuous by dominated convergence. Thus both the
conditional contrast and its spectral series are continuous in $t$.
Choose a countable dense subset of the original index set, intersect the
sets of $L^2$ agreement for its elements, and then use continuity to
extend agreement to every $t$. Intersecting with the preceding Borel
set gives the claimed common full-measure restriction. The moment,
identification, and curvature conditions are preserved by this
restriction, and any resulting smaller target index set inherits all
bounds. The stipulated measurable, separable versions are retained.
\end{proof}

\begin{lemma}[Uniform fractional-source truncation]
\label{lem:automatic-spectral-truncation}
Under \aref{M1}, \aref{M2}, \aref{K1}, and \aref{R1}, put
\[
 v_{r,s}(x):=(\lambda_j^{r/2}\phi_j(x))_{j>s}\in\ell^2,
 \qquad r_{t,s}:=\dot h_t-g_{t,s}.
\]
On the common full-measure set,
\[
 r_{t,s}(x)=\langle a_{>s}(t),v_{r,s}(x)\rangle_{\ell^2},\qquad
 \|v_{r,s}\|_{L^2(\bbP_X;\ell^2)}^2=\sum_{j>s}\lambda_j^r.
\]
Consequently
\[
 \left\|\sup_t|r_{t,s}|\right\|_{L^2(\bbP_X)}\le b_{r,s}\longrightarrow0.
\]
If $\lambda_j\asymp j^{-q}$, summability is equivalent to $qr>1$ and
$b_{r,s}^2\lesssim s^{-(qr-1)}$.
\end{lemma}

\begin{proof}
The first identity follows from Lemma~\ref{lem:population-rkhs-surrogate}.
Tonelli and $\bbE\phi_j(X)^2=1$ give the vector-valued norm identity.
Taking the supremum in Cauchy--Schwarz before integrating gives the
bound with $B_r$. Summability yields its limit. Integral comparison of
$\sum_{j>s}j^{-qr}$ proves the polynomial assertion.
\end{proof}

\begin{proposition}[Hilbert responses: fractional source and noise]
\label{prop:hilbert-automatic-complexity}
Let $\cY=\mathcal H_Y$ be a separable Hilbert space and assume \aref{M1}.
Then \aref{R2} holds without any source assumption. For source regularity,
let $0<r\le1$, $\sum_j\lambda_j^r<\infty$, and suppose
\[
 m_c:=\bbE[Y\mid X=\cdot]-\bbE Y
       =\sum_{j\ge1}\theta_j\phi_j
       \quad\hbox{in }L^2(\bbP_X;\mathcal H_Y),\qquad
 M_r^2:=\sum_{j\ge1}\frac{\|\theta_j\|_{\mathcal H_Y}^2}{\lambda_j^r}<\infty.
\]
Then \aref{R1} holds with any fixed $B_r\ge2M_r$ and the appropriate
common versions. In the contrast and prediction bounds one may replace
$b_{r,s}$ by
\[
 b^{\mathrm H}_{r,s}:=2\left(\sum_{j>s}\|\theta_j\|_{\mathcal H_Y}^2\right)^{1/2}
 \le2M_r\lambda_{s+1}^{r/2},
\]
while using $G^{\mathrm H}_{r,s}:=2M_r\lambda_s^{(r-1)/2}$ in the ridge
and representation terms (or the larger $G_{r,s}$). The nondegenerate
case $M_r>0$ is understood in the latter notation; if $M_r=0$, the
centered signal terms are zero. For $\lambda_j\asymp j^{-q}$,
$(b^{\mathrm H}_{r,s})^2\lesssim s^{-qr}$.
The conclusions also hold for a closed convex subset of a separable
Hilbert space with its induced distance.
\end{proposition}

\begin{proof}
For $u_t=(y-z)/\|y-z\|$, the squared-norm identity gives
\[
 \psi_t(v)=-2\langle u_t,v-z\rangle,\qquad
 \dot h_t(x)=-2\langle u_t,m_c(x)\rangle,
 \qquad c_j(t)=-2\langle u_t,\theta_j\rangle.
\]
Hence $\sum_j c_j(t)^2/\lambda_j^r\le4M_r^2$, uniformly in $t$, and
closure membership follows from the vector expansion. The scalar
versions are justified by Lemma~\ref{lem:population-rkhs-surrogate};
equivalently the vector series converges absolutely in norm almost
everywhere by the source bound and $\sum_j\lambda_j^r\phi_j(x)^2<\infty$.
For the residual, define the deterministic Hilbert-valued tail
$V_s(x)=\sum_{j>s}\theta_j\phi_j(x)$. Orthogonality gives
\[
 \|V_s\|_{L^2(\bbP_X;\mathcal H_Y)}^2
 =\sum_{j>s}\|\theta_j\|^2
 \le \lambda_{s+1}^r M_r^2.
\]
Moreover $r_{t,s}=-2\langle u_t,V_s\rangle$ and
$\bbE P_L\|V_s\|^2=\|V_s\|_{L^2}^2$. Markov's inequality and the
Hilbert-valued smoother bound of Lemma~\ref{lem:hilbert-smoother-stability}
give $\|\widehat S V_s\|_{L^2}=O_p(\|V_s\|_{L^2})$. The triangle
inequality and $\sup_t\|u_t\|=1$ therefore yield
\[
 \left\|\sup_t|(\widehat S-I)r_{t,s}|\right\|_2
 =O_p\big(b^{\mathrm H}_{r,s}\big).
\]
This proves the claimed replacement for the fitted residual, not merely
a pointwise approximation bound. The finite-source norm bound is the
one already proved, with $B_r=2M_r$.

To check \aref{R2}, condition on $x_1,\ldots,x_k$ and set
$\epsilon_i=Y_i-\bbE[Y_i\mid X=x_i]$. Then
\[
 \sup_t\left|\sum_iw_i\zeta_t(x_i,Y_i)\right|
 \le2\left\|\sum_iw_i\epsilon_i\right\|_{\mathcal H_Y}.
\]
Conditional independence and centering give
\[
 \bbE\left[\left\|\sum_iw_i\epsilon_i\right\|^2\,\middle|\,x_1,\ldots,x_k\right]
 =\sum_iw_i^2\bbE[\|\epsilon_i\|^2\mid X=x_i]
 \le\sigma_Y^2\sum_iw_i^2,
\]
where $\sigma_Y^2=\sup_{x\in\cX_0}\bbE[\|Y-y_*\|^2\mid X=x]<\infty$.
Thus $C_M=2\sigma_Y$ suffices. If the response is restricted to a closed
convex subset, its Bochner conditional mean remains in that subset,
all distance identities above are unchanged, and the same argument
applies. This includes the closed convex quantile representation of
univariate $W_2$ responses \citep{bigot2017geodesic}.
\end{proof}

\begin{proposition}[An abstract entropy sufficient condition for general metric responses]
\label{prop:metric-response-entropy}
Suppose that $(\mathfrak T,d_\zeta)$ is a separable totally bounded
semimetric space, fix $t_0\in\mathfrak T$, and write
$D_\zeta=\operatorname{diam}(\mathfrak T,d_\zeta)$. Assume that the processes
below have measurable $d_\zeta$-separable versions. Assume also that,
uniformly in $x\in\cX_0$,
\begin{align}
 \|\zeta_t(x,Y)-\zeta_{t'}(x,Y)\|_{\psi_2(\bbP_{Y\mid X=x})}
 &\le L_\zeta d_\zeta(t,t'),
 \label{eq:metric-entropy-increment}\\
 \|\zeta_{t_0}(x,Y)\|_{\psi_2(\bbP_{Y\mid X=x})}
 &\le \sigma_\zeta,
 \notag
\end{align}
and that
\[
 J_\zeta
 :=\int_0^{D_\zeta}
 \sqrt{1+\log N(\epsilon,\mathfrak T,d_\zeta)}\,d\epsilon
 <\infty.
\]
Then Condition~\nameref{par:R2} holds with
\[
 C_M\le C\{\sigma_\zeta+L_\zeta J_\zeta\},
\]
where $C$ is numerical. In particular,
the entropy integral is finite if, for some finite $v>0$ and
$A\ge D_\zeta$,
\[
 \log N(\epsilon,\mathfrak T,d_\zeta)
 \le v\log(A/\epsilon),\qquad 0<\epsilon\le D_\zeta.
\]

Alternatively, the abstract hypotheses can be verified directly from the
response spaces as follows. Suppose that $(\cM,d_\cY)$ and
$(\cY,d_\cY)$ are totally bounded, with diameters $D_\cM$ and $D_\cY$, and
give $\mathfrak T\subset\cM\times\cY$ the product semimetric
\[
 d_\oplus((z,y),(z',y'))
 :=d_\cY(z,z')+d_\cY(y,y').
\]
Assume that the processes have measurable $d_\oplus$-separable versions and,
for some $t_0\in\mathfrak T$, uniformly in $x\in\cX_0$,
\[
 \|\zeta_t(x,Y)-\zeta_{t'}(x,Y)\|_{\psi_2(\bbP_{Y\mid X=x})}
 \le L_\oplus d_\oplus(t,t'),
 \qquad
 \|\zeta_{t_0}(x,Y)\|_{\psi_2(\bbP_{Y\mid X=x})}
 \le\sigma_\oplus.
\]
If
\[
 J_\cM+J_\cY
 :=\int_0^{D_\cM}\!\sqrt{1+\log N(\epsilon,\cM,d_\cY)}\,d\epsilon
 +\int_0^{D_\cY}\!\sqrt{1+\log N(\epsilon,\cY,d_\cY)}\,d\epsilon
 <\infty,
\]
then Condition~\nameref{par:R2} holds with
\[
 C_M\le C\{\sigma_\oplus+L_\oplus(J_\cM+J_\cY)\}.
\]
\end{proposition}

\begin{proof}
Fix $x_1,\ldots,x_k\in\cX_0$ and $\mathbf w\in\mathbb R^k$, and condition on
these covariates. Put
\[
 Z_t:=\sum_{i=1}^k w_i\zeta_t(x_i,Y_i).
\]
The summands are conditionally independent and centered. The standard
Orlicz-norm inequality for sums of independent centered sub-Gaussian random
variables and \eqref{eq:metric-entropy-increment} give
\[
 \|Z_t-Z_{t'}\|_{\psi_2}
 \le C\left\{\sum_{i=1}^k w_i^2
 \|\zeta_t(x_i,Y_i)-\zeta_{t'}(x_i,Y_i)\|_{\psi_2}^2\right\}^{1/2}
 \le CL_\zeta\|\mathbf w\|_2d_\zeta(t,t').
\]
Likewise,
$\|Z_{t_0}\|_{\psi_2}\le C\sigma_\zeta\|\mathbf w\|_2$.
Let $d_Z(t,t'):=\|Z_t-Z_{t'}\|_{\psi_2}$. The preceding display gives
$d_Z\le CL_\zeta\|\mathbf w\|_2d_\zeta$. We give the required $L^2$ chaining
calculation explicitly. Put $D=D_\zeta$.
If $D=0$, then $d_Z(t,t_0)=0$ for every $t$ and the assumed separable
version gives $Z_t=Z_{t_0}$ almost surely on the determining subclass, so
the desired bound is immediate. Hence suppose $D>0$, put
$\delta_m=2^{-m}D$, and let $T_m$
be a finite $\delta_m$-net of $\mathfrak T$ with
$|T_m|=N(\delta_m,\mathfrak T,d_\zeta)$, take $T_0=\{t_0\}$, and choose
$\pi_m:\mathfrak T\to T_m$ such that
$d_\zeta(t,\pi_m(t))\le\delta_m$, with $\pi_0(t)=t_0$.

For centered sub-Gaussian variables $V_1,\ldots,V_M$, put
$\sigma_M:=\bigvee_{r\le M}\|V_r\|_{\psi_2}$. If $\sigma_M=0$, all the
variables vanish almost surely and the bound below is immediate. Otherwise,
the union bound gives
\[
 \bbP\!\left(\bigvee_{r\le M}|V_r|>u\right)
 \le2M\exp\!\left\{-\frac{cu^2}{\sigma_M^2}\right\},
\]
and integration of
$\bbE\bigvee_{r\le M}|V_r|^2
=\int_0^\infty 2u\,\bbP(\bigvee_{r\le M}|V_r|>u)\,du$,
split at $u_0=C\sigma_M\sqrt{1+\log M}$, implies
\[
 \left\|\bigvee_{r\le M}|V_r|\right\|_{L^2}
 \le C\sqrt{1+\log M}
       \bigvee_{r\le M}\|V_r\|_{\psi_2}.
\]
Apply this finite maximum bound at level $m$ to the increments
$Z_u-Z_v$ indexed by pairs
$(u,v)=(\pi_m(t),\pi_{m-1}(t))$ that occur as $t$ varies. There are at
most $|T_m||T_{m-1}|$ such pairs, and their $d_\zeta$-length is at most
$\delta_m+\delta_{m-1}$. Since
$N(\delta_{m-1},\mathfrak T,d_\zeta)
\le N(\delta_m,\mathfrak T,d_\zeta)$, it follows that
\[
 \left\|\bigvee_{t\in\mathfrak T}
 |Z_{\pi_m(t)}-Z_{\pi_{m-1}(t)}|\right\|_{L^2}
 \le C L_\zeta\|\mathbf w\|_2\delta_{m-1}
 \sqrt{1+\log N(\delta_m,\mathfrak T,d_\zeta)}.
\]
For
$f(\epsilon):=\sqrt{1+\log N(\epsilon,\mathfrak T,d_\zeta)}$,
monotonicity of the covering number gives, term by term,
\[
 \delta_{m-1}f(\delta_m)
 =2\delta_m f(\delta_m)
 \le4\int_{\delta_{m+1}}^{\delta_m}f(\epsilon)\,d\epsilon.
\]
Therefore,
\[
 \sum_{m\ge1}\delta_{m-1}
 \sqrt{1+\log N(\delta_m,\mathfrak T,d_\zeta)}
 \le
 C\int_0^D
 \sqrt{1+\log N(\epsilon,\mathfrak T,d_\zeta)}\,d\epsilon
 =CJ_\zeta.
\]
For any finite subset $F$ of the fixed countable determining subclass, the
telescoping identity for $Z_{\pi_M(t)}-Z_{t_0}$, Minkowski's inequality, and
$d_Z(\pi_M(t),t)\to0$ yield, upon first letting $M\to\infty$,
\[
 \left\|\bigvee_{t\in F}|Z_t-Z_{t_0}|\right\|_{L^2}
 \le CL_\zeta\|\mathbf w\|_2J_\zeta.
\]
Letting $F$ increase to that determining subclass and using monotone
convergence proves
\[
 \left\|\,\sup_{t\in\mathfrak T}|Z_t-Z_{t_0}|\right\|_{L^2}
 \le CL_\zeta\|\mathbf w\|_2J_\zeta.
\]
This is the $L^2$ form of the usual sub-Gaussian entropy bound. The triangle
inequality and the base-point bound prove the asserted abstract
constant bound. Finally,
$\int_0^{D_\zeta}\sqrt{1+v\log(A/\epsilon)}\,d\epsilon<\infty$, proving the
polynomial-entropy assertion.

For the response-space specialization, form the product of an
$\epsilon/4$-net of $\cM$ and an $\epsilon/4$-net of $\cY$. In each product
cell that intersects $\mathfrak T$, choose one point of the intersection as
its center. Every point in that cell is within \(\epsilon\) of the chosen
center in \(d_\oplus\). Therefore, even when covering centers are required to
belong to \(\mathfrak T\),
\[
 N(\epsilon,\mathfrak T,d_\oplus)
 \le
 N(\epsilon/4,\cM,d_\cY)
 N(\epsilon/4,\cY,d_\cY).
\]
Since $\sqrt{1+a+b}\le\sqrt{1+a}+\sqrt{1+b}$ for $a,b\ge0$ and
$\operatorname{diam}(\mathfrak T,d_\oplus)\le D_\cM+D_\cY$, splitting the
integral at the two diameters and changing variables gives
\[
 \int_0^{\operatorname{diam}(\mathfrak T,d_\oplus)}
 \sqrt{1+\log N(\epsilon,\mathfrak T,d_\oplus)}\,d\epsilon
 \le C(J_\cM+J_\cY).
\]
Apply the abstract part with $d_\zeta=d_\oplus$ to obtain
the asserted response-space constant bound.
\end{proof}

\begin{remark}[Relation to Fr\'echet-regression covering assumptions]
The local entropy assumptions used in global and local Fr\'echet regression
have the form
\[
 \int_0^1\sqrt{1+\log N(\delta\epsilon,
 B_\delta(m_\oplus(x)),d_{\cY})}\,d\epsilon=O(1),
\]
pointwise or uniformly in the predictor
\citep{petersen2019frechet,chen2022uniform}.  Proposition~
\ref{prop:metric-response-entropy} is the corresponding condition for the
object actually indexed in the present proof: the normalized conditional
noise class $\{\zeta_t:t\in\mathfrak T\}$. The response-space specialization
above makes the transfer explicit when the normalized contrast map is
uniformly Lipschitz in $d_\oplus$ and the two global entropy integrals are
finite. This response-space route is a convenient sufficient verification.
Since the normalized contrast divides by $d_{\cY}(y,z)$, a global
product-metric Lipschitz bound can be restrictive as $y$ approaches $z$. In
such settings the abstract $d_\zeta$ condition can be verified directly.
Local covering arguments can also be paired with localization before they
imply the global maximal inequality in Condition~\nameref{par:R2}. The
Hilbert proposition above gives a separate verification.
\end{remark}

\subsection{Empirical spectral and design lemmas}
\label{app:spectral-lemmas}

\begin{lemma}[Fixed-design four-component decomposition]
\label{lem:fixed-design-decomposition}
On $\mathcal E_s$, write
$\ell_{t,s}=\widehat\Pi_{\le s}g_{t,s}$,
$d_{t,s}=(I-\widehat\Pi_{\le s})g_{t,s}$, and
$\mathbf c_{t,s}=\widehat E_s^*g_{t,s}$. For $x\in\cX_0$,
\begin{equation}
\label{eq:directional-fixed-design-decomposition}
\begin{split}
 \widehat C_{s,\xi}(t;x)-\dot C_\oplus(t;x)
 ={}&\sum_{i\in L}a_{s,\xi,i}(x)\zeta_t(X_i,Y_i)
       +(\widehat S-I)r_{t,s}(x)\\
    &+(\widehat S-I)d_{t,s}(x)
       -\xi\mathbf q_L(x)^\top\mathbf A_s\mathbf c_{t,s}.
\end{split}
\end{equation}
\end{lemma}

\begin{proof}
The weights reproduce constants, so the difference is the first term
plus $(\widehat S-I)\dot h_t$. Use
$\dot h_t=r_{t,s}+d_{t,s}+\ell_{t,s}$. In the common RKHS,
\[
 \ell_{t,s}(x)=\widehat{\boldsymbol\Phi}_s(x)^\top\mathbf c_{t,s}
       +\langle\ell_{t,s},\overline K_N^0\rangle.
\]
The constant is absorbed by the unpenalized intercept. Centered ridge
therefore gives
\[
 \widehat S\ell_{t,s}(x)=P_L\ell_{t,s}
       +\mathbf q_L(x)^\top\mathbf A_s\mathbf C_{L,s}\mathbf c_{t,s}
       =\ell_{t,s}(x)-\xi\mathbf q_L(x)^\top\mathbf A_s\mathbf c_{t,s}.
\]
This proves the identity. In particular, the learned-subspace component
is reproduced exactly at zero ridge on an invertible design; no
independent labeled signal-fluctuation term is introduced.
\end{proof}

\subsection{Prediction-weighted spectral and fixed-design bounds}
\label{app:prediction-weighted-analysis}

This subsection establishes the spectral and fixed-design bounds used in
Theorem~\ref{thm:contrast-error}.
The regression arguments use either the basic spectral bounds in
Proposition~\ref{prop:primitive-spectral-verification} or the relative
spectral bounds in Proposition~\ref{prop:sharp-spectral-verification},
under their respective assumptions.
Condition~\nameref{par:R1} supplies fractional-source coordinates, and the
labeled-versus-full-sample design comparison uses sampling without
replacement. These ingredients are combined with the response-noise,
truncation, and representation bounds in
Section~\ref{app:additional-rate-results}.

\begin{proposition}[Basic spectral stability]
\label{prop:primitive-spectral-verification}
Under Conditions~\nameref{par:K1}--\nameref{par:K3},
$\bbP(\mathcal E_s)\to1$, and the full-sample empirical eigensystem satisfies
\[
    \frac{\widehat\lambda_s}{\lambda_s}\xrightarrow{p}1,
    \qquad
    \|\mathbf B_s^{-1/2}\overline{\mathbf C}_s\mathbf B_s^{-1/2}\|_\op=O_p(1),
\]
\[
    \|\overline K_N^0\|_{\mathcal H_{\kappa_0}}=O_p(N^{-1/2}),
\]
and
\[
    \|\mathcal T_{\kappa_0}^{1/2}(\widehat\Pi_{\le s}-\Pi_{\le s})\|_\HS
    =O_p\!\left(
      \frac{1}
      {\sqrt N\,(\lambda_s-\lambda_{s+1})}
    \right)
    =o_p(1).
\]
\end{proposition}

\begin{proof}
Write $Z_i=K_{X_i}^0$ and
\[
    \widetilde{\mathcal T}_N:=\frac1N\sum_{i=1}^NZ_i\otimes Z_i,
    \qquad
    \Delta_N:=\widehat{\mathcal T}_{\kappa_0,N}-\mathcal T_{\kappa_0}
    =(\widetilde{\mathcal T}_N-\mathcal T_{\kappa_0})
      -\overline K_N^0\otimes\overline K_N^0.
\]
Condition~\nameref{par:K1} and Proposition~
\ref{prop:app-centered-rkhs-realization} give
$\|Z_i\|^2\le C_\kappa<\infty$. Independence in the Hilbert space of
Hilbert--Schmidt operators therefore yields
\[
 \bbE\|\widetilde{\mathcal T}_N-\mathcal T_{\kappa_0}\|_\HS^2
 =\frac1N\bbE\|Z\otimes Z-\mathcal T_{\kappa_0}\|_\HS^2
 \le\frac{C}{N}.
\]
Also, because $\bbE Z=0$,
\[
    \bbE\|\overline K_N^0\|^2
    =\frac{\operatorname{tr}(\mathcal T_{\kappa_0})}{N}.
\]
Consequently,
\[
    \|\Delta_N\|_\op=O_p(N^{-1/2}),
    \qquad
    \|\overline K_N^0\|=O_p(N^{-1/2}).
\]
Condition~\nameref{par:K3} states
$N^{-1/2}=o(\lambda_s-\lambda_{s+1})$. Combining it with the covariance
bound above gives
\[
    \frac{\|\Delta_N\|_\op}
    {\lambda_s-\lambda_{s+1}}=o_p(1).
\]
Weyl's inequality now gives, with probability tending to one,
\[
    \widehat\lambda_s-\widehat\lambda_{s+1}
    \ge \lambda_s-\lambda_{s+1}-2\|\Delta_N\|_\op>0.
\]
This proves $\bbP(\mathcal E_s)\to1$. Moreover,
$|\widehat\lambda_s-\lambda_s|\le\|\Delta_N\|_\op$; since
$\lambda_s-\lambda_{s+1}\le\lambda_s$, it follows that
$\widehat\lambda_s/\lambda_s\to_p1$.

On the same event, the Davis--Kahan theorem gives
\[
 \|\widehat\Pi_{\le s}-\Pi_{\le s}\|_\op
 \le\frac{C\|\Delta_N\|_\op}
 {\lambda_s-\lambda_{s+1}}.
\]
Put $A=\widehat\Pi_{\le s}-\Pi_{\le s}$. Since $A$ is self-adjoint,
$A^2\preceq\|A\|_\op^2\mathrm{Id}$, and therefore
\[
\begin{split}
 \|\mathcal T_{\kappa_0}^{1/2}A\|_\HS^2
 &=\operatorname{tr}(A\mathcal T_{\kappa_0}A)
   =\operatorname{tr}\{\mathcal T_{\kappa_0}^{1/2}A^2
   \mathcal T_{\kappa_0}^{1/2}\}\\
 &\le \|A\|_\op^2\operatorname{tr}(\mathcal T_{\kappa_0}).
\end{split}
\]
Hence
\[
 \|\mathcal T_{\kappa_0}^{1/2}(\widehat\Pi_{\le s}-\Pi_{\le s})\|_\HS
 =O_p\!\left(
   \frac{1}{\sqrt N\,(\lambda_s-\lambda_{s+1})}
 \right).
\]
The exceptional-event convention preserves this $O_p$ conclusion.

It remains to verify the pooled query bound. On $\mathcal E_s$, for every
$\mathbf a\in\mathbb R^s$, put $h=\widehat E_s\mathbf a$. By
\eqref{eq:supp-overline-C},
\[
\begin{split}
 \mathbf a^\top\overline{\mathbf C}_s\mathbf a
 &=\langle h,(\mathcal T_{\kappa_0}
     +\overline K_N^0\otimes\overline K_N^0)h\rangle\\
 &\le \mathbf a^\top\widehat{\boldsymbol\Lambda}_s\mathbf a
   +(\|\Delta_N\|_\op+\|\overline K_N^0\|^2)\|\mathbf a\|^2.
\end{split}
\]
The last coefficient is $o_p(\lambda_s)$, while
$\widehat\lambda_s/\lambda_s\to_p1$. Thus, with probability tending to one,
$\overline{\mathbf C}_s\preceq C\widehat{\boldsymbol\Lambda}_s$, and consequently
\[
    \|\mathbf B_s^{-1/2}\overline{\mathbf C}_s\mathbf B_s^{-1/2}\|_\op\le C.
\]
The centered labeled inverse is controlled separately below.
\end{proof}

\begin{lemma}[Consequences of the cross-boundary polynomial spectrum]
\label{lem:cross-boundary-spectrum-consequences}
Under Conditions~\nameref{par:K2} and \nameref{par:K5}, the following bounds hold at every
selected complete spectral-block cutoff $s$:
\begin{align}
    \mathcal N(\lambda_{s+1})
    &\lesssim s,
    \label{eq:derived-effective-dimension}\\
    \mathfrak r_s^{\partial}
    &\lesssim s\log(es),
    \label{eq:derived-block-relative-complexity}\\
    \mathcal U_s
    &\lesssim s^2\log(es),
    \label{eq:derived-unweighted-projector-sum}\\
    \mathcal W_s^{+}
    &\lesssim s^{2-q}\log(es).
    \label{eq:derived-leading-pca-excess-scale}
\end{align}
\end{lemma}

\begin{proof}
For $s=1$, Condition~\nameref{par:K5} gives
$\lambda_1-\lambda_k\gtrsim1-k^{-q}$ and $\lambda_k\lesssim k^{-q}$ for
$k>1$. Since $q>1$, the series defining all four quantities are bounded by
constants. It therefore suffices to consider $s\ge2$.
Polynomial decay gives
$\lambda_s\asymp\lambda_{s+1}\asymp s^{-q}$. Therefore
\[
 \mathcal N(\lambda_{s+1})
 \le s+\lambda_{s+1}^{-1}\sum_{k>s}\lambda_k
 \lesssim s+s^q\sum_{k>s}k^{-q}
 \lesssim s,
\]
where the last step uses $q>1$.

For the relative rank, Condition~\nameref{par:K5} gives
\[
 \mathfrak r_s^{\partial}
 \lesssim
 \sum_{j\le s}\frac{j^{-q}}{j^{-q}-(s+1)^{-q}}
 +
 \sum_{k>s}\frac{k^{-q}}{s^{-q}-k^{-q}}.
\]
If $j\le s/2$, the first summand is $O(1)$.  If $s/2<j\le s$,
the mean-value theorem gives
\[
 j^{-q}-(s+1)^{-q}
 \ge c_q(s+1-j)s^{-q-1},
\]
so the summand is at most $Cs/(s+1-j)$.  Similarly, the second
summand is at most $Cs/(k-s)$ for $s<k\le2s$, while
\[
 \sum_{k>2s}\frac{k^{-q}}{s^{-q}-k^{-q}}
 \lesssim s^q\sum_{k>2s}k^{-q}
 \lesssim s.
\]
Summing the two harmonic series proves
\eqref{eq:derived-block-relative-complexity}.

For the one-sided excess scale, polynomial decay and the relative-rank
calculation above give
\[
 \sum_{k>s}\lambda_k\lesssim s^{1-q},
 \qquad
 \sum_{j\le s}\frac{\lambda_j}{\lambda_j-\lambda_{s+1}}
 \lesssim s\log(es).
\]
Multiplication proves
\eqref{eq:derived-leading-pca-excess-scale}.

It remains to treat $\mathcal U_s$.  Split the cross-boundary
index set into
\[
\begin{aligned}
 \mathcal I_1&:=\{j\le s/2,\ k>s\},\\
 \mathcal I_2&:=\{s/2<j\le s,\ k>2s\},\\
 \mathcal I_3&:=\{s/2<j\le s<k\le2s\}.
\end{aligned}
\]
On $\mathcal I_1\cup\mathcal I_2$, $k\ge2j$ and hence
$j^{-q}-k^{-q}\ge(1-2^{-q})j^{-q}$.  Consequently,
\[
 \frac{\lambda_j\lambda_k}{(\lambda_j-\lambda_k)^2}
 \lesssim j^q k^{-q}.
\]
Integral comparison then gives
\[
\begin{split}
 \sum_{\mathcal I_1\cup\mathcal I_2}j^qk^{-q}
 &\lesssim
 \left(\sum_{j\le s}j^q\right)
 \left(\sum_{k>s}k^{-q}\right)
 \lesssim s^2.
\end{split}
\]

On $\mathcal I_3$, all eigenvalues are comparable to $s^{-q}$ and the
relative-gap condition plus the mean-value theorem gives
\[
 \lambda_j-\lambda_k
 \gtrsim j^{-q}-k^{-q}
 \gtrsim(k-j)s^{-q-1}.
\]
Also $\lambda_j-\lambda_{s+1}\le\lambda_j\lesssim s^{-q}$. Hence
\[
 \frac{\lambda_j\lambda_k}{(\lambda_j-\lambda_k)^2}
 \lesssim\frac{s^2}{(k-j)^2}.
\]
Writing $a=s+1-j$ and $b=k-s$ shows
\[
 \sum_{\mathcal I_3}\frac1{(k-j)^2}
 \le
 \sum_{a=1}^{s}\sum_{b=1}^{s}
 \frac1{(a+b-1)^2}
 \lesssim\log(es).
\]
Combining the three regions proves
\eqref{eq:derived-unweighted-projector-sum}.
\end{proof}

\begin{lemma}[Consequences of the tuning regime]
\label{lem:tuning-consequences}
Under Conditions~\nameref{par:K2}, \nameref{par:K5}, and
\eqref{eq:sharp-tuning-regime},
\begin{equation}
\label{eq:tuning-effective-dimension-smallness}
    \mathfrak r_s^{\partial}\sqrt{\frac{\log N}{N}}\longrightarrow0,
    \qquad
    \frac{\mathcal N(\lambda_{s+1})+\log N}{N}\longrightarrow0.
\end{equation}
The following moment-based smallness conditions also hold:
\begin{equation}
\label{eq:tuning-moment-smallness}
    \frac{\mathcal N(\lambda_{s+1})}{\sqrt N}\longrightarrow0,
    \qquad
    \frac{(\mathfrak r_s^{\partial})^2}{N}\longrightarrow0,
    \qquad
    \frac{s^2}{N(\log N)^2}\longrightarrow0.
\end{equation}
\end{lemma}

\begin{proof}
Lemma~\ref{lem:cross-boundary-spectrum-consequences} gives
$\mathcal N(\lambda_{s+1})\lesssim s$ and
$\mathfrak r_s^{\partial}\lesssim s\log(es)$. The first restriction in
\eqref{eq:sharp-tuning-regime} therefore gives the first limit in
\eqref{eq:tuning-effective-dimension-smallness}. Its square states
\[
    \frac{s^2\log^2(es)\log N}{N}\longrightarrow0.
\]
Since $N\to\infty$ and $\log(es)\ge1$, this implies all three limits in
\eqref{eq:tuning-moment-smallness}. Finally,
$\mathcal N(\lambda_{s+1})/N\lesssim s/N\le s/n\to0$ and
$\log N/N\to0$, proving the remaining limit.
\end{proof}

\begin{lemma}[Moment bounds for fixed spectral transforms]
\label{lem:spectral-transform-moments}
Assume Condition~\nameref{par:K4}. Write
$\mathcal T=\mathcal T_{\kappa_0}$, $Z=K_X^0$, and
$e_j=\sqrt{\lambda_j}\phi_j$ in the fixed eigenbasis satisfying that
condition. Let $Z_1,\ldots,Z_N$ be independent copies of $Z$, and put
\[
    \widetilde{\mathcal T}_N=N^{-1}\sum_{i=1}^N Z_i\otimes Z_i,
    \qquad
    \bar Z_N=N^{-1}\sum_{i=1}^N Z_i,
    \qquad
    \widehat{\mathcal T}_N
      =\widetilde{\mathcal T}_N-\bar Z_N\otimes\bar Z_N.
\]
Let $B$ be a fixed bounded self-adjoint operator diagonal in this
basis, with $Be_j=\beta_j e_j$, and define
\[
    b(B):=\operatorname{tr}(B\mathcal T B)
         =\sum_{j\ge1}\lambda_j\beta_j^2<\infty.
\]
There is a constant $C$ depending only on $C_{X,8}$ such that
\begin{align}
    \bbE\|BZ\|^4
      &\le Cb(B)^2,
      \label{eq:spectral-transform-fourth-moment}\\
    \bbE\|B(\widetilde{\mathcal T}_N-\mathcal T)B\|_\HS^2
      &\le \frac{Cb(B)^2}{N},
      \label{eq:spectral-transform-covariance-moment}\\
    \bbE\|B\bar Z_N\|^2
      &=\frac{b(B)}{N},\qquad
    \bbE\|B\bar Z_N\|^4\le\frac{Cb(B)^2}{N^2},
      \label{eq:spectral-transform-mean-moments}\\
    \bbE\|B(\widehat{\mathcal T}_N-\mathcal T)B\|_\HS^2
      &\le\frac{Cb(B)^2}{N}.
      \label{eq:spectral-transform-centered-moment}
\end{align}
In particular, the operator norms of the two transformed covariance
perturbations are $O_p(b(B)/\sqrt N)$, and
$\|B\bar Z_N\|^2=O_p(b(B)/N)$. These bounds hold uniformly for fixed
choices of $B$ that may depend on $n,N,s$; no uniform bound on $\|B\|_\op$
across sample sizes is needed. This does not assert a bound for data-dependent
choices of $B$ or for a supremum over transforms. If $b(B)=0$, all the
displayed quantities vanish almost surely.
\end{lemma}

\begin{proof}
Put $\eta_j=\lambda_j^{-1/2}\langle e_j,Z\rangle=\phi_j(X)$. Then
$\bbE\eta_j=0$, $\bbE\eta_j^2=1$, and
\[
    \sup_j\bbE|\eta_j|^4\le C_{X,8}^{1/2}=:C_4.
\]
The coordinates need not be independent. With
$a_j=\lambda_j\beta_j^2$, the spectral expansion of $Z$ and Tonelli's
theorem give
\[
\begin{split}
    \bbE\|BZ\|^4
    &=\sum_{j,k\ge1}a_ja_k\bbE(\eta_j^2\eta_k^2)\\
    &\le\sum_{j,k\ge1}a_ja_k
          (\bbE\eta_j^4\,\bbE\eta_k^4)^{1/2}
     \le C_4 b(B)^2.
\end{split}
\]
Independence of the observations in the Hilbert space of Hilbert--Schmidt
operators now yields
\[
\begin{split}
    \bbE\|B(\widetilde{\mathcal T}_N-\mathcal T)B\|_\HS^2
    &=\frac1N\bbE\|BZ\otimes BZ-B\mathcal T B\|_\HS^2\\
    &\le\frac1N\bbE\|BZ\|^4.
\end{split}
\]
Since $\bbE Z=0$, the second-moment identity for the sample mean gives
$\bbE\|B\bar Z_N\|^2=b(B)/N$.

For its fourth moment, put
$\bar\eta_j=N^{-1}\sum_{i=1}^N\eta_j(X_i)$. Independence and centering
across observations give the exact identity
\[
    \bbE\bar\eta_j^4
    =N^{-4}\{N\bbE\eta_j^4+3N(N-1)\}
    \le\frac{C}{N^2}.
\]
Thus $\bbE(\bar\eta_j^2\bar\eta_k^2)\le C/N^2$ by
Cauchy--Schwarz. Tonelli's theorem then gives
\[
    \bbE\|B\bar Z_N\|^4
    =\sum_{j,k\ge1}a_ja_k
         \bbE(\bar\eta_j^2\bar\eta_k^2)
    \le\frac{Cb(B)^2}{N^2}.
\]
All infinite sums used here have nonnegative summands, so the same
calculations apply to the infinite-dimensional space without any exchange of
conditionally convergent series.

Finally,
\[
    B(\widehat{\mathcal T}_N-\mathcal T)B
    =B(\widetilde{\mathcal T}_N-\mathcal T)B
       -(B\bar Z_N)\otimes(B\bar Z_N),
\]
and the squared Hilbert--Schmidt norm of the rank-one term is
$\|B\bar Z_N\|^4$. The squared triangle inequality proves
\eqref{eq:spectral-transform-centered-moment}. Markov's inequality and
$\|\cdot\|_\op\le\|\cdot\|_\HS$ give the asserted stochastic orders.
The constants in these moment bounds do not depend on the diagonal
coefficients of $B$, which justifies their use for sample-size-dependent
transforms chosen without using the observations. When $b(B)=0$, $\bbE\|BZ\|^2=0$, and all the
conclusions are immediate.
\end{proof}

\begin{lemma}[Regularized covariance and empirical-centering stability]
\label{lem:derived-covariance-stability}
Under Conditions~\nameref{par:K2}, \nameref{par:K4}--\nameref{par:K5}, and
\eqref{eq:sharp-tuning-regime}, put $\mu=\lambda_{s+1}$. Then
\begin{equation}
\label{eq:derived-preconditioned-stability}
    \chi_{s,N}
    :=\left\|
        (\mathcal T_{\kappa_0}+\mu\mathrm{Id})^{-1/2}
        (\widetilde{\mathcal T}_N-\mathcal T_{\kappa_0})
        (\mathcal T_{\kappa_0}+\mu\mathrm{Id})^{-1/2}
      \right\|_\op
    =O_p\!\left(\frac{\mathcal N(\mu)}{\sqrt N}\right)
    =o_p(1),
\end{equation}
\begin{equation}
\label{eq:derived-centering-rate}
    \|\overline K_N^0\|=O_p(N^{-1/2}),
\end{equation}
and
\begin{equation}
\label{eq:derived-regularized-mean-stability}
    \nu_{s,N}
    :=\left\|
       (\mathcal T_{\kappa_0}+\mu\mathrm{Id})^{-1/2}\overline K_N^0
      \right\|^2
    =O_p\!\left(\frac{\mathcal N(\mu)}{N}\right)
    =o_p(1).
\end{equation}
Consequently, with $\delta_{s,N}:=\chi_{s,N}+\nu_{s,N}=o_p(1)$,
\begin{equation}
\label{eq:derived-centered-preconditioned-stability}
 \left\|
   (\mathcal T_{\kappa_0}+\mu\mathrm{Id})^{-1/2}
   (\widehat{\mathcal T}_{\kappa_0,N}-\mathcal T_{\kappa_0})
   (\mathcal T_{\kappa_0}+\mu\mathrm{Id})^{-1/2}
 \right\|_\op
 \le\delta_{s,N}.
\end{equation}
\end{lemma}

\begin{proof}
The population operator
$B_\mu=(\mathcal T_{\kappa_0}+\mu\mathrm{Id})^{-1/2}$ is bounded for
each $s$ and diagonal in the fixed population eigenbasis. Its transformed
covariance trace is
\[
    b(B_\mu)
    =\operatorname{tr}\{\mathcal T_{\kappa_0}
           (\mathcal T_{\kappa_0}+\mu\mathrm{Id})^{-1}\}
    =\mathcal N(\mu).
\]
Lemma~\ref{lem:spectral-transform-moments} therefore gives
\[
    \bbE\chi_{s,N}^2\le\frac{C\mathcal N(\mu)^2}{N},
    \qquad
    \bbE\nu_{s,N}=\frac{\mathcal N(\mu)}{N}.
\]
Markov's inequality gives the displayed stochastic orders, and
Lemma~\ref{lem:tuning-consequences} makes both quantities $o_p(1)$.
These conclusions use only coordinatewise moments, not a uniform tail
bound for arbitrary linear combinations of eigenfunctions.
Also,
\[
    \bbE\|\overline K_N^0\|^2
    =\frac{\operatorname{tr}(\mathcal T_{\kappa_0})}{N},
\]
which proves \eqref{eq:derived-centering-rate}. Finally,
$\widehat{\mathcal T}_{\kappa_0,N}
 =\widetilde{\mathcal T}_N-\overline K_N^0\otimes\overline K_N^0$,
and the regularized rank-one correction has operator norm $\nu_{s,N}$.
The triangle inequality proves
\eqref{eq:derived-centered-preconditioned-stability}.
\end{proof}

\begin{lemma}[Empirical eigenvalue comparability]
\label{lem:empirical-eigenvalue-comparability}
Under Conditions~\nameref{par:K2}, \nameref{par:K4}--\nameref{par:K5}, and
\eqref{eq:sharp-tuning-regime},
\begin{equation}
\label{eq:empirical-eigenvalue-comparability}
    \bigvee_{1\le j\le s+1}
    \left|\frac{\widehat\lambda_j}{\lambda_j}-1\right|
    =o_p(1).
\end{equation}
In particular, $\widehat\lambda_s>0$ with probability tending to one.
\end{lemma}

\begin{proof}
Put $\mu=\lambda_{s+1}$ and use $\delta_{s,N}$ from
Lemma~\ref{lem:derived-covariance-stability}. On the event
$\delta_{s,N}<1/4$, \eqref{eq:derived-centered-preconditioned-stability}
gives the quadratic-form bounds
\[
 (1-\delta_{s,N})\mathcal T_{\kappa_0}-\delta_{s,N}\mu\mathrm{Id}
 \preceq\widehat{\mathcal T}_{\kappa_0,N}
 \preceq(1+\delta_{s,N})\mathcal T_{\kappa_0}
          +\delta_{s,N}\mu\mathrm{Id}.
\]
For $j\le s+1$, the min--max principle and $\mu\le\lambda_j$ imply
\[
 (1-2\delta_{s,N})\lambda_j
 \le(1-\delta_{s,N})\lambda_j-\delta_{s,N}\mu
 \le\widehat\lambda_j
 \le(1+\delta_{s,N})\lambda_j+\delta_{s,N}\mu
 \le(1+2\delta_{s,N})\lambda_j.
\]
Since $\delta_{s,N}=o_p(1)$, this proves the assertion for the raw
empirical eigenvalues, including repeated eigenvalues.
\end{proof}

\begin{lemma}[Empirical cutoff separation and weighted good events]
\label{lem:empirical-cutoff-separation}
Under Conditions~\nameref{par:K2}, \nameref{par:K4}--\nameref{par:K5}, and
\eqref{eq:sharp-tuning-regime},
\[
    \bbP\!\left(\widehat\lambda_s-\widehat\lambda_{s+1}
       \ge\tfrac34(\lambda_s-\lambda_{s+1})\right)\longrightarrow1,
    \qquad \bbP(\mathcal E_s)\longrightarrow1.
\]
Put
\[
    \mu_s=\lambda_{s+1},\qquad
    g_s=\lambda_s-\lambda_{s+1},\qquad
    \Delta_N=\mathcal T_{\kappa_0}-\widehat{\mathcal T}_{\kappa_0,N},\qquad
    S_s=\sum_{j\le s}(\lambda_j-\mu_s)^{-1/2}\Pi_j.
\]
The following bounds also hold:
\begin{equation}
\label{eq:centered-pca-eigenvalue-good-event}
    \bbP\!\left(\widehat\lambda_{s+1}-\mu_s>g_s/2\right)
    \le\frac{C(\mathfrak r_s^\partial)^2}{N}=o(1),
\end{equation}
\begin{equation}
\label{eq:centered-pca-weighted-good-event}
    \bbP\{\|S_s\Delta_NS_s\|_\op\le1/4\}
    \ge1-\frac{C(\mathfrak r_s^\partial)^2}{N}\longrightarrow1.
\end{equation}
\end{lemma}

\begin{proof}
Write $\mathcal T=\mathcal T_{\kappa_0}$ and define the midpoint and its
spectral transform by
\[
    c_s:=\frac{\lambda_s+\lambda_{s+1}}2,
    \qquad
    R_s^{\mathrm{mid}}:=|\mathcal T-c_s\mathrm{Id}|^{-1/2}.
\]
Since $c_s$ lies in a spectral gap and $c_s>0$, $R_s^{\mathrm{mid}}$ is a bounded
population spectral transform for every $s$. For $j\le s$,
$\lambda_j-c_s\ge(\lambda_j-\lambda_{s+1})/2$; for $k>s$,
$c_s-\lambda_k\ge(\lambda_s-\lambda_k)/2$. Hence
\[
\begin{split}
    b(R_s^{\mathrm{mid}})
    &=\sum_{j\ge1}\frac{\lambda_j}{|\lambda_j-c_s|}\\
    &\le2\left\{
        \sum_{j\le s}\frac{\lambda_j}{\lambda_j-\lambda_{s+1}}
        +\sum_{k>s}\frac{\lambda_k}{\lambda_s-\lambda_k}
       \right\}
     =2\mathfrak r_s^\partial.
\end{split}
\]
The centered covariance bound in Lemma~\ref{lem:spectral-transform-moments}
and Markov's inequality give, with
$\epsilon_{s,N}:=\|R_s^{\mathrm{mid}}\Delta_N R_s^{\mathrm{mid}}\|_\op$,
\begin{equation}
\label{eq:midpoint-relative-good-event}
    \mathcal G_{s,N}^{\mathrm{mid}}:=\{\epsilon_{s,N}\le1/4\},
    \qquad
    \bbP\{(\mathcal G_{s,N}^{\mathrm{mid}})^c\}
    \le\frac{C(\mathfrak r_s^\partial)^2}{N}=o(1).
\end{equation}
The final limit follows from Lemma~\ref{lem:tuning-consequences}.

On $\mathcal G_{s,N}^{\mathrm{mid}}$, for every $h$,
\[
    |\langle h,\Delta_Nh\rangle|
    \le\epsilon_{s,N}
        \langle h,|\mathcal T-c_s\mathrm{Id}|h\rangle.
\]
On the population leading subspace the absolute-value operator is
$\mathcal T-c_s\mathrm{Id}$; on its orthogonal complement it is
$c_s\mathrm{Id}-\mathcal T$. Since $1-\epsilon_{s,N}>0$, the min--max
principle, applied to these two fixed subspaces, yields
\[
\begin{aligned}
    \widehat\lambda_s
    &\ge(1-\epsilon_{s,N})\lambda_s+\epsilon_{s,N}c_s
      =\lambda_s-\epsilon_{s,N}g_s/2,\\
    \widehat\lambda_{s+1}
    &\le(1-\epsilon_{s,N})\lambda_{s+1}+\epsilon_{s,N}c_s
      =\lambda_{s+1}+\epsilon_{s,N}g_s/2.
\end{aligned}
\]
Consequently,
$\widehat\lambda_s-\widehat\lambda_{s+1}
 \ge(1-\epsilon_{s,N})g_s>0$ and
$\widehat\lambda_{s+1}-\mu_s\le g_s/8$ on this event. This proves
$\mathcal G_{s,N}^{\mathrm{mid}}\subseteq\mathcal E_s$ and
\eqref{eq:centered-pca-eigenvalue-good-event}.

For the weighted event, the diagonal operator
\[
    D_s:=|\mathcal T-c_s\mathrm{Id}|^{1/2}S_s
\]
satisfies
\[
    \|D_s\|_\op^2
    =\max_{j\le s}\frac{\lambda_j-c_s}{\lambda_j-\mu_s}\le1.
\]
Because $S_s=R_s^{\mathrm{mid}}D_s=D_s R_s^{\mathrm{mid}}$,
\[
    \|S_s\Delta_NS_s\|_\op
    =\|D_s(R_s^{\mathrm{mid}}\Delta_N R_s^{\mathrm{mid}})D_s\|_\op
    \le\epsilon_{s,N}\le1/4
\]
on $\mathcal G_{s,N}^{\mathrm{mid}}$. Equation~
\eqref{eq:midpoint-relative-good-event} therefore also proves
\eqref{eq:centered-pca-weighted-good-event}. In particular, all the
required good events hold together with probability tending to one,
without an exponential-tail assumption.
\end{proof}

\begin{lemma}[Ridge-leverage stability]
\label{lem:sharp-leverage-stability}
Under Conditions~\nameref{par:K2}, \nameref{par:K4}--\nameref{par:K5}, and
\eqref{eq:sharp-tuning-regime}, with
$\mathbf B_s=\widehat{\boldsymbol\Lambda}_s+\xi \mathbf I_s$,
\[
    \left\|\mathbf B_s^{-1/2}\overline{\mathbf C}_s\mathbf B_s^{-1/2}\right\|_\op=O_p(1).
\]
\end{lemma}

\begin{proof}
On $\mathcal E_s^c$, $\overline{\mathbf C}_s=0$ by the proof convention. Put
$\mu=\lambda_{s+1}$. By Lemmas~\ref{lem:derived-covariance-stability}--\ref{lem:empirical-cutoff-separation},
the event
\[
 \mathcal E_s\cap
 \{\delta_{s,N}\le1/4,\ \nu_{s,N}\le1,\
   \widehat\lambda_s\ge\lambda_s/2\}
\]
has probability tending to one. Work on this event. For
$h=\widehat E_s\mathbf a$, \eqref{eq:derived-centered-preconditioned-stability}
yields
\[
 (1-\delta_{s,N})\langle h,\mathcal T_{\kappa_0}h\rangle
 \le \mathbf a^\top\widehat{\boldsymbol\Lambda}_s\mathbf a+\delta_{s,N}\mu\|\mathbf a\|^2.
\]
Here $\mu\|\mathbf a\|^2\le2\mathbf a^\top\widehat{\boldsymbol\Lambda}_s\mathbf a$, so
\[
 \langle h,(\mathcal T_{\kappa_0}+\mu\mathrm{Id})h\rangle
 \le C \mathbf a^\top\widehat{\boldsymbol\Lambda}_s\mathbf a.
\]
Cauchy--Schwarz in regularized coordinates and
\eqref{eq:derived-regularized-mean-stability} then give
\[
 |\langle h,\overline K_N^0\rangle|^2
 \le\nu_{s,N}
     \langle h,(\mathcal T_{\kappa_0}+\mu\mathrm{Id})h\rangle
 \le C\nu_{s,N}\mathbf a^\top\widehat{\boldsymbol\Lambda}_s\mathbf a.
\]
Combining these inequalities with \eqref{eq:supp-overline-C} proves
$\overline{\mathbf C}_s\preceq C\widehat{\boldsymbol\Lambda}_s$ on the displayed event. Hence
\[
 \mathbf B_s^{-1/2}\overline{\mathbf C}_s\mathbf B_s^{-1/2}
 \preceq C \mathbf B_s^{-1/2}\widehat{\boldsymbol\Lambda}_s\mathbf B_s^{-1/2}
 \preceq C\mathbf I_s,
\]
which proves the claim.
\end{proof}

\begin{lemma}[Conditional second-moment lifting]
\label{lem:conditional-moment-lifting}
Assume Condition~\nameref{par:K1}. Let
$\mathbf B\in\mathbb R^{s\times s}$ be positive semidefinite and
$\mathscr X_{N,L}$-measurable.  Conditionally on $\mathscr X_{N,L}$,
\[
    \left\|
      \widehat{\boldsymbol\Phi}_{s}(X)^\top
      \mathbf B\widehat{\boldsymbol\Phi}_{s}(X)
    \right\|_{L^1(\bbP_X)}
    =
    \operatorname{tr}(\mathbf B\overline{\mathbf C}_{s}).
\]
For every $\mathscr X_{N,L}$-measurable Hilbert--Schmidt operator $A$,
\begin{equation}
    \left\|\|AK_X^0\|\right\|_{L^2(\bbP_X)}
    =\{\operatorname{tr}(A\mathcal T_{\kappa_0}A^*)\}^{1/2},
    \label{eq:conditional-operator-moment-lifting-population}
\end{equation}
\end{lemma}

\begin{proof}
For every $h\in\mathcal H_{\kappa_0}$, the covariance identity gives
\[
    \|\langle h,K_X^0\rangle\|_{L^2}
    =
    \langle h,\mathcal T_{\kappa_0}h\rangle^{1/2}.
\]
The quadratic identity follows by taking the conditional expectation of
\(\widehat{\boldsymbol\Phi}_{s}(X)^\top
\mathbf B\widehat{\boldsymbol\Phi}_{s}(X)\).  For the operator identity, diagonalize the
positive trace-class operator \(A^*A=\sum_\ell \mu_\ell u_\ell\otimes u_\ell\).
Then
\[
    \bbE\|AK_X^0\|^2
    =
    \sum_\ell\mu_\ell
    \bbE|\langle u_\ell,K_X^0\rangle|^2
    =
    \operatorname{tr}(A\mathcal T_{\kappa_0}A^*),
\]
which proves \eqref{eq:conditional-operator-moment-lifting-population}.
\end{proof}

\begin{lemma}[Leading-side PCA excess bound]
\label{lem:deterministic-leading-pca-excess}
Write $\mathcal T=\mathcal T_{\kappa_0}$,
$\Delta_N=\mathcal T-\widehat{\mathcal T}_{\kappa_0,N}$,
$\mu=\lambda_{s+1}$, and $g_s=\lambda_s-\lambda_{s+1}>0$.
On $\mathcal E_s$, suppose that
\[
    \widehat\lambda_{s+1}\le\mu+g_s/2,
    \qquad
    \|S_s\Delta_NS_s\|_\op\le1/4,
    \qquad
    S_s=\sum_{j\le s}(\lambda_j-\mu)^{-1/2}\Pi_j.
\]
Then
\begin{equation}
\label{eq:centered-pca-deterministic-leading-bound}
    \sum_{j\le s}(\lambda_j-\mu)
       \|\Pi_j\widehat\Pi_{\le s}^{\perp}\|_\HS^2
    \le16\sum_{j\le s}
       \frac{\|\Pi_j\Delta_N\Pi_{>s}\|_\HS^2}{\lambda_j-\mu}.
\end{equation}
The bound holds whenever the displayed conditions are met; no tail or moment
condition is required.
\end{lemma}

\begin{proof}
The bound is the $d=r=s$, $\mu=\lambda_{s+1}$ specialization of
\citet[arXiv version~3, equation~(3.7)]{reiss2016pca}.
We give a direct proof to separate
the operator inequality from the probabilistic control of the good events.
Let $P=\Pi_{\le s}$, $Q=\Pi_{>s}$, and
\[
    A=(\mathcal T-\mu\mathrm{Id})|_{\operatorname{ran}(P)}.
\]
Let $I_P$ denote the identity on $\operatorname{ran}(P)$. The operator $A$
on this subspace is positive definite with $A\succeq g_s I_P$, and
$S_s=A^{-1/2}P$, extended by zero on $\operatorname{ran}(Q)$. Choose an orthonormal basis
$(v_\ell)_\ell$ of $\operatorname{ran}(\widehat\Pi_{\le s}^{\perp})$
consisting of eigenvectors of the empirical covariance, with eigenvalues
$\nu_\ell\le\widehat\lambda_{s+1}$. This includes an orthonormal basis of
its null space; the empirical covariance is finite rank.
Put $u_\ell=A^{1/2}Pv_\ell$. Projecting
$(\mathcal T-\Delta_N)v_\ell=\nu_\ell v_\ell$ onto
$\operatorname{ran}(P)$ and multiplying by $A^{-1/2}$ gives
\[
    \{I_P+(\mu-\nu_\ell)A^{-1}
          -S_s\Delta_NS_s|_{\operatorname{ran}(P)}\}u_\ell
    =S_s\Delta_NQv_\ell.
\]
If $\nu_\ell\le\mu$, the sum of the first two terms is at least $I_P$.
Otherwise $0<\nu_\ell-\mu\le g_s/2$ and $A^{-1}\preceq g_s^{-1}I_P$,
so that sum is at least $I_P/2$. In either case the operator in braces is
at least $I_P/4$. Its inverse therefore has operator norm at most $4$, and
\[
    \|u_\ell\|\le4\|S_s\Delta_NQv_\ell\|.
\]
Summing the squares over this orthonormal basis gives
\[
\begin{split}
    \sum_{j\le s}(\lambda_j-\mu)
       \|\Pi_j\widehat\Pi_{\le s}^{\perp}\|_\HS^2
    &=\sum_\ell\|u_\ell\|^2\\
    &\le16\|S_s\Delta_NQ\widehat\Pi_{\le s}^{\perp}\|_\HS^2\\
    &\le16\|S_s\Delta_NQ\|_\HS^2
     =16\sum_{j\le s}
       \frac{\|\Pi_j\Delta_N\Pi_{>s}\|_\HS^2}{\lambda_j-\mu}.
\end{split}
\]
All these Hilbert--Schmidt norms are finite because $S_s$ is finite rank.
This proves the claim, including when the empirical covariance has a
nontrivial null space.
\end{proof}

\begin{lemma}[Relative PCA control of the empirically centered projector]
\label{lem:relative-pca-projector-control}
Define the leading-side PCA excess
\[
    \mathcal E_{s,N}^{\mathrm{PCA},+}
    :=
    \sum_{j\le s}(\lambda_j-\lambda_{s+1})
    \|\Pi_j\widehat\Pi_{\le s}^\perp\|_\HS^2.
\]
Under Conditions~\nameref{par:K2}, \nameref{par:K4}--\nameref{par:K5}, and the
relative-spectral tuning regime \eqref{eq:sharp-tuning-regime},
\begin{align}
    \|\widehat\Pi_{\le s}-\Pi_{\le s}\|_\HS^2
    &=O_p\!\left(\frac{\mathcal U_s\log N}{N}\right)
     =O_p\!\left(\frac{s^2\log(es)\log N}{N}\right),
    \label{eq:relative-pca-unweighted-bound}\\
    \mathcal E_{s,N}^{\mathrm{PCA},+}
    &=O_p\!\left(\frac{\mathcal W_s^{+}}{N}\right)
     =O_p\!\left(\frac{s^{2-q}\log(es)}{N}\right).
    \label{eq:relative-pca-leading-excess-bound}
\end{align}
Consequently,
\begin{equation}
\label{eq:relative-pca-prediction-weighted-bound}
    \|\mathcal T_{\kappa_0}^{1/2}(\widehat\Pi_{\le s}-\Pi_{\le s})\|_\HS
    =
    O_p\!\left(
      \sqrt{\frac{s^{2-q}\log(es)\log N}{N}}
    \right).
\end{equation}
\end{lemma}

\begin{proof}
Put $Z=K_X^0$, $\mathcal T=\mathcal T_{\kappa_0}$, and
\[
 \widetilde{\mathcal T}_N:=N^{-1}\sum_{i=1}^N Z_i\otimes Z_i,
 \qquad
 \bar Z_N:=N^{-1}\sum_{i=1}^N Z_i.
\]
Thus the perturbation generating the projector used in the estimator is
\begin{equation}
\label{eq:centered-relative-perturbation}
 \widehat{\mathcal T}_{\kappa_0,N}-\mathcal T
 =(\widetilde{\mathcal T}_N-\mathcal T_{\kappa_0})
  -\bar Z_N\otimes\bar Z_N.
\end{equation}
Recall $e_j=\sqrt{\lambda_j}\phi_j$. The standardized
Karhunen--Lo\`eve coordinate of $Z$ in direction $e_j$ is
\[
    \eta_j(X)
    :=\lambda_j^{-1/2}\langle e_j,Z\rangle
    =\phi_j(X).
\]
Condition~\nameref{par:K4}, in the fixed population eigenbasis, gives
\begin{equation}
\label{eq:relative-pca-eighth-moment}
    \sup_{j\ge1}\bbE|\eta_j|^8\le C_{X,8}<\infty.
\end{equation}
Thus the coordinate-moment assumption in
\citet[Lemma~1]{jirak2020perturbation} holds with moment order $8$.
Independence of different coordinates is not needed.

We next verify the finite-sample event required by the block version of the
relative perturbation theorem of
\citet{jirak2020perturbation,jirak2023relative}. To make the terminology
explicit, recall the population spectral decomposition
\[
    \mathcal T=\sum_{j\ge1}\lambda_j\Pi_j,
    \qquad \Pi_j=e_j\otimes e_j.
\]
A \emph{spectral block} is an index set \(B\subset\mathbb N\) representing
the population invariant subspace
\(\operatorname{ran}(\Pi_B)\), where
\[
    \Pi_B:=\sum_{j\in B}\Pi_j,
    \qquad
    b_B:=\operatorname{tr}(\Pi_B\mathcal T)=\sum_{j\in B}\lambda_j.
\]
Thus \(\Pi_B(\widehat{\mathcal T}_{\kappa_0,N}-\mathcal T)\Pi_C\) is the \((B,C)\) rectangular
submatrix of the perturbation when the operator is written in the population
eigenbasis; explicitly,
\[
 \Pi_B(\widehat{\mathcal T}_{\kappa_0,N}-\mathcal T)\Pi_C
 =\sum_{j\in B}\sum_{k\in C}
   \langle e_j,(\widehat{\mathcal T}_{\kappa_0,N}-\mathcal T)e_k\rangle
   e_j\otimes e_k.
\]
Here ``block'' refers to a group of population eigendirections, not to a
block of observations.

The target projector is indexed by
\[
    I:=\{1,\ldots,s\}.
\]
Thus $\Pi_I=\Pi_{\le s}$.
The set \(I\) is generally a union of spectral blocks rather than a single
block. Condition~\nameref{par:K2} ensures that no population eigenspace is
split at its boundary because \(\lambda_s>\lambda_{s+1}\). To construct the
comparison partition, let
\[
 k_0:=\bigwedge\{k>s:\lambda_k\le\lambda_s/2\}.
\]
Condition~\nameref{par:K5} implies $k_0\asymp s$. In the notation of the
cited block theorem, set
\[
 I':=\{1,\ldots,k_0-1\}.
\]
For every distinct eigenvalue \(\mu\) among
\(\lambda_1,\ldots,\lambda_{k_0-1}\), define
\[
    B(\mu):=\{j<k_0:\lambda_j=\mu\},
    \qquad
    B_{\mathrm{tail}}:=\{k_0,k_0+1,\ldots\}.
\]
The nonempty sets \(B(\mu)\), together with \(B_{\mathrm{tail}}\), form the
comparison partition. In particular, repeated eigenvalues remain in the same
block, while all sufficiently small eigenvalues are collected in the single
tail block.

\paragraph{Specialized block perturbation principle.}
The form of \citet[Theorem~2]{jirak2020perturbation} used below can be read as
follows. Suppose the perturbation between every pair of blocks $B,C$ satisfies
\[
 \|\Pi_B(\widehat{\mathcal T}_{\kappa_0,N}-\mathcal T)\Pi_C\|_\HS
 \le x\sqrt{b_Bb_C},
\]
the tail is separated from the target directions, and
$x\mathfrak r_s^\partial\le1/8$. Then, on the event that the empirical leading
index set corresponds to $I=\{1,\ldots,s\}$, the theorem gives
\[
 \|\widehat\Pi_{\le s}-\Pi_{\le s}\|_\HS^2
 \le Cx^2
 \sum_{j\le s}\sum_{k>s}
 \frac{\lambda_j\lambda_k}{(\lambda_j-\lambda_k)^2}
 =Cx^2\mathcal U_s.
\]
Thus the remainder of this part verifies the three hypotheses in the present
notation; it does not introduce a new perturbation argument.

For this partition, the three inputs to be verified are the following:
\begin{enumerate}[label=\textup{(\roman*)}]
\item for every pair of blocks \(B,C\), the relative block bound
\[
 \left\|\Pi_B\{(\widetilde{\mathcal T}_N-\mathcal T_{\kappa_0})
       -\bar Z_N\otimes\bar Z_N\}\Pi_C\right\|_\HS
 \le x_N\sqrt{b_Bb_C};
\]
\item for every \(i\in I\) and \(k\notin I'\), the far-block separation
\[
    |\lambda_i-\lambda_k|\ge\lambda_i/2;
\]
\item the contraction condition
\[
    x_N\mathfrak r_s^\partial\le\frac18
\]
for all sufficiently large \(N\), where the relative coefficient \(x_N\) is
specified below.
\end{enumerate}

The far-block condition is immediate. For \(i\in I\) and \(k\notin I'\),
the definition of \(k_0\) gives
\[
 \lambda_k\le\lambda_s/2\le\lambda_i/2,
 \qquad
 |\lambda_i-\lambda_k|\ge\lambda_i/2.
\]
It remains to establish the relative block bound and the contraction
condition. For this purpose, let \(B,C\) be any two blocks in the displayed
partition. Then \citet[Lemma~1]{jirak2020perturbation}, applied to each block pair
and followed by a union bound, gives
\begin{equation}
\label{eq:relative-pca-uncentered-block-event}
 \|\Pi_B(\widetilde{\mathcal T}_N-\mathcal T_{\kappa_0})\Pi_C\|_\HS
 \le C\frac{t}{\sqrt N}\sqrt{b_Bb_C}.
\end{equation}
For $t\ge1$ the exceptional probability is at most
\[
 Ck_0^2\{N^{-1}t^{-4}+e^{-t^2}\}.
\]
This records explicitly the moment order, the far-block index, and the
failure probability of the block event used below.

Empirical centering is controlled on the same partition by fourth moments.
For every block $B$, including the infinite tail block,
Lemma~\ref{lem:spectral-transform-moments} applied to $\Pi_B$ gives
\[
    \bbE\|\Pi_B\bar Z_N\|^4\le\frac{Cb_B^2}{N^2}.
\]
All blocks are nonempty and have $b_B>0$. For $N\ge3$, put
\[
    u_N:=\sqrt{\frac{\log N}{N}},\qquad
    \mathcal C_N:=\bigcap_B
       \{\|\Pi_B\bar Z_N\|^2\le b_Bu_N\}.
\]
There are at most $k_0+1$ blocks, so Markov's inequality and a union bound
yield
\begin{equation}
\label{eq:relative-pca-centering-failure}
    \bbP(\mathcal C_N^c)
    \le\frac{C(k_0+1)}{N^2u_N^2}
    =\frac{C(k_0+1)}{N\log N}.
\end{equation}
On $\mathcal C_N$, simultaneously for every pair of blocks,
\begin{equation}
\label{eq:relative-pca-centering-block-bound}
    \|\Pi_B(\bar Z_N\otimes\bar Z_N)\Pi_C\|_\HS
    =\|\Pi_B\bar Z_N\|\,\|\Pi_C\bar Z_N\|
    \le u_N\sqrt{b_Bb_C}.
\end{equation}
Only the relative scale $u_N$ is needed by the block theorem; no
exponential-confidence bound for the sample mean is used.

Take $t_N=\sqrt{4\log N}$ in
\eqref{eq:relative-pca-uncentered-block-event}. On its intersection with
$\mathcal C_N$, the combined perturbation satisfies
\begin{equation}
\label{eq:relative-pca-combined-block-event}
    \|\Pi_B\{(\widetilde{\mathcal T}_N-\mathcal T_{\kappa_0})
           -\bar Z_N\otimes\bar Z_N\}\Pi_C\|_\HS
    \le x_N\sqrt{b_Bb_C}
    \quad\text{for every }B,C,
\end{equation}
where $x_N=C\sqrt{\log N/N}$, with a sufficiently large fixed
constant depending only on $C_{X,8}$. Let $\mathcal B_N$ denote the
simultaneous event in \eqref{eq:relative-pca-combined-block-event}.
The two union bounds give
\begin{equation}
\label{eq:relative-pca-block-failure}
\begin{split}
    \bbP(\mathcal B_N^c)
    &\le Ck_0^2\{N^{-1}t_N^{-4}+e^{-t_N^2}\}
          +\frac{C(k_0+1)}{N\log N}\\
    &\le C\left\{
          \frac{s^2}{N(\log N)^2}
          +\frac{s^2}{N^4}
          +\frac{s+1}{N\log N}\right\}
     =o(1).
\end{split}
\end{equation}
Here $k_0\asymp s$, the first term tends to zero by
Lemma~\ref{lem:tuning-consequences}, and the remaining terms vanish because
$s<N$ and $N\to\infty$.
For the selected leading index set $I=\{1,\ldots,s\}$, the separation
quantity in the cited block theorem is $\mathfrak r_s^\partial$. Lemmas~
\ref{lem:cross-boundary-spectrum-consequences} and
\ref{lem:tuning-consequences} give
\[
 x_N\mathfrak r_s^\partial=o(1).
\]
In particular, \(x_N\mathfrak r_s^\partial\le1/8\) for all sufficiently
large \(N\), as required by the block perturbation theorem. The relative
block bounds above apply simultaneously to every pair of blocks in the
near/tail partition.
The block failure bounds above and Lemma~
\ref{lem:empirical-cutoff-separation} show that the intersection of
$\mathcal E_s$ with all the block events has probability tending to one. On
this intersection, all inputs required by the relative
perturbation theorem are present: the cross-block inequalities are
\eqref{eq:relative-pca-combined-block-event}, the far-block separation follows
from the definition of \(k_0\), and
\(x_N\mathfrak r_s^\partial\le1/8\) gives the contraction condition.
Therefore \citet[Theorem~2]{jirak2020perturbation} applies to the displayed
block partition and to the combined perturbation
$(\widetilde{\mathcal T}_N-\mathcal T_{\kappa_0})-\bar Z_N\otimes\bar Z_N$. With
$I=\{1,\ldots,s\}$, the event \(\mathcal E_s\) identifies the empirical
projector indexed by \(I\) with \(\widehat\Pi_{\le s}\), and the theorem gives the
bound
\[
    \|\widehat\Pi_{\le s}-\Pi_{\le s}\|_\HS^2
    \le Cx_N^2\mathcal U_s.
\]
Since $x_N^2=O(\log N/N)$, Lemma~
\ref{lem:cross-boundary-spectrum-consequences} proves
\eqref{eq:relative-pca-unweighted-bound}. Notice that this argument retains
the logarithmic confidence factor throughout.

We next prove the one-sided excess using
Lemma~\ref{lem:deterministic-leading-pca-excess}. Put
\[
    \mu_s:=\lambda_{s+1},\qquad
    \Delta_N:=\mathcal T-\widehat{\mathcal T}_{\kappa_0,N}
       =-(\widetilde{\mathcal T}_N-\mathcal T_{\kappa_0})
          +\bar Z_N\otimes\bar Z_N.
\]
Work on $\mathcal G_{s,N}^{\mathrm{mid}}$ from
\eqref{eq:midpoint-relative-good-event}. The proof of
Lemma~\ref{lem:empirical-cutoff-separation} shows that this event is
contained in $\mathcal E_s$ and satisfies both hypotheses of the
excess lemma. Hence
\eqref{eq:centered-pca-deterministic-leading-bound} applies on an event
whose probability tends to one. No covariance-relative sub-Gaussian
assumption is used in this step.
It remains to control the right-hand side. Since the population covariance
has zero cross-blocks and the observations are independent,
\[
 \bbE\|\Pi_j(\widetilde{\mathcal T}_N-\mathcal T_{\kappa_0})\Pi_{>s}\|_\HS^2
 =\frac1N\bbE[\|\Pi_jZ\|^2\|\Pi_{>s}Z\|^2]
 \le\frac{C\lambda_j}{N}\sum_{k>s}\lambda_k.
\]
Indeed,
\[
 \|\Pi_jZ\|^2=\lambda_j\eta_j^2,
 \qquad
 \|\Pi_{>s}Z\|^2=\sum_{k>s}\lambda_k\eta_k^2,
\]
so Tonelli's theorem, Cauchy--Schwarz, and
\eqref{eq:relative-pca-eighth-moment} give
\[
 \bbE[\|\Pi_jZ\|^2\|\Pi_{>s}Z\|^2]
 =\lambda_j\sum_{k>s}\lambda_k\bbE(\eta_j^2\eta_k^2)
 \le C\lambda_j\sum_{k>s}\lambda_k.
\]
Hence
\[
 \sum_{j\le s}
 \frac{\|\Pi_j(\widetilde{\mathcal T}_N-\mathcal T_{\kappa_0})\Pi_{>s}\|_\HS^2}
      {\lambda_j-\mu_s}
 =O_p(\mathcal W_s^+/N).
\]
For the rank-one centering term, put
$U_j=\langle e_j,\bar Z_N\rangle$ and
$\bar\eta_j=N^{-1}\sum_{i=1}^N\eta_j(X_i)$, so
$U_j=\sqrt{\lambda_j}\bar\eta_j$. Since the summands are independent and
centered, expansion of the fourth power and
\eqref{eq:relative-pca-eighth-moment} give
\[
 \bbE\bar\eta_j^4
 =N^{-4}\left\{N\bbE\eta_j^4
   +3N(N-1)(\bbE\eta_j^2)^2\right\}
 \le C/N^2.
\]
Cauchy--Schwarz therefore yields
\[
 \bbE(U_j^2U_k^2)
 \le\lambda_j\lambda_k
   (\bbE\bar\eta_j^4\,\bbE\bar\eta_k^4)^{1/2}
 \le C\lambda_j\lambda_k/N^2.
\]
Moreover,
\[
 \|\Pi_j(\bar Z_N\otimes\bar Z_N)\Pi_{>s}\|_\HS^2
 =U_j^2\sum_{k>s}U_k^2.
\]
By this identity and Tonelli's theorem,
\[
 \bbE\sum_{j\le s}
 \frac{\|\Pi_j(\bar Z_N\otimes\bar Z_N)\Pi_{>s}\|_\HS^2}
      {\lambda_j-\mu_s}
 \le
 \frac{C}{N^2}
 \left(\sum_{k>s}\lambda_k\right)
 \left(\sum_{j\le s}\frac{\lambda_j}{\lambda_j-\mu_s}\right)
 =\frac{C\mathcal W_s^+}{N^2}.
\]
By the elementary inequality
\[
 \|\Pi_j\Delta_N\Pi_{>s}\|_\HS^2
 \le2\|\Pi_j(\widetilde{\mathcal T}_N-\mathcal T_{\kappa_0})\Pi_{>s}\|_\HS^2
    +2\|\Pi_j(\bar Z_N\otimes\bar Z_N)\Pi_{>s}\|_\HS^2,
\]
the nonnegative quantity
\[
    V_{s,N}:=\sum_{j\le s}
       \frac{\|\Pi_j\Delta_N\Pi_{>s}\|_\HS^2}{\lambda_j-\mu_s}
\]
satisfies $\bbE V_{s,N}\le C\mathcal W_s^+/N$.
For every $M>0$, the excess lemma and Markov's inequality
therefore give
\[
    \bbP\!\left(\mathcal E_{s,N}^{\mathrm{PCA},+}
                >M\mathcal W_s^+/N\right)
    \le\bbP\{(\mathcal G_{s,N}^{\mathrm{mid}})^c\}
       +\bbP\!\left(16V_{s,N}>M\mathcal W_s^+/N\right)
    \le o(1)+\frac{C}{M}.
\]
Thus $\mathcal E_{s,N}^{\mathrm{PCA},+}=O_p(\mathcal W_s^+/N)$.
No independence between the good event and $V_{s,N}$ is required.
Lemma~\ref{lem:cross-boundary-spectrum-consequences} now gives
\eqref{eq:relative-pca-leading-excess-bound}.

It remains to convert these two PCA quantities into the prediction-weighted
norm through an operator identity. On $\mathcal E_s$, put
\[
    a_j:=\|\widehat\Pi_{\le s}^\perp e_j\|^2\quad(j\le s),
    \qquad
    b_k:=\|\widehat\Pi_{\le s} e_k\|^2\quad(k>s).
\]
Since $\Pi_{\le s}$ and $\widehat\Pi_{\le s}$ are orthogonal projectors of the same finite
rank,
\[
    \sum_{j\le s}a_j
    =\sum_{k>s}b_k
    =\frac12\|\widehat\Pi_{\le s}-\Pi_{\le s}\|_\HS^2.
\]
Moreover,
\begin{align*}
 \|\mathcal T_{\kappa_0}^{1/2}(\widehat\Pi_{\le s}-\Pi_{\le s})\|_\HS^2
 &=
 \sum_{j\le s}\lambda_j a_j+\sum_{k>s}\lambda_k b_k\\
 &\le
 \sum_{j\le s}(\lambda_j-\lambda_{s+1})a_j
 +\lambda_{s+1}\left(\sum_{j\le s}a_j+\sum_{k>s}b_k\right)\\
 &=
 \mathcal E_{s,N}^{\mathrm{PCA},+}
 +\lambda_{s+1}\|\widehat\Pi_{\le s}-\Pi_{\le s}\|_\HS^2.
\end{align*}
Using $\lambda_{s+1}\asymp s^{-q}$ together with
\eqref{eq:relative-pca-unweighted-bound} and
\eqref{eq:relative-pca-leading-excess-bound} proves
\eqref{eq:relative-pca-prediction-weighted-bound}.  Lemma~
\ref{lem:empirical-cutoff-separation} gives
$\bbP(\mathcal E_s^c)\to0$, and the exceptional-event convention preserves
the $O_p$ conclusion. These are probability-order bounds under the
eighth-moment condition; the proof does not assert exponential failure
probabilities or unconditional mean bounds for the empirical projector.
\end{proof}

\begin{proposition}[Relative spectral stability]
\label{prop:sharp-spectral-verification}
Under Conditions~\nameref{par:K1}--\nameref{par:K2}, \nameref{par:K4}, and
\nameref{par:K5}, and under the relative-spectral tuning regime
\eqref{eq:sharp-tuning-regime}, the full-sample empirical eigensystem
constructed in Section~\ref{app:empirical-spectral-construction} satisfies
the following spectral bounds:
\[
    \frac{\widehat\lambda_s}{\lambda_s}\xrightarrow{p}1,\qquad
    \|\mathbf B_s^{-1/2}\overline{\mathbf C}_{s}\mathbf B_s^{-1/2}\|_{\op}=O_p(1),
\]
\[
    \|\mathcal T_{\kappa_0}^{1/2}(\widehat\Pi_{\le s}-\Pi_{\le s})\|_{\HS}
    =O_p\!\left(
      \sqrt{\frac{s^{2-q}\log(es)\log N}{N}}
    \right),
\]
\[
    \|\overline K_N^0\|_{\mathcal H_{\kappa_0}}
    =O_p(N^{-1/2}).
\]
\end{proposition}

\begin{proof}
Let
\[
    \widetilde{\mathcal T}_N
    :=
    \frac1N\sum_{i=1}^N K_{X_i}^0\otimes K_{X_i}^0.
\]
The full-sample empirical centered covariance operator generated by
Section~\ref{app:empirical-spectral-construction} is
\[
    \widehat{\mathcal T}_{{\kappa_0},N}
    =
    \widetilde{\mathcal T}_N
    -\overline K_N^0
     \otimes\overline K_N^0,
\]
so its leading eigenpairs are exactly
\((\widehat\lambda_j,\widehat e_j)\) on $\mathcal E_s$.
Lemma~\ref{lem:empirical-eigenvalue-comparability} gives
$\widehat\lambda_s/\lambda_s\xrightarrow{p}1$, and
Lemma~\ref{lem:sharp-leverage-stability} gives the ridge-leverage bound.
Equation~\eqref{eq:derived-centering-rate} gives
$\|\overline K_N^0\|=O_p(N^{-1/2})$, while
Lemma~\ref{lem:relative-pca-projector-control} gives
\[
    \|\mathcal T_{\kappa_0}^{1/2}(\widehat\Pi_{\le s}-\Pi_{\le s})\|_\HS
    =O_p\!\left(
      \sqrt{\frac{s^{2-q}\log(es)\log N}{N}}
    \right).
\]
The displayed projector rate converges to zero. For $1<q<2$,
\[
 \frac{s^{2-q}\log(es)\log N}{N}
 \le
 \frac{s^2\log^2(es)\log N}{N}=o(1)
\]
by the first tuning restriction. For $q\ge2$, the squared rate is at most
$\log(es)\log N/N=o(1)$.

Labeled fitting stability follows from Lemma~\ref{lem:labeled-gram-stability},
not by treating its inverse as constant when averaging over $L$.
\end{proof}

For the remaining regression bounds, choose either the \emph{basic spectral
route}, consisting of Conditions~\nameref{par:K1}--\nameref{par:K3},
the tuning regime \eqref{eq:basic-tuning-regime}, and
Proposition~\ref{prop:primitive-spectral-verification}, or the
\emph{relative spectral route}, consisting of Conditions~\nameref{par:K1}--\nameref{par:K2},
\nameref{par:K4}, and \nameref{par:K5}, the relative-spectral tuning regime
\eqref{eq:sharp-tuning-regime}, and
Proposition~\ref{prop:sharp-spectral-verification}. Define
$\widetilde\rho_{r,s,n,N,\xi}$ as
$\rho_{r,s,n,N,\xi}$ or $\rho_{r,s,n,N,\xi}^{\star}$.

\begin{lemma}[Labeled means and covariance]
\label{lem:labeled-full-sample-design-comparison}
For $f_{n,N}=(N-n)/\{n(N-1)\}$, conditional on all pooled covariates,
\begin{equation}
\label{eq:labeled-full-sample-design-comparison}
\begin{aligned}
 \bbE_L[\overline{\boldsymbol\Phi}_L\overline{\boldsymbol\Phi}_L^\top]
 &=f_{n,N}\widehat{\boldsymbol\Lambda}_s,\\
 \bbE_L[\mathbf C_{L,s}]
 &=\frac{N(n-1)}{n(N-1)}\widehat{\boldsymbol\Lambda}_s,\qquad
 \bbE_L\|\mathbf B_s^{-1/2}\overline{\boldsymbol\Phi}_L\|^2\le s/n.
\end{aligned}
\end{equation}
\end{lemma}
\begin{proof}
The inclusion probabilities are $n/N$ and $n(n-1)/\{N(N-1)\}$.
Use $P_N\widehat{\boldsymbol\Phi}_s=0$ and
$P_N\widehat{\boldsymbol\Phi}_s\widehat{\boldsymbol\Phi}_s^\top
=\widehat{\boldsymbol\Lambda}_s$. Subtract the mean outer product from
$\bbE_L\mathbf C^u_{L,s}=\widehat{\boldsymbol\Lambda}_s$ and whiten.
\end{proof}

\begin{lemma}[Derived stability of the centered labeled covariance]
\label{lem:labeled-gram-stability}
Under either spectral route, put
\[
 \eta_{L,N}:=\|\mathbf B_s^{-1/2}(\mathbf C_{L,s}
                     -\widehat{\boldsymbol\Lambda}_s)\mathbf B_s^{-1/2}\|_\op.
\]
Then $s/n\to0$ and $\eta_{L,N}=o_p(1)$. For any fixed $0<\delta<1$, with probability
tending to one,
\begin{equation}
\label{eq:labeled-inverse-sandwich}
 (1-\delta)\mathbf B_s\preceq\mathbf C_{L,s}+\xi I
 \preceq(1+\delta)\mathbf B_s,\qquad
 \|\mathbf B_s^{1/2}\mathbf A_s\mathbf B_s^{1/2}\|_\op
 \le(1-\delta)^{-1}.
\end{equation}
For $N=n$ these hold exactly with $\eta_{L,N}=0$.
\end{lemma}
\begin{proof}
First, the basic conditions already imply $s/n\to0$. Put
$C_{\rm tr}:=\operatorname{tr}(\mathcal T_{\kappa_0})<\infty$ and
$g_s:=\lambda_s-\lambda_{s+1}$. Since $s\lambda_s\le C_{\rm tr}$
and $\xi/\lambda_s\to0$, along indices with $N>n$ we have eventually
\[
 \frac{s}{n}
 \le\frac{C_{\rm tr}}{n\lambda_s}
 \le\frac{2C_{\rm tr}}{n(\lambda_s+\xi)}\longrightarrow0
\]
by the labeled condition in \eqref{eq:basic-tuning-regime}. Along indices with $N=n$,
$g_s\le\lambda_s\le C_{\rm tr}/s$ gives
\[
 \frac{s^2}{n}\le\frac{C_{\rm tr}^2}{n g_s^2}\longrightarrow0
\]
by the basic eigengap condition, which also implies $s/n\to0$.
For the relative route, $s/n\to0$ follows from its retained assumption
$s^2/n\to0$.

If $N=n$, the centered labeled covariance equals the pooled diagonal
covariance identically, so $\eta_{L,N}=0$ and the sandwich is exact.
It remains to prove the assertions along indices with $N>n$.
For the basic route, condition on the pooled covariates and set
$v_i=\mathbf B_s^{-1/2}\widehat{\boldsymbol\Phi}_s(X_i)$,
$\ell_N=\max(1,\max_i\|v_i\|^2)$ and
$S_N=P_Nvv^\top\preceq I$. The centered finite population
$v_iv_i^\top-S_N$ has norm at most $2\ell_N$ and matrix variance at most
$\ell_N I$. The without-replacement matrix Bernstein bound
\citep[Theorem~1]{gross2010note} gives, with conditional probability
at least $1-\varepsilon$,
\[
 \|P_Lvv^\top-S_N\|_\op
 \le C\left\{\sqrt{\frac{\ell_N\log(2s/\varepsilon)}n}
               +\frac{\ell_N\log(2s/\varepsilon)}n\right\}.
\]
Write $Z_N$ for the left-hand side and $b_N(\varepsilon)$ for the
right-hand side, which is measurable with respect to the pooled covariates.
The zero-extension convention on $\mathcal E_s^c$ gives $Z_N=0$ there,
so the conditional probability bound holds on that event as well.
Bounded kernel and eigenvalue comparability give
$\ell_N\le 1+4B_\kappa/(\widehat\lambda_s+\xi)$.
The labeled condition in \eqref{eq:basic-tuning-regime} gives
$b_N(\varepsilon)=o_p(1)$ for each fixed $\varepsilon\in(0,1)$.
For every fixed $a>0$, integrating the conditional bound yields
\[
 \bbP(Z_N>a)
 \le \bbP\{Z_N>b_N(\varepsilon)\}
       +\bbP\{b_N(\varepsilon)>a\}
 \le \varepsilon+\bbP\{b_N(\varepsilon)>a\}.
\]
Thus $\limsup_{n,N\to\infty}\bbP(Z_N>a)\le\varepsilon$.
Letting $\varepsilon\downarrow0$ proves $Z_N=o_p(1)$ unconditionally.
Centering subtracts $(P_Lv)(P_Lv)^\top$, whose norm is $O_p(s/n)=o_p(1)$
by Lemma~\ref{lem:labeled-full-sample-design-comparison}.

For the relative route, work directly in the population-centered RKHS.
Put $\mu=\lambda_{s+1}$, $R_\mu=(\mathcal T_{\kappa_0}+\mu I)^{-1/2}$
and $d_\mu=\mathcal N(\mu)\lesssim s$, the last inequality following
from \aref{K5}. For $Z=K_X^0$, \aref{K4} and Cauchy--Schwarz give
\[
 \bbE\|R_\mu Z\|^4
 =\sum_{j,k}\frac{\lambda_j\lambda_k\,\bbE[\phi_j^2\phi_k^2]}
                   {(\lambda_j+\mu)(\lambda_k+\mu)}
 \le C d_\mu^2.
\]
For an iid sample of size $m$, let
$\widetilde{\mathcal T}_m=P_mZ\otimes Z$ and
$\widehat{\mathcal T}_m=\widetilde{\mathcal T}_m-\bar Z_m\otimes\bar Z_m$.
Hilbert--Schmidt variance and the mean identity yield
\[
 \bbE\|R_\mu(\widetilde{\mathcal T}_m-\mathcal T_{\kappa_0})R_\mu\|_\HS^2
 \le Cd_\mu^2/m,\qquad
 \bbE\|R_\mu\bar Z_m\|^2=d_\mu/m.
\]
Consequently
$\delta_m:=\|R_\mu(\widehat{\mathcal T}_m-\mathcal T_{\kappa_0})R_\mu\|_\op
=O_p(s/\sqrt m)$.
Apply this separately to $m=n,N$: the labeled observations are marginally
iid, and no independence between these two samples is required.
The assumed $s^2/n\to0$, with $N\ge n$, makes both discrepancies vanish.

The transfer to empirical coordinates does not presume the conclusion.
On $\delta_N<1/2$,
\[
 (1-\delta_N)(\mathcal T_{\kappa_0}+\mu I)
 \preceq\widehat{\mathcal T}_N+\mu I.
\]
Then $\widehat\lambda_s\ge(1-2\delta_N)\lambda_s$.
So $\mu\le2\widehat\lambda_s$ with probability tending to one.
It follows that
$U_N=(\mathcal T_{\kappa_0}+\mu I)^{1/2}\widehat E_s\mathbf B_s^{-1/2}$
has 
\[
\|U_N\|_{\mathrm{op}}^2
\le \frac{1}{1-\delta_N}
\left(1+\frac{\mu}{\widehat\lambda_s}\right)
=O_p(1).
\]
Labeled centering and the empirical
eigenvector equations give the exact identities
\[
 \mathbf C_{L,s}=\widehat E_s^*\widehat{\mathcal T}_L\widehat E_s,\qquad
 \widehat{\boldsymbol\Lambda}_s
       =\widehat E_s^*\widehat{\mathcal T}_N\widehat E_s.
\]
Thus
\[
 \eta_{L,N}\le\|U_N\|_\op^2(\delta_n+\delta_N)
 =O_p(s/\sqrt n+s/\sqrt N)=o_p(1).
\]
This avoids assuming independence of learned eigenvectors and labeled data.
The sandwich follows by adding $I$ to the whitened discrepancy.
\end{proof}

\begin{lemma}[Integrated query and weight bounds]
\label{lem:labeled-weight-bound}
Under either spectral route, let
$\mathbf Q_L=\bbE_{X^\circ}[\mathbf q_L(X^\circ)
\mathbf q_L(X^\circ)^\top\mid\mathscr X_{N,L}]$. Then
\[
 \|\mathbf B_s^{-1/2}\mathbf Q_L\mathbf B_s^{-1/2}\|_\op=O_p(1),\qquad
 \bbE_{X^\circ}\|\mathbf a_{s,\xi}(X^\circ)\|^2=O_p(s/n).
\]
\end{lemma}
\begin{proof}
Use $\mathbf Q_L\preceq2\overline{\mathbf C}_s+
2\overline{\boldsymbol\Phi}_L\overline{\boldsymbol\Phi}_L^\top$,
the applicable pooled query bound and
\eqref{eq:labeled-full-sample-design-comparison}. Since
$\mathbf A_s\mathbf C_{L,s}\mathbf A_s
=\mathbf A_s-\xi\mathbf A_s^2\preceq\mathbf A_s$,
\[
 \bbE_{X^\circ}\|\mathbf a_{s,\xi}(X^\circ)\|^2
 \le n^{-1}+n^{-1}\operatorname{tr}(\mathbf A_s\mathbf Q_L)
 =O_p(s/n)
\]
by \eqref{eq:labeled-inverse-sandwich}. The random inverse is bounded
before averaging, not removed from an expectation over $L$.
\end{proof}

\begin{lemma}[$L^2$ conditional response-noise bound]
\label{lem:sharp-response-noise}
Under either spectral route and \aref{R2},
\[
 \left\|\sup_t|\widehat C_{s,\xi}(t;X^\circ)
                 -\overline C^L_{s,\xi}(t;X^\circ)|\right\|_{2,X\mid\mathscr D_{n,N}}
 =O_p(\sqrt{s/n}).
\]
\end{lemma}
\begin{proof}
Conditional on covariates, $L$ and a query, the difference equals
$\sum_{i\in L}a_{s,\xi,i}(X^\circ)\zeta_t(X_i,Y_i)$.
The weights are fixed under this conditioning and do not depend on the responses. Apply \aref{R2},
square and integrate over the independent query. The conditional expected
squared norm is bounded by
$C_M^2\bbE_{X^\circ}\|\mathbf a_{s,\xi}(X^\circ)\|^2=O_p(s/n)$.
Conditional Markov's inequality on events where this bound is at most
$Cs/n$ proves the assertion. No inverse moments on exceptional events
are required for this probability-order conclusion.
\end{proof}

\begin{lemma}[Hilbert-valued stability of the fitted smoother]
\label{lem:hilbert-smoother-stability}
Under either spectral route, let $\mathcal V$ be a separable Hilbert space.
There is a nonnegative $\mathscr X_{N,L}$-measurable random constant
$C_{n,N}=O_p(1)$ such that, simultaneously for all inputs
$(V_i)_{i\in L}\in\mathcal V^n$,
\[
 \left\|\sum_{i\in L}a_{s,\xi,i}(\cdot)V_i
       \right\|_{L^2(\bbP_X;\mathcal V)}^2
 \le C_{n,N}\,\frac1n\sum_{i\in L}\|V_i\|_{\mathcal V}^2.
\]
The inputs may depend on the pooled covariates and on $L$. No independence
between the inputs and the weights is required for this inequality.
\end{lemma}

\begin{proof}
Put $\mathbf D_L=\mathbf C_{L,s}+\xi\mathbf I_s$ and retain
$\mathbf Q_L=\bbE_{X^\circ}[\mathbf q_L(X^\circ)\mathbf q_L(X^\circ)^\top
\mid\mathscr X_{N,L}]$. Lemmas~\ref{lem:labeled-gram-stability} and
\ref{lem:labeled-weight-bound} imply
\[
 c_{n,N}:=\|\mathbf D_L^{-1/2}\mathbf Q_L\mathbf D_L^{-1/2}\|_\op=O_p(1).
\]
For the fixed input values set $\bar V_L=n^{-1}\sum_i V_i$ and define
$\widehat B:\mathbb R^s\to\mathcal V$ by
\[
 \widehat B\mathbf D_L
   =\frac1n\sum_{i\in L}(V_i-\bar V_L)\otimes\mathbf q_L(X_i).
\]
This is the unique slope in the Hilbert-valued centered ridge fit, and
its fitted function is $\bar V_L+\widehat B\mathbf q_L(x)$.
Expansion of its penalized least-squares objective at the optimum gives
\[
 0\le \frac1n\sum_i\|V_i-\bar V_L\|^2
       -\|\widehat B\mathbf D_L^{1/2}\|_\HS^2.
\]
Since $n^{-1}\sum_i\|V_i\|^2=\|\bar V_L\|^2+
 n^{-1}\sum_i\|V_i-\bar V_L\|^2$, it follows that
\[
 \|\bar V_L\|^2+\|\widehat B\mathbf D_L^{1/2}\|_\HS^2
 \le\frac1n\sum_i\|V_i\|^2.
\]
Query integration and $\mathbf Q_L\preceq c_{n,N}\mathbf D_L$ yield
\[
\begin{split}
 \|\bar V_L+\widehat B\mathbf q_L\|_{L^2}^2
 &\le2\|\bar V_L\|^2+2\|\widehat B\mathbf Q_L^{1/2}\|_\HS^2\\
 &\le2\max\{1,c_{n,N}\}
       \big(\|\bar V_L\|^2+\|\widehat B\mathbf D_L^{1/2}\|_\HS^2\big).
\end{split}
\]
Take $C_{n,N}=2\max\{1,c_{n,N}\}$. This constant does not depend on
$\mathcal V$ or the input values. All inequalities are deterministic
conditional on the design, which proves the simultaneous assertion.
\end{proof}

\begin{lemma}[Propagation of the fractional-source residual]
\label{lem:approximation-residual-bound}
Under either spectral route and \aref{M1}, \aref{M2}, and \aref{R1},
\[
 \left\|\sup_t|(\widehat S-I)r_{t,s}|\right\|_{2,X\mid\mathscr X_{N,L}}
 =O_p(b_{r,s}).
\]
Under the Hilbert-source assumptions of
Proposition~\ref{prop:hilbert-automatic-complexity}, the right side can
be replaced by $O_p(b^{\mathrm H}_{r,s})$.
\end{lemma}

\begin{proof}
Use the common $\ell^2$ representation in
Lemma~\ref{lem:automatic-spectral-truncation}. Linearity gives
\[
 \sup_t|(\widehat S-I)r_{t,s}(x)|
 \le B_r\|(\widehat S-I)v_{r,s}(x)\|_{\ell^2}.
\]
The vector $v_{r,s}$ is deterministic because the population kernel,
its eigenfunctions, and $s$ are specified without observing the data.
The labeled covariates are marginally iid, so
\[
 \bbE P_L\|v_{r,s}\|^2=\|v_{r,s}\|_{L^2}^2=\sum_{j>s}\lambda_j^r.
\]
Markov's inequality and Lemma~\ref{lem:hilbert-smoother-stability} give
$\|\widehat S v_{r,s}\|_{L^2}=O_p(\|v_{r,s}\|_{L^2})$. The triangle
inequality proves the result. No conditional-mean-zero assumption is
made on this residual. For Hilbert responses use the deterministic
vector tail $\sum_{j>s}\theta_j\phi_j$ instead; its squared $L^2$ norm
is $\sum_{j>s}\|\theta_j\|^2$, as established in
Proposition~\ref{prop:hilbert-automatic-complexity}.
\end{proof}

\begin{lemma}[Fitted representation residual]
\label{lem:sharp-subspace-perturbation}
Under either spectral route and \aref{R1}, let $e_{N,s}$ be the
deterministic spectral scale for that route:
\[
 e_{N,s}^{\mathrm{basic}}=\{\sqrt N(\lambda_s-\lambda_{s+1})\}^{-1},
 \qquad
 e_{N,s}^{\mathrm{ref}}=\left\{\frac{s^{2-q}\log(es)\log N}{N}\right\}^{1/2}.
\]
Then
\[
 \left\|\sup_t|(\widehat S-I)d_{t,s}|\right\|_{2,X\mid\mathscr X_{N,L}}
 =O_p(G_{r,s}e_{N,s}).
\]
This bound accounts for fitting the data-dependent representation
residual on the same pooled covariates used to learn the eigenspace.
\end{lemma}

\begin{proof}
Work on the event $\mathcal E_s$, whose probability tends to one, and
write $T=\mathcal T_{\kappa_0}$, $P=\Pi_{\le s}$, and
$Q=\widehat\Pi_{\le s}$. Set
\[
 A_{r,s}:=\sum_{j\le s}\lambda_j^{(r-1)/2}(e_j\otimes e_j),\qquad
 w_t:=\sum_{j\le s}a_j(t)e_j,\qquad
 V_{s,N}(x):=A_{r,s}(I-Q)K_x^0.
\]
Here $e_j=\sqrt{\lambda_j}\phi_j$, $g_{t,s}=A_{r,s}w_t$, and
$\sup_t\|w_t\|\le B_r$. Self-adjointness gives
\[
 d_{t,s}(x)=\langle w_t,V_{s,N}(x)\rangle,\qquad
 \sup_t|(\widehat S-I)d_{t,s}(x)|
 \le B_r\|(\widehat S-I)V_{s,N}(x)\|.
\]
Since $A_{r,s}P=A_{r,s}$ and
$\|A_{r,s}\|_\op=\lambda_s^{(r-1)/2}$,
\begin{equation}
\label{eq:source-representation-population}
 \|V_{s,N}\|_{L^2(\bbP_X;\mathcal H_{\kappa_0})}^2
 =\|A_{r,s}(P-Q)T^{1/2}\|_\HS^2
 \le\lambda_s^{r-1}E_{N,s}^2.
\end{equation}
We next control its \emph{pooled empirical} norm; an iid expectation
identity cannot be applied to the data-dependent $V_{s,N}$.

Let $\widetilde T_N=P_N(K^0\otimes K^0)$, with
$P_Nf=N^{-1}\sum_{i=1}^Nf(X_i)$. Under the basic route,
$\|\widetilde T_N-T\|_\op=O_p(N^{-1/2})=o_p(\lambda_s)$ by
\aref{K3}. Under the refined route,
Lemma~\ref{lem:derived-covariance-stability}, with
$\mu=\lambda_{s+1}\le\lambda_s$, gives
\[
 \widetilde T_N\preceq(1+\chi_{s,N})T+\chi_{s,N}\mu I,
 \qquad\chi_{s,N}=o_p(1).
\]
Thus in either case the event
$\widetilde T_N\preceq2(T+\lambda_s I)$ has probability tending to one.
On that event,
\[
 P_N\|V_{s,N}\|^2
 \le2\|A_{r,s}(I-Q)T^{1/2}\|_\HS^2
       +2\lambda_s\|A_{r,s}(I-Q)\|_\HS^2.
\]
For the second term, the spectral decomposition of $T$ gives
\[
 \lambda_s\|P(I-Q)\|_\HS^2
 =\lambda_s\sum_{j\le s}\|(I-Q)e_j\|^2
 \le\sum_{j\le s}\lambda_j\|(I-Q)e_j\|^2
 \le E_{N,s}^2.
\]
Combining this with \eqref{eq:source-representation-population} yields
\begin{equation}
\label{eq:source-representation-pooled}
 P_N\|V_{s,N}\|^2\le4\lambda_s^{r-1}E_{N,s}^2
\end{equation}
on an event of probability tending to one.

Conditional on all pooled covariates, $V_{s,N}$ is fixed and the labeled
set is a uniform subset. Therefore the exact identity
\[
 \bbE_L[P_L\|V_{s,N}\|^2\mid X_1,\ldots,X_N]
       =P_N\|V_{s,N}\|^2
\]
holds, also when $N=n$. Since $E_{N,s}=O_p(e_{N,s})$, conditional
Markov on the good events in \eqref{eq:source-representation-pooled}
gives $P_L\|V_{s,N}\|^2=O_p(\lambda_s^{r-1}e_{N,s}^2)$.
Lemma~\ref{lem:hilbert-smoother-stability} now controls
$\|\widehat S V_{s,N}\|_{L^2}$ by the square root of this empirical
quantity, while \eqref{eq:source-representation-population} controls
the unfitted query residual. Multiplying their sum by $B_r$ proves the
claim. This argument never treats the learned projector or the labeled
inverse as independent of the data that define them.
\end{proof}

\begin{lemma}[Ridge bias in the learned subspace]
\label{lem:sharp-fixed-design}
Under either spectral route and \aref{R1},
\[
 \left\|\sup_t\left|\xi\mathbf q_L(\cdot)^\top
       \mathbf A_s\mathbf c_{t,s}\right|\right\|_{2,X\mid\mathscr X_{N,L}}
 =O_p\left(G_{r,s}\frac{\xi}{\lambda_s+\xi}\right).
\]
\end{lemma}

\begin{proof}
The common RKHS synthesis is an isometry on $\mathcal E_s$, so
$\sup_t\|\mathbf c_{t,s}\|\le G_{r,s}$.
The centered feature bound gives $\sup_x\|\mathbf q_L(x)\|\le4\sqrt{B_\kappa}$.
The labeled covariance sandwich and empirical eigenvalue comparability
imply $\|\mathbf A_s\|_\op=O_p((\lambda_s+\xi)^{-1})$.
Cauchy--Schwarz therefore proves the asserted bound even uniformly in
$x$. The exceptional event has vanishing probability. The rank-$s$
norm bound appears only in this ridge term and in the representation
term, not as a multiplier of the response-noise bound.
\end{proof}

\subsection{Normalized-contrast oracle inequality and rate interpretations}
\label{app:additional-rate-results}

For $x\in\cX_0$, define
\[
    R_{s,\xi}^{\oplus}(x)
    :=\sup_{y\ne m_\oplus(x)}
    \big|\widehat C_{s,\xi}((m_\oplus(x),y);x)
    -\dot C_\oplus((m_\oplus(x),y);x)\big|.
\]

\begin{theorem}[Normalized-contrast oracle inequality]
\label{thm:contrast-error}
Under Conditions~\aref{M1}, \aref{M2}, and \aref{R1}--\aref{R2}, and
under either spectral route,
\[
 \bbE_{X^\circ}\left[R_{s,\xi}^{\oplus}(X^\circ)^2
          \mid\mathscr D_{n,N}\right]^{1/2}
 =O_p(\widetilde\rho_{r,s,n,N,\xi}).
\]
Under the source assumptions of
Proposition~\ref{prop:hilbert-automatic-complexity}, replace $b_{r,s}$
by $b^{\mathrm H}_{r,s}$ and $G_{r,s}$ by $G^{\mathrm H}_{r,s}$.
\end{theorem}

\begin{proof}
On $\mathcal E_s$, Lemma~\ref{lem:fixed-design-decomposition} gives
four terms. Their indexed query norms are, respectively,
\[
 O_p\!\left(\sqrt{\frac{s}{n}}\right),\qquad
 O_p(b_{r,s}),\qquad
 O_p(G_{r,s}e_{N,s}),\qquad
 O_p\!\left(G_{r,s}\frac{\xi}{\lambda_s+\xi}\right),
\]
by Lemmas~\ref{lem:sharp-response-noise},
\ref{lem:approximation-residual-bound},
\ref{lem:sharp-subspace-perturbation}, and \ref{lem:sharp-fixed-design}.
The last three terms depend only on the covariates and $L$, so their
query norms do not change when conditioning on $\mathscr D_{n,N}$
instead of $\mathscr X_{N,L}$. The supremum in $R_{s,\xi}^{\oplus}(x)$
is over a subset of $\mathfrak T$. Taking the supremum before applying
the $L^2$ triangle inequality therefore yields the sum of the four
bounds. Proposition~\ref{prop:primitive-spectral-verification} or
\ref{prop:sharp-spectral-verification} supplies the appropriate
$e_{N,s}$, producing the stated scale. Since $\bbP(\mathcal E_s^c)\to0$,
this is also the unconditional probability-order conclusion under the
specified fallback convention. The vector-valued source uses the
sharper residual calculation in
Proposition~\ref{prop:hilbert-automatic-complexity}; all other steps
are unchanged.
\end{proof}

\begin{remark}[Closure membership and quantitative regularity]
\label{rmk:approximation-tradeoff}
The model uses a single fixed kernel, not a sequence of learned RKHSs.
A target in $\overline{\mathcal H_{\kappa_0}}^{\,L^2(\bbP_X)}\setminus\mathcal H_{\kappa_0}$ is approximated
by $g_{t,s}=\Pi_{\le s}\dot h_t$; its RKHS norm may diverge as $s\to\infty$.
The source condition quantifies this through $G_{r,s}$ and the common
Hilbert-valued residual through $b_{r,s}$. It is stronger than pointwise
closure membership: the latter alone supplies neither uniform
approximation of the contrast family nor a convergence rate.

The use of $\ell^2$ in the residual proof does not impose a Hilbert
response space. It encodes the coefficients of scalar contrast
regressions. Lemma~\ref{lem:hilbert-smoother-stability} controls this
entire residual class and Lemma~\ref{lem:sharp-subspace-perturbation}
handles data reuse, rather than bounding an arbitrary residual family
by a $\sqrt{s}$-inflated envelope. The target remains the
true $m_\oplus$; no minimizer of an RKHS-projected criterion is substituted.
The theorems do not cover targets at positive distance from the closure,
and are not simultaneous guarantees for data-selected kernels, source
classes, or tuning sequences. Example~\ref{ex:bhv-closure-source}
verifies a non-RKHS BHV model, not every experimental generating law.
\end{remark}

\begin{remark}[Spectral-decay regimes and labeled tuning]
\label{rmk:spectral-decay-regimes}
Let $\lambda_j\asymp j^{-q}$, $qr>1$, and use the refined route. The
squared source, variance, and representation terms are
\[
 s^{-(qr-1)},\qquad s/n,\qquad
 s^{2-qr}\log(es)\log N/N.
\]
Balancing the first two gives $s\asymp n^{1/(qr)}$; the sufficient
labeled condition $s^2/n\to0$ then requires $qr>2$. This is a limitation
of the present covariance-stability argument, not a necessity or a
minimax transition. No amount of unlabeled data removes that labeled
condition.

For $1<qr\le2$ and $N\ge n$, take an admissible complete-block cutoff
$s\asymp\sqrt n/(\log n)^2$ and the ridge choice in
Remark~\ref{rmk:ridge-sequence}. Then
$s\log(es)\sqrt{\log N/N}=O((\log n)^{-1/2})\to0$,
$s^2/n\to0$, and the bias dominates the other squared terms. Under
quadratic curvature this gives the valid, not claimed optimal, rate
\[
 \mathfrak R_{s,\xi,\varepsilon_{\mathrm{opt}}}
 =O_p\big(n^{-(qr-1)/2}(\log n)^{2(qr-1)}
          +\varepsilon_{\mathrm{opt}}\big).
\]
Indeed the representation-to-bias ratio is bounded by
$s\log(es)\log N/N=o(1)$, and $s/n$ is no larger than the displayed
bias order. For other curvature exponents, substitute the corresponding
contrast scale into Theorem~\ref{thm:frechet-error}.

For the basic route with $N>n$, the automatic balanced cutoff instead
fails $\log(2s)/(n\lambda_s)\to0$ when $r\le1$. The basic theorem is
still useful at smaller ranks. For example, with a gap lower bound
$\lambda_s-\lambda_{s+1}\gtrsim s^{-q-1}$, choose $s\asymp n^c$,
$0<c<1/q$, and $N/s^{3q-qr+2}\to\infty$. Together with the same ridge
choice these conditions make all four basic error terms vanish. A
balanced-rate claim for that route would require an additional,
stronger approximation or design-stability argument.
\end{remark}

\begin{remark}[Curvature regimes]
Let $p_M=2/(\beta_M-1)$. For a vanishing contrast scale
$\widetilde\rho_{r,s,n,N,\xi}$, the effective power in the prediction
bound is $p_M\wedge2$; quadratic curvature gives its square. Only an
$L^2$ query bound is used, including when $p_M>2$, because the inversion
inequality truncates that higher power on its small-error event.
\end{remark}

\begin{remark}[Sufficient ridge choices]
\label{rmk:ridge-sequence}
Under polynomial decay,
$G_{r,s}^2\lesssim s^{q(1-r)}$. For the general metric source, the
choice
\[
 0<\xi=O\big(s^{-(3q-1)/2}\big)
\]
ensures $\xi/\lambda_s\to0$ and
$\{G_{r,s}\xi/(\lambda_s+\xi)\}^2\lesssim s^{-(qr-1)}$.
For the vector-valued Hilbert source, use
$0<\xi=O(s^{-3q/2})$ to make the squared ridge term no larger than
$s^{-qr}$. These are sufficient deterministic choices, not statements
about validation-selected regularization.
\end{remark}

\begin{corollary}[Polynomial rate for a Hilbert-valued source]
\label{cor:hilbert-fractional-rate}
Assume \aref{M1}, \aref{M2}--\aref{M3}, \aref{K1},
\aref{K4}--\aref{K5}, and the source conditions of
Proposition~\ref{prop:hilbert-automatic-complexity}, with $qr>1$.
Use the minimizer convention of Theorem~\ref{thm:frechet-error} and take
\[
 N\ge n,\qquad s\asymp n^{1/(qr+1)},\qquad
 0<\xi=O(s^{-3q/2}).
\]
Then the refined tuning conditions hold, \aref{R2} is automatic, and
under the automatic quadratic curvature,
\[
 \mathfrak R_{s,\xi,\varepsilon_{\mathrm{opt}}}
 =O_p\big(n^{-qr/(qr+1)}+\varepsilon_{\mathrm{opt}}\big).
\]
\end{corollary}

\begin{proof}
Proposition~\ref{prop:hilbert-automatic-complexity} and
Theorem~\ref{thm:contrast-error} give the squared scale
\[
 s^{-qr}+s/n+
 s^{q(1-r)}\{\xi/(\lambda_s+\xi)\}^2+
 s^{2-qr}\log(es)\log N/N.
\]
The first two terms balance at $n^{-qr/(qr+1)}$, and the chosen ridge
has the same or smaller order. Since $qr>1$, both
$s^2/n$ and $s\log(es)\sqrt{\log N/N}$ tend to zero. For large $n$,
$\log N/N\le\log n/n$ when $N\ge n$, and the representation-to-bias
ratio is at most $Cs^2(\log n)^2/n=o(1)$. The source supplies
\aref{R1}, the proposition supplies \aref{R2}, and the quadratic
argmin bound in Theorem~\ref{thm:frechet-error}, using the refined scale,
proves the result.
\end{proof}

\subsection{Escape and argmin inversion}
\label{app:argmin-lemmas}

\begin{lemma}[Uniform linear escape]
\label{lem:uniform-linear-escape}
Under Conditions~\nameref{par:M1} and \nameref{par:M2}--\nameref{par:M3}, there exist
constants $\eta_0\in(0,\eta_M)$ and
$c_{\mathrm{esc}},R_0>0$ such that, uniformly in $x\in\cX_0$,
\[
    \Delta_\oplus(y,x)
    \ge
    c_{\mathrm{esc}}d_{\cY}(y,m_\oplus(x))
\]
whenever $d_{\cY}(y,m_\oplus(x))\ge\eta_0$.
Moreover,
\[
 \Delta_\oplus(y,x)
 \ge\tfrac12d_{\cY}(y,m_\oplus(x))^2
\]
whenever \(d_{\cY}(y,m_\oplus(x))\ge R_0\).
\end{lemma}

\begin{proof}
Fix $x\in\cX_0$. Let \(y_*\) be the reference point in Condition~\nameref{par:M1}, and set
\[
    \sigma_*:=
    \sup_{x\in\cX_0}M_\oplus(y_*,x)^{1/2}<\infty.
\]
Since \(m_\oplus(x)\) minimizes \(M_\oplus(\cdot,x)\),
\(M_\oplus(m_\oplus(x),x)\le M_\oplus(y_*,x)\). Minkowski's inequality also
gives
\[
    d_{\cY}(m_\oplus(x),y_*)
    \le M_\oplus(m_\oplus(x),x)^{1/2}
       +M_\oplus(y_*,x)^{1/2}
    \le2\sigma_*.
\]
Thus both the target and its conditional root-mean-square radius are uniformly
bounded by Condition~\nameref{par:M1}, which is the bound needed for the
escape argument.
Let $m_x=m_\oplus(x)$ and $r=d_{\cY}(y,m_x)$. Minkowski's inequality gives
\[
    M_\oplus(y,x)^{1/2}
    \ge r-M_\oplus(m_x,x)^{1/2}.
\]
Thus, whenever $r\ge4\sigma_*$,
\[
    \Delta_\oplus(y,x)
    \ge r^2-2r\sigma_*\ge\frac12r^2.
\]
Choose $\eta_0\in(0,\eta_M)$. Global separation gives
\[
    a_0
    :=
    \inf_{x\in\cX_0}\inf_{d_{\cY}(y,m_x)\ge\eta_0}
    \Delta_\oplus(y,x)>0.
\]
Let $R_0:=4\sigma_*\lor\eta_0$. On
$\eta_0\le r\le R_0$, the excess is at least
$(a_0/R_0)r$. On $r\ge R_0$, it is at least
$(\eta_0/2)r$. Taking
$c_{\mathrm{esc}}=(a_0/R_0)\land(\eta_0/2)$ proves the claim.
\end{proof}

\begin{lemma}[Argmin inversion for approximate minimizers]
\label{lem:normalized-argmin-inversion}
Assume Condition~\nameref{par:M3} and the conclusion of
Lemma~\ref{lem:uniform-linear-escape}. There exist constants
$C<\infty$ and $\eta_1\in(0,1]$, independent of $x$, the data, and the sample
sizes, with the following property. For $x\in\cX_0$ and $0\le\eta\le\eta_1$,
suppose $\widehat m_x$ satisfies
\[
    \widehat M_{s,\xi}(\widehat m_x,x)
    \le
    \inf_{y\in\cY}\widehat M_{s,\xi}(y,x)+\eta.
\]
Put
\[
    m_x=m_\oplus(x),
    \qquad
    d_x=d_{\cY}(\widehat m_x,m_x),
    \qquad
    R_x=R_{s,\xi}^{\oplus}(x).
\]
Then
\begin{equation}
\label{eq:deterministic-inversion}
    d_x^2
    \le
    C\left[
        R_x^{2/(\beta_M-1)}
        \ind\{R_x<c_{\mathrm{esc}}/2\}
        +R_x^2\ind\{R_x\ge c_{\mathrm{esc}}/2\}
        +\eta^{2/\beta_M}
    \right].
\end{equation}
\end{lemma}

\begin{proof}
Set $\eta_1:=1\land(c_{\mathrm{esc}}\eta_0/4)$ and take $0\le\eta\le\eta_1$.
If \(d_x=0\), then \eqref{eq:deterministic-inversion} is immediate. Assume
\(d_x>0\). Then \(t_x=(m_x,\widehat m_x)\in\mathfrak T\), and approximate
optimality gives
\[
 \widehat M_{s,\xi}(\widehat m_x,x)
 -\widehat M_{s,\xi}(m_x,x)\le\eta.
\]
By the definitions of \(R_x\), \(\widehat C_{s,\xi}\), and
\(\dot C_\oplus\),
\[
 \left|\{\widehat M_{s,\xi}(\widehat m_x,x)
 -\widehat M_{s,\xi}(m_x,x)\}
 -\Delta_\oplus(\widehat m_x,x)\right|
 \le R_xd_x.
\]
Combining the last two displays yields
\begin{equation}
\label{eq:approximate-basic-inequality}
    \Delta_\oplus(\widehat m_x,x)
    \le R_xd_x+\eta.
\end{equation}
If $R_x<c_{\mathrm{esc}}/2$ and $d_x\ge\eta_0$, the uniform escape bound
gives
\[
    c_{\mathrm{esc}}d_x
    \le R_xd_x+\eta
    <\frac{c_{\mathrm{esc}}}{2}d_x+\eta.
\]
Thus $d_x<2\eta/c_{\mathrm{esc}}\le\eta_0/2$, contradicting
$d_x\ge\eta_0$. Hence
$d_x<\eta_0$, and local curvature together with
\eqref{eq:approximate-basic-inequality} gives
\[
    c_Md_x^{\beta_M}\le R_xd_x+\eta.
\]
Young's inequality yields
\[
    R_xd_x
    \le
    \frac{c_M}{2}d_x^{\beta_M}
    +C R_x^{\beta_M/(\beta_M-1)}.
\]
Consequently,
\[
    d_x^{\beta_M}
    \le C\{R_x^{\beta_M/(\beta_M-1)}+\eta\}.
\]
Using $(u+v)^{2/\beta_M}\le
C_{\beta_M}(u^{2/\beta_M}+v^{2/\beta_M})$ for $u,v\ge0$ gives
\[
    d_x^2
    \le
    C\left\{R_x^{2/(\beta_M-1)}+\eta^{2/\beta_M}\right\},
\]
which proves the first part of \eqref{eq:deterministic-inversion}.

Now suppose $R_x\ge c_{\mathrm{esc}}/2$. If
$d_x<R_0$, then $d_x^2\le
4R_0^2c_{\mathrm{esc}}^{-2}R_x^2$. If
$d_x\ge R_0$, the quadratic lower bound in
Lemma~\ref{lem:uniform-linear-escape} and
\eqref{eq:approximate-basic-inequality} give
\[
    \frac12d_x^2\le R_xd_x+\eta
    \le\frac14d_x^2+R_x^2+\eta.
\]
Thus $d_x^2\le4R_x^2+4\eta$. Since
$R_x\ge c_{\mathrm{esc}}/2$ and $\eta\le1$,
\[
    4\eta
    \le \frac{16}{c_{\mathrm{esc}}^2}R_x^2.
\]
This proves the large-$R_x$ part with an explicit constant and
completes the proof.
\end{proof}

\section{Proofs of the Main-Text Results}
\label{app:main-result-proofs}

This section collects proofs of results stated in the main text; auxiliary
results are proved immediately after their statements.

\begin{proof}[Proof of Theorem~\ref{thm:frechet-error}]
Put
\[
    p_M:=\frac{2}{\beta_M-1},
    \qquad
    \rho:=\rho_{r,s,n,N,\xi},
    \qquad
    R_X:=R_{s,\xi}^{\oplus}(X^\circ).
\]
Since $\varepsilon_{\mathrm{opt}}\to0$, it satisfies
$0\le\varepsilon_{\mathrm{opt}}\le\eta_1$ for all sufficiently large sample
sizes, where $\eta_1$ is from
Lemma~\ref{lem:normalized-argmin-inversion}. Apply that lemma for every
$x\in\cX_0$ with $\eta=\varepsilon_{\mathrm{opt}}$. Since
$\bbP_X(\cX_0)=1$ and $X^\circ$ is independent of the data, integrating over
$X^\circ$ gives the same prediction-risk inequality; no condition on the
conditional law outside $\cX_0$ is used. Thus
\begin{equation}
\label{eq:frechet-risk-proof-integrated-inversion}
 \mathfrak R_{s,\xi,\varepsilon_{\mathrm{opt}}}
 \le C\bigg[
 \bbE_{X^\circ}[R_X^{p_M}
       \ind(R_X<c_{\mathrm{esc}}/2)\mid\mathscr D_{n,N}]
 +\bbE_{X^\circ}[R_X^2
       \ind(R_X\ge c_{\mathrm{esc}}/2)\mid\mathscr D_{n,N}]
 +\varepsilon_{\mathrm{opt}}^{2/\beta_M}
 \bigg].
\end{equation}
Theorem~\ref{thm:contrast-error} states precisely that
\[
    \bbE_{X^\circ}[R_X^2\mid\mathscr D_{n,N}]^{1/2}
    =O_p(\rho).
\]

Suppose first that $p_M\le2$. Conditional Lyapunov inequality yields
\[
 \bbE_{X^\circ}[R_X^{p_M}\mid\mathscr D_{n,N}]
 \le
 \left\{\bbE_{X^\circ}[R_X^2\mid\mathscr D_{n,N}]\right\}^{p_M/2}
 =O_p(\rho^{p_M}),
\]
while the second expectation in
\eqref{eq:frechet-risk-proof-integrated-inversion} is bounded by
$O_p(\rho^2)$ using the preceding second-moment bound.

If $p_M>2$, then on $\{R_X<c_{\mathrm{esc}}/2\}$,
\[
    R_X^{p_M}
    \le(c_{\mathrm{esc}}/2)^{p_M-2}R_X^2.
\]
Hence both expectations in
\eqref{eq:frechet-risk-proof-integrated-inversion} are $O_p(\rho^2)$.
Combining the two cases, and retaining the harmless $\rho^{p_M}$ term in
both regimes, proves
\[
    \mathfrak R_{s,\xi,\varepsilon_{\mathrm{opt}}}
    =
    O_p\!\left(
      \rho^{2/(\beta_M-1)}+\rho^2
      +\varepsilon_{\mathrm{opt}}^{2/\beta_M}
    \right).
\]
\end{proof}

\begin{proof}[Proof of Theorem~\ref{thm:sharp-frechet-error}]
Under the relative spectral route, Theorem~\ref{thm:contrast-error} gives
\[
    \bbE_{X^\circ}[R_{s,\xi}^{\oplus}(X^\circ)^2
      \mid\mathscr D_{n,N}]^{1/2}
    =O_p(\rho_{r,s,n,N,\xi}^{\star}).
\]
Repeating the pointwise inversion and conditional moment argument in the
proof of Theorem~\ref{thm:frechet-error}, with
$\rho_{r,s,n,N,\xi}^{\star}$ in place of
$\rho_{r,s,n,N,\xi}$, proves the stated rate.
\end{proof}

\begin{proof}[Proof of Corollary~\ref{cor:polynomial-rate}]
Since $qr>2$, the chosen cutoff satisfies
\[
 \frac{s^2}{n}\asymp n^{2/(qr)-1}\to0,
 \qquad
 s\log(es)\sqrt{\frac{\log N}{N}}
 \lesssim n^{1/(qr)-1/2}(\log n)^{3/2}\to0
\]
for all $N\ge n$. The latter inequality uses that $u\mapsto\log u/u$
is decreasing for sufficiently large $u$. Also
\[
 \frac{\xi}{\lambda_s}=O(s^{-(q-1)/2})\to0.
\]
Thus all three refined tuning restrictions are verified, rather than
inferred from the desired error rate.

The source and labeled-noise terms satisfy
\[
 b_{r,s}^2+\frac{s}{n}
 \lesssim s^{-(qr-1)}+\frac{s}{n}
 \lesssim n^{-(qr-1)/(qr)}.
\]
The increasing source norm is retained explicitly in the ridge term:
\[
 \left(G_{r,s}\frac{\xi}{\lambda_s+\xi}\right)^2
 \lesssim s^{q(1-r)}s^{-(q-1)}
 =s^{-(qr-1)}.
\]
Similarly the squared representation term is bounded by
\[
 G_{r,s}^2\frac{s^{2-q}\log(es)\log N}{N}
 \lesssim \frac{s^{2-qr}\log(es)\log N}{N}.
\]
Its ratio to $s^{-(qr-1)}$ is at most
$Cs\log(es)\log N/N\lesssim n^{1/(qr)-1}(\log n)^2\to0$.
Using $(\sum_{j=1}^4u_j)^2\le4\sum_{j=1}^4u_j^2$ for nonnegative
$u_j$ now yields
\[
 (\rho_{r,s,n,N,\xi}^{\star})^2\lesssim n^{-(qr-1)/(qr)}.
\]
Theorem~\ref{thm:sharp-frechet-error} gives
\[
 \mathfrak R_{s,\xi,\varepsilon_{\mathrm{opt}}}
 =O_p\left(n^{-\frac{qr-1}{qr(\beta_M-1)}}
           +n^{-\frac{qr-1}{qr}}
           +\varepsilon_{\mathrm{opt}}^{2/\beta_M}\right).
\]
When $\beta_M=2$ the first two terms have the same order, proving the
stated specialization.
\end{proof}

\clearpage
\section{Hybrid Stability and Prediction Mixing}
\label{app:th-stability-proof}

This section proves Theorem~\ref{thm:th-stability} for the actual
concatenated-feature ridge fit. All operators and their operator or
Hilbert--Schmidt norms below are on the common RKHS
$\mathcal H_{\kappa_0}$, unless another space is displayed. We retain the
shared-law sampling scheme, common numerical tuning, and measurable
spectral conventions of the theorem. The proof uses population
Hadamard curvature, not convexity or uniqueness of the signed fitted
Fr\'echet criteria.

\begin{proof}[Proof of Theorem~\ref{thm:th-stability}]
Write $d=d_\cY$, $g_s=\lambda_s-\lambda_{s+1}$, and
$f_{n,N}=(N-n)/\{n(N-1)\}$. The source and kernel are fixed, and all
probability orders concern the joint law of the training data.

\paragraph{1. Common-space representation of the fitted kernels.}
Put $Z_i=K_{X_i}^0$ and, for $B=L,[N]$, define
\[
 \overline Z_B=\frac1{|B|}\sum_{i\in B}Z_i,\qquad
 z_L(x)=K_x^0-\overline Z_L,\qquad
 \widehat T_B=\frac1{|B|}\sum_{i\in B}
           (Z_i-\overline Z_B)\otimes(Z_i-\overline Z_B).
\]
Let $\widehat\Pi_B$ denote the leading rank-$s$ projector of $\widehat T_B$ on
its empirical-gap event, with the zero fallback of
Appendix~\ref{app:empirical-spectral-construction} otherwise. Denote
its common-RKHS synthesis map by $\widehat E_s^B$; on the fallback
both this map and $\widehat\Pi_B$ are zero. Thus
$\widehat\Pi_B=\widehat E_s^B(\widehat E_s^B)^*$.
The empirical features have the representation
\[
 \widehat{\boldsymbol\Phi}_s^B(x)
   =(\widehat E_s^B)^*(K_x^0-\overline Z_B).
\]
Subtracting their labeled average cancels the representation-specific
centering exactly:
\begin{equation}
\label{eq:th-proof-recentering}
 \mathbf q_L^B(x)=(\widehat E_s^B)^*z_L(x).
\end{equation}
This is the common-space identity already used in the original
spectral proof, applied separately to the two samples.
For $\alpha\in[0,1]$, set
\[
 W_\alpha(u,v)=\sqrt{1-\alpha}\,\widehat E_s^Lu
                    +\sqrt\alpha\,\widehat E_s^{[N]}v,
 \qquad Q_\alpha=W_\alpha W_\alpha^*
                    =(1-\alpha)\widehat\Pi_L+\alpha \widehat\Pi_{[N]}.
\]
The concatenated centered feature is $W_\alpha^*z_L(x)$, so its kernel
is $\langle z_L(x),Q_\alpha z_L(x')\rangle$. In particular,
\begin{equation}
\label{eq:th-proof-kernel-difference}
 Q_\alpha-\widehat\Pi_{[N]}=(1-\alpha)(\widehat\Pi_L-\widehat\Pi_{[N]}),\qquad 0\preceq Q_\alpha\preceq\mathrm{Id}.
\end{equation}
This identity describes the sum kernel before the joint ridge refit;
it is not an average of two fitted coefficients. 
\paragraph{2. Nested covariance and eigenspace comparison.}
For deterministic vectors $A_1,\ldots,A_N$ in any Hilbert space, let
$\overline A_B=|B|^{-1}\sum_{i\in B}A_i$. Sampling $L$ uniformly without
replacement gives
\begin{equation}
\label{eq:th-proof-finite-population}
 \bbE_L\!\left[\|\overline A_L-\overline A_{[N]}\|^2\right]
 =f_{n,N}\frac1N\sum_{i=1}^N\|A_i-\overline A_{[N]}\|^2.
\end{equation}
Indeed, for $a_i=A_i-\overline A_{[N]}$, one has $\sum_i a_i=0$,
$\bbP(i\in L)=n/N$, and
$\bbP(i,j\in L)=n(n-1)/\{N(N-1)\}$ for $i\ne j$. Expanding the
squared norm of $n^{-1}\sum_i\ind\{i\in L\}a_i$ and using
$\sum_{i\ne j}\langle a_i,a_j\rangle=-\sum_i\|a_i\|^2$ proves
\eqref{eq:th-proof-finite-population}.

Let $C_\kappa=\sup_{x\in\mathcal S}\|K_x^0\|^2\le4B_\kappa$ and
$\widetilde T_B=|B|^{-1}\sum_{i\in B}Z_i\otimes Z_i$.
Applying \eqref{eq:th-proof-finite-population}, conditional on all
pooled covariates, to $Z_i$ and then to the Hilbert--Schmidt operators
$Z_i\otimes Z_i$ yields
\[
 \bbE_L[\|\overline Z_L-\overline Z_{[N]}\|^2\mid X_1,\ldots,X_N]
 \le C_\kappa f_{n,N},\qquad
 \bbE_L[\|\widetilde T_L-\widetilde T_{[N]}\|_\HS^2\mid X_1,\ldots,X_N]
 \le C_\kappa^2 f_{n,N}.
\]
Since $\widehat T_B=\widetilde T_B-\overline Z_B\otimes\overline Z_B$,
$\|\overline Z_B\|\le\sqrt{C_\kappa}$, and
\[
 \|u\otimes u-v\otimes v\|_\HS
 \le(\|u\|+\|v\|)\|u-v\|,
\]
we obtain
\begin{equation}
\label{eq:th-proof-nested-covariance}
 \bbE_L[\|\widehat T_L-\widehat T_{[N]}\|_\HS^2\mid X_1,\ldots,X_N]
 \le10C_\kappa^2 f_{n,N}.
\end{equation}
Consequently, along sample sizes with $N>n$,
$\|\widehat T_L-\widehat T_{[N]}\|_\op=O_p(\sqrt{f_{n,N}})$.
This calculation does not treat the two covariance estimates as independent.

Put $A=\widehat T_{[N]}$, $B=\widehat T_L$,
$P=\widehat\Pi_{[N]}$, $P'=\widehat\Pi_L$, and
$\delta=\|A-B\|_\op$. Write $a_s,a_{s+1}$ for the pooled empirical
eigenvalues at the cutoff. The event
\[
 \mathcal G_{N,s}:=\{a_s-a_{s+1}\ge3g_s/4\}
\]
has probability tending to one under either spectral route. For the
basic route, Proposition~\ref{prop:primitive-spectral-verification}
and \aref{K3} give
$\|A-\mathcal T_{\kappa_0}\|_\op=o_p(g_s)$, so Weyl's inequality
applies. For the refined route, use
Lemma~\ref{lem:empirical-cutoff-separation}. No corresponding
labeled-sample event is required.

We compare projectors on $\mathcal G_{N,s}$ in two cases, without a
dimension factor. If $\delta\le g_s/8$, Weyl's inequality gives a
labeled empirical gap of at least
$a_s-a_{s+1}-2\delta\ge g_s/2$. Both projectors therefore use their
leading spectral subspaces, not the zero fallback.
Let $A_\perp=A|_{\operatorname{ran}(\mathrm{Id}-P)}$ and
$B_+=B|_{\operatorname{ran}P'}$. Their spectra are separated by at
least $a_s-a_{s+1}-\delta\ge5g_s/8$. For
$X=(\mathrm{Id}-P)P'$ and $b=(a_s+a_{s+1})/2$, invariance gives
\[
 A_\perp X-XB_+=(\mathrm{Id}-P)(A-B)P'=:E,\qquad
 X=-\int_0^\infty e^{t(A_\perp-b\mathrm{Id})}
                  E e^{-t(B_+-b\mathrm{Id})}\,dt.
\]
The integral identity follows by differentiation; both exponential
factors decay. Hence
$\|(\mathrm{Id}-P)P'\|_\op\le8\delta/(5g_s)$.
Interchanging $A,B$ gives the same bound for
$(\mathrm{Id}-P')P$. Since
$P-P'=P(\mathrm{Id}-P')-(\mathrm{Id}-P)P'$, we obtain
$\|P-P'\|_\op\le16\delta/(5g_s)$.

If $\delta>g_s/8$, both $P,P'$ are orthogonal projections, including
the possible zero fallback for $P'$, so
$\|P-P'\|_\op\le1<8\delta/g_s$. Combining the two cases gives
\[
 \|\widehat\Pi_L-\widehat\Pi_{[N]}\|_\op
 \le\frac{8}{g_s}\|\widehat T_L-\widehat T_{[N]}\|_\op
 \quad\text{on }\mathcal G_{N,s}.
\]
Together with \eqref{eq:th-proof-nested-covariance} and
$\bbP(\mathcal G_{N,s}^c)\to0$, this proves
\begin{equation}
\label{eq:th-proof-projector-rate}
 \delta_\Pi:=\|\widehat\Pi_L-\widehat\Pi_{[N]}\|_\op
 =O_p\!\left(\frac{\sqrt{f_{n,N}}}{g_s}\right),\qquad N>n.
\end{equation}
If $N=n$, the two empirical covariances and their common-convention
projectors are identical, so $\delta_\Pi=0$ exactly. This handles
sequences that also contain sample sizes with $N=n$ without assigning
an ambiguous probability order to a zero normalizing sequence.

\paragraph{3. Sensitivity of the actual ridge refit.}
Define $S:\mathcal H_{\kappa_0}\to\mathbb R^n$ by
$(Sh)_i=n^{-1/2}\langle z_L(X_i),h\rangle$, $i\in L$.
For a positive contraction $Q$, set
\[
 A_Q=SQS^*+\xi I_n,\qquad v_Q(x)=A_Q^{-1}SQz_L(x),\qquad
 \Gamma_\xi=\frac{2C_\kappa}{\xi}
                      \left(1+\frac{C_\kappa}{\xi}\right).
\]
For scalar observations $b_i$, put
$\overline b_L=n^{-1}\sum_{i\in L}b_i$ and
$b_c=n^{-1/2}(b_i-\overline b_L)_{i\in L}$.
If the centered features are $W^*z_L(x)$ and $WW^*=Q$, their centered
coefficient objective is
\[
 \min_\theta\ \frac12\|b_c-SW\theta\|_2^2+\frac\xi2\|\theta\|_2^2.
\]
The normal equations and the identity
$(F^*F+\xi I)^{-1}F^*=F^*(FF^*+\xi I)^{-1}$, with $F=SW$, give the
fitted scalar function
\begin{equation}
\label{eq:th-proof-scalar-fit}
 \widehat b_Q(x)=\overline b_L+\langle b_c,v_Q(x)\rangle_{\mathbb R^n}.
\end{equation}
This applies to $W=W_\alpha$ and to either endpoint synthesis map,
including the zero fallback, and keeps the cross-block regression
inverse in the concatenated fit.

The centered feature bounds give
\[
 \|S\|_\op^2\le\|S\|_\HS^2
 =n^{-1}\sum_{i\in L}\|Z_i\|^2-\|\overline Z_L\|^2\le C_\kappa,
 \qquad \sup_{x\in\mathcal S}\|z_L(x)\|\le2\sqrt{C_\kappa}.
\]
For $D=Q_1-Q_2$ and $R_j=(SQ_jS^*+\xi I_n)^{-1}$, positivity gives
$\|R_j\|_\op\le\xi^{-1}$ and the resolvent identity gives
$R_1-R_2=-R_1SDS^*R_2$. Consequently,
\[
 v_{Q_1}(x)-v_{Q_2}(x)
 =R_1SDz_L(x)-R_1SDS^*R_2SQ_2z_L(x),
\]
and, because $\|Q_2\|_\op\le1$,
\begin{equation}
\label{eq:th-proof-resolvent-bound}
 \sup_{x\in\mathcal S}\|v_{Q_1}(x)-v_{Q_2}(x)\|_2
 \le\left(\frac{2C_\kappa}{\xi}+\frac{2C_\kappa^2}{\xi^2}\right)
                          \|Q_1-Q_2\|_\op
 =\Gamma_\xi\|Q_1-Q_2\|_\op.
\end{equation}
These are deterministic identities and inequalities for the realized
data; they require no independence between the representations, the
fitting design, and the scalar observations.

\paragraph{4. Uniform control of fitted normalized contrasts.}
For $t\in\mathfrak T$, write
\[
 \overline\psi_{t,L}=n^{-1}\sum_{i\in L}\psi_t(Y_i),\qquad
 \psi_{t,c}=n^{-1/2}(\psi_t(Y_i)-\overline\psi_{t,L})_{i\in L},\qquad
 U_L=\sup_{t\in\mathfrak T}\|\psi_{t,c}\|_2.
\]
Linearity of \eqref{eq:th-proof-scalar-fit} and constant reproduction
show that the normalized criterion fitted with $Q$ is
\[
 \widehat C_Q(t;x)
   =d_t+\overline\psi_{t,L}+\langle\psi_{t,c},v_Q(x)\rangle.
\]
Thus, by \eqref{eq:th-proof-kernel-difference} and
\eqref{eq:th-proof-resolvent-bound},
\begin{equation}
\label{eq:th-proof-contrast-discrepancy}
 \sup_{x\in\mathcal S}\sup_{t\in\mathfrak T}
 |\widehat C_{Q_\alpha}(t;x)-\widehat C_{\widehat\Pi_{[N]}}(t;x)|
 \le(1-\alpha)B_{n,N},\qquad B_{n,N}:=\Gamma_\xi U_L\delta_\Pi.
\end{equation}
All indexed suprema use the paper's measurable, separable versions.

To control $U_L$, put
$\sigma_*^2=\sup_{x\in\cX_0}\bbE[d^2(Y,y_*)\mid X=x]$, which is finite
under \aref{M1}. For $z=m_\oplus(x_z)$, minimization of the conditional
Fr\'echet function and conditional Minkowski give
$d(z,y_*)\le2\sigma_*$. Lemma~\ref{lem:directional-algebra} therefore
implies the common envelope
\[
 \sup_{t\in\mathfrak T}|\psi_t(v)|
 \le2d(v,y_*)+4\sigma_*.
\]
Empirical centering decreases each squared empirical norm, so
\[
 U_L^2\le n^{-1}\sum_{i\in L}\sup_{t\in\mathfrak T}|\psi_t(Y_i)|^2
 \le n^{-1}\sum_{i\in L}\{2d(Y_i,y_*)+4\sigma_*\}^2.
\]
The expectation of the last expression is at most $40\sigma_*^2$.
Markov's inequality gives $U_L=O_p(1)$. In particular, no interchange
of an indexed supremum and an expectation is being used. Since
$\Gamma_\xi\lesssim\xi^{-1}+\xi^{-2}$,
\eqref{eq:th-proof-projector-rate} yields
\begin{equation}
\label{eq:th-proof-B-rate}
 B_{n,N}=O_p(\chi_{s,n,N,\xi})\quad(N>n),
 \qquad B_{n,n}=0.
\end{equation}

\paragraph{5. Population curvature, TH interpolation, and risk.}
Let $A_{\mathrm N}(x)=R_{s,\xi}^{\oplus}(x)$ be the pure pooled
normalized-contrast error from
Section~\ref{app:additional-rate-results}. For the RH fit define
\[
 A_\alpha(x)=\sup_{y\ne m_\oplus(x)}
 \left|\widehat C_{Q_\alpha}((m_\oplus(x),y);x)
             -\dot C_\oplus((m_\oplus(x),y);x)\right|.
\]
Then \eqref{eq:th-proof-contrast-discrepancy} gives
$A_\alpha(x)\le A_{\mathrm N}(x)+(1-\alpha)B_{n,N}$.
For any measurable $\varepsilon_{\mathrm{opt}}$-minimizer
$\widehat m^\alpha(x)$ of its own fitted criterion, let
$d_\alpha(x)=d(\widehat m^\alpha(x),m_\oplus(x))$.
Approximate optimality and the definition of the normalized error yield
\[
 \Delta_\oplus(\widehat m^\alpha(x),x)
 \le A_\alpha(x)d_\alpha(x)+\varepsilon_{\mathrm{opt}}.
\]
The population Hadamard variance inequality in
Proposition~\ref{prop:automatic-quadratic-curvature} bounds the left
side below by $d_\alpha(x)^2$. Solving this quadratic inequality gives
\begin{equation}
\label{eq:th-proof-endpoint-distance}
 d_\alpha(x)
 \le A_\alpha(x)+\sqrt{\varepsilon_{\mathrm{opt}}}
 \le A_{\mathrm N}(x)+(1-\alpha)B_{n,N}
                            +\sqrt{\varepsilon_{\mathrm{opt}}}.
\end{equation}
The same calculation applies at $\alpha=0,1$ to the labeled-only and
pooled endpoints, including separately selected approximate minimizers
with the common error bound. No convexity of these signed fitted
criteria has been assumed.

For two points $u,v$ in a Hadamard space, their constant-speed geodesic
$\gamma_{u,v}(\eta)$ minimizes
$F_\eta(y)=(1-\eta)d^2(y,u)+\eta d^2(y,v)$. Indeed the CAT(0)
inequality gives
\[
 F_\eta(y)\ge d^2(y,\gamma_{u,v}(\eta))
                  +\eta(1-\eta)d^2(u,v),\qquad
 F_\eta(\gamma_{u,v}(\eta))=\eta(1-\eta)d^2(u,v).
\]
It follows that every $\varepsilon_{\mathrm{hyb}}$-minimizer is within
$\sqrt{\varepsilon_{\mathrm{hyb}}}$ of this geodesic point. Moreover,
convexity of distance along Hadamard geodesics gives
\[
 d(\gamma_{u,v}(\eta),m)
 \le(1-\eta)d(u,m)+\eta d(v,m).
\]
For completeness, the common-endpoint CAT(0) comparison bounds
$d(\gamma_{u,v}(\eta),\gamma_{u,m}(\eta))$ by $\eta d(v,m)$;
the triangle inequality and
$d(\gamma_{u,m}(\eta),m)=(1-\eta)d(u,m)$ give the displayed bound.

Apply these facts with $u=\widehat m^{\mathrm L}(x)$,
$v=\widehat m^{\alphaH}(x)$, and $m=m_\oplus(x)$.
Combining \eqref{eq:th-proof-endpoint-distance} at $\alpha=0$ and
$\alpha=\alphaH$ gives
\begin{equation}
\label{eq:th-proof-pointwise-bound}
 d(\widehat m^{\mathrm{TH}}_{\alphaH,\eta}(x),m_\oplus(x))
 \le A_{\mathrm N}(x)+(1-\alphaH\eta)B_{n,N}
       +\sqrt{\varepsilon_{\mathrm{opt}}}+\sqrt{\varepsilon_{\mathrm{hyb}}},
\end{equation}
since $(1-\eta)+\eta(1-\alphaH)=1-\alphaH\eta$.
Writing $a_{\mathrm N}=\|A_{\mathrm N}\|_{L^2(\bbP_X)}$, Minkowski's
inequality and the pooled endpoint bound imply
\begin{align}
 \sqrt{\mathfrak R(\widehat m^{\mathrm{TH}}_{\alphaH,\eta})}
 &\le a_{\mathrm N}+(1-\alphaH\eta)B_{n,N}
       +\sqrt{\varepsilon_{\mathrm{opt}}}+\sqrt{\varepsilon_{\mathrm{hyb}}},
       \label{eq:th-proof-root-risk}\\
 \sqrt{\mathfrak D(\widehat m^{\mathrm{TH}}_{\alphaH,\eta},\widehat m^{\mathrm N})}
 &\le 2a_{\mathrm N}+(1-\alphaH\eta)B_{n,N}
       +2\sqrt{\varepsilon_{\mathrm{opt}}}+\sqrt{\varepsilon_{\mathrm{hyb}}}.
       \label{eq:th-proof-root-discrepancy}
\end{align}
The second inequality uses the triangle inequality through $m_\oplus$.
Theorem~\ref{thm:contrast-error}, under the hypotheses of the selected
basic or refined main theorem, gives $a_{\mathrm N}=O_p(\rho_{\mathrm N})$.
All preceding pathwise bounds hold simultaneously for every
$(\alphaH,\eta)\in[0,1]^2$, using the joint selection and uniform
tolerances of Section~\ref{app:hybrid-extension-details}.
In particular, $a_{\mathrm N}$ and $B_{n,N}$ do not depend on the weights.
Squaring \eqref{eq:th-proof-root-risk}--\eqref{eq:th-proof-root-discrepancy}
gives, for every pair,
\[
 \mathfrak R(\widehat m^{\mathrm{TH}}_{\alphaH,\eta})
 +\mathfrak D(\widehat m^{\mathrm{TH}}_{\alphaH,\eta},\widehat m^{\mathrm N})
 \le20\{a_{\mathrm N}^2+(1-\alphaH\eta)^2B_{n,N}^2
                   +\varepsilon_{\mathrm{opt}}+\varepsilon_{\mathrm{hyb}}\}.
\]
Here $\rho_{\mathrm N}>0$ since it contains $\sqrt{s/n}$.
Put $V_{n,N}=B_{n,N}/\chi_{s,n,N,\xi}$ for $N>n$ and $V_{n,n}=0$;
the latter uses $B_{n,n}=0$. By \eqref{eq:th-proof-B-rate},
\[
 Z_{n,N}:=20\max\{1,(a_{\mathrm N}/\rho_{\mathrm N})^2,V_{n,N}^2\}
 =O_p(1).
\]
Under our measurable-version conventions, this is one measurable factor
bounding the ratio in \eqref{eq:th-stability-bound} for all weight pairs.
Thus no measurability of an uncountable supremum over selected fits is
needed. Substitution of any measurable, data-selected pair gives the
same bound under the joint law of training and validation data.
If $(1-\alphaH\eta)\chi_{s,n,N,\xi}=O_p(\rho_{\mathrm N})$, the stated
risk-order consequence follows.
\end{proof}

\paragraph{Boundary cases and the meaning of stability.}
With exact two-point interpolation, returning the shared endpoint fit
at $\eta=1$ gives RH without an additional optimization error. At
$\alphaH=1$, reusing the pooled endpoint gives PH. When $N=n$ both
spectral maps are identical, so $Q_\alpha=\widehat\Pi_{[N]}$ for every $\alpha$;
the fitted criteria therefore agree for all weights. A common exact
response-minimizer rule and exact interpolation then yield identical
predictors. Arbitrary approximate selectors need not coincide.
In general Hadamard spaces, \eqref{eq:th-proof-root-discrepancy} compares
predictors through the population target; it is not a Lipschitz theorem
for a selected signed Fr\'echet minimizer. Nor does it establish an oracle
inequality for validation selection or a risk ordering between endpoints.

\paragraph{Balanced cutoffs and informative stability bounds.}
The pooled good event above follows from the selected prediction
theorem's assumptions, without requiring $\sqrt n\,g_s\to\infty$.
Thus the balanced cutoffs in Corollaries~\ref{cor:polynomial-rate}
and~\ref{cor:hilbert-fractional-rate} are admissible for the projector
comparison. Under their stated ridge choices, however, the discrepancy
scale diverges along sample sizes with $N>n$. Indeed, since $N,n$ are
integers with $N\ge n+1$,
\[
 \sqrt{f_{n,N}}\ge\frac1n,\qquad
 g_s\le\lambda_s\lesssim s^{-q},\qquad
 \chi_{s,n,N,\xi}\gtrsim\frac{\xi^{-2}s^q}{n}.
\]
For Corollary~\ref{cor:polynomial-rate},
$\xi=O(s^{-(3q-1)/2})$ and $s\asymp n^{1/(qr)}$ give
\[
 \chi_{s,n,N,\xi}\gtrsim\frac{s^{4q-1}}{n}
 \asymp n^{(4q-1)/(qr)-1}\longrightarrow\infty.
\]
For Corollary~\ref{cor:hilbert-fractional-rate},
$\xi=O(s^{-3q/2})$ and $s\asymp n^{1/(qr+1)}$ instead give
\[
 \chi_{s,n,N,\xi}\gtrsim\frac{s^{4q}}{n}
 \asymp n^{4q/(qr+1)-1}\longrightarrow\infty.
\]
Both exponents are positive for $q>1$ and $0<r\le1$.
At $N=n$, $\chi=0$ by definition. Hence the stability theorem is
applicable at the balanced cutoffs, but its discrepancy term does not
certify a vanishing risk bound for fixed $\alphaH\eta<1$. This is a
limitation of this bound, not a claim that the hybrid risk diverges.
To retain the pooled risk order via this bound, it suffices that
$(1-\alphaH\eta)\chi_{s,n,N,\xi}=O_p(\rho_{\mathrm N})$, up to
optimization error. For a fixed finite mixing grid, this condition
forces selection of $\alphaH\eta=1$ with probability tending to one
when $\rho_{\mathrm N}\to0$ and $N>n$. PH has a separate endpoint-based
rate below; the unfavorable $\chi$ bound does not preclude it.
The stability term is minimized at $\alphaH\eta=1$, so it does not
establish improvement over the pooled endpoint.

\paragraph{A compatible tuning example.}
Under \aref{K5} and $qr>1$, choose, along admissible complete-block cutoffs,
\[
 s\asymp n^c,\qquad \xi\asymp s^{-(3q-1)/2},\qquad
 0<c<\frac1{8q+qr-1}.
\]
Write $\nu_s=s^{-(qr-1)/2}$. Since $g_s\gtrsim s^{-q-1}$ and
$f_{n,N}\le n^{-1}$, the two-sided ridge order gives
\[
 \chi_{s,n,N,\xi}\lesssim s^{4q}/\sqrt n=o(\nu_s).
\]
The basic tuning conditions follow from
$\xi/\lambda_s\asymp s^{-(q-1)/2}\to0$ and
$\log(2s)/\{n(\lambda_s+\xi)\}\lesssim s^q\log(2s)/n\to0$.
In the basic scale, $b_{r,s}$ and
$G_{r,s}\xi/(\lambda_s+\xi)$ are $O(\nu_s)$, whereas
$\sqrt{s/n}/\nu_s=s^{qr/2}/\sqrt n\to0$ and
$G_{r,s}/(\sqrt N g_s\nu_s)\lesssim s^{(3q+1)/2}/\sqrt n\to0$.
Thus, under the remaining basic-theorem hypotheses,
$\rho_{\mathrm N}=O(\nu_s)$ and Theorem~\ref{thm:th-stability} bounds
both risk and discrepancy by
$O_p(n^{-c(qr-1)}+\varepsilon_{\mathrm{opt}}+\varepsilon_{\mathrm{hyb}})$.
This verifies that the conditions can hold simultaneously; it is not
an optimal-rate claim or the balanced regime of
Corollary~\ref{cor:polynomial-rate}.

\begin{proposition}[Exact geodesic prediction mixing]
\label{prop:hybrid-geodesic-gain}
Let $\cY$ be Hadamard. Conditional on the data used to fit and select
measurable predictors $u,v$, let $(X^\circ,Y^\circ)\sim\bbP$ be an
independent test pair and define
\[
 \mathcal L(g):=\bbE[d_\cY^2(Y^\circ,g(X^\circ))\mid\text{fitted data}].
\]
Suppose $\mathcal L(u),\mathcal L(v)<\infty$. For $\eta\in[0,1]$, let
$g_\eta(x)$ be the exact geodesic point at fraction $\eta$ from $u(x)$
to $v(x)$. Then
\[
 \mathcal L(g_\eta)
 \le(1-\eta)\mathcal L(u)+\eta\mathcal L(v)
       -\eta(1-\eta)\mathfrak D(u,v).
\]
The same statement holds with $\mathfrak R$ in place of $\mathcal L$
when the endpoint target risks are finite. It applies to exact TH with
$(u,v)=(\widehat m^{\mathrm L},\widehat m^{\alphaH})$ and exact PH with
$(u,v)=(\widehat m^{\mathrm L},\widehat m^{\mathrm N})$, without
common-tuning or eigengap assumptions.
\end{proposition}

\begin{proof}
Integrate the CAT(0) squared-distance inequality along the geodesic,
using $Y^\circ$ as the comparison point for $\mathcal L$, or
$m_\oplus(X^\circ)$ for $\mathfrak R$.
\end{proof}

The inequality holds for every mixing weight, including a weight selected
from the fitted data, but supplies no oracle guarantee for that selection.
For example, if $\mathfrak D(u,v)>|\mathcal L(u)-\mathcal L(v)|$, some
interior weight makes the upper bound strictly smaller than both endpoint
losses. With $\varepsilon_{\mathrm{hyb}}>0$, the proposition concerns the
exact geodesic point, not the approximate output: Step~5 above bounds the
distance between them by $\sqrt{\varepsilon_{\mathrm{hyb}}}$ and handles
its contribution separately in the stability bound.

\begin{proof}[Proof of Corollary~\ref{cor:ph-balanced-rate}]
The labeled subset consists of $n$ iid labeled pairs under the sampling
scheme. Apply Corollary~\ref{cor:polynomial-rate} at the prescribed
common kernel, rank and ridge twice: with representation size $n$ for
$\widehat m^{\mathrm L}$ and size $N\ge n$ for $\widehat m^{\mathrm N}$.
Since $qr>2$, the refined tuning conditions hold for both sizes at
$s\asymp n^{1/(qr)}$. Hadamard curvature gives $\beta_M=2$, so both
endpoint risks are $O_p(a_n)$, where
$a_n=n^{-(qr-1)/(qr)}+\varepsilon_{\mathrm{opt}}$.
Proposition~\ref{prop:hybrid-geodesic-gain} gives, simultaneously for
every $\eta\in[0,1]$,
\[
 \mathfrak R(\widehat m^{\mathrm{PH}}_\eta)
 \le(1-\eta)\mathfrak R(\widehat m^{\mathrm L})
      +\eta\mathfrak R(\widehat m^{\mathrm N})
 \le\max\{\mathfrak R(\widehat m^{\mathrm L}),
          \mathfrak R(\widehat m^{\mathrm N})\}=O_p(a_n).
\]
The maximum of these two risks is one measurable bound for every
mixing weight. Substituting a measurable validation-selected
$\widehat\eta$ therefore preserves the same order; no independence
between the endpoint fits or between a fit and $\widehat\eta$ is needed.
Prediction risk is evaluated at a fresh query independent of all
fitting and selection data. This proves the corollary without using
$\chi$ or a labeled eigengap condition.
\end{proof}

\paragraph{Hilbert sources and approximate interpolation.}
Under Corollary~\ref{cor:hilbert-fractional-rate}, the same argument,
with its fixed common tuning for both endpoints, gives the uniform PH
rate $O_p(n^{-qr/(qr+1)}+\varepsilon_{\mathrm{opt}})$.
More generally, if $\widetilde m^{\mathrm{PH}}_\eta$ is an
$\varepsilon_{\mathrm{hyb}}$-minimizer of the two-point objective,
Step~5 gives
\[
 \mathfrak R(\widetilde m^{\mathrm{PH}}_\eta)
 \le2\mathfrak R(\widehat m^{\mathrm{PH}}_\eta)
      +2\varepsilon_{\mathrm{hyb}}
 \le2\max\{\mathfrak R(\widehat m^{\mathrm L}),
           \mathfrak R(\widehat m^{\mathrm N})\}
      +2\varepsilon_{\mathrm{hyb}},
\]
uniformly over the weights under a common interpolation tolerance.
Thus approximate interpolation adds $O_p(\varepsilon_{\mathrm{hyb}})$
to either PH rate. These conclusions bound PH by the worse endpoint,
not the better one, and give no validation oracle inequality. The
geometric inequality itself allows arbitrary fitted endpoints, but
the polynomial rates require the two endpoint risk bounds. In
particular, they do not establish those bounds for the independently
validation-tuned kernels, ranks and ridges used in the experiments,
and they do not supply a balanced-rate result for a jointly refitted RH.

\end{document}